\documentclass{llncs}
\usepackage{paralist}
\usepackage{citesort,url}
\usepackage{graphicx}
\usepackage{amssymb,amsmath}
\newif\ifpxfont\pxfonttrue
\ifpxfont
\usepackage{pxfonts}
\else
\usepackage{stmaryrd}
\fi
\usepackage{qtree}
\usepackage{color}
\usepackage{bcprules,proof,myproof}
\usepackage[all]{xy}
\newif\ifdraft\draftfalse
\newif\iffull\fulltrue
\newif\ifsubmission\submissionfalse
\iffull
\fi
\ifdraft
\definecolor{DarkGreen}{RGB}{0,100,0}
\newcommand\nk[1]{\textcolor{red}{[#1 -nk]}}
\newcommand\changed[1]{\textcolor{red}{#1}}
\newcommand\wtnb[1]{\textcolor{blue}{[#1 -wtnb]}}
\newcommand\tk[1]{\textcolor{DarkGreen}{[#1 -tk]}}
\newcommand\tkchanged[1]{\textcolor{DarkGreen}{#1}}
\else
\newcommand\nk[1]{}
\newcommand\changed[1]{#1}
\newcommand\wtnb[1]{}
\newcommand\tk[1]{}
\newcommand\tkchanged[1]{#1}
\fi
\newcommand\IT{\bigwedge}
\newenvironment{myproof}{\begin{proof}}{\qed\end{proof}}
\newcommand\judg[2]{#1\pHFL #2}
\newcommand\HFLp[1]{\textrm{HFL}\(_{#1}\)}
\newcommand\HFLZ{\HFLp{\INTSET}}
\newcommand\HFLN{\HFLp{\mathbf{Nat}}}
\newcommand\HFLP{\HFLp{\emptyset}}
\newcommand\citeN[1]{\cite{#1}}

\newcommand\pGlb{\textstyle\bigwedge}
\newcommand\pLub{\textstyle\bigvee}
\newcommand\LUB{\bigsqcup}
\newcommand\Lub{\textstyle\bigsqcup}
\newcommand\GLB{\bigsqcap}
\newcommand\Glb{\textstyle\bigsqcap}
\newcommand\Ra{\Rightarrow}
\newcommand\obsolete[1]{}
\newcommand\rulesp{\ \\[-1.5ex]}
\newcommand\Cont{\mathit{Cont}}
\newcommand\Cocont{\mathit{Cocont}}
\newcommand\eqcont[1]{=_{\mathit{cont},#1}}
\newcommand\eqcocont[1]{=_{\mathit{cocont},#1}}
\newcommand\Eqcont{=_{\mathit{cont}}}
\newcommand\Eqcocont{=_{\mathit{cocont}}}
\newcommand\Incseq[1]{{#1}^{(\uparrow\omega)}}
\newcommand\Decseq[1]{{#1}^{(\downarrow\omega)}}
\newcommand\COL{\mathbin{:}}

\newcommand\DEF{=}
\newcommand\dom{\mathit{dom}}
\newcommand\codom{\mathit{codom}}

\newcommand\IFF{\Leftrightarrow}
\newcommand\OP{\mathbin{\mathtt{op}}}
\newcommand\INT{\typInt}
\newcommand\Tint{\typInt}
\newcommand\imply{\Rightarrow}
\newcommand\lab{a}
\newcommand\labseq{\seq{\ell}}
\newcommand\Lab[1]{\mathtt{#1}}
\newcommand\p{\vdash}

\newcommand\TRUE{\mathbf{true}}
\newcommand\FALSE{\mathbf{false}}
\newcommand\formT{\TRUE}
\newcommand\formF{\FALSE}
\newcommand\arity{\mathtt{arity}}
\newcommand\order{\mathtt{order}}
\newcommand\FAIL{\mathtt{fail}}
\newcommand\lts{\mathtt{L}}
\newcommand\seq[1]{\widetilde{#1}}
\newcommand\set[1]{\{#1\}}
\newcommand\Vunit{(\,)}

\newcommand\READ{\mathtt{read}}

\newcommand\CLOSE{\mathtt{close}}

\newcommand\D{\mathcal{D}} 
\newcommand\red{\longrightarrow}
\newcommand\PredSet{\mathbf{Pred}}
\newcommand\HES{\mathcal{E}}
\newcommand\HESf{\Phi}
\newcommand\toHFL{\mathit{toHFL}}
\newcommand\form{\varphi}
\newcommand\Some[1]{\langle{#1}\rangle}
\newcommand\typProp{\bullet}
\newcommand\typInt{\mathtt{int}}
\newcommand\All[1]{[#1]}
\newcommand\munu{\alpha}
\newcommand\HFLte{\Delta}
\newcommand\HFLae{\form}
\newcommand\pHFL{\p_{\mathtt{H}}}
\newcommand\typ{\tau}
\newcommand\etyp{\sigma}
\newcommand\sqleq{\sqsubseteq}
\newcommand\INTSET{\mathbf{Z}}
\newcommand\sem[1]{\llbracket #1 \rrbracket}
\newcommand\Sem[2]{\sem{#2}_{#1}}
\newcommand\csem[1]{\sem{#1}_{\mathit{cont}}}
\newcommand\cocsem[1]{\sem{#1}_{\mathit{cocont}}}
\newcommand\HFLenv{\rho}
\newcommand\LFP[1]{\mathbf{lfp}_{#1}}
\newcommand\GFP[1]{\mathbf{gfp}_{#1}}
\newcommand\St{U}
\newcommand\init{\mathtt{init}}

\newcommand\stunique{\st_\star}
\newcommand\st{\mathtt{s}}
\newcommand\Act{A}
\newcommand\TR{\mathbin{\longrightarrow}}
\newcommand\Ar[1]{\stackrel{#1}{\TR}}

\newcommand\Pred[2]{\stackrel{#1}\longrightarrow_{#2}}
\newcommand\Preds[3][*]{\mathbin{\stackrel{#2}{\longrightarrow}\!\!\!{\,}^{#1}_{#3}}}

\newcommand\Traces{\mathbf{Traces}}
\newcommand\FinTraces{\mathbf{FinTraces}}
\newcommand\FullTraces{\mathbf{FullTraces}}
\newcommand\InfTraces{\mathbf{InfTraces}}
\newcommand\term{t}
\newcommand\termalt{s} 
\newcommand\termaltu{u} 
\newcommand\Pae{\term} 
\newcommand\ifexp[2]{\mathbf{if}\ #1\ \mathbf{then}\ #2\ \mathbf{else}\ }
\newcommand\unitexp{(\,)}
\newcommand\evexp[2]{\mathbf{event}\ #1; #2}
\newcommand\evatom[1]{\mathbf{event}\ #1}
\newcommand\evname[1]{\mathtt{#1}}
\newcommand\Tunit{\star}
\newcommand\Pst{\kappa}
\newcommand\Pest{\eta}
\newcommand\pST{\p}
\newcommand\STE{\mathcal{K}}
\newcommand\nondet{\Box}
\newcommand\EvSet{\mathbf{Ev}}
\newcommand\mainfun{\mathbf{main}}
\newcommand\prog{P}
\newcommand\Funcs{\mathbf{funs}}
\newcommand\progd{D}
\newcommand\Assert{\mathbf{assert}}
\newcommand\Sum{\mathit{sum}}
\newcommand\redv[1]{\longrightarrow_{#1}}
\newcommand\nonredv[1]{\mathbin{\ \ \not\!\!\longrightarrow_{#1}}}
\newcommand\redvd[1]{\longrightarrow_{#1,\mathtt{dist}}}
\newcommand\negredvd[1]{\mathbin{\ \ \not\!\!\!\longrightarrow_{#1,\mathtt{dist}}}}
\newcommand\redvds[1]{\longrightarrow^*_{#1,\mathtt{dist}}}
\newcommand\redvs[1]{\longrightarrow^*_{#1}}
\newcommand\EC{E} 
\newcommand\Hole{[\,]}
\newcommand\PWA{\mathcal{A}} 
\newcommand\INF{\mathbf{Inf}}
\newcommand\Callseq{\mathbf{Callseq}}
\newcommand\InfCallseq{\mathbf{InfCallseq}}
\newcommand\CSA{\mathit{csa}}

\newcommand\MAX{\mathbf{max}}

\newcommand\livetrans[1]{{#1}^{\dagger_{\CSA}}}
\newcommand\ITE{\Gamma} 
\newcommand\RaiseP{\uparrow} 
\newcommand\ESRel{\Rightarrow} 
\newcommand\Atype{\theta}
\newcommand\Itype{\rho}
\newcommand\App[1]{{#1}\;}
\newcommand\dup[1]{\mathtt{dup}(#1)}
\newcommand\pInter{\p_{\PWA}}
\newcommand\qinit{q_I}
\newcommand\qinitA[1]{q_{I,#1}}

\newcommand\Pfun{\Omega}

\newcommand\may{\textit{may}}
\newcommand\must{\textit{must}}
\newcommand\Path{\textit{path}}
\newcommand\Must{\textbf{Must}}
\newcommand\trMay[1]{{#1}^{\dagger_{\may}}}
\newcommand\trMust[1]{{#1}^{\dagger_{\must}}}
\newcommand\trPath[1]{{#1}^{\dagger_{\Path}}}
\newcommand\A{A} 

\newcommand{\SGame}{\mathcal{G}}
\newcommand{\Nodes}{V}
\newcommand{\Edges}{E}
\newcommand{\Strategy}{\mathcal{S}}
\newcommand{\CounterGame}{\overline{\SGame}}

\newcommand{\CallSeq}{\rightsquigarrow}
\newcommand{\CallSeqN}[1]{\stackrel{#1}{\CallSeq}}

\newcommand{\RTyPath}{\pi}

\newcommand\Types[1]{\mathbf{Types}_{#1}}

\newcommand{\CounterStrategy}{\overline{\Strategy}}

\newcommand{\PLeft}{\mathtt{L}}
\newcommand{\PRight}{\mathtt{R}}

\newcommand{\TopEnv}{\Xi}

 \newcommand\APP{\mathtt{app}}
\newcommand{\Lang}{\mathcal{L}}
\newcommand{\runseq}{R}

\newcommand\Mark{\sharp}
\newcommand\recall[1]{{\CallSeq}_{#1}}
\newcommand\Nat{\mathbb{N}}
\newcommand{\EnvRel}{\lhd}

\newcommand{\ZEnv}[1]{[#1]}
\newcommand{\TopEnvZ}{\ZEnv{\TopEnv}}
\newcommand{\TopEnvZM}{\ZEnv{\TopEnv^\Mark}}

\newcommand{\enondet}{\circ} 
\newcommand\eESRel{\Rrightarrow} 

\newcommand\InfruleS[3]{\begin{minipage}{#1\textwidth}\infrule{#2}{#3}\end{minipage}}

\begin{document}

\title{Higher-Order Program Verification\\ via HFL Model Checking}         

\author{Naoki Kobayashi\and
Takeshi Tsukada\and
Keiichi Watanabe}
\institute{The University of Tokyo}

\maketitle

\begin{abstract}
There are two kinds of higher-order extensions of model checking: HORS model checking
and HFL model checking. Whilst the former has been applied to automated verification of
higher-order functional programs, applications of the latter have not been well studied.
In the present paper, we show that various verification
problems for functional programs,
including may/must-reachability, trace properties, and linear-time temporal properties (and their negations),
can be naturally reduced to (extended) HFL model checking.
The reductions yield a sound and complete logical characterization of those program properties.
Compared with the previous approaches based on HORS model checking, our approach
provides a more uniform, streamlined method for higher-order program 
\iffull
verification.\footnote{A shorter version of this article is published in Proceedings of ESOP 2018.}
\else
verification.
\fi
\end{abstract}

\section{Introduction}
\label{sec:intro}

There are two kinds of higher-order extensions of model checking in the literature:
HORS model checking~\cite{Knapik02FOSSACS,Ong06LICS} and HFL model checking~\cite{Viswanathan04}.
The former is concerned about whether the tree generated by a given higher-order tree grammar
called a higher-order recursion scheme (HORS) satisfies the property expressed by
a given modal \(\mu\)-calculus formula (or a tree automaton), and the latter is concerned about
whether a given finite state system satisfies the property expressed by a given formula of
higher-order modal fixpoint logic (HFL), a higher-order extension of the modal \(\mu\)-calculus.
Whilst HORS model checking has been applied to automated verification of higher-order functional
programs~\cite{Kobayashi09POPL,
Kobayashi13JACM,KSU11PLDI,Ong11POPL,UnnoTK13,
Kuwahara2015Nonterm,Watanabe16ICFP}, there have been few studies on applications of HFL model checking to
program/system verification. Despite that HFL has been introduced more than 10 years ago, 
we are only aware of applications to 
 assume-guarantee reasoning~\cite{Viswanathan04} and 
 process equivalence checking~\cite{LangeLG14}.

In the present paper, we show that various verification problems for
higher-order functional programs can actually be reduced to (extended) HFL model checking
in a rather natural manner.
We briefly explain the idea of our reduction below.\footnote{In this section, we use
only a fragment of HFL that can be expressed in the modal \(\mu\)-calculus.
Some familiarity with the modal \(\mu\)-calculus~\cite{Kozen83} would help.} 
We translate a program to an HFL formula that says ``the program
has a valid behavior'' (where the \emph{validity} of a behavior depends on each verification problem).
Thus, a program is actually mapped to a \emph{property}, and a program property is mapped to a system to be
verified; this has been partially inspired by the recent work of Kobayashi et al.~\cite{Kobayashi17POPL},
where HORS model checking problems have been translated to HFL model checking problems by switching 
the roles of models and properties.

For example, consider a simple program fragment \(\READ(x); \CLOSE(x)\) that reads and then
closes a file (pointer)
\(x\). The transition system in Figure~\ref{fig:file} shows a valid access protocol to 
read-only files. Then, the property that a read operation is allowed in the current state
can be expressed by a formula of the form \(\Some{\READ}\varphi\), which
says that the current state has a \(\READ\)-transition, after which \(\varphi\) is satisfied. 
Thus, the program \(\READ(x); \CLOSE(x)\) being valid is expressed as
\(\Some{\READ}\Some{\CLOSE}\TRUE\),\footnote{Here, for the
sake of simplicity, we assume that we are interested in the usage of the single file pointer \(x\), 
so that the name \(x\) can be ignored in HFL formulas; usage of multiple files can be tracked by using
the technique of \citeN{Kobayashi09POPL}.} which is indeed satisfied by the initial state
\(q_0\) of the transition system in Figure~\ref{fig:file}. Here, we have just replaced
the operations \(\READ\) and \(\CLOSE\) of the program with the corresponding modal operators
\(\Some{\READ}\) and \(\Some{\CLOSE}\). We can also naturally deal with branches and recursions.
For example, consider the program \(\CLOSE(x)\nondet (\READ(x); \CLOSE(x))\), 
where \(e_1\nondet e_2\) represents a non-deterministic choice between \(e_1\) and \(e_2\).
Then the property that the program always accesses \(x\) in a valid manner
can be expressed by \((\Some{\CLOSE}\TRUE) \land (\Some{\READ}\Some{\CLOSE}\TRUE)\).
 Note that
we have just replaced the non-deterministic branch with the logical conjunction, as we wish here
to require that the program's behavior is valid in \emph{both} branches. We can also deal with
conditional branches if HFL is extended with predicates; \(\ifexp{b}{\CLOSE(x)}{(\READ(x);\CLOSE(x))}\)
can be translated to \((b\imply \Some{\CLOSE}\TRUE) \land (\neg b\imply \Some{\READ}\Some{\CLOSE}\TRUE)\).
Let us also consider
the recursive function \(f\) defined by:
\[{f\,x}={\CLOSE(x)\nondet (\READ(x); \READ(x); f x)},\]
Then, the program \(f\,x\) being valid can be represented by using a (greatest) fixpoint formula:
\[\nu F.(\Some{\CLOSE}\TRUE)\land (\Some{\READ}\Some{\READ}F).\]
If the state \(q_0\) satisfies this formula (which is indeed the case), then
we know that all the file accesses made by \(f\,x\) are valid.
So far, we have used only the modal \(\mu\)-calculus formulas. If we wish to express the validity
of higher-order programs, we need HFL formulas; such examples are given later.

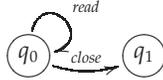
\begin{figure}[tb]
\begin{center}
\[
\xymatrix@!R=2pc{%
 *+<1pc>[o][F-]{q_0}  \ar@(u,r)^{read} \ar@/_/[r]^{close}
& *+<1pc>[o][F-]{q_1}  }\]
\end{center}
\caption{File access protocol}
\label{fig:file}
\end{figure}

We generalize the above idea and formalize reductions from various classes of
verification problems for 
simply-typed higher-order functional programs with recursion, integers and non-determinism
-- including verification of may/must-reachability, trace properties, and linear-time temporal properties 
(and their negations) -- 
to (extended) HFL model checking where HFL is extended with integer predicates,
and prove soundness and completeness of the reductions. 
Extended HFL model checking problems obtained by the reductions are (necessarily) undecidable in general,
but for finite-data programs (i.e., programs that consist of only functions and data from finite data domains
such as Booleans), the reductions yield \emph{pure} HFL model checking problems, which are decidable~\cite{Viswanathan04}. 

Our reductions provide sound and complete logical characterizations of 
a wide range of program properties mentioned above.
Nice properties of the logical characterizations include: 
(i)~(like verification conditions for Hoare triples,) once the logical characterization 
is obtained as an HFL formula, purely logical reasoning can be used to prove or disprove it
(without further referring to the program semantics);
for that purpose, one may use theorem provers with various degrees of automation,
ranging from interactive ones like Coq, semi-automated ones requiring some annotations,
to fully automated ones (though the
latter two are yet to be implemented),
(ii)~(unlike the standard verification condition generation for Hoare triples using invariant 
annotations) the logical characterization
can \emph{automatically} be computed, without any annotations,\footnote{This does not mean that
invariant discovery is unnecessary; invariant discovery is just postponed to the later phase of discharging verification conditions,
so that it can be uniformly performed among various verification problems.}
(iii)~standard logical reasoning can be applied based on the semantics of formulas; for example,
co-induction and induction can be used for proving \(\nu\)- and \(\mu\)-formulas respectively,
and (iv) thanks to the completeness, the set of program properties characterizable by HFL formula
is closed under negations; for example, from a formula characterizing may-reachability,
one can obtain a formula characterizing non-reachability by just taking the De Morgan dual.


Compared with previous approaches based on HORS model checking~\cite{Kobayashi13JACM,KSU11PLDI,Ong11POPL,Tobita12FLOPS,Kuwahara2015Nonterm}, our approach based on (extended) HFL model checking provides
more uniform, streamlined methods for higher-order program verification. 
HORS model checking provides sound and complete verification methods for
\emph{finite-data} programs~\cite{Kobayashi09POPL,Kobayashi13JACM}, but for infinite-data programs,
other techniques such as 
predicate abstraction~\cite{KSU11PLDI} and program transformation~\cite{Kuwahara2014Termination,MTSUK16POPL}
had to be combined to obtain sound (but incomplete) reductions to HORS model checking.
Furthermore, the techniques were different for each of program properties, 
such as reachability~\cite{KSU11PLDI}, termination~\cite{Kuwahara2014Termination}, 
non-termination~\cite{Kuwahara2015Nonterm}, fair termination~\cite{MTSUK16POPL}, and 
fair non-termination~\cite{Watanabe16ICFP}. In contrast, our reductions are sound and complete even 
for infinite-data programs. Although the obtained HFL model checking problems are undecidable in general,
the reductions allow us to treat various program properties uniformly; all the verifications
are boiled down to the issue of how to prove \(\mu\)- and \(\nu\)-formulas (and as remarked above,
we can use induction and co-induction to deal with them). Technically, our reduction to HFL model checking
may actually be considered an extension of HORS model checking 
in the following sense.
HORS model checking algorithms~\cite{Ong06LICS,KO09LICS} usually consist of two phases,
one for computing a kind of higher-order ``procedure summaries'' in the form of variable profiles~\cite{Ong06LICS}
or intersection types~\cite{KO09LICS},
and the other for 
nested least/greatest fixpoint computations.
Our reduction from program verification to extended HFL model checking (the reduction given in Section~\ref{sec:liveness},
in particular) can be regarded as an extension of the first phase to deal with infinite data domains, where
the problem for the second phase is expressed in the form of extended HFL model checking:
\iffull
see Appendix~\ref{sec:hors-vs-hfl} for more details.
\else
see \cite{ESOP2018full} for more details.
\fi

The rest of this paper is structured as follows.
Section~\ref{sec:pre} introduces HFL extended with integer predicates and
defines the HFL model checking problem.
Section~\ref{sec:examples} informally demonstrates some examples of reductions from
program verification problems to HFL model checking.
Section~\ref{sec:lang} introduces a functional language used to formally discuss the reductions
in later sections. Sections~\ref{sec:reachability}, \ref{sec:path},
and \ref{sec:liveness} consider may/must-reachability, trace properties,
and temporal properties respectively, and present (sound and complete) 
reductions from verification of those properties to HFL model checking. 
Section~\ref{sec:related} discusses related work, and Section~\ref{sec:conc} concludes the paper.
\iffull
Proofs are found in Appendices.
\else
Proofs are found in an extended version~\cite{ESOP2018full}.
\fi

\section{(Extended) HFL}
\label{sec:pre}
In this section, we introduce an extension of higher-order modal fixpoint logic (HFL)~\cite{Viswanathan04}
with integer predicates (which we call \HFLZ{}; we often drop the subscript and
\iffull just \fi
write HFL, as in Section~\ref{sec:intro}), and
define the \HFLZ{} model checking problem. The set of integers can actually be replaced by
another infinite set \(X\) of data (like the set of natural numbers  or the set of finite trees)
to yield \HFLp{X}. 

\subsection{Syntax}
For a map \(f\), we write \(\dom(f)\) and \(\codom(f)\) for the domain and codomain of \(f\)
respectively.
We write \(\INTSET\) for the set of integers, ranged over by the meta-variable \(n\) below.
We assume a set \(\PredSet\)
of primitive predicates on integers, ranged over by \(p\). We write \(\arity(p)\) for 
the arity of \(p\). We assume that \(\PredSet\) contains standard integer predicates such as \(=\) and \(<\),
and also assume that, for each predicate \(p\in \PredSet\), 
there also exists a predicate \(\neg p\in \PredSet\) such that,
for any integers \(n_1,\ldots,n_k\),
\(p(n_1,\ldots,n_k)\) holds if and only if \(\neg p(n_1,\ldots,n_k)\) does not hold;
thus, \(\neg p(n_1,\ldots,n_k)\) should be parsed as \((\neg p)(n_1,\ldots,n_k)\), but
can semantically be interpreted as \(\neg (p(n_1,\ldots,n_k))\).

The syntax of \emph{\HFLZ{} formulas} is given by:
\[
\begin{array}{l}
\form \mbox{ (formulas) }::= 
n \mid \form_1\OP \form_2  \mid\TRUE \mid \FALSE \mid 
p(\HFLae_1,\ldots,\HFLae_k) 
\mid  \form_1\vee\form_2\mid\form_1\wedge \form_2 \\\qquad\qquad\qquad
\mid X 
\mid\Some{\lab}\form\mid\All{\lab}\form 
\mid \mu X^{\typ}.\form \mid \nu X^{\typ}.\form 
\mid \lambda X\COL{\etyp}.\form \mid \form_1\; \form_2 \\
\typ \mbox{ (types) }::=  \typProp \mid \etyp\to\typ \qquad \etyp \mbox{ (extended types) }::=  \typ \mid \typInt
\end{array}
\]
Here, \(\OP\) ranges over a set of binary operations on integers, such as \(+\),
and \(X\) ranges over a denumerable set of variables.
We have extended the original HFL~\cite{Viswanathan04} with integer expressions (\(n\) and 
\(\form_1\OP\form_2\)), and atomic formulas \(p(\form_1,\ldots,\form_k)\) on integers
(here, the arguments of integer operations or predicates
 will be restricted to integer expressions by the type system introduced below).
Following \citeN{Kobayashi17POPL}, we have omitted negations, as any formula can be transformed
to an equivalent negation-free formula~\cite{DBLP:journals/corr/Lozes15}.

We explain the meaning of each formula informally; the formal semantics is given in Section~\ref{sec:HFL-semantics}.
Like modal \(\mu\)-calculus~\cite{Kozen83,Automata}, each formula expresses a property of a labeled transition system.
The first line of the syntax of formulas consists of the standard constructs of predicate logics.
On the second line, as in the standard modal \(\mu\)-calculus,
\(\Some{\lab}\form\) means that there exists an \(\lab\)-labeled transition to
a state that satisfies \(\form\). The formula \(\All{\lab}\form\) means that after any \(\lab\)-labeled transition,
\(\form\) is satisfied.
The formulas \(\mu X^{\typ}.\form\) and  \(\nu X^{\typ}.\form\) represent the least and greatest fixpoints
respectively
(the least and greatest \(X\) that \(X=\form\)) respectively; 
unlike the modal \(\mu\)-calculus, \(X\) may range over not only propositional variables but
also higher-order predicate variables (of type \(\typ\)). The \(\lambda\)-abstractions 
\(\lambda X\COL{\etyp}.\form\) and applications \(\form_1\; \form_2\) are used to manipulate higher-order
predicates. We often omit type annotations in 
\(\mu X^{\typ}.\form\), \(\nu X^{\typ}.\form\) and \(\lambda X\COL{\etyp}.\form\),
and just write 
\(\mu X.\form\), \(\nu X.\form\) and \(\lambda X.\form\).

\begin{example}
\label{ex:HFL-anbn}
Consider 
\(\form_{\mathtt{ab}}\,\form\) where
\(\form_{\mathtt{ab}} = \mu X^{\typProp\to\typProp}.\lambda Y\COL\typProp.Y\lor \Some{\Lab{a}}(X(\Some{\Lab{b}}Y))\).
We can expand the formula as follows:
\[
\begin{array}{l}
\form_{\mathtt{ab}}\,\form = (\lambda Y.\typProp.Y\lor \Some{\Lab{a}}(\form_{\mathtt{ab}}(\Some{\Lab{b}}Y)))\form
= \form \lor \Some{\Lab{a}}(\form_{\mathtt{ab}}(\Some{\Lab{b}}\form))\\\qquad
= \form \lor \Some{\Lab{a}}(\Some{\Lab{b}}\form \lor \Some{\Lab{a}}(\form_{\mathtt{ab}}(\Some{\Lab{b}}\Some{\Lab{b}}\form))) = \cdots,
\end{array}
\]
and obtain 
\(
\form \lor (\Some{\Lab{a}}\Some{\Lab{b}}\form) \lor
(\Some{\Lab{a}}\Some{\Lab{a}}\Some{\Lab{b}}\Some{\Lab{b}}\form) \lor  \cdots
\).
Thus, the formula means that there is a transition sequence of the form \(\Lab{a}^n\Lab{b}^n\) for 
some \(n\geq 0\) that leads to a state satisfying \(\form\).
\end{example}

Following \citeN{Kobayashi17POPL},
we exclude out unmeaningful formulas such as \((\Some{\lab}\TRUE)+1\) by using 
a simple type 
\iffull
system.\footnote{The original type system of \citeN{Viswanathan04} was more complex due to the presence of negations.}
\else
system.
\fi
The types \(\typProp\), \(\INT\), and \(\etyp\to\typ\) describe propositions, integers, and 
(monotonic) functions from 
\(\etyp\) to \(\typ\), respectively. Note that the integer type \(\INT\) may occur only in an argument position;
this restriction is required to ensure that least and greatest fixpoints are well-defined.
The typing rules for formulas are given in Figure~\ref{fig:HFLtyping}. In the figure, \(\HFLte\) denotes
a type environment, which is a finite map from variables to (extended) types. Below we consider
\iffull
only well-typed formulas, i.e., formulas \(\form\) such that \(\HFLte\pHFL \form:\typ\) for some \(\HFLte\)
and \(\typ\).
\else
only well-typed formulas.
\fi
\begin{figure}
\begin{multicols}{2}
\typicallabel{HT-Int}
\infrule[HT-Int]{}{\HFLte\pHFL n:\typInt}
\infrule[HT-Op]{\HFLte\pHFL \HFLae_i:\typInt\mbox{ for each $i\in\set{1,2}$}}
{\HFLte\pHFL \HFLae_1\OP \HFLae_2:\typInt}
\infrule[HT-True]{}{\HFLte\pHFL \formT:\typProp}
\infrule[HT-False]{}{\HFLte\pHFL \formF:\typProp}
\rulesp
\infrule[HT-Pred]{\arity(p)=k\\ \HFLte\pHFL \HFLae_i:\typInt\mbox{ for each $i\in\set{1,\ldots,k}$}}
{ \HFLte\pHFL p(\HFLae_1,\ldots,\HFLae_k):\typProp}

\infrule[HT-Var]{}{\HFLte,X\COL\etyp\pHFL X:\etyp}

\infrule[HT-Or]{\judg{\HFLte}{\form_i:\typProp}\mbox{ for each $i\in\set{1,2}$}}
{\judg{\HFLte}{\form_1\vee\form_2:\typProp}}

\infrule[HT-And]{\judg{\HFLte}{\form_i:\typProp}\mbox{ for each $i\in\set{1,2}$}}
{\judg{\HFLte}{\form_1\wedge\form_2:\typProp}}

\infrule[HT-Some]{\judg{\HFLte}{\form:\typProp}}{\judg{\HFLte}{\Some{a}\form:\typProp}}

\infrule[HT-All]{\judg{\HFLte}{\form:\typProp}}{\judg{\HFLte}{\All{a}\form:\typProp}}

\infrule[HT-Mu]{\judg{\HFLte,X:\typ}{\form:\typ}}
{\judg{\HFLte}{\mu X^\typ.\ \form :\typ}}

\infrule[HT-Nu]{\judg{\HFLte,X:\typ}{\form:\typ}}
{\judg{\HFLte}{\nu X^\typ.\ \form :\typ}}

\infrule[HT-Abs]{\judg{\HFLte,X:\etyp}{\form:\typ}}
{\judg{\HFLte}{\lambda X\COL{\etyp}.\ \form :\etyp\to\typ}}

\infrule[HT-App]{\judg{\HFLte}{\form_1:\etyp\to\typ}\quad\judg{\HFLte}{\form_2:\etyp}}
{\judg{\HFLte}{\form_1\ \form_2 :\typ}}
\end{multicols}
\caption{Typing Rules for \HFLZ{} Formulas}
\label{fig:HFLtyping}
\end{figure}

\subsection{Semantics and \HFLZ{} Model Checking}
\label{sec:HFL-semantics}

We now define the formal semantics of \HFLZ{} formulas.
A \emph{labeled transition system} (LTS) is a quadruple 
\(\lts = (\St{}, \Act{}, \mathbin{\TR}, \st_\init)\), where \(\St{}\) is a finite set of states,
\(\Act{}\) is a finite set of actions, \(\mathbin{\TR} \subseteq \St{}\times\Act{}\times\St\) is
a labeled transition relation, and \(\st_\init\in\St\) is the initial state.
We write \(\st_1\stackrel{\lab}{\TR}\st_2\) when \((\st_1,\lab,\st_2)\in \TR\).

For an LTS \(\lts=(\St{}, \Act{}, \mathbin{\TR}, \st_\init)\) and an extended type \(\etyp\),
we define the partially ordered set \((\D_{\lts,\etyp}, \sqleq_{\lts,\etyp})\) inductively by:
\[
\begin{array}{l}
\D_{\lts,\typProp} = {2^\St}\qquad \sqleq_{\lts,\typProp} = \subseteq\qquad
\D_{\lts,\INT} = \INTSET \qquad \sqleq_{\lts,\INT} = \set{(n,n) \mid n\in \INTSET}\\
\D_{\lts,\etyp\to\typ}=\set{f\in \D_{\lts,\etyp}\to \D_{\lts,\typ} \mid
    \forall x,y.(x\sqleq_{\lts,\etyp}y \imply f\,x\sqleq_{\lts,\typ}f\,y)}\\
\sqleq_{\lts,\etyp\to\typ} = \set{(f,g) \mid \forall x\in \D_{\lts,\etyp}.f(x)\sqleq_{\lts,\typ}g(x)}
\end{array}
\]
Note that \((\D_{\lts,\typ}, \sqleq_{\lts,\typ})\) forms a complete lattice
(but \((\D_{\lts,\INT},\sqleq_{\lts,\INT})\) does not). We write \(\bot_{\lts,\typ}\) and
\(\top_{\lts,\typ}\) for the least and greatest elements of \(\D_{\lts,\typ}\) 
(which are \(\lambda\seq{x}.\emptyset\) and \(\lambda\seq{x}.\St\)) respectively.
We sometimes omit the subscript \(\lts\) below.
Let \(\Sem{\lts}{\HFLte}\) be the set of functions (called \emph{valuations}) that maps
\(X\) to an element of \(\D_{\lts,\etyp}\) for each \(X\COL\etyp\in \HFLte\). 
For an HFL formula \(\form\) such that \(\HFLte\pHFL \form:\etyp\),
we define \(\Sem{\lts}{\HFLte\pHFL \form:\etyp}\) as a map from
\(\Sem{\lts}{\HFLte}\) to \(\D_{\etyp}\), by induction on the derivation\footnote{Note that
the derivation of each judgment \(\HFLte\pHFL \form:\etyp\) is unique if there is any.} of \(\HFLte\pHFL \form:\etyp\),
as follows.
\allowdisplaybreaks[1]
\begin{align*}
&\Sem{\lts}{\HFLte\pHFL n:\INT}(\HFLenv) = n\qquad
\Sem{\lts}{\HFLte\pHFL \formT:\typProp}(\HFLenv) =  \St 
\qquad
\Sem{\lts}{\HFLte\pHFL \formF:\typProp}(\HFLenv) =  \emptyset
\\&
\Sem{\lts}{\HFLte\pHFL \HFLae_1\OP \HFLae_2:\INT}(\HFLenv) =
  (\Sem{\lts}{\HFLte\pHFL \HFLae_1:\INT}(\HFLenv)) \sem{\OP}(\Sem{\lts}{\HFLte\pHFL \HFLae_2:\INT}(\HFLenv)) \\&
\Sem{\lts}{\HFLte\pHFL p(\HFLae_1,\ldots,\HFLae_k):\typProp}(\HFLenv) =\\&
  \left\{\begin{array}{ll}
    \St & \mbox{if $(\Sem{\lts}{\HFLte\pHFL \HFLae_1:\INT}(\HFLenv),\ldots,\Sem{\lts}{\HFLte\pHFL \HFLae_k:\INT}(\HFLenv))\in \sem{p}$}\\
    \emptyset & \mbox{otherwise}
  \end{array}\right.\\&
\Sem{\lts}{\HFLte,X:\etyp\pHFL X:\etyp}(\HFLenv)  
=\HFLenv(X) \\&
\Sem{\lts}{\HFLte\pHFL\form_1\vee\form_2:\typProp}(\HFLenv)  
=\Sem{\lts}{\HFLte\pHFL\form_1:\typProp}(\HFLenv)\cup\Sem{\lts}{\HFLte\pHFL\form_2:\typProp}(\HFLenv)  
\\&
\Sem{\lts}{\HFLte\pHFL\form_1\wedge\form_2:\typProp}(\HFLenv)  
=\Sem{\lts}{\HFLte\pHFL\form_1:\typProp}(\HFLenv)\cap\Sem{\lts}{\HFLte\pHFL\form_2:\typProp}(\HFLenv)  
\\&
\Sem{\lts}{\HFLte\pHFL \Some{a}\form:\typProp}(\HFLenv) 
=\{\st\mid \exists \st'\in \Sem{\lts}{\HFLte\pHFL\form:\typProp}(\HFLenv).\ \st\Ar{a}\st'\}
\\&
\Sem{\lts}{\HFLte\pHFL \All{a}\form:\typProp}(\HFLenv) 
=\{\st\mid \forall \st'\in \St.\
(\st\Ar{a}\st' \mbox{ implies } \st'\in\Sem{\lts}{\HFLte\pHFL\form:\typProp}(\HFLenv))\}
\\&
\Sem{\lts}{\HFLte\pHFL \mu X^\typ.\form:\typ}(\HFLenv)  
=\LFP{\lts,\typ} (\Sem{\lts}{\HFLte\pHFL\lambda X\COL{\typ}.\ \form:\typ\to\typ}(\HFLenv))
\\&
\Sem{\lts}{\HFLte\pHFL \nu X^\typ.\form:\typ}(\HFLenv)  
=\GFP{\lts,\typ} (\Sem{\lts}{\HFLte\pHFL\lambda X\COL{\typ}.\ \form:\typ\to\typ}(\HFLenv))
\\&
\Sem{\lts}{\HFLte\pHFL\lambda X\COL{\etyp}.\ \form:\etyp\to\typ}(\HFLenv) 
= \set{(v, \Sem{\lts}{\HFLte,X\COL\etyp\pHFL\form:\typ}(\HFLenv[X\mapsto v])) \mid v\in \D_{\lts,\etyp}}
\\&
\Sem{\lts}{\HFLte\pHFL \form_1\ \form_2:\typ}(\HFLenv) 
=\Sem{\lts}{\HFLte\pHFL \form_1:\etyp\to\typ}(\HFLenv)(\Sem{\lts}{\HFLte\pHFL\form_2:\etyp}(\HFLenv)) \\
\end{align*}
Here, \(\sem{\OP}\) denotes the binary function on integers represented by \(\OP\)
and \(\sem{p}\) denotes the \(k\)-ary relation on integers represented by \(p\).
The least/greatest fixpoint operators \(\LFP{\lts,\typ}\) and \(\GFP{\lts,\typ}\) are defined 
\iffull
by:
\[
\begin{array}{l}
\LFP{\lts,\typ}(f) = \bigsqcap_{\lts,\typ}\set{x\in \D_{\lts,\typ} \mid f(x)\sqleq_{\lts,\typ} x}\qquad
\GFP{\lts,\typ}(f) = \bigsqcup_{\lts,\typ}\set{x\in \D_{\lts,\typ} \mid x\sqleq_{\lts,\typ} f(x)}\\
\end{array}
\]
\else
by 
\(\LFP{\lts,\typ}(f) = \bigsqcap_{\lts,\typ}\set{x\in \D_{\lts,\typ} \mid f(x)\sqleq_{\lts,\typ} x}\)
and \(
\GFP{\lts,\typ}(f) = \bigsqcup_{\lts,\typ}\set{x\in \D_{\lts,\typ} \mid x\sqleq_{\lts,\typ} f(x)}\).
\fi
Here, \(\bigsqcup_{\lts,\typ}\) and \(\bigsqcap_{\lts,\typ}\) respectively denote 
the least upper bound and the greatest lower bound
with respect to \(\sqleq_{\lts,\typ}\).

We often omit the subscript \(\lts\) and write 
\(\sem{\HFLte\pHFL \form:\etyp}\) for
\(\Sem{\lts}{\HFLte\pHFL \form:\etyp}\).
For a closed formula, i.e., a formula well-typed under the empty type environment \(\emptyset\), we often write
\(\Sem{\lts}{\form}\) or just \(\sem{\form}\) for 
\(\Sem{\lts}{\emptyset\pHFL \form:\etyp}(\emptyset)\).

\begin{example}
\label{ex:HFL-semantics}
For the LTS \(\lts_{\mathit{file}}\) in Figure~\ref{fig:file}, we have:
\[
\begin{array}{l}
\sem{\nu X^{\typProp}.(\Some{\CLOSE}\TRUE \land \Some{\READ}X)} = \\\quad
\GFP{\lts,\typProp}(\lambda x\in \D_{\lts,\typProp}.\sem{X\COL\typProp\p \Some{\CLOSE}\TRUE \land \Some{\READ}X:\typProp}(\set{X\mapsto x}))
=
\set{q_0}.
\end{array}\]
In fact, \(x=\set{q_0}\in \D_{\lts,\typProp}\) satisfies the equation:
\(\Sem{\lts}{X\COL\typProp \p \Some{\CLOSE}\TRUE \land \Some{\READ}X:\typProp}(\set{X\mapsto x})
= x\),
and \(x=\set{q_0}\in \D_{\lts,\typProp}\) is the greatest such element.

Consider 
\iffull
the LTS \(\lts_1\) in Figure~\ref{fig:lts1}
\else
the following LTS \(\lts_1\):
\begin{center}
\vspace*{-2ex}
\[
\xymatrix@!R=2pc{%
 *+<1pc>[o][F-]{q_0} \ar@/^/[r]^{a} 
& *+<1pc>[o][F-]{q_1}  \ar@/^/[r]^{b} 
& *+<1pc>[o][F-]{q_2}  \ar@/^/[l]^{c} }
\]
\end{center}
\fi
and \(\form_{\mathtt{ab}}\,(\Some{c}\TRUE)\) where \(\form_{\mathtt{ab}}\) is the one introduced in Example~\ref{ex:HFL-anbn}.
Then, \(\Sem{\lts_1}{\form_{\mathtt{ab}}\,(\Some{c}\TRUE)} = \set{q_0,q_2}\).

\iffull
Consider the formula \(\form_2 = \form_{\mathtt{even}}\,(\Some{c}\TRUE)\,0\), where
\(\form_{\mathtt{even}}\) is:
\[\nu X^{\typProp\to\typProp}.\lambda Y\COL\typProp.\lambda Z\COL\INT.
(\mathtt{even}(Z)\land Y)\lor \Some{\Lab{a}}(X\,(\Some{\Lab{b}}Y)\,(Z+1)).\]
Here, \(\mathtt{even}\in \PredSet\) is a unary predicate on integers such that \(\mathtt{even}(n)\) if and only
if \(n\) is even.
Then, \(\form_2\) denotes  the set of states from which there is a transition sequence of
the form \(\Lab{a}^{2n}\Lab{b}^{2n}\) to a state where a \(\Lab{c}\)-labeled transition is possible.
Thus, \(\Sem{\lts_1}{\form_2}=\set{q_2}\).
\fi
\iffull
\begin{figure}
\begin{center}
\[
\xymatrix@!R=2pc{%
 *+<1pc>[o][F-]{q_0} \ar@/^/[r]^{a} 
& *+<1pc>[o][F-]{q_1}  \ar@/^/[r]^{b} 
& *+<1pc>[o][F-]{q_2}  \ar@/^/[l]^{c} }
\]
\end{center}
\caption{LTS \(\lts_1\)}
\label{fig:lts1}
\end{figure}
\fi
\end{example}

\begin{definition}[\HFLZ{} model checking]
For a closed formula \tkchanged{\(\form\)} of type \(\typProp\),
we write \(\lts, \st\models\form\) if \(\st\in \Sem{\lts}{\form}\),
and write \(\lts\models \form\) if \(\st_\init\in \Sem{\lts}{\form}\).
\emph{\HFLZ{} model checking} is the problem of, given \(\lts\) and \(\form\),
deciding whether \(\lts\models \form\) holds.
\end{definition}

The \HFLZ{}  model checking problem is
\emph{un}decidable, due to the presence of integers; in fact, the semantic domain \(\D_{\lts,\etyp}\) is not finite
for \(\etyp\) that contains \(\INT\). The undecidability is obtained as a corollary of the soundness and completeness
of the reduction from the may-reachability problem to HFL model checking discussed in Section~\ref{sec:reachability}.
For the fragment of pure HFL (i.e., \HFLZ{} without integers, which we write \HFLP{} below), 
the model checking problem  is decidable~\cite{Viswanathan04}.

The \emph{order} of an \HFLZ{} model checking problem \(\lts\stackrel{?}{\models}\form\) is 
the highest order of types of subformulas of \(\form\), where the order of a type is defined by:
\(\order(\typProp)=\order(\INT) = 0\) and \(\order(\etyp\to\typ) = \max(\order(\etyp)+1,\order(\typ))\).
The complexity of order-\(k\) \HFLP{} model checking is 
\(k\)-EXPTIME complete~\cite{DBLP:journals/lmcs/AxelssonLS07},
but polynomial time in the size of HFL formulas under the assumption that the other parameters 
(the size of LTS and the largest size of types used in formulas) are fixed~\cite{Kobayashi17POPL}.

\begin{remark}
Though we do not have quantifiers on integers as primitives, we can encode them using fixpoint operators.
Given a formula \(\form:\INT\to \typProp\), we can express
\(\exists x\COL\INT.\form(x)\) and \(\forall x\COL\INT.\form(x)\) by
\((\mu X^{\INT\to\typProp}.\lambda x\COL\INT.\form(x)\lor X(x-1)\lor X(x+1))0\)
and 
\((\nu X^{\INT\to\typProp}.\lambda x\COL\INT.\form(x)\land X(x-1)\land X(x+1))0\)
respectively.
\end{remark}
\subsection{HES}
As in \cite{Kobayashi17POPL}, we often write an \HFLZ{} formula as 
a sequence of fixpoint equations, called a \emph{hierarchical equation system} (HES).

\begin{definition}
\label{def:HES}
An (extended) \emph{hierarchical equation system} (HES) is a pair \((\HES,\form)\) where \(\HES\) is a 
sequence of fixpoint equations, of the 
form:
\iffull
\[
 X_1^{\typ_1} =_{\munu_1} \form_1;  \cdots; X_n^{\typ_n} =_{\munu_n} \form_n.
\]
Here, \(\munu_i\in\set{\mu,\nu}\). 
\else
\(
 X_1^{\typ_1} =_{\munu_1} \form_1;  \cdots; X_n^{\typ_n} =_{\munu_n} \form_n
\),
where \(\munu_i\in\set{\mu,\nu}\). 
\fi
We assume that \(X_1\COL\typ_1,\ldots,X_n\COL\typ_n \pHFL \form_i:\typ_i\)
holds for each \(i\in\set{1,\ldots,n}\), and 
that \(\form_1,\ldots,\form_n,\form\) do not contain any fixpoint operators.
\end{definition}

The HES \(\HESf = (\HES, \form)\) 
represents
the \HFLZ{} formula \(\toHFL(\HES,\form)\) 
defined inductively by:
\iffull
\[
\begin{array}{l}
\toHFL(\epsilon, \form) = \form\\
\toHFL(\HES;X^\typ=_\munu \form', \form) = \toHFL([\munu X^\typ.\form'/X]\HES, [\munu X^\typ.\form'/X]\form)
\end{array}
\]
\else
\(\toHFL(\epsilon, \form) = \form\) and
\(\toHFL(\HES;X^\typ=_\munu \form', \form) = \toHFL([\munu X^\typ.\form'/X]\HES, [\munu X^\typ.\form'/X]\form)\).
\fi
Conversely, every \HFLZ{} formula can be easily converted to an equivalent HES.
In the rest of the paper, we often represent an \HFLZ{} formula in the form of HES, and just call it
an \HFLZ{} formula. We write \(\sem{\HESf}\) for \(\sem{\toHFL(\HESf)}\).
An HES
\( (X_1^{\typ_1} =_{\munu_1} \form_1;  \cdots; X_n^{\typ_n} =_{\munu_n} \form_n, \form)\) can be normalized to
\( (X_0^{\typ_0} =_{\nu} \form;X_1^{\typ_1} =_{\munu_1} \form_1;  \cdots; X_n^{\typ_n} =_{\munu_n} \form_n, X_0)\) 
where \(\typ_0\) is the type of \(\form\). Thus, we sometimes call
just a sequence of equations
\(X_0^{\typ_0} =_{\nu} \form;X_1^{\typ_1} =_{\munu_1} \form_1;  \cdots; X_n^{\typ_n} =_{\munu_n} \form_n\) 
an HES, with the understanding that ``the main formula'' is the first variable \(X_0\).
Also, we often write \(X^\typ\;x_1\,\cdots\, x_k =_\munu \form\) for
the equation \(X^\typ =_\munu \lambda x_1.\cdots\lambda x_k.\form\). We often omit type annotations
and just write \(X =_\munu \form\) for \(X^\typ =_\munu \form\).

\begin{example}
\label{ex:HES}
\iffull
The formula 
\(\form_2 = \form_{\mathtt{ab}}\,(\Some{c}\TRUE)\) in Example~\ref{ex:HFL-semantics} is expressed
as the following HES:
\[
\left(\,X =_\mu \lambda Y\COL\typProp.Y\lor \Some{\Lab{a}}(X(\Some{\Lab{b}}Y)), \quad X\,(\Some{c}\TRUE)\,\right).
\]
\fi
The formula
\(\nu X.\mu Y.\Some{\Lab{b}}X\lor \Some{\Lab{a}}Y\) (which means that
the current state has  a transition sequence of the form \((\Lab{a}^*\Lab{b})^\omega\)) 
is expressed as the following HES:
\[
\left((X =_\nu Y; Y=_\mu \Some{\Lab{b}}X\lor \Some{\Lab{a}}Y), \quad X\right).
\]
\iffull
Note that the order of the equations matters. 
\(\left((Y=_\mu \Some{\Lab{b}}X\lor \Some{\Lab{a}}Y); X =_\nu Y, \quad X\right)\)
represents the \HFLZ{} formula \(\nu X.\mu Y.\Some{\Lab{b}}\nu X.Y\lor \Some{\Lab{a}}Y\equiv
\mu Y.\Some{\Lab{b}}Y\lor \Some{\Lab{a}}Y\), which is equivalent to \(\FALSE\).
\fi
\end{example}

\section{Warming Up}
\label{sec:examples}
\label{sec:resource-usage}

To help readers get more familiar with \HFLZ{} and the idea of reductions,
we give here some variations of the examples of verification of file-accessing programs 
in Section~\ref{sec:intro},
which are instances of the ``resource usage verification problem''~\cite{IK05TOPLAS}.
General reductions will be discussed
in Sections~\ref{sec:reachability}--\ref{sec:liveness}, after the target language is set up in
Section~\ref{sec:lang}.

Consider the following OCaml-like program, which uses exceptions.
\begin{quote}
\begin{verbatim}
let readex x = read x; (if * then () else raise Eof) in
let rec f x = readex x; f x in
let d = open_in "foo" in try f d with Eof -> close d
\end{verbatim}
\end{quote}
Here, \texttt{*} represents a non-deterministic boolean value.
The function \(\texttt{readex}\) reads the file pointer \(x\), and then non-deterministically
raises an end-of-file (\texttt{Eof}) exception. 
The main expression (on the third line) first opens file ``foo'', 
 calls \texttt{f} to read the file repeatedly, and closes the file upon an end-of-file exception.
Suppose, as in the example of Section~\ref{sec:intro},
 we wish to verify that the file ``foo'' is accessed following the protocol in Figure~\ref{fig:file}.

First, we can remove exceptions by 
representing an exception handler as a special 
continuation~\cite{Blume08APLAS}:
\begin{quote}
\begin{verbatim}
let readex x h k = read x; (if * then k() else h()) in
let rec f x h k = readex x h (fun _ -> f x h k) in
let d = open_in "foo" in f d (fun _ -> close d) (fun _ -> ())
\end{verbatim}
\end{quote}
Here, we have added to each function two parameters \texttt{h} and \texttt{k},
which represent an exception handler and a (normal) continuation respectively.

Let \(\HESf\) be \((\HES,F\;\TRUE\;(\lambda r.\Some{\CLOSE}\TRUE)\;(\lambda r.\TRUE))\) where \(\HES\) is:
\newcommand\Readex{\mathit{Readex}}
\[
\begin{array}{l}
\Readex\;x\;h\;k =_\nu \Some{\READ}(k\,\TRUE \land h\,\TRUE);\\
F\;x\;h\;k =_\nu \Readex\;x\;h\;(\lambda r.F\;x\;h\;k). 
\end{array}
\]
Here, we have just replaced read/close operations with the modal operators
\(\Some{\READ}\) and \(\Some{\CLOSE}\), non-deterministic choice with 
a logical conjunction, and the unit value \((\,)\) with \(\TRUE\).
Then, \(\lts_{\mathit{file}}\models \HESf\) if and only if the program performs only valid
accesses to the file (e.g., it does not access the file after a close operation),
where \(\lts_{\mathit{file}}\) is the LTS shown in Figure~\ref{fig:file}.
The correctness of the reduction can be informally understood by observing that 
there is a close correspondence between reductions of the program and
those of the HFL formula above, and when the program reaches a read command
\(\READ\;x\), the corresponding formula is of the form \(\Some{\READ}\cdots\),
meaning that the read operation is valid in the current state; a similar condition holds also for
close operations. We will present a general translation and prove its correctness
in Section~\ref{sec:path}.

Let us consider another example, which uses integers:
\begin{quote}
\begin{verbatim}
let rec f y x k = if y=0 then (close x; k()) 
                  else (read x; f (y-1) x k) in
let d = open_in "foo" in f n d (fun _ -> ())
\end{verbatim}
\end{quote}
Here, \(\mathtt{n}\) is an integer constant.
The function \(\mathtt{f}\) reads \(\mathtt{x}\) \(\mathtt{y}\) times, and 
then calls the continuation \(\mathtt{k}\).
Let \(\lts'_{\mathit{file}}\) be the LTS obtained by adding to 
\(\lts_{\mathit{file}}\) a new state \(q_2\) and
the transition \(q_1\stackrel{\Lab{end}}{\TR}q_2\) (which intuitively 
means that
a program is allowed to terminate in the state \(q_1\)), and
let \(\HESf'\) be \((\HES',F\;n\; \TRUE\; (\lambda r.\Some{\Lab{end}}\TRUE))\) where \(\HES'\) is:
\[
\begin{array}{l}
F\;y\; x\;k =_\mu (y=0\imply \Some{\CLOSE}(k\,\TRUE))\land 
   (y\neq 0\imply \Some{\READ}(F\;(y-1)\;x\;k)). 
\end{array}
\]
Here, \(p(\form_1,\ldots,\form_k)\imply \form\) is an abbreviation of
\(\neg p(\form_1,\ldots,\form_k)\lor \form\).
Then, \(\lts'_{\mathit{file}}\models \HESf'\) if and only if
(i) the program performs only valid accesses to the file,  (ii) it eventually terminates,
and (iii) the file is closed when the program terminates.
 Notice the use of \(\mu\) instead of \(\nu\) above; by using \(\mu\), we can
express liveness properties.
The property \(\lts'_{\mathit{file}}\models \HESf'\) indeed holds for \(n\geq 0\),
but not for \(n<0\). In fact, \(F\;n\;x\;k\) is equivalent to
\(\FALSE\) for \(n<0\), and \(\Some{\READ}^n\Some{\CLOSE}(k\;\TRUE)\) for \(n\geq 0\).

\section{Target Language}
\label{sec:lang}

This section sets up, as the target of program verification,
a call-by-name\footnote{Call-by-value programs can be handled by applying the CPS
transformation before applying the reductions to HFL model checking.}
higher-order functional language extended with events.
The language is essentially the same as the one used by Watanabe et al.~\cite{Watanabe16ICFP} for discussing fair 
non-termination.

\subsection{Syntax and Typing}
We assume a finite set \(\EvSet\) of names called \emph{events}, ranged over by \(\lab\),
and a denumerable set of variables, ranged over by \(x,y,\ldots\).
Events are used to express temporal properties of programs.
We write \(\seq{x}\) (\(\seq{\term}\), resp.) for a sequence of variables (terms, resp.),
and write \(|\seq{x}|\) for the length of the sequence.

A \emph{program} is a pair \((\progd, \term)\) consisting of a set \(\progd\) of function definitions
\( \set{f_1\;\seq{x}_1 = \term_1,\ldots,f_n\;\seq{x}_n=\term_n}\) and a term \(\term\).
The set of \emph{terms}, ranged over by \(\term\), is defined by:
\[
\begin{array}{l}
\term ::= \unitexp \mid x \mid n \mid \term_1\OP\term_2
\mid \evexp{\lab}{\term}\mid \ifexp{p(\term'_1,\ldots,\term'_k)}{\term_1}{\term_2}\\\qquad
      \mid \term_1\term_2 
\mid \term_1\nondet \term_2.\\
\end{array}
\]
Here, \(n\) and \(p\) range over the sets of integers and integer predicates as in HFL formulas.
The expression \(\evexp{\lab}{\term}\) raises an event \(\lab\), and then evaluates \(\term\).
Events are used to encode program properties of interest. For example, an assertion \(\Assert(b)\)
can be expressed as \(\ifexp{b}{\unitexp}{(\evexp{\FAIL}{\Omega})}\), where \(\FAIL\) is an event 
that expresses an assertion failure and \(\Omega\) is a non-terminating term. If 
program termination is
of interest, one can insert ``\(\evatom{\evname{end}}\)'' to every termination point
and check whether an \(\evname{end}\) event occurs.
The expression \(\term_1\nondet\term_2\) evaluates \(\term_1\) or \(\term_2\) in a non-deterministic manner;
it can be used to model, e.g., unknown inputs from an environment.
We use the meta-variable \(\prog\) for programs.
When \(\prog = (\progd,\term)\) with \(\progd = 
\set{f_1\;\seq{x}_1 = \term_1,\ldots,f_n\;\seq{x}_n=\term_n}\),
we write \(\Funcs(\prog)\) for \(\set{f_1,\ldots,f_n}\)
(i.e., the set of function names defined in \(\prog\)). 
Using \(\lambda\)-abstractions, we sometimes write \(f=\lambda \seq{x}.\term\)
for the function definition \(f\;\seq{x}=\term\).
We also regard \(\progd\)
as a map from function names to terms, and write \(\dom(\progd)\) for 
\(\set{f_1,\ldots,f_n}\) and \(\progd(f_i)\) for \(\lambda \seq{x}_i.\term_i\).

Any program \((\progd,\term)\) can be normalized to \((\progd\cup\set{\mainfun=\term},\mainfun)\) where
\(\mainfun\) is a name for the ``main'' function. We sometimes write just \(\progd\) for a program
\((\progd,\mainfun)\), with the understanding that \(\progd\) contains a definition of \(\mainfun\).

We restrict the syntax of expressions using a type system.
The set of \emph{simple types}, ranged over by \(\Pst\), is defined by:
\[
\begin{array}{l}
\Pst ::= \Tunit \mid \Pest \to \Pst \qquad
\qquad \Pest ::= \Pst \mid \Tint.
\end{array}
\]
The types \(\Tunit\), \(\Tint\), and \(\Pest\to\Pst\) describe the unit value,
integers, and functions from \(\Pest\) to \(\Pst\) respectively.
Note that \(\Tint\) is allowed to occur only in argument positions. 
We defer typing rules to 
\iffull
Appendix~\ref{sec:lang-typing}, 
\else
\cite{ESOP2018full},
\fi
as they are standard,
except that we require that the righthand side of each
function definition must have type \(\Tunit\); this restriction, as well as the restriction that
\(\Tint\) occurs only in argument positions, does 
not lose generality, as those conditions
can be ensured by applying CPS transformation.
We consider below only well-typed programs.

\subsection{Operational Semantics}
We define the labeled transition relation
\(\term\Pred{\ell}{\progd}\term'\), where \(\ell\) is either \(\epsilon\) or an event name,
as the least relation closed under the rules in Figure~\ref{fig:os}. 
We implicitly assume that the program \((\progd,\term)\) is well-typed, and this assumption is
maintained throughout reductions by the standard type preservation 
\iffull property (which we omit to prove). \else property. \fi
In the rules for if-expressions, 
\(\sem{\term'_i}\) represents the integer value denoted by \(\term'_i\);
note that the well-typedness of \((\progd,\term)\) guarantees that
 \(\term'_i\) must be arithmetic expressions consisting of integers and integer operations; thus, 
\(\sem{\term'_i}\) is well defined.
We often omit the subscript \(\progd\) when it is clear from the context.
We write \(\term\Preds{\ell_1\cdots\ell_k}{\progd}\term'\) if
\(\term\Pred{\ell_1}{\progd}\cdots \Pred{\ell_k}{\progd}\term'\). Here, \(\epsilon\) is
treated as an empty sequence; thus, for example, we write \(\term\Preds{ab}{\progd}\term'\)
if \(\term\Pred{a}{\progd}\Pred{\epsilon}{\progd}\Pred{b}{\progd}\Pred{\epsilon}{\progd}\term'\).
\begin{figure}[tb]
 \begin{multicols}{2}
 \InfruleS{0.21}{\ }{
 \evexp{\lab}{\term}  \Pred{\lab}{\progd}\term 
 }
 \InfruleS{0.3}{
 f \seq{x} = \termaltu \in \progd \andalso
 |\seq{x}|=|\seq{\term}|
 } {
 f\;\seq{\term} \Pred{\epsilon}{\progd} [\seq{\term}/\seq{x}]\termaltu
 }
 \infrule{i\in\set{1,2}} {
 \term_1 \nondet \term_2 \Pred{\epsilon}{\progd} \term_i
 }
 \infrule{
  (\sem{\term'_1},\ldots,\sem{\term'_k})\in\sem{p} 
 } {
 \ifexp{p(\term'_1,\dots, \term'_k)}{\term_1}\term_2 \Pred{\epsilon}{\progd} \term_1
 }
 \infrule{
  (\sem{\term'_1},\ldots,\sem{\term'_k})\not\in \sem{p} 
 } {
 \ifexp{p(\term'_1,\dots, \term'_k)}{\term_1}\term_2 \Pred{\epsilon}{\progd} \term_2
 }
 \end{multicols}
\caption{Labeled Transition Semantics}
\label{fig:os}
\end{figure}

\noindent
For a program \(\prog=(\progd,\term_0)\), 
we define the set \(\Traces(\prog) (\subseteq \EvSet^*\cup\EvSet^\omega)\)
of \emph{traces} by: 
\[
\begin{array}{l}
\Traces(\progd,\term_0) =
\set{\ell_0\cdots \ell_{n-1} \in (\set{\epsilon}\cup\EvSet)^* 
\mid 
\forall i\in\set{0,\ldots,n-1}. \term_{i}\Pred{\ell_i}{\progd} \term_{i+1}}\\\qquad\qquad\qquad
\cup
\set{\ell_0\ell_1\cdots \in (\set{\epsilon}\cup\EvSet)^\omega
\mid 
\forall i\in \omega. \term_i\Pred{\ell_i}{\progd} \term_{i+1}}.
\end{array}
\]
Note that since the label \(\epsilon\) is regarded as an empty sequence, 
\(\ell_0\ell_1\ell_2 = \lab\lab\) if \(\ell_0=\ell_2=\lab\) and \(\ell_1=\epsilon\), and
an element of \((\set{\epsilon}\cup\EvSet)^\omega\) is regarded as that of 
\(\EvSet^*\cup\EvSet^\omega\).
We write
\(\FinTraces(\prog)\) and \(\InfTraces(\prog)\) for
\(\Traces(\prog)\cap\EvSet^*\) and \(\Traces(\prog)\cap \EvSet^\omega\) respectively.
The set of \emph{full traces} \(\FullTraces(\progd,\term_0)(\subseteq \EvSet^*\cup\EvSet^\omega)\) is defined as:
\[
\begin{array}{l}
\set{\ell_0\cdots \ell_{n-1} \in (\set{\epsilon}\cup\EvSet)^* \mid 
\term_n=\Vunit\land
\forall i\in\set{0,\ldots,n-1}. \term_{i}\Pred{\ell_i}{\progd} \term_{i+1}}\\\qquad
\cup
\set{\ell_0\ell_1\cdots\in (\set{\epsilon}\cup\EvSet)^\omega\mid 
\forall i\in \omega. \term_i\Pred{\ell_i}{\progd} \term_{i+1}}.
\end{array}
\]
\begin{example}
The last example in Section~\ref{sec:intro} is modeled as 
\(P_{\mathit{file}} = (D, f\,(\,))\), where 
\(D = \set{f\,x = (\evexp{\CLOSE}\Vunit)\nondet (\evexp{\READ}\evexp{\READ}f\,x)}\).
We have:
\[
\begin{array}{l}
\Traces(P) = \set{\READ^{n} \mid n\geq 0} \cup \set{\READ^{2n}\CLOSE \mid n\geq 0} \cup 
 \set{\READ^\omega}\\
\FinTraces(\prog) = \set{\READ^{n} \mid n\geq 0} \cup \set{\READ^{2n}\CLOSE \mid n\geq 0} \\
\InfTraces(\prog) =  \set{\READ^\omega}\ 
\FullTraces(P) = \set{\READ^{2n}\CLOSE \mid n\geq 0} \cup \set{\READ^\omega}.
\end{array}
\]
\end{example}



\section{May/Must-Reachability Verification}
\label{sec:reachability}
\label{SEC:REACHABILITY}

Here we consider the following problems:
\begin{itemize}
\item May-reachability: ``Given a program \(\prog\) and an event \(\lab\),
may \(\prog\) raise 
\(\lab\)?''
\item Must-reachability: ``Given a program \(P\) and an event \(\lab\),
must \(\prog\) raise 
\(\lab\)?''
\end{itemize}
Since we are interested in a particular event \(\lab\), we restrict here the event set \(\EvSet\)
 to a singleton set of the form \(\set{a}\). Then, the may-reachability is
formalized as \(a\stackrel{?}{\in}\Traces(P)\), whereas the must-reachability is formalized as
``does every trace in \(\FullTraces(P)\) contain \(a\)?''
We encode both problems into the validity of \HFLZ{} formulas (without any modal operators
\(\Some{\lab}\) or \(\All{\lab}\)), or the \HFLZ{} model checking of
those formulas against a trivial model (which consists of a single state without any transitions).
Since our reductions are sound and complete, the characterizations
of their negations --non-reachability and may-non-reachability-- can also be obtained immediately.
Although these are the simplest classes of properties among those
discussed in Sections~\ref{sec:reachability}--\ref{sec:liveness}, they are already large
enough to accommodate many program properties discussed in the literature, including
lack of assertion failures/uncaught exceptions~\cite{KSU11PLDI} (which can be characterized as non-reachability;
recall the encoding of assertions in Section~\ref{sec:lang}), termination~\cite{Ledesma-Garza2012,Kuwahara2014Termination}
(characterized as must-reachability), and non-termination~\cite{Kuwahara2015Nonterm}
(characterized as may-non-reachability).

\subsection{May-Reachability}
\label{sec:mayreach}
As in the examples in Section~\ref{sec:resource-usage}, we translate a program to a formula
that says ``the program may raise an event \(\lab\)'' in a compositional manner. 
For example, \(\evexp{\lab}{\term}\) can be
translated to \(\TRUE\) (since the event will surely be raised immediately), and
\(\term_1\nondet\term_2\) can be translated to \(\term_1^\dagger \lor \term_2^\dagger\)
where \(\term_i^\dagger\) is the result of the translation of \(\term_i\) (since
only one of \(\term_1\) and \(\term_2\) needs to raise an event).

\begin{definition}
Let \(\prog=(\progd,\term)\) be a program.
\(\HESf_{\prog,\may}\) is the HES \((\trMay{\progd}, \trMay{\term})\), where
\(\trMay{\progd}\) and \(\trMay{\term}\) are defined by:
\[
\begin{array}{l}
\trMay{\set{f_1\,\seq{x}_1=\term_1,\ldots,f_n\,\seq{x}_n=\term_n}} =
\left(f_1\;\seq{x}_1=_\mu \trMay{\term_1};\cdots; f_n\;\seq{x}_n=_\mu \trMay{\term_n}\right)\\
\trMay{\Vunit} = \FALSE \qquad \trMay{x} = x \qquad  \trMay{n}=n\qquad
\trMay{(\term_1\OP\term_2)}=\trMay{\term_1}\OP\trMay{\term_2}\\
\trMay{(\ifexp{p(\term'_1,\ldots,\term'_k)}{\term_1}{\term_2})}= \\\qquad\qquad
  (p(\trMay{\term'_1},\ldots,\trMay{\term'_k})\land \trMay{\term_1})
  \lor (\neg p(\trMay{\term'_1},\ldots,\trMay{\term'_k})\land \trMay{\term_2})\\
\trMay{(\evexp{\lab}{\term})} = \TRUE\quad 
\trMay{(\term_1\term_2)} = \trMay{\term_1}\trMay{\term_2}\quad
\trMay{(\term_1\nondet\term_2)} = \trMay{\term_1}\lor \trMay{\term_2}.\\
\end{array}
\]
\end{definition}
Note that, in the definition of \(\trMay{\progd}\), the order of function definitions in \(\progd\) does not matter
(i.e., the resulting HES is unique up to the semantic equality), since all the fixpoint variables are bound by \(\mu\).
\begin{example}
Consider the program:
\[P_{\mathit{loop}} = (\set{\mathit{loop}\;x = \mathit{loop}\;x},
\mathit{loop}(\evexp{\lab}{(\,)})).\] It is translated to
the HES \(\HESf_{\mathit{loop}} = (\mathit{loop}\;x=_\mu \mathit{loop}\;x, \mathit{loop}(\TRUE))\).
Since \(\mathit{loop} \equiv \mu \mathit{loop}.\lambda x.\mathit{loop}\;x\) is equivalent to
\(\lambda x.\FALSE\), \(\HESf_{\mathit{loop}}\) is equivalent to \(\FALSE\).
In fact, \(P_{\mathit{loop}}\) never raises an event \(\lab\) (recall that our language is call-by-name).
\end{example}
\begin{example}
Consider the program \(\prog_{\mathit{sum}}=(\progd_{\mathit{sum}},\mainfun)\) where \(\progd_{\mathit{sum}}\) is:
\[
\begin{array}{l}
\mainfun = \Sum\ n\ (\lambda r.\Assert(r\geq n))\\
\Sum\ x\ k = \ifexp{x=0}{k\, 0}{\Sum\ (x-1)\ (\lambda r.k(x+r))}\\
\end{array}
\]
Here, \(n\) is some integer constant, and \(\Assert(b)\) is the macro introduced in Section~\ref{sec:lang}. We have used \(\lambda\)-abstractions for the sake of readability.
The function \(\Sum\) is a CPS version of a function that computes the summation of integers
from \(1\) to \(x\). The main function computes the sum \(r=1+\cdots + n\), and asserts 
 \(r\geq n\).
It is translated to the HES \(\HESf_{\prog_2,\may} = (\HES_{\Sum},\mainfun)\) where \(\HES_{\Sum}\) is:
\[
\begin{array}{l}
\mainfun =_\mu \Sum\ n\ (\lambda r.(r\geq n\land \FALSE)\lor (r<n\land \TRUE));\\
\Sum\ x\ k =_\mu (x=0\land k\, 0)\lor (x\neq 0\land \Sum\ (x-1)\ (\lambda r.k(x+r))).
\end{array}
\]
Here, \(n\) is treated as a constant. Since the shape of the formula does not depend on the value of \(n\), 
the property ``an assertion failure may occur for some \(n\)'' can be expressed by
\(\exists n.\HESf_{\prog_2,\may}\).
\iffull
Thanks to the completeness of the encoding (Theorem~\ref{th:mayreach} below),
the lack of assertion failures can be characterized by \(\forall n.\HESf\),
where \(\HESf\) is the De Morgan dual of the above HES:
\[
\begin{array}{l}
\mainfun =_\nu \Sum\ n\ (\lambda r.(r< n\lor \TRUE)\land (r\geq n \lor \FALSE))\\
\Sum\ x\ k =_\nu (x\neq 0\lor k\, 0)\land (x= 0\lor \Sum\ (x-1)\ (\lambda r.k(x+r))).
\end{array}
\]
\nk{A note on the general encoding of properties of functions of the form
\(f: x:\{\Tint|P(x)\} \to \{r:\Tint|Q(x,r)\}\).}
\fi
\qed
\end{example}

The following theorem states that \(\HESf_{\prog,\may}\) is a complete characterization of
the may-reachability of \(\prog\).
\begin{theorem}
\label{th:mayreach}
Let \(\prog\) be a program. Then, \(\lab\in\Traces(\prog)\) if and only if
\(\lts_0 \models \HESf_{\prog,\may}\) for \(\lts_0 = (\set{\stunique},\emptyset,\emptyset,\stunique)\).
\end{theorem}

\iffull
To prove the theorem, we first show the theorem for recursion-free programs
and then lift it to arbitrary programs by using the continuity of functions represented in
the fixpoint-free fragment of \HFLZ{} formulas. 
To show that the theorem holds for recursion-free programs, 
See Appendix~\ref{sec:proofs-mayreach} for a concrete proof.
\else
A proof of the theorem above is found in \cite{ESOP2018full}.
We only provide an outline.
We first show the theorem for recursion-free programs
and then lift it to arbitrary programs by using the continuity of functions represented in
the fixpoint-free fragment of \HFLZ{} formulas. 
To show the theorem for recursion-free programs, we define the reduction relation
\(\term\redv{\progd}\term'\) by:\\[0.3ex]
\quad \InfruleS{0.3}{
 f \seq{x} = \termaltu \in \progd \andalso
 |\seq{x}|=|\seq{\term}|
 } {
 \EC[f\;\seq{\term}] \redv{\progd} \EC[[\seq{\term}/\seq{x}]\termaltu]
 }
 \InfruleS{0.6}{
  (\sem{\term'_1},\ldots,\sem{\term'_k})\in\sem{p} 
 } {
 \EC[\ifexp{p(\term'_1,\dots, \term'_k)}{\term_1}\term_2] \redv{\progd} \EC[\term_1]
 }
 \InfruleS{0.6}{
  (\sem{\term'_1},\ldots,\sem{\term'_k})\not\in \sem{p} 
 } {
 \EC[\ifexp{p(\term'_1,\dots, \term'_k)}{\term_1}\term_2] \redv{\progd} \EC[\term_2]
 }\\[0.3ex]
Here, \(E\) ranges over the set of evaluation contexts given by
\(\EC ::= \Hole \mid \EC\nondet \term \allowbreak\mid \term\nondet \EC \mid \evexp{\lab}\EC\).
The reduction relation differs from 
the labeled transition relation given in Section~\ref{sec:lang}, in that
\(\nondet\) and \(\evexp{\lab}{\cdots}\) are not eliminated.
By the definition of the translation, the theorem holds for programs in normal form (with respect to
the reduction relation), and the semantics of translated HFL formulas is preserved by the reduction relation;
thus the theorem holds for 
 recursion-free programs, as they are strongly normalizing.
\fi

\subsection{Must-Reachability}
\label{sec:mustreach}
The characterization of must-reachability can be 
obtained by an easy modification of the 
characterization of may-reachability: we just need to replace branches with logical conjunction.
\begin{definition}
Let \(\prog=(\progd,\term)\) be a program.
\(\HESf_{\prog,\must}\) is the HES \((\trMust{\progd}, \trMust{\term})\), where
\(\trMust{\progd}\) and \(\trMust{\term}\) are defined by:
\[
\begin{array}{l}
\trMust{\set{f_1\,\seq{x}_1=\term_1,\ldots,f_n\,\seq{x}_n=\term_n}} =
\left(f_1\;\seq{x}_1=_\mu \trMust{\term_1};\cdots; f_n\;\seq{x}_n=_\mu \trMust{\term_n}\right)\\
\trMust{\Vunit} = \FALSE \qquad \trMust{x} = x \qquad  \trMust{n}=n\qquad
\trMust{(\term_1\OP\term_2)}=\trMust{\term_1}\OP\trMust{\term_2}\\
\trMust{(\ifexp{p(\term'_1,\ldots,\term'_k)}{\term_1}{\term_2})}= \\\qquad\qquad
  (p(\trMust{\term'_1},\ldots,\trMust{\term'_k})\imply \trMust{\term_1})
  \land (\neg p(\trMust{\term'_1},\ldots,\trMust{\term'_k})\imply \trMust{\term_2})\\
\trMust{(\evexp{\lab}{\term})} = \TRUE\quad 
\trMust{(\term_1\term_2)} = \trMust{\term_1}\trMust{\term_2}\ 
\trMust{(\term_1\nondet\term_2)} = \trMust{\term_1}\land \trMust{\term_2}.\\
\end{array}
\]
Here, \(p(\form_1,\ldots,\form_k)\imply \form\) is a shorthand for
\(\neg p(\form_1,\ldots,\form_k)\lor \form\).
\end{definition}

\begin{example}
\newcommand\LOOP{\mathtt{loop}}
Consider \(\prog_{\LOOP} = (D, \LOOP\,m\,n)\) where \(D\) is:
\[
\begin{array}{l}
\LOOP\; x\; y = \ifexp{x\leq 0\lor y\leq 0}{(\evexp{\mathtt{end}} \Vunit)\\\qquad\qquad}
  {(\LOOP\; (x-1)\;(y*y)) \nondet (\LOOP\;x\;(y-1))}
\end{array}
\]
Here, the event \(\mathtt{end}\) is used to signal the termination of the program.
The function \(\LOOP\) non-deterministically updates the values of \(x\) and \(y\)
until either \(x\) or \(y\) becomes non-positive. The must-termination of the program
is characterized by \(\HESf_{\prog_{\LOOP,\must}} = (\HES, \LOOP\,m\,n)\) where \(\HES\) is:
\[
\begin{array}{l}
\LOOP\;x\;y =_\mu (x\leq 0\lor y\leq 0\imply \TRUE)\\\qquad\qquad\qquad
  \land (\neg(x\leq 0\lor y\leq 0) \imply (\LOOP\;(x-1)\;(y*y))\land (\LOOP\;x\;(y-1))).
\end{array}
\]
\end{example}

We write \(\Must_\lab(\prog)\) if every \(\pi\in \FullTraces(\prog)\) contains \(\lab\).
The following theorem, which can be proved in a manner similar to 
Theorem~\ref{th:mayreach}, guarantees that \(\HESf_{\prog,\must}\) is indeed 
a sound and complete characterization of the 
 must-reachability.
\begin{theorem}
\label{th:must-reachability}
Let \(\prog\) be a program. Then, \(\Must_\lab(\prog)\) if and only if
\(\lts_0 \models \HESf_{\prog,\must}\) for \(\lts_0 = (\set{\stunique},\emptyset,\emptyset,\stunique)\).
\end{theorem}
\iffull
The proof is given in 
Appendix~\ref{sec:proofs-mustreach}.
\fi

\section{Trace Properties}
\label{sec:path}
\label{SEC:PATH}
Here we consider the verification problem:
``Given a (non-\(\omega\)) regular language \(L\) and a program \(\prog\),
does \emph{every} finite event sequence of \(\prog\) belong to \(L\)?
(i.e. \(\FinTraces(\prog)\stackrel{?}{\subseteq} L\))'' 
and reduce it to an \HFLZ{} model checking problem. 
The verification of file-accessing programs considered in Section~\ref{sec:resource-usage}
may be considered an instance of the 
\iffull
problem.\footnote{The last example in 
Section~\ref{sec:resource-usage} is actually a combination with the must-reachability problem.}
\else
problem.
\fi

Here we assume that the language
\(L\) is closed under the prefix operation; this does not lose generality because \(\FinTraces(P)\)
is also closed under the prefix operation. We write \(\A_L=
(Q,\Sigma,\delta,q_0,F)\) for
the minimal, deterministic automaton with no dead states (hence the transition function
\(\delta\) may be partial). 
Since \(L\) is prefix-closed and the automaton is minimal, 
\(w\in L\) if and only if \(\hat{\delta}(q_0,w)\) is defined (where \(\hat{\delta}\) is defined by:
\(\hat{\delta}(q,\epsilon)=q\) and \(\hat{\delta}(q,aw) = \hat{\delta}(\delta(q,a),w)\)).
We use the corresponding LTS \(\lts_L = (Q, \Sigma, \set{(q,a,q') \mid \delta(q,a)=q'}, q_0)\)
as the model of the reduced \HFLZ{} model checking problem.

Given the LTS \(\lts_L\) above, whether an event sequence \(\lab_1\cdots\lab_k\) belongs to \(L\) 
can be expressed as \(\lts_L \stackrel{?}{\models} \Some{\lab_1}\cdots\Some{\lab_k}\TRUE\).
Whether all the event sequences in \(\set{\lab_{j,1}\cdots\lab_{j,k_j}\mid j\in\set{1,\ldots,n}}\)
belong to \(L\) can be expressed as
\(\lts_L \stackrel{?}{\models} \bigwedge_{j\in\set{1,\ldots,n}} 
\Some{\lab_{j,1}}\cdots\Some{\lab_{j,k_j}}\TRUE\).
We can lift these translations for event sequences to
the translation from a program (which can be considered a description of a set of event sequences)
to an \HFLZ{} formula, as follows.
\begin{definition}
Let \(\prog=(\progd,\term)\) be a program.
\(\HESf_{\prog,\Path}\) is the HES \((\trPath{\progd}, \trPath{\term})\), where
\(\trPath{\progd}\) and \(\trPath{\term}\) are defined by:
\[
\begin{array}{l}
\trPath{\set{f_1\,\seq{x}_1=\term_1,\ldots,f_n\,\seq{x}_n=\term_n}} =
\left(f_1\;\seq{x}_1=_\nu \trPath{\term_1};\cdots; f_n\;\seq{x}_n=_\nu \trPath{\term_n}\right)\\
\trPath{\Vunit} = \TRUE \qquad \trPath{x} = x \qquad  \trPath{n}=n\qquad
\trPath{(\term_1\OP\term_2)}=\trPath{\term_1}\OP\trPath{\term_2}\\
\trPath{(\ifexp{p(\term'_1,\ldots,\term'_k)}{\term_1}{\term_2})}= \\\qquad\qquad
  (p(\trPath{\term'_1},\ldots,\trPath{\term'_k})\imply \trPath{\term_1})
  \land (\neg p(\trPath{\term'_1},\ldots,\trPath{\term'_k})\imply \trPath{\term_2})\\
\trPath{(\evexp{\lab}{\term})} = \Some{\lab}\trPath{\term}\quad 
\trPath{(\term_1\term_2)} = \trPath{\term_1}\trPath{\term_2}\ 
\trPath{(\term_1\nondet\term_2)} = \trPath{\term_1}\land \trPath{\term_2}.\\
\end{array}
\]
\end{definition}

\begin{example}
The last program discussed in Section~\ref{sec:resource-usage} is modeled as \(\prog_2 = (\progd_2,f\; m\; g)\), 
where \(m\) is an integer constant and \(\progd_2\) consists of:
\[
\begin{array}{l}
f\; y\; k = \ifexp{y=0}{(\evexp{\CLOSE}{k\,\Vunit})}{(\evexp{\READ}{f\;(y-1)\;k})}\\
g\; r = \evexp{\mathtt{end}}{\Vunit}
\end{array}
\]
Here, we have modeled accesses to the file, and termination as events.
Then, \(\HESf_{\prog_2,\Path} = (\HES_{\prog_2,\Path}, f\;m\;g)\) where 
\(\HES_{\prog_2,\Path}\) is:\footnote{Unlike in Section~\ref{sec:examples}, the variables are 
bound by \(\nu\) since we are not concerned with the termination property here.}
\[
\begin{array}{l}
f\, n\; k =_\nu (n= 0\imply \Some{\CLOSE}{(k\,\TRUE)})\land (n\neq 0\imply \Some{\READ}(f\;(n-1)\;k))\\
g\;r =_\nu \Some{\mathtt{end}}{\TRUE}.
\end{array}
\]
Let \(L\) be the prefix-closure of \(\READ^*\cdot \CLOSE\cdot\mathtt{end}\). Then \(\lts_L\)
is \(\lts'_{\mathit{file}}\) in Section~\ref{sec:resource-usage}, and
 \(\FinTraces(\prog_2){\subseteq} L\) can be verified by 
checking \(\lts_L{\models} \HESf_{\prog_2,\Path}\).
\qed
\end{example}

\begin{theorem}
\label{th:path}
Let \(\prog\) be a program and \(L\) be a regular, prefix-closed language.
Then, \(\FinTraces(\prog)\subseteq L\) if and only if \(\lts_L\models \HESf_{\prog,\Path}\).
\end{theorem}

As in Section~\ref{sec:reachability}, we first prove the theorem for programs in normal form,
and then lift it
to recursion-free programs by using the preservation of the semantics of \HFLZ{} formulas by
reductions, and further to arbitrary programs by using the (co-)continuity of the functions
represented by fixpoint-free \HFLZ{} formulas. 
\iffull
The proof is given in Appendix~\ref{sec:proofs-trace}.
\else
See \cite{ESOP2018full} for a concrete proof.
\fi

\section{Linear-Time Temporal Properties}
\label{sec:liveness}
\label{SEC:LIVENESS}

This section considers the following problem:
 ``Given a program \(\prog\) and an \(\omega\)-regular word language \(L\), does 
\(\InfTraces(P){\cap} L = \emptyset\) hold\(?\)''
From the viewpoint of program verification, \(L\) represents the set of ``bad'' behaviors.
\newcommand\LOOP{\mathit{loop}}
This can be considered an extension of the problems considered in the previous 
\iffull 
sections.\footnote{Note that
finite traces can be turned into infinite ones by inserting a dummy event for every function call
and replacing each occurrence of the unit value \(\Vunit\) with \(\LOOP(\,)\) where \(\LOOP\,x=\evexp{\Lab{dummy}}\LOOP\,x\).}
\else sections.
\fi 

The reduction to HFL model checking is more involved than those in
the previous sections. To see the difficulty, consider the
 program \(P_0\):
\[ \left(\set{f = \ifexp{c}{(\evexp{\Lab{a}}f)}{(\evexp{\Lab{b}}f)}}, \quad f\right),\]
where \(c\) is some boolean expression.
Let \(L\) be the complement of \((\Lab{a}^*\Lab{b})^\omega\), i.e., 
the set of infinite sequences that contain only finitely many 
\(\Lab{b}\)'s. Following Section~\ref{sec:path}
(and noting that
\(\InfTraces(P){\cap} L = \emptyset\) is equivalent to 
\(\InfTraces(P) \subseteq (\Lab{a}^*\Lab{b})^\omega\) in this case),
 one may be tempted to 
prepare an LTS like the one in Figure~\ref{fig:abomega} (which corresponds to the transition function
of a (parity) word automaton accepting \((\Lab{a}^*\Lab{b})^\omega\)),
and translate the program to an HES \(\HESf_{P_0}\) of the form:
\[ \left(f =_\alpha (c\imply \Some{\Lab{a}}f) \land (\neg c\imply \Some{\Lab{b}}f), \quad f\right),\]
where \(\alpha\) is \(\mu\) or \(\nu\). 
However, such a translation would not work.
If \(c=\TRUE\), then \(\InfTraces(P_0)=\Lab{a}^\omega\), hence 
\(\InfTraces(P_0)\cap L\neq\emptyset\);
thus, \(\alpha\) should be \(\mu\) for \(\HESf_{P_0}\) to be unsatisfied.
If \(c=\FALSE\), however, \(\InfTraces(P_0)=\Lab{b}^\omega\), hence 
\(\InfTraces(P_0)\cap L=\emptyset\); thus,
\(\alpha\) must be \(\nu\) for \(\HESf_{P_0}\) to be satisfied. 
\begin{figure}[tb]
\begin{center}
\[
\xymatrix@!R=2pc{%
 *+<1pc>[o][F-]{q_a}  \ar@(d,l)^{a} \ar@/^/[r]^{b}
& *+<1pc>[o][F-]{q_b} \ar@/^/[l]^{\Lab{a}} \ar@(u,r)^{b} }\]
\end{center}
\caption{LTS for \((a^*b)^\omega\)}
\label{fig:abomega}
\end{figure}
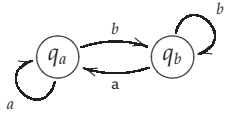

The example above suggests that
 we actually need to distinguish between the two occurrences of \(f\) in the body
of \(f\)'s definition. Note that in the then- and else-clauses respectively, \(f\) is called 
after different events \(\Lab{a}\) and \(\Lab{b}\). This difference is important, since we are 
interested in whether \(\Lab{b}\) occurs infinitely often.
We thus duplicate \(f\), and replace the program with the following program \(\prog_{\mathit{dup}}\):
\[\begin{array}{ll}
(\{f_b = \ifexp{c}{(\evexp{\Lab{a}}f_a)}{(\evexp{\Lab{b}}f_b)}, \\
 \quad f_a = \ifexp{c}{(\evexp{\Lab{a}}f_a)}{(\evexp{\Lab{b}}f_b)}\}, & f_b).
\end{array}
\]
For checking \(\InfTraces(P_0)\cap L = \emptyset\),
it is now sufficient to check that \(f_b\) is recursively called infinitely often.
We can thus obtain the following HES:
\[\begin{array}{ll}
((f_b =_\nu (c\imply \Some{\Lab{a}}f_a)\land (\neg c\imply \Some{\Lab{b}}f_b); 
 \quad f_a =_\mu (c\imply \Some{\Lab{a}}f_a)\land (\neg c\imply \Some{\Lab{b}}f_b)), \quad & f_b).
\end{array}
\]
Note that \(f_b\) and \(f_a\) are bound by \(\nu\) and \(\mu\) respectively, reflecting the fact
that \(\Lab{b}\) should occur infinitely often, but \(\Lab{a}\) need not. If \(c=\TRUE\),
the formula is equivalent to \(\nu f_b.\Some{\Lab{a}}\mu f_a.\Some{\Lab{a}}f_a\), which is false.
If \(c=\FALSE\), then the formula is equivalent to \(\nu f_b.\Some{\Lab{b}}f_b\), 
which is satisfied by by the LTS in Figure~\ref{fig:abomega}.

The general translation is more involved due to the presence of higher-order functions,
but, as in the example above, the overall translation consists of two steps.
We first replicate functions according to what events may occur between two recursive calls,
and reduce the problem \(\InfTraces(P)\cap L\stackrel{?}{=} \emptyset\)
to a problem of analyzing which functions are recursively called infinitely often,
which we call a \emph{call-sequence analysis}. We can then reduce the call-sequence analysis
to HFL model checking in a rather straightforward manner (though the proof of the correctness
is non-trivial).
The resulting HFL formula actually does not contain modal operators.\footnote{In the example above,
we can actually remove \(\Some{\Lab{a}}\) and \(\Some{\Lab{b}}\), as information about events has
been taken into account when \(f\) was duplicated.}
So, as in Section~\ref{sec:reachability}, the resulting problem is the validity checking of
HFL formulas without modal operators.

In the rest of this section, we first introduce the call-sequence analysis problem and
its reduction to HFL model checking in Section~\ref{sec:callsequence}.
We then show how to reduce the temporal verification problem \(\InfTraces(P)\cap L\stackrel{?}{=} \emptyset\)
to an instance of the call-sequence analysis problem in Section~\ref{sec:to-callsequence}.

\subsection{Call-sequence analysis}
\label{sec:callsequence}
We define the call-sequence analysis and reduce it to an
 HFL model-checking problem.
As mentioned above, in the call-sequence analysis,
we are interested in analyzing which functions are \emph{recursively called} infinitely often.
Here, we say that \(g\) is \emph{recursively called from} \(f\), if
\(f\,\seq{s}\Pred{\epsilon}{\progd} [\seq{s}/\seq{x}]t_f\Preds{\labseq}{\progd} g\,\seq{t}\),
where \(f\,\seq{x}=t_f\in \progd\) and \(g\) ``originates from'' \(t_f\) (a more formal
definition will be given in Definition~\ref{def:recall} below). For example, consider the following program \(\prog_{\mathit{app}}\),
which is a twisted version of 
\(\prog_{\mathit{dup}}\) above.
\[\begin{array}{ll}
(\{\APP\,h\,x = h\,x,\\\quad
f_b\,x = \ifexp{x> 0}{(\evexp{\Lab{a}}\APP\,f_a\,(x-1))}{(\evexp{\Lab{b}}\APP\,f_b\,5)}, \\
 \quad f_a\,x = \ifexp{x> 0}{(\evexp{\Lab{a}}\APP\,f_a\,(x-1))}{(\evexp{\Lab{b}}\APP\,f_b\,5)}\},& f_b\,5).
\end{array}
\]
Then \(f_a\) is
``recursively called'' from \(f_b\) in \(f_b\,5 \Preds{\Lab{a}}{\progd} \APP\,f_a\,4\Preds{\epsilon}{\progd}
f_a\,4\) (and so is \(\APP\)).
We are interested in infinite chains 
of recursive calls
\(f_0f_1f_2\cdots\), and which functions may occur infinitely often in each chain.
For instance, the program above has the unique infinite chain \((f_b f_a^5)^\omega\), in which 
both \(f_a\) and \(f_b\) occur infinitely often. (Besides the infinite chain, the program has finite
chains like \(f_b\,\APP\); note that the chain cannot be extended further, as the body of \(\APP\) does not have
any occurrence of recursive functions: \(\APP,f_a\) and \(f_b\).)

We define the notion of ``recursive calls'' and call-sequences formally below.


\begin{definition}[recursive call relation, call sequences]
\label{def:recall}
Let \(\prog=(\progd, f_1\,\seq{s})\) be a program, with \(\progd= \{ f_i\,\tilde{x}_i = u_i \}_{1 \le i \le n}\).
We define \( \progd^{\Mark} := \progd \cup \{ f^{\Mark}_i\,\tilde{x} = u_i \}_{1 \le i \le n} \) where \( f^{\Mark}_1, \dots, f^{\Mark}_n \) are fresh symbols. (Thus, \( \progd^{\Mark} \) has two copies of each function symbol, one of which is marked by \(\Mark\).)
For the terms \(\seq{t}_i\) and \(\seq{t}_j\) that do not contain marked symbols,
we write \(f_i\,\seq{t}_i \recall{\progd} f_j\,\seq{t}_j\) if
(i) \([\seq{t}_i/\seq{x}_i][f_1^\Mark/f_1,\ldots,f_n^\Mark/f_n]u_i \Preds{\labseq}{\progd^\Mark} f_j^\Mark\,\seq{t}'_j\)
and (ii) \(\seq{t}_j\) is obtained by erasing all the marks in \(\seq{t}'_j\).
We write \(\Callseq(P)\) for the set of (possibly infinite) sequences of function symbols:
\[ \set{f_1\,g_1\,g_2\dots \mid f_1\,\seq{s}\recall{\progd} g_1\,\seq{t}_1\recall{\progd} g_2\,\seq{t}_2\recall{\progd}\cdots}.\]
We write \(\InfCallseq(P)\) for the subset of \(\Callseq(\prog)\) consisting of infinite sequences, i.e.,
\(\Callseq(P)\cap \set{f_1,\ldots,f_n}^\omega\).
\end{definition}

For example, for \(\prog_{\mathit{app}}\) above,  
\(\Callseq(P)\) is the prefix closure of
\(\set{(f_bf_a^5)^\omega}\cup 
\set{s\cdot \APP \mid \mbox{$s$ is a non-empty finite prefix of \((f_bf_a^5)^\omega\)}}\),
and
\(\InfCallseq(P)\) is the singleton set  \(\set{(f_bf_a^5)^\omega}\).

\begin{definition}[Call-sequence analysis]
\label{def:call-sequence-analysis}
  A \emph{priority assignment} for a program \( P \) is a function \( \Pfun \COL \Funcs(P) \to \mathbb{N} \)
from the set of function symbols of \( P \) to the set \(\mathbb{N}\) of natural numbers.
We write \(\models_{\CSA} (\prog,\Pfun)\) if
every infinite call-sequence \(g_0g_1g_2\dots\in
\InfCallseq(P)\) 
satisfies the parity condition w.r.t.~\( \Pfun \), i.e.,
the largest number occurring infinitely often in \( \Pfun(g_0) \Pfun(g_1) \Pfun(g_2) \dots \) is even.
\emph{Call-sequence analysis} is the problem of, given a program \( P \) with a priority assignment \( \Pfun \),
deciding whether \(\models_{\CSA} (\prog,\Pfun)\) holds.
\end{definition}

For example, for \(\prog_{\mathit{app}}\) and the priority assignment \(\Pfun_{\mathit{app}} =
\set{\APP\mapsto 3, f_a\mapsto 1, f_b\mapsto 2}\),
\(\models_\CSA(\prog_{\mathit{app}}, \Pfun_{\mathit{app}})\) holds.

The call-sequence analysis can naturally be reduced to
 HFL model checking against the trivial LTS \( \lts_0 = (\set{\stunique}, \emptyset, \emptyset, \stunique) \)
(or validity checking).
\begin{definition}
  Let \(\prog=(D,t)\)
be a program and \(\Pfun\) be a priority assignment for \(\prog\).
The HES \(\HESf_{(\prog,\Pfun),\CSA}\) is \((\livetrans{\progd}, \livetrans{t})\), where 
\(\livetrans{\progd}\) and \(\livetrans{t}\) are defined by:
  \[
  \begin{array}{l}
\livetrans{
\set{f_1\,\seq{x}_1=\term_1,\ldots,f_n\,\seq{x}_n=\term_n}}
= \left(f_1\;\seq{x}_1=_{\munu_1} \livetrans{\term_1};\cdots; f_n\;\seq{x}_n=_{\munu_n} \livetrans{\term_n}\right)\\
   \livetrans{\unitexp} = \TRUE \qquad
    \livetrans{x} = x \qquad
    \livetrans{n} = n \qquad
    \livetrans{(\term_1 \OP \term_2)} = \livetrans{\term_1} \OP \livetrans{\term_2} \\
    \livetrans{(\ifexp{p(\Pae_1',\ldots,\Pae_k')}{\term_1}{\term_2})} = \\\qquad
(p(\livetrans{\Pae_1'},\ldots,\livetrans{\Pae_k'}) \imply \livetrans{\term_1})\land (\neg p(\livetrans{\Pae_1'},\ldots,\livetrans{\Pae_k'}) \imply \livetrans{\term_2})\\
    \livetrans{(\evexp{\lab}{\term})} = \livetrans{\term} \qquad
    \livetrans{(\term_1\,\term_2)} = \livetrans{\term_1}\,\livetrans{\term_2} \qquad
    \livetrans{(\term_1 \nondet \term_2)} = \livetrans{\term_1} \land \livetrans{\term_2}.
  \end{array}
  \]
Here, we assume that \(\Pfun(f_i) \ge \Pfun(f_{i+1})\) for each \(i \in \{1,\dots,n-1\}\),
and \(\munu_i = \nu \) if \(\Pfun(f_i)\) is even and \(\mu\) otherwise.
\end{definition}

The following theorem states the soundness and completeness of the reduction. 
\iffull
See Appendix~\ref{sec:proof-csa} for a proof.
\else
See \cite{ESOP2018full} for a proof.
\fi
\begin{theorem}
\label{theorem:callseq}
  Let \( P \) 
be a program and \( \Pfun \) be a priority assignment for \(P\).
Then \(\models_\CSA (P,\Pfun)\) if and only if \( \lts_0 \models \HESf_{(P,\Pfun),\CSA}\).
\end{theorem}
\begin{example}
For \(\prog_{\mathit{app}}\) and \(\Pfun_{\mathit{app}}\) above,
\(\livetrans{(\prog_{\mathit{app}},\Pfun_{\mathit{app}})} = (\HES, f_b\,5)\), where: \(\HES\) is:
\[
\begin{array}{ll}
 &\APP\,h\,x =_\mu h\,x;\quad
f_b\,x =_\nu (x> 0\imply \APP\,f_a\,(x-1))\land (x\leq 0\imply \APP\,f_b\,5);\\ &
f_a\,x =_\mu (x> 0\imply \APP\,f_a\,(x-1))\land (x\leq 0\imply \APP\,f_b\,5).
\end{array}
\]
Note that \(\lts_0\models \livetrans{(\prog_{\mathit{app}},\Pfun_{\mathit{app}})}\) holds.
\end{example}


\subsection{From temporal verification to call-sequence analysis}
\label{sec:to-callsequence}

This subsection shows a reduction from the temporal verification problem \(\InfTraces(\prog) \cap L \stackrel{?}{=} \emptyset\)
to a call-sequence analysis problem \(\stackrel{?}{\models}_{\CSA}(\prog',\Pfun)\).

For the sake of simplicity, we assume without loss of 
\iffull
generality\footnote{As noted at the beginning of this section,
every finite trace can be turned into an infinite trace by inserting (fresh) dummy events. Then,
\(\InfTraces(\prog)\cap L = \emptyset\) holds if and only if
\(\InfTraces(\prog')\cap L'=\emptyset\), where \(\prog'\) is the program obtained from \(\prog\) by inserting dummy
events, and \(L'\) is the set of all event sequences obtained by inserting dummy events 
into a sequence in \(L\).
}
\else
generality
\fi
that every program \(\prog=(\progd, \term)\) in this section is non-terminating and
every infinite reduction sequence produces infinite events,
so that \(\FullTraces(\prog) = \InfTraces(\prog)\) holds.
We also assume that the \(\omega\)-regular language \(L\) for the temporal verification problem
is specified by using a non-deterministic,
parity word automaton~\cite{Automata}. We recall the definition of
non-deterministic,
parity word automata below.

\begin{definition}[Parity automaton]
\label{def:pwa}
 A \emph{non-deterministic parity word automaton} 
\iffull
(NPW)\footnote{Note that non-deterministic B\"{u}chi automata
may be viewed as instances of non-deterministic parity word automata, where there are only two priorities
\(1\) and \(2\), and accepting and non-accepting states have priorities \(2\) and \(1\) respectively. 
We also note that the classes of
deterministic parity, non-deteterministic parity, and non-deteterministic B\"{u}chi word automata
accept the same class of \(\omega\)-regular languages; here we opt for non-deteterministic parity word automata,
because the translations from the others to NPW are trivial but the other directions may blow up the size of automata.}
\else
\fi
is a quintuple \(\PWA = (Q, \Sigma, \delta, \qinit, \Omega)\) where
 \begin{inparaenum}[(i)]
  \item \(Q\) is a finite set of states;
  \item \(\Sigma\) is a finite alphabet;
  \item \(\delta\), called a transition function, is a \emph{total} map from \(Q\times \Sigma\) to \(2^Q\);
  \item \(\qinit \in Q\) is the initial state; and
  \item \(\Omega \in Q \to \tkchanged{\mathbb{N}}\) is the priority function.
 \end{inparaenum}
 A \emph{run} of \(\PWA\) on an \(\omega\)-word
\(a_0 a_1 \dots \in
\Sigma^{\omega}\) is
 an infinite sequence of states
\(\rho = \rho(0) \rho(1) \dots \in
Q^{\omega}\) such that
 \begin{inparaenum}[(i)]
  \item \(\rho(0) = \qinit\), and
  \item \(\rho(i+1) \in \delta(\rho(i),a_i)\) for each \(i\in\omega\).
 \end{inparaenum}
 An \(\omega\)-word \(w \in \Sigma^{\omega}\) is \emph{accepted} by \(\PWA\) if,
there exists a run \(\rho\) of \(\PWA\) on \(w\) such that
\(
  \MAX\{\Omega(q) \mid q \in \INF(\rho)\} \text{ is even}
 \),
 where \(\INF(\rho)\) is the set of states that occur infinitely often in \(\rho\).
 We write \( \Lang(\PWA) \) for the set of \( \omega \)-words accepted by \( \PWA \).
\end{definition}
For technical convenience, we assume below that \(\delta(q,a)\neq \emptyset\) for every \(q\in Q\) and \(a\in\Sigma\);
this does not lose generality since if \(\delta(q,a)=\emptyset\), we can introduce a new ``dead'' state \(q_{\mathit{dead}}\)
(with priority 1) and change \(\delta(q,a)\) to \(\set{q_{\mathit{dead}}}\). Given a parity automaton \(\PWA\),
we refer to each component of \(\PWA\)
by \(Q_{\PWA}\), \(\Sigma_{\PWA}\), \(\delta_{\PWA}\), \(\qinitA{{\PWA}}\)  and \(\Pfun_{\PWA}\).

\begin{example}\label{ex:pwa}
Consider the automaton \(\PWA_{ab}=(\set{q_a,q_b}, \set{\Lab{a},\Lab{b}}, \delta, q_a, \Omega)\),
where \(\delta\) is as given in Figure~\ref{fig:abomega},
\(\Omega(q_a)=0\), and \(\Omega(q_b)=1\). Then, \(\Lang(\PWA_{ab})=
\overline{(\Lab{a}^*\Lab{b})^\omega} = (\Lab{a}^*\Lab{b})^*\Lab{a}^\omega\).
\end{example}
The goal of this subsection is, given a program \(\prog\) and a parity word automaton
\(\PWA\), to construct another program \(\prog'\) and a priority assignment \(\Pfun\) for \(\prog'\),
such that \(\InfTraces(P)\cap \Lang(\PWA)=\emptyset\) if and only if \(\models_{\CSA}(\prog',\Pfun)\).

Note that a necessary and sufficient condition for
\(\InfTraces(P)\cap\Lang(\PWA)=\emptyset\) is that no trace in \(\InfTraces(P)\)
has a run whose priority sequence satisfies the parity condition; in other words, for every sequence in
\(\InfTraces(P)\), and for every run for the sequence,
 the largest priority that occurs in the associated priority sequence is odd.
As explained at the beginning of this section,
we reduce this condition to a call sequence analysis problem by appropriately duplicating functions
in a given program. For example, recall the program \(P_0\):
\[ \left(\set{f = \ifexp{c}{(\evexp{\Lab{a}}f)}{(\evexp{\Lab{b}}f)}}, f\right).\]
It is translated to \(P_0'\):
\[\begin{array}{ll}
(\{f_b = \ifexp{c}{(\evexp{\Lab{a}}f_a)}{(\evexp{\Lab{b}}f_b)}, \\
 \quad f_a = \ifexp{c}{(\evexp{\Lab{a}}f_a)}{(\evexp{\Lab{b}}f_b)}\},& f_b),
\end{array}
\]
where \(c\) is some (closed) boolean expression.
Since the largest priorities encountered before calling \(f_a\) and \(f_b\) (since the last recursive
call) respectively are \(0\) and \(1\), we assign those priorities plus 1 (to flip odd/even-ness)
to \(f_a\) and \(f_b\) respectively.
Then, the problem of \(\InfTraces(P_0)\cap \Lang(\PWA) = \emptyset\) 
is reduced to
\(\models_{\CSA} (P'_0, \set{f_a\mapsto 1,f_b\mapsto 2})\).
Note here that the priorities of
\(f_a\) and \(f_b\) represent \emph{summaries} of the priorities (plus one) that occur in the run of the automaton until
\(f_a\) and \(f_b\) are respectively called since the last recursive call; thus, the largest priority
of states that occur infinitely often in the run for an infinite trace is equivalent to the largest priority that
occurs infinitely often in the sequence of summaries \((\Pfun(f_1)-1)(\Pfun(f_2)-1)(\Pfun(f_3)-1)\cdots\)
computed from a corresponding call sequence \(f_1f_2f_3\cdots\).

Due to the presence of higher-order functions, the general reduction is more complicated than the example above.
First, we need to replicate not only function symbols, but also arguments. For example, consider
the following variation \(P_1\) of \(P_0\) above:
\[ \left(\set{g\, k = \ifexp{c}{(\evexp{\Lab{a}}k)}{(\evexp{\Lab{b}}k)}, \quad f = g\,f}, \quad f\right).\]
Here, we have just made the calls to \(f\) indirect, by preparing the function \(g\).
Obviously, the two calls to \(k\) in the body of \(g\) must be distinguished from each other, since different priorities
are encountered before the calls. Thus, we duplicate the argument \(k\), and obtain the following
program \(P'_1\):
\[
\begin{array}{l}
 (\set{g\, k_a\,k_b = \ifexp{c}{(\evexp{\Lab{a}}k_a)}{(\evexp{\Lab{b}}k_b)}, f_a = g\,f_a\,f_b,
 f_b = g\,f_a\,f_b}, \\\ 
f_a).
\end{array}
\]
Then, for the priority assignment \(\Pfun = \set{f_a\mapsto 1, f_b\mapsto 2, g\mapsto 1}\),
\(\InfTraces(P_1)\cap \Lang(\PWA_{ab})=\emptyset\) if and only if \(\models_{\CSA} (P_1', \Pfun)\).
Secondly, we need to take into account not only the priorities of states visited by \(\PWA\), but also
the states themselves. For example, if we have a function definition \(f\,h = h(\evexp{\Lab{a}}f\,h)\),
the largest priority encountered before \(f\) is recursively called in the body of \(f\) depends on
 the priorities encountered inside \(h\), \emph{and also} the state of \(\PWA\) when \(h\) uses the argument
\(\evexp{\Lab{a}}f\) (because the state after the \(\Lab{a}\) event depends on the previous state in general).
We, therefore, use \emph{intersection types} (a la Kobayashi and Ong's intersection types
for HORS model checking~\cite{KO09LICS}) to represent summary information on
how each function traverses states of
the automaton, and replicate each function and its arguments for each type.
We thus formalize the translation as an intersection-type-based program transformation;
related transformation techniques are found in \cite{KMS13HOSC,Carayol12LICS,DBLP:conf/csl/TsukadaO14,DBLP:conf/fsttcs/Haddad13,DBLP:conf/csl/GrelloisM15}.

\begin{definition}
 Let \(\PWA = (Q, \Sigma, \delta, \qinit, \Omega)\) be a non-deterministic parity word automaton.
 Let \(q\) and \(m\) range over \(Q\) and the set \(\codom(\Omega)\)
of priorities respectively. 
 The set \(\Types{\PWA}\) of \emph{intersection types}, ranged over by \(\Atype\), is defined by:
 \[
 \begin{array}{l}
  \Atype ::= q \mid \Itype \to \Atype \qquad
  \qquad \Itype ::= \INT \mid \bigwedge_{1 \le i \le k} (\Atype_i, m_i)
 \end{array}
\]
We assume a certain total order \(<\) on \(\Types{\PWA}\times \Nat\), and require that
in \(\bigwedge_{1 \le i \le k} (\Atype_i, m_i)\),
\((\Atype_i,m_i)<(\Atype_j,m_j)\) holds for each \(i<j\).
\end{definition}
We often write \((\Atype_1,m_1)\land \cdots \land(\Atype_k,m_k)\) for
\(\bigwedge_{1 \le i \le k} (\Atype_i, m_i)\), and \(\top\) when \(k=0\).
Intuitively,
the type \(q\) describes expressions of simple type \(\Tunit\), which
may be evaluated when the automaton \(\PWA\) is in the state \(q\)
(here, we have in mind an execution of the \emph{product} of a program and the automaton,
where the latter takes events produced by the program and changes its states).
The type \((\bigwedge_{1 \le i \le k} (\Atype_i, m_i))\to \Atype\) describes functions
that take an argument, use it according to types \(\Atype_1,\ldots,\Atype_k\), and return
a value of type \(\Atype\). Furthermore, the part \(m_i\) describes that the argument may
be used as a value of type \(\Atype_i\) only when the largest priority visited since the function
is called is \(m_i\). For example, given the automaton in Example~\ref{ex:pwa},
the function \(\lambda x.(\evexp{\Lab{a}}x)\) may have types \((q_a,0)\to q_a\)
and \((q_a,0)\to q_b\), because the body may be executed from state \(q_a\) or \(q_b\)
(thus, the return type may be any of them), but \(x\) is used only when the automaton is in state \(q_a\)
and the largest priority visited is \(1\). In contrast,
\(\lambda x.(\evexp{\Lab{b}}x)\) have types \((q_b,1)\to q_a\) and \((q_b,1)\to q_b\).


Using the intersection types above, we shall define a type-based transformation relation of the form
\(\Gamma \pInter \term:\Atype \ESRel \term'\), where \(\term\) and \(\term'\)
are the source and target terms of the transformation, and \(\Gamma\),
called an \emph{intersection type environment}, is
a finite set of type bindings of the form \(x \COL \INT\) or \(x \COL (\Atype, m, m')\).
We allow multiple type bindings for a variable \( x \) except for \( x \COL \INT \) (i.e.~if \( x \COL \INT \in \ITE \), then this must be the unique type binding for \(x \) in \( \ITE \)).
The binding \(x \COL (\Atype, m, m')\) means that \(x\) should be used as a value of type \(\Atype\)
when the largest priority visited is \(m\); \(m'\) is auxiliary information used to record the largest priority
encountered so far.

The transformation relation \(\Gamma \pInter \term:\Atype \ESRel \term'\) is inductively
defined by the rules in Figure~\ref{fig:inter}. (For technical convenience, we have extended
terms with \(\lambda\)-abstractions; they may occur only at top-level function definitions.)
In the figure, \([k]\) denotes the set \(\set{i\in \Nat\mid 1\leq i\leq k}\).
The operation \(\ITE \RaiseP m\) used in the figure is defined by:
\[
 \ITE \RaiseP m \DEF \{x \COL \INT \mid x \COL \INT \in \ITE \} \cup \{x \COL (\Atype, m_1, \MAX(m_2, m)) \mid x \COL (\Atype, m_1, m_2) \in \ITE \}
\]
The operation is applied when the priority \(m\) is encountered,
in which case the largest priority encountered is updated accordingly.
The key rules are \rn{IT-Var}, \rn{IT-Event}, \rn{IT-App}, and \rn{IT-Abs}.
In \rn{IT-Var}, the variable \(x\) is replicated for each type; in the target of the translation,
\(x_{\Atype,m}\) and \(x_{\Atype',m'}\) are treated as different variables if \((\Atype,m)\neq (\Atype',m')\).
The rule \rn{IT-Event} reflects the state change caused by the event \(a\) to the type and the type environment.
Since the state change may be non-deterministic, we transform \(t\) for each of the next states \(q_1,\ldots,q_n\),
and combine the resulting terms with non-deterministic choice. 
The rule \rn{IT-App} and \rn{IT-Abs} replicates function arguments for each type. In addition,
in \rn{IT-App}, the operation \(\ITE\RaiseP m_i\) reflects the fact that \(t_2\) is used as a value of
type \(\Atype_i\) after the priority \(m_i\) is encountered.
The other rules just transform terms in a compositional manner.
If target terms are ignored, the entire rules
are close to those of Kobayashi and Ong's type system for HORS model checking~\cite{KO09LICS}.

\begin{figure}[t]
 \begin{multicols}{2}
\typicallabel{IT-Unit}
 \infrule[IT-Unit]{}{
  \ITE \pInter \unitexp \COL q \ESRel \unitexp
 }

 \infrule[IT-VarInt]{}{
  \ITE, x \COL \INT \pInter x \COL \INT \ESRel x_{\INT}
 }

 \infrule[IT-Var]{}{
  \ITE, x \COL (\Atype, m, m) \pInter x \COL \Atype \ESRel x_{\Atype, m}
 }


 \infrule[IT-Int]{}{
  \ITE \pInter n \COL \INT \ESRel n
 }

 \infrule[IT-Op]{
  \ITE \pInter \term_1 \COL \INT \ESRel \term'_1 \andalso
  \ITE \pInter \term_2 \COL \INT \ESRel \term'_2
 } {
  \ITE \pInter \term_1 \OP \term_2 \COL \INT \ESRel \term'_1 \OP \term'_2
 }
 \infrule[IT-If]{
  \ITE \pInter \term_i \COL \INT \ESRel \term_i' \quad\text{(for each \(i \in [k]\))} \\
  \ITE \pInter \term_{k+1} \COL q \ESRel \term'_{k+1} \\
  \ITE \pInter \term_{k+2} \COL q \ESRel \term'_{k+2}\\
 \seq{t} = t_1,\ldots,t_k\andalso \seq{\term}'=t_1',\ldots,t_k'
 } {
  \ITE \pInter \ifexp{p(\seq{\term})}{\term_{k+1}}{\term_{k+2}} \COL q \\
  \ESRel \ifexp{p(\seq{\term'})}{\term'_{k+1}}{\term'_{k+2}}
 }
 \infrule[IT-Event]{
  \delta_\A(q, \lab) = \{q_1, \dots, q_k\} \\
  \ITE \RaiseP \Pfun_\A(q_i) \pInter \term \COL q_i \ESRel \term'_i \quad\text{(for each \(i \in [k]\))}
 } {
  \ITE
  \pInter
  (\evexp{\lab}{\term}) \COL q  \ESRel (\evexp{\lab}{\term'_1 \nondet \cdots \nondet \term'_k})
 }

 \infrule[IT-NonDet]{
  \ITE \pInter \term_1 \COL q \ESRel \term'_1 \andalso
  \ITE \pInter \term_2 \COL q \ESRel \term'_2
 } {
  \ITE \pInter \term_1 \nondet \term_2 \COL q \ESRel \term'_1 \nondet \term'_2
 }

  \infrule[IT-AppInt]{
  \ITE \pInter \term_1 \COL \INT \to \Atype \ESRel \term'_1 \\
  \ITE \pInter \term_2 \COL \INT \ESRel \term'_2
  } {
  \ITE \pInter \App{\term_1}{\term_2} \COL \Atype \ESRel \App{\term'_1}{\term'_2}
  }

 \infrule[IT-App]{
 \ITE \pInter \term_1 \COL \bigwedge_{1 \le i \le k} (\Atype_i, m_i) \to \Atype \ESRel \term'_1\\
 \ITE \RaiseP m_i \pInter \term_2 \COL \Atype_i \ESRel \term'_{2,i}  \: \text{(for each \(i \in [k]\))}
 } {
 \ITE
 \pInter
 \App{\term_1}{\term_2} : \Atype \ESRel \term'_1 \; \term'_{2,1} \; \dots \; \term'_{2,k}
 }

 \infrule[IT-AbsInt]{
   \ITE, x : \Tint \pInter t : \Atype \Rightarrow t' \andalso
   x \notin \dom(\ITE)
 }{
  \ITE \pInter \lambda x. t : \Tint \to \Atype \Rightarrow \lambda x_{\Tint}. t'
 }

 \infrule[IT-Abs]{
   \ITE \cup\set{ x \COL (\Atype_i, m_i, 0) \mid i\in [k]} \pInter t : \Atype' \Rightarrow t'
   \\
   x \notin \dom(\ITE)
 }{
   \ITE \pInter \lambda x. t : \bigwedge_{1 \le i \le k} (\Atype_i, m_i) \to \Atype'
   \\ \Rightarrow \lambda x_{\Atype_1,m_1} \dots x_{\Atype_k,m_k}. t'
 }
 \end{multicols}
 \caption{Type-based Transformation Rules for Terms}
 \label{fig:inter}
\end{figure}

We now define the transformation for programs.
A \emph{top-level type environment} \( \TopEnv \) is a finite set of type bindings of the form \( x : (\Atype, m) \).
Like intersection type environments, \( \TopEnv \) may have more than one binding for each variable.
We write \( \TopEnv \pInter t : \Atype \) to mean
\( \{ x : (\theta, m, 0) \mid x : (\theta, m) \in \TopEnv \} \pInter t : \Atype \).
For a set \( D \) of function definitions, we write \( \TopEnv \pInter D \Rightarrow D' \) if
\( \dom(D') = \{\, f_{\Atype,m} \mid f : (\Atype, m) \in \TopEnv \,\} \) and \( \TopEnv \pInter D(f) : \Atype \Rightarrow D'(f_{\Atype,m}) \) for every \( f \COL (\Atype, m) \in \TopEnv \).
For a program \( P = (D, t) \), we write \( \TopEnv \pInter P \Rightarrow (P',\Pfun') \) if
\( P' = (D', t') \), \( \TopEnv \pInter D \Rightarrow D' \) and \( \TopEnv \pInter t : \qinit \Rightarrow t' \),
with \(\Pfun'(f_{\Atype,m})=m\changed{+1}\) for each \(f_{\Atype,m}\in \dom(D')\).
We just write \(\pInter P\Rightarrow (P',\Pfun')\) if
\( \TopEnv \pInter P \Rightarrow (P',\Pfun') \) holds for some \(\TopEnv\).

\begin{example}
\label{ex:tr}
Consider the automaton \(\PWA_{ab}\) in Example~\ref{ex:pwa},
and the program \(\prog_2 = (\progd_2, f\,5)\) where \(\progd_2\) consists of the following
function definitions:
\[
\begin{array}{l}
g\; k = (\evexp{\Lab{a}}k)\nondet (\evexp{\Lab{b}}k),\\
f\;x = \ifexp{x>0}{g\,(f(x-1))}{(\evexp{\Lab{b}}f\,5)}.
\end{array}
\]
Let \(\TopEnv\) be:
\iffull
\[
\begin{array}{l}
g\COL ((q_a,0)\land (q_b,1)\to q_a, 0), \quad
g\COL ((q_a,0)\land (q_b,1) \to q_b, 0),\\
f\COL (\INT\to q_a, 0), 
\quad f\COL (\INT\to q_b, 1)
\end{array}
\]
\else
\(\set{g\COL ((q_a,0)\land (q_b,1)\to q_a, 0), 
g\COL ((q_a,0)\land (q_b,1) \to q_b, 0),
f\COL (\INT\to q_a, 0), 
f\COL (\INT\to q_b, 1)}
\).
\fi
Then, \(\TopEnv\pInter \prog_1 \Rightarrow ((\progd'_2, f_{\INT\to q_a, 0}\,5),\Pfun)\) where:
\[
\begin{array}{l}
\progd'_2 = \{
g_{(q_a,0)\land (q_b,1)\to q_a, 0}\; k_{q_a,0}\; k_{q_b,1} = \term_g, \quad
g_{(q_a,0)\land (q_b,1)\to q_b, 0}\; k_{q_a,0}\; k_{q_b,1} = \term_g, \\\qquad\quad
f_{\INT\to q_a, 0}\;x_\INT = \term_{f,q_a}, \quad
f_{\INT\to q_b, 1}\;x_\INT = \term_{f,q_b}\}\\
\term_g = (\evexp{\Lab{a}}k_{q_a,0})\nondet (\evexp{\Lab{b}}k_{q_b,1}),\\
\term_{f,q} =
\ifexp{x_\INT>0}{\\\qquad\qquad g_{(q_a,0)\land (q_b,1)\to q, 0}\,(f_{\INT\to q_a, 0}(x_\INT-1))\,
(f_{\INT\to q_b, 1}(x_\INT-1))\\\qquad\ }
{(\evexp{\Lab{b}}f_{\INT\to q_b, 1}\,5)}, \hfill\mbox{ (for each \(q\in\set{q_a,q_b}\))}\\
\Pfun = \set{g_{(q_a,0)\land (q_b,1)\to q_a, 0}\mapsto 1,
g_{(q_a,0)\land (q_b,1)\to q_b, 0}\mapsto 1, f_{\INT\to q_a, 0}\mapsto 1, 
f_{\INT\to q_b, 1}\mapsto 2}.
\end{array}
\]
\iffull
Appendix~\ref{sec:ex-tr-derivation} shows how \(\term_g\) and \(\term_f\) are derived.
\fi
Notice that \(f\), \(g\), and the arguments of \(g\) have been duplicated. 
Furthermore, whenever \(f_{\Atype,m}\) is called, the largest priority that has been encountered
since the last recursive call is \(m\). For example, in the then-clause of
\(f_{\INT\to q_a, 0}\),
\(f_{\INT\to q_b, 1}(x-1)\) may be called through \(g_{(q_a,0)\land (q_b,1)\to q_a, 0}\).
Since \(g_{(q_a,0)\land (q_b,1)\to q_a, 0}\) uses the second argument only after an event \(\Lab{b}\),
the largest priority encountered is \(1\).
This property is important for the correctness of our reduction.
\end{example}

\iffull
The following theorem claims the soundness and completeness of our reduction.
See 
Appendix~\ref{sec:proof-liveness} 
for a proof.
\else
The following theorems below claim that our reduction is sound and complete, and that
there is an effective algorithm for the reduction: see \cite{ESOP2018full} for proofs.
\fi
\begin{theorem}\label{thm:linevess:eff-sel-sound-and-complete}
  Let \( P\) 
be a program and \( \PWA \) be a parity automaton.
  Suppose that \( \TopEnv \pInter P \Rightarrow (P',\Pfun) \).
Then \( \InfTraces(P) \cap \Lang(\PWA) = \emptyset \) if and only if
\(\models_\CSA (P',\Pfun)\).
\end{theorem}

\iffull
Furthermore, one can effectively find an appropriate transformation.
\fi
\begin{theorem}\label{thm:liveness:eff-sel-effective}
  For every \( P \) and \( \PWA \), one can effectively construct \( \TopEnv \), \( P' \) and \(\Pfun\)
such that \( \TopEnv \pInter P \Rightarrow (P',\Pfun) \).
\end{theorem}
\iffull
See Appendix~\ref{sec:proof-effectiveness} for a proof sketch.
A proof of the above theorem
is given in Appendix~\ref{sec:proof-effectiveness}.
The proof also implies that the reduction from temporal property verification to call-sequence analysis
can be performed in polynomial time.
\else
The proof of Theorem~\ref{thm:liveness:eff-sel-effective} above
 also implies that the reduction from temporal property verification to call-sequence analysis
can be performed in polynomial time.
\fi
Combined with the reduction from call-sequence analysis to HFL model checking, we have thus obtained
a polynomial-time reduction from the temporal verification problem \(\InfTraces(P)\stackrel{?}{\subseteq}
\Lang(\PWA)\) to HFL model checking.

\section{Related Work}
\label{sec:related}

As mentioned in Section~\ref{sec:intro}, our reduction from program verification problems to
HFL model checking problems has been partially inspired by the translation of Kobayashi et al.~\cite{Kobayashi17POPL}
from HORS model checking to HFL model checking. As in their translation (and unlike in
previous applications of HFL model checking~\cite{Viswanathan04,LangeLG14}), our translation 
switches the roles of properties and models (or programs) to be verified. Although a combination of
their translation with Kobayashi's reduction from program verification to HORS model checking~\cite{Kobayashi09POPL,Kobayashi13JACM}
yields an (indirect) translation from \emph{finite-data} programs to pure HFL model checking problems,
the combination does not work for infinite-data programs. In contrast, our translation is sound and complete even for
infinite-data programs. Among the translations in Sections~\ref{sec:reachability}--\ref{sec:liveness},
the translation in Section~\ref{sec:to-callsequence} shares some similarity to their translation, in that functions 
and their arguments are replicated for each priority. The actual translations are
however quite different; ours is type-directed and optimized for a given automaton, whereas
their translation is not. This difference comes from the difference of the goals: the goal of \citeN{Kobayashi17POPL} was to
clarify the relationship between HORS and HFL, hence their translation was designed to be
independent of an automaton.
The proof of the correctness of our translation in Section~\ref{sec:liveness} is much more involved 
\iffull(cf. Appendix~\ref{sec:game-characterization} and \ref{sec:proof-liveness}), 
\fi
due to the need
for dealing with integers. Whilst the proof of \cite{Kobayashi17POPL} could reuse the type-based characterization of 
HORS model checking~\cite{KO09LICS}, we had to generalize arguments in both \cite{KO09LICS} and \cite{Kobayashi17POPL}
to work on infinite-data programs.

Lange et al.~\citeN{LangeLG14} have shown that various process equivalence checking problems (such as
bisimulation and trace equivalence) can be reduced to (pure) HFL model checking problems.
The idea of their reduction is quite different from ours. They reduce processes to LTSs, 
whereas we reduce programs to HFL formulas. 

Major approaches to automated or semi-automated higher-order program verification have been
HORS model checking~\cite{Kobayashi09POPL,KSU11PLDI,Kobayashi13JACM,Ong11POPL,Kuwahara2014Termination,MTSUK16POPL,Watanabe16ICFP},
(refinement) type systems~\cite{Jhala08,Skalka08,UnnoTK13,Terauchi10POPL,Unno09PPDP,zhu_2015,Koskinen14,Hofmann14CSL}, 
Horn clause solving~\cite{Bjorner15,Ramsay17}, and their combinations. As already discussed in Section~\ref{sec:intro},
compared with the HORS model checking approach, our new approach provides more uniform, streamlined methods. 
Whilst the HORS model checking approach is for fully automated verification, our approach enables various degrees
of automation: after verification problems are automatically translated to \HFLZ{} formulas, one can prove them (i)~interactively
 using a proof assistant like 
 Coq (see \iffull Appendix~\ref{sec:coq}\else \cite{ESOP2018full}\fi), (ii)~semi-automatically, by letting users provide 
hints for induction/co-induction and discharging the rest of proof obligations by (some extension of) an SMT solver, or
(iii)~fully automatically by recasting the techniques used in the HORS-based approach; for example, to deal with
the \(\nu\)-only fragment of \HFLZ{}, we can reuse the technique of predicate abstraction~\cite{KSU11PLDI}.
For a more technical comparison between the HORS-based approach and our HFL-based approach, see 
\iffull Appendix~\ref{sec:hors-vs-hfl}.
\else \cite{ESOP2018full}. \fi

As for type-based approaches~\cite{Jhala08,Skalka08,UnnoTK13,Terauchi10POPL,Unno09PPDP,zhu_2015,Koskinen14,Hofmann14CSL}, 
most of the refinement type systems are (i) restricted to safety properties, and/or (ii) incomplete.
A notable exception is the recent work of Unno et al.~\cite{Unno18POPL}, which provides a relatively complete
type system for the classes of properties discussed in Section~\ref{sec:reachability}. Our approach deals with
a wider class of properties (cf. Sections~\ref{sec:path} and \ref{sec:liveness}). Their ``relative completeness'' property
relies on Godel coding of functions, which cannot be exploited in practice.

The reductions from program verification to Horn clause solving have recently been 
advocated~\cite{Bjorner12,Bjorner13,Bjorner15}
or used~\cite{Jhala08,Unno09PPDP} (via refinement type inference problems) by 
a number of researchers.
Since Horn clauses can be expressed in a fragment of HFL without modal operators, fixpoint alternations 
(between \(\nu\) and \(\mu\)), and higher-order predicates,
 our reductions to HFL model checking may be viewed as extensions of those approaches.
Higher-order predicates and fixpoints over them allowed us to provide sound and complete characterizations of
properties of higher-order programs for a wider class of properties. 
Bj{\o}rner et al.~\citeN{Bjorner13} proposed an alternative
approach to obtaining a complete characterization of safety properties, which defunctionalizes higher-order programs
by using algebraic data types and then reduces the problems to (first-order) Horn clauses. A disadvantage
of that approach is that control flow information of higher-order programs is also encoded into algebraic data
types; hence even for finite-data higher-order programs, the Horn clauses obtained by the reduction
belong to an undecidable fragment.  In contrast, our reductions yield pure HFL model checking problems for
finite-data programs.
Burn et al.~\citeN{Ramsay17} have recently advocated the use of \emph{higher-order} (constrained) Horn clauses
for verification of safety properties (i.e., which correspond to the negation of may-reachability properties
discussed in Section~\ref{sec:mayreach} of the present paper) of higher-order programs. They interpret recursion using the least fixpoint
semantics, so their higher-order Horn clauses roughly corresponds to a fragment of the \HFLZ{}
without modal operators and fixpoint alternations. They have not shown a general, concrete reduction from
safety property verification to higher-order Horn clause solving.

The characterization of the reachability problems in Section~\ref{sec:reachability} in terms of
formulas without modal operators is a reminiscent of predicate transformers~\cite{Dijkstra75,Hesselink89}
used for computing the weakest preconditions of imperative programs. In particular, 
\citeN{Blass87} and \citeN{Hesselink89} respectively used least fixpoints to express
weakest preconditions for while-loops and recursions.


\section{Conclusion}
\label{sec:conc}

We have shown that various verification problems for higher-order functional programs can
be naturally reduced to (extended) HFL model checking problems. In all the reductions,
a program is mapped to an HFL formula expressing the property that the behavior of the program is correct.
For developing verification tools for higher-order functional programs,
our reductions allow us to focus on the development of (automated or semi-automated) \HFLZ{} model checking tools
(or, even more simply, theorem provers for \HFLZ{} without modal operators, as the reductions of
Section~\ref{sec:reachability} and \ref{sec:liveness} yield HFL formulas without modal operators). To this end,
we have developed a prototype model checker for pure HFL (without integers), which will be reported in a separate paper.
Work is under way to develop \HFLZ{} model checkers by recasting the techniques~\cite{KSU11PLDI,Kuwahara2014Termination,Kuwahara2015Nonterm,Watanabe16ICFP} developed for the HORS-based approach, which, together with the reductions presented
in this paper, would yield fully automated verification tools.
We have also started building a Coq library for interactively proving \HFLZ{} formulas, 
as briefly discussed in \iffull Appendix~\ref{sec:coq}. \else \cite{ESOP2018full}. \fi
As a final remark, although one may fear that our reductions may map program verification problems to ``harder'' problems due to
the expressive power of \HFLZ{}, 
it is actually not the case at least for the classes of problems in Section~\ref{sec:reachability} and \ref{sec:path},
which use the only alternation-free fragment of \HFLZ{}. The model checking problems for \(\mu\)-only or
\(\nu\)-only \HFLZ{} are semi-decidable and co-semi-decidable respectively, like the source verification problems
of may/must-reachability and their negations of closed programs.

\subsubsection*{Acknowledgment}
We would like to thank anonymous referees for useful comments.
This work was supported by JSPS KAKENHI Grant Number JP15H05706 and JP16K16004.


\newpage
\appendix
\section*{Appendix}
\section{Typing Rules for Programs}
\label{sec:lang-typing}
The type judgments for expressions and programs are of the form
\(\STE\pST \term:\Pest\) and 
\(\STE \pST \prog\), where \(\STE\) is a finite map from variables to types.
The typing rules are shown in Figure~\ref{fig:simple-typing}.
We write \( \pST \prog\) if \(\STE \pST \prog\) for some \(\STE\).
\begin{figure}
\begin{minipage}[t]{0.4\textwidth}
\vspace*{-5ex}
\infrule[LT-Unit]{}{\STE\pST \unitexp:\Tunit}
\rulesp
\infrule[LT-Var]{}{\STE,x\COL\Pest\pST x:\Pest}
\rulesp
\infrule[LT-Int]{}{\STE\pST n:\Tint}
\rulesp
\infrule[LT-Op]{\STE\pST \term_1:\Tint \andalso \STE\pST\term_2:\Tint}{\STE\pST \term_1\OP\term_2:\Tint}
\rulesp
\infrule[LT-Ev]{\STE\pST \term:\Tunit}{\STE\pST \evexp{\lab}{\term}:\Tunit}
\end{minipage}
\begin{minipage}[t]{0.59\textwidth}
\infrule[LT-If]{\arity(p)=k\andalso
  \STE\pST \term'_i:\Tint\mbox{ for each $i\in\set{1,\ldots,k}$}\\
  \STE\pST \term_j:\Tunit\mbox{ for each $j\in\set{1,2}$}}
{\STE\pST \ifexp{p(\term'_1,\ldots,\term'_k)}{\term_1}{\term_2}:\Tunit}
\rulesp
\infrule[LT-App]{\STE\pST \term_1:\Pest\to\Pst\andalso \STE\pST \term_2:\Pest}{\STE\pST \term_1\term_2:\Pst}
\rulesp
\infrule[LT-NonDet]{\STE\pST \term_1:\Tunit\andalso \STE\pST \term_2:\Tunit}{\STE\pST \term_1\nondet\term_2:\Tunit}
\rulesp
\infrule[LT-Prog]{
  \STE = f_1\COL \Pst_1,\ldots, f_n\COL\Pst_n \andalso  \STE\pST \term:\Tunit\\
\STE, \seq{x}_i:\seq{\Pest}_i\pST \term_i:\Tunit \andalso
   \Pst_i = \seq{\Pest}_i\to \Tunit\mbox{ for each $i\in\set{1,\ldots,n}$}\\
}
{\STE \pST (\set{f_1\;\seq{x}_1 = \term_1,\ldots,f_n\;\seq{x}_n=\term_n},\term)}
\end{minipage}
\caption{Typing Rules for Expressions and Programs}
\label{fig:simple-typing}
\end{figure}

\section{Proofs for Section~\ref{sec:reachability}}
\label{sec:proofs-reachability}
\subsection{Proofs for Section~\ref{sec:mayreach}}
\label{sec:proofs-mayreach}

To prove the theorem, we define the reduction relation \(\term\redv{\progd}\term'\) as given in 
 Figure~\ref{fig:red-semantics}. It differs from the labeled transition semantics in that
\(\nondet\) and \(\evexp{\lab}{\cdots}\) are not eliminated; this semantics is more convenient for
establishing the relationship between a program and a corresponding \HFLZ{} formula.
It should be clear that
\(\lab\in \Traces(\progd,\term)\) if and only if \(\term\redvs{\progd} \EC[\evexp{\lab}{\term'}]\) for some \(\term'\).
\begin{figure}
\[
\EC\mbox{(evaluation contexts)} ::= \Hole \mid \EC\nondet \term \mid \term\nondet \EC \mid \evexp{\lab}\EC
\]
\begin{multicols}{2}
 \infrule[R-Fun]{
 f \seq{x} = \termaltu \in \progd \andalso
 |\seq{x}|=|\seq{\term}|
 } {
 \EC[f\;\seq{\term}] \redv{\progd} \EC[[\seq{\term}/\seq{x}]\termaltu]
 }
 \infrule[R-IfT]{
  (\sem{\term'_1},\ldots,\sem{\term'_k})\in\sem{p} 
 } {
 \EC[\ifexp{p(\term'_1,\dots, \term'_k)}{\term_1}\term_2] \redv{\progd} \EC[\term_1]
 }
 \infrule[R-IfF]{
  (\sem{\term'_1},\ldots,\sem{\term'_k})\not\in \sem{p} 
 } {
 \EC[\ifexp{p(\term'_1,\dots, \term'_k)}{\term_1}\term_2] \redv{\progd} \EC[\term_2]
 }
\end{multicols}
\caption{Reduction Semantics}
\label{fig:red-semantics}
\end{figure}

We shall first prove the theorem for recursion-free programs. 
Here, a program \(\prog=(\progd,\term)\) is \emph{recursion-free} if
the transitive closure of the relation \(\set{(f_i,f_j)\in \dom(\progd)\times \dom(\progd)\mid 
\mbox{\(f_j\) occurs in \(\progd(f_i)\)}}\) is irreflexive. To this end, we prepare a few lemmas.

The following lemma says that the semantics of \HFLZ{} formulas is preserved by reductions of
the corresponding programs.
\begin{lemma}
\label{lem:mayreachsem-preserved-by-reduction}
Let \((\progd,\term)\) be a program and \(\lts\) be an LTS. If \(\term\redv{\progd}\term'\), then
\(\Sem{\lts}{\trMay{(\progd,\term)}} = \Sem{\lts}{\trMay{(\progd,\term')}}\).
\end{lemma}

\begin{proof} 
Let \(\progd = \set{f_1\,\seq{x}_1=\term_1,\ldots,f_n\,\seq{x}_n=\term_n}\),
and \((F_1,\ldots,F_n)\) be the least fixpoint of
\[\lambda (X_1,\ldots,X_n).(\sem{\lambda \seq{x}_1.\trMay{\term_1}}(\set{f_1\mapsto X_1,\ldots,f_n\mapsto X_n}),\ldots,
\sem{\lambda \seq{x}_n.\trMay{\term_n}}(\set{f_1\mapsto X_1,\ldots,f_n\mapsto X_n})).\]
By the Beki\'{c} property, \(\sem{\trMay{(\progd,\term)}}
= \sem{\trMay{\term}}(\set{f_1\mapsto F_1,\ldots,f_n\mapsto F_n})\).
Thus, it suffices to show that \(\term\redv{\progd}\term'\) implies
\(\sem{\trMay{\term}}(\HFLenv)=
\sem{\trMay{\term'}}(\HFLenv)\)
for \(\HFLenv=\set{f_1\mapsto F_1,\ldots,f_n\mapsto F_n}\).
We show it by case analysis on the rule used for deriving \(\term\redv{\progd}\term'\).
\begin{itemize}
\item Case \rn{R-Fun}: In this case, \(\term = E[f_i\,\seq{s}]\) and \(\term'=E[[\seq{s}/\seq{x}_i]t_i]\).
Since \((F_1,\ldots,F_n)\) is a fixpoint, we have:
\[
\begin{array}{l}
\sem{f_i\,\seq{s}}(\HFLenv)
 = F_i(\sem{\seq{s}}(\HFLenv))\\
 = \sem{\lambda \seq{x}_i.t_i}(\HFLenv) (\sem{\seq{s}}(\HFLenv))\\
 = \sem{[\seq{s}/\seq{x}_i]t_i}(\HFLenv)
\end{array}
\]
Thus, we have 
\(\sem{\trMay{\term}}(\HFLenv)=
\sem{\trMay{\term'}}(\HFLenv)\) as required.
\item Case \rn{R-IfT}:
In this case, \(\term = E[\ifexp{p(\termalt_1',\ldots,\termalt_k')}{\termalt_1}{\termalt_2}]\) and \(\term' = E[\termalt_1]\)
with \((\sem{\termalt_1'},\ldots,\sem{\termalt_k'})\in \sem{p}\). 
Thus, \(\trMay{\term}=\trMay{E}[(p(\termalt_1',\ldots,\termalt_k')\land \trMay{\termalt_1})\lor (\neg p(\termalt_1',\ldots,\termalt_k')\land \trMay{\termalt_2})]\).
Since \((\sem{\termalt_1'},\ldots,\sem{\termalt_k'})\in \sem{p}\), 
\((\sem{\termalt_1'},\ldots,\sem{\termalt_k'})\not \in \sem{\neg p}\). Thus,
\(\sem{\trMay{\term}}(\HFLenv) = 
\sem{\trMay{E}[(\TRUE\land \trMay{\termalt_1})\lor (\FALSE\land \trMay{\termalt_2})]}(\HFLenv)
= 
\sem{\trMay{E}[\trMay{\termalt_1}]}(\HFLenv)\).
We have thus 
\(\sem{\trMay{\term}}(\HFLenv)=
\sem{\trMay{\term'}}(\HFLenv)\) as required.
\item Case \rn{R-IfF}:
Similar to the above case.
\end{itemize}
\qed
\end{proof}

The following lemma says that Theorem~\ref{th:mayreach} holds for programs in normal form.
\begin{lemma}
\label{lem:mayreach-base}
Let \((\progd,\term)\) be a program and \(\term\mathbin{\nonredv{\progd}}\). Then,
\(\lts_0 \models \trMay{(\progd,\term)}\) if and only if \(\term=\EC[\evexp{\lab}{\term'}]\) for some
evaluation context \(E\) and \(\term'\).
\end{lemma}

\begin{proof} 
The proof proceeds by induction on the structure of \(\term\).
By the condition \(\term\nonredv{\progd}\) and the (implicit) assumption that
\(\pST (\progd,\term)\), \(\term\) is generated by the following grammar:
\[
\term ::= \Vunit \mid \evexp{\lab}\term' \mid \term_1\nondet\term_2 .
\]
\begin{itemize}
\item Case \(\term=\Vunit\): The result follows immediately, 
as \(\term\) is not of the form \(\EC[\evexp{\lab}{\term'}]\), and 
\(\trMay{\term}=\FALSE\).
\item Case \(\term=\evexp{\lab}\term'\): 
The result follows immediately, as
\(\term\) is of the form \(\EC[\evexp{\lab}{\term'}]\), and 
\(\trMay{\term}=\TRUE\).
\item Case \(\term_1\nondet\term_2\):
Because \(\trMay{(\term_1\nondet\term_2)} = \trMay{\term_1}\lor \trMay{\term_2}\),
\(\lts_0 \models \trMay{(\progd,\term)}\) if and only if
\(\lts_0 \models \trMay{(\progd,\term_i)}\) for some \(i\in\set{1,2}\).
By the induction hypothesis, the latter is equivalent to 
the property that \(\term_i\) is of the form \(\EC[\evexp{\lab}{\term'}]\) for some \(i\in\set{1,2}\),
which is equivalent to the property that 
\(\term\) is of the form \(\EC'[\evexp{\lab}{\term'}]\).
\end{itemize}
\qed
\end{proof}

The following lemma says that Theorem~\ref{th:mayreach} holds for recursion-free programs;
this is an immediate corollary of Lemmas~\ref{lem:mayreachsem-preserved-by-reduction} and
\ref{lem:mayreach-base}, and the strong normalization property of the simply-typed \(\lambda\)-calculus.
\begin{lemma}
\label{lem:mayreach-for-recfree-program}
Let \(\prog\) be a recursion-free program. Then, \(\lab\in\Traces(\prog)\) if and only if
\(\lts_0 \models \HESf_{\prog,\may}\) for \(\lts_0 = (\set{\stunique},\emptyset,\emptyset,\stunique)\).
\end{lemma}

\begin{proof} 
Since \(\prog=(\progd,\term_0)\) is recursion-free, there exists a finite, normalizing reduction sequence
\(\term_0 \redvs{\progd} \term \nonredv{\progd}\). We show the required property by induction on
the length \(n\) of this reduction sequence.
\begin{itemize}
\item Case \(n=0\): Since \(\term_0\nonredv{\progd}\),
\(\lab\in\Traces(\prog)\) if and only if \(\term_0=E[\evexp{\lab}{\term}]\) for some \(E\) and \(\term\).
Thus, the result follows immediately from Lemma~\ref{lem:mayreach-base}.
\item Case \(n>0\): In this case, 
\(\term_0 \redv{\progd} \term_1\redvs{\progd} \term \).
By the induction hypothesis, 
\(\lab\in\Traces(\progd,\term_1)\) if and only if 
\(\lts_0 \models \HESf_{(\progd,\term_1),\may}\).
Thus, by the definition of the reduction semantics and Lemma~\ref{lem:mayreachsem-preserved-by-reduction},
\(\lab\in\Traces(\progd,\term_0)\) if and only if 
\(\lab\in\Traces(\progd,\term_1)\), if and only if 
\(\lts_0 \models \HESf_{(\progd,\term_1),\may}\), if and only if 
\(\lts_0 \models \HESf_{(\progd,\term_0),\may}\).
\end{itemize}
\qed
\end{proof}

To prove Theorem~\ref{th:mayreach} for arbitrary programs, we use the fact that the semantics
of \(\trMay{P}\) may be approximated by \(\trMay{P^{(i)}}\), where \(P^{(i)}\) is the recursion-free
program obtained by unfolding recursion functions \(i\) times (a more formal definition will be given later).
To guarantee the correctness of this finite approximation, we need to introduce
 a slightly non-standard notion of (\(\omega\)-)continuous functions below.
\begin{definition}
\label{def:continuity}
For an LTS 
\(\lts = (\St{}, \Act{}, \mathbin{\TR}, \st_\init)\) and a type \(\etyp\),
the set of \emph{continuous} elements \(\Cont_{\lts,\etyp}\subseteq \D_{\lts,\etyp}\)
and the equivalence relation
\(\eqcont{\lts,\etyp}\subseteq
\Cont_{\lts,\etyp}\times \Cont_{\lts,\etyp}\)
are  defined by induction on \(\etyp\) as follows.
\[
\begin{array}{l}
\Cont_{\lts,\typProp} = \D_{\lts,\typProp}\qquad
\Cont_{\lts,\INT} = \D_{\lts,\INT}\\
\Cont_{\lts,\etyp\to\typ} = \set{f\in \D_{\lts,\etyp\to\typ}\mid 
{  \forall x_1,x_2\in \Cont_{\lts,\etyp}.(x_1\eqcont{\lts,\etyp}x_2\imply
  f(x_1)\eqcont{\lts,\typ}f(x_2))}\\\qquad\qquad\qquad\qquad
\land 
\forall \set{y_i}_{i\in \omega}\in \Incseq{\Cont_{\lts,\etyp}}. 
f(\bigsqcup_{i\in\omega} y_i)\eqcont{\lts,\typ}\bigsqcup_{i\in\omega}f(y_i)
}.\\
\eqcont{\lts,\typProp} \,=\, \set{(x,x)\mid x\in \Cont_{\lts,\typProp}}\qquad
  \eqcont{\lts,\INT} \,=\, \set{(x,x)\mid x\in \Cont_{\lts,\INT}}\\
    \eqcont{\lts,\etyp\to\typ} \,=\,
    \set{(f_1,f_2)\mid f_1,f_2\in \Cont_{\lts,\etyp\to\typ}\\\qquad\qquad\qquad\qquad \land
      \forall x_1, x_2\in \Cont_{\lts,\etyp}.(x_1\eqcont{\lts,\etyp}x_2\imply
      f_1(x_1)\eqcont{\lts,\typ}f_2(x_2)}.
\end{array}
\]
Here, \(\Incseq{\Cont_{\lts,\etyp}}\) denotes the set of increasing infinite sequences
\(a_0\sqsubseteq a_1\sqsubseteq a_2 \sqsubseteq \cdots\) consisting of elements of 
\(\Cont_{\lts,\etyp}\).
We just write \(\Eqcont{}\) for \(\eqcont{\lts,\etyp}\)  when \(\lts\) and \(\etyp\) are clear
from the context.
\end{definition}
\begin{remark}
Note that we require that a continuous function returns a continuous element only if 
its argument is. To see the need for this requirement, consider an LTS with a singleton state set
\(\set{\st}\), and the function:
\(f = \lambda x\COL((\INT\to\typProp)\to\typProp).\lambda y\COL (\INT\to\typProp).x\,y\).
One may expect that \(f\) is continuous in the usual sense (i.e., \(f\) preserves the limit), but for
the function \(g\) defined by \[g(p) = \left\{ 
\begin{array}{ll}
\set{\st} & \mbox{if $\st\in p(n)$ for every \(n\geq 0\)}\\
\emptyset & \mbox{otherwise}
\end{array}
\right.,\]
 \(f(g)\) is \emph{not} continuous. In fact, 
let \(p_i\) be \(\set{(n,\set{\st}) \mid 0\leq n\leq i}\cup
\set{(n,\emptyset)\mid n<0\lor n>i}\). Then \(f(g)(p_i) = \emptyset\) for every \(i\)
but \(f(g)(\sqcup_{i\in\omega} p_i)=\set{\st}\).
The (non-continuous) function \(g\) above can be expressed by
\((\nu X.\lambda n.\lambda p.(p(n) \land (X\, (n+1)\, p))) 0\).
\qed
\end{remark}

\begin{lemma}
  \label{lem:continuity-of-limit}
  If \(\set{x_i}_{i\in\omega}, \set{y_i}_{i\in\omega} \in \Incseq{\Cont_{\lts,\etyp}}\)
  and \(x_i\eqcont{\lts,\etyp} y_i\) for each \(i\in \omega\), then
  \(\LUB_{i\in\omega} x_i \eqcont{\lts,\etyp}\LUB_{i\in\omega} y_i\).
\end{lemma}
\begin{proof}
  The proof proceeds by induction on \(\etyp\).
  The base case, where \(\etyp=\typProp\) or \(\etyp=\INT\), is trivial,
  as \(\Cont_{\lts,\etyp}=\D_{\lts,\etyp}\) and \( \eqcont{\lts,\etyp}\) is the identity relation.
  Let us consider the induction step, where \(\etyp = \etyp_1\to\typ\).
  We first check that \(\LUB_{i\in\omega} x_i\in \Cont_{\lts,\etyp}\).
  To this end, suppose \(z_1\eqcont{\lts,\etyp_1}z_2\).
  By the continuity of \(x_i\) and the assumption \(z_1\eqcont{\lts,\etyp_1}z_2\),
  we have \(x_i z_1 \eqcont{\lts,\typ}x_i z_2\) for each \(i\).
  By the induction hypothesis, we have
  \(\LUB_{i\in\omega} (x_i z_1) \eqcont{\lts,\typ}\LUB_{i\in\omega} (x_i z_2)\).
  Therefore, we have:
  \[(\LUB_{i\in\omega} x_i)z_1 =
  \LUB_{i\in\omega} (x_i z_1)\eqcont{\lts,\typ}\LUB_{i\in\omega} (x_i z_2)
  = (\LUB_{i\in\omega} x_i)z_2\]
  as required.
  To check the second condition for \(\LUB_{i\in\omega} x_i\in \Cont_{\lts,\etyp}\), suppose that
  \(\set{z_i}_{i\in\omega} \in \Incseq{\Cont_{\lts,\etyp_1}}\).
  We need to show \(
  (\LUB_{i\in\omega} x_i)(\LUB_{j\in\omega}z_j) \eqcont{\lts,\typ} \LUB_{j\in\omega} ((\LUB_{i\in\omega} x_i)z_j)\).
  By the continuity of \(x_i\) and the induction hypothesis, we have indeed:
  \[
  \begin{array}{l}
    (\LUB_{i\in\omega} x_i)(\LUB_{j\in\omega}z_j)
     = \LUB_{i\in\omega} (x_i(\LUB_{j\in\omega}z_j))\\
     \eqcont{\lts,\typ} \LUB_{i\in\omega} (\LUB_{j\in\omega}(x_i z_j))
     = \LUB_{j\in\omega} (\LUB_{i\in\omega} (x_i z_j))
     = \LUB_{j\in\omega} ((\LUB_{i\in\omega} x_i) z_j).
  \end{array}
  \]
  Thus, we have proved \(\LUB_{i\in\omega} x_i\in \Cont_{\lts,\etyp}\).
  The proof of \(\LUB_{i\in\omega} y_i\in \Cont_{\lts,\etyp}\) is the same.

  It remains to check that \(z\eqcont{\lts,\etyp_1}w\) implies
  \((\LUB_{i\in\omega}x_i)z \eqcont{\lts,\typ}(\LUB_{i\in\omega}y_i)w\).
  Suppose \(z\eqcont{\lts,\etyp_1}w\).
  Then we have:
  \[\textstyle (\LUB_{i\in\omega}x_i)z =
    \LUB_{i\in\omega} (x_i z) 
    \eqcont{\lts,\typ} \LUB_{i\in\omega} (y_i w)
    = (\LUB_{i\in\omega}y_i)w.\]
    as required.
    Note that \(x_i z\eqcont{\lts,\typ}y_i w\) follows from
    the assumptions \(x_i\eqcont{\lts,\etyp}y_i\) and \(z\eqcont{\lts,\etyp_1}w\),
     and then we have applied the induction hypothesis to obtain 
    \(\LUB_{i\in\omega} (x_i z) 
    \eqcont{\lts,\typ} \LUB_{i\in\omega} (y_i w)\).
    This completes the proof for the induction step. \qed
\end{proof}

The following lemma guarantees the continuity of the functions expressed by fixpoint-free \HFLZ{} formulas.
\begin{lemma}[continuity of fixpoint-free functions]
\label{lem:continuity}
Let \(\lts\) be an LTS.
If \(\form\) is a closed, fixpoint-free \HFLZ{} formula of type \(\typ\), then 
\(\sem{\form} \in \Cont_{\lts,\typ}\).
\end{lemma}

\begin{proof} 
We write \(\csem{\HFLte}\) for the set of 
valuations: \(\set{\HFLenv\in \sem{\HFLte} \mid \HFLenv(f)\in \Cont_{\lts,\etyp}\mbox{ for each
    \(f\COL \etyp\in \HFLte\)}}\), and \(\eqcont{\lts,\HFLte}\) for:
\[\set{(\HFLenv_1,\HFLenv_2) \in \csem{\HFLte}
\times \csem{\HFLte}\mid
   \HFLenv_1(x)\eqcont{\lts,\etyp}\HFLenv_2(x)\mbox{ for every \(x\COL\etyp\in\HFLte\)}}.\]
We show the following property by induction on the derivation of
\(\HFLte \p \form:\etyp\):
\begin{quote}
If \(\form\) is fixpoint-free and
\(\HFLte \p \form:\etyp\), then 
\begin{itemize}
\item[(i)] \(\HFLenv_1\eqcont{\lts,\HFLte}\HFLenv_2\) imply
  \(\sem{\HFLte\p \form:\etyp}(\HFLenv_1)\eqcont{\lts,\etyp}
  \sem{\HFLte\p \form:\etyp}(\HFLenv_2)\).
\item[(ii)] For any increasing sequence of valuations \(\HFLenv_0 \sqsubseteq \HFLenv_1\sqsubseteq \HFLenv_2 \sqsubseteq 
\cdots \) such that \(\HFLenv_i\in \csem{\HFLte}\) for each \(i\in\omega\), 
\(\sem{\HFLte\p \form:\etyp}(\sqcup_{i\in\omega}\HFLenv_i) \eqcont{\lts,\etyp} \LUB_{i\in\omega}
\sem{\HFLte\p \form:\etyp}(\HFLenv_i)\).
\end{itemize}
\end{quote}
Then, the lemma would follow as a special case, where \(\HFLte = \emptyset\).
We perform case analysis on the last rule used for deriving \(\HFLte \p \form:\etyp\).
We discuss only the main cases;
the other cases are similar or straightforward.
\begin{itemize}
\item Case \rn{HT-Var}: In this case, \(\form=X\) and \(\HFLte=\HFLte',X\COL \etyp\). 
The condition (i) follows immediately by:
\(\sem{\HFLte\p \form:\etyp}(\HFLenv_1)= \HFLenv_1(X)\eqcont{\lts,\etyp}\HFLenv_2(X)
= \sem{\HFLte\p \form:\etyp}(\HFLenv_2)\).
To see (ii), suppose 
\(\HFLenv_0 \sqsubseteq \HFLenv_1\sqsubseteq \HFLenv_2 \sqsubseteq 
\cdots \), with \(\HFLenv_i\in \csem{\HFLte}\) for each \(i\in\omega\). 
Then,  we have
\[
\sem{\HFLte\p \form:\etyp}(\sqcup_{i\in\omega}\HFLenv_i) 
= (\sqcup_{i\in\omega}\HFLenv_i)(X)
= \sqcup_{i\in\omega}(\HFLenv_i(X))
= \sqcup_{i\in\omega}(\sem{\HFLte\p \form:\etyp}(\HFLenv_i)).
\]
By Lemma~\ref{lem:continuity-of-limit},
\(\sqcup_{i\in\omega}(\HFLenv_i(X))\in \Cont_{\lts,\etyp}\). We have thus
\(\sem{\HFLte\p \form:\etyp}(\sqcup_{i\in\omega}\HFLenv_i) 
\eqcont{\lts,\etyp}
\sqcup_{i\in\omega}(\sem{\HFLte\p \form:\etyp}(\HFLenv_i))\) as required.
\item Case \rn{HT-Some}:
In this case, \(\form=\Some{\lab}\form'\) with \(\HFLte\p \form:\typProp\) and \(\etyp=\typProp\).
The condition (i) is trivial since \(\Cont_{\lts,\typProp} = \D_{\lts,\typProp}\).
We also have the condition (ii) by:
\[
\begin{array}{l}
\sem{\HFLte\p \Some{\lab}{\form'}:\etyp}(\sqcup_{i\in\omega}\HFLenv_i) 
=\set{\st\mid \exists \st'\in \sem{\HFLte\pHFL\form':\typProp}(\sqcup_{i\in\omega}\HFLenv_i).\ \st\Ar{a}\st'}\\
= \set{\st\mid \exists \st'\in \sqcup_{i\in\omega}(\sem{\HFLte\pHFL\form':\typProp}(\HFLenv_i)).\ \st\Ar{a}\st'}
\qquad\hfill \mbox{ (by induction hypothesis)}\\
= \sqcup_{i\in\omega}
\set{\st\mid \exists \st'\in \sem{\HFLte\pHFL\form':\typProp}(\HFLenv_i).\ \st\Ar{a}\st'}
= \sqcup_{i\in\omega}\sem{\HFLte\pHFL\Some{\lab}\form':\typProp}(\HFLenv_i).
\end{array}
\]
\item Case \rn{HT-All}:
In this case, \(\form=\All{\lab}\form'\) with \(\HFLte\p \form:\typProp\) and \(\etyp=\typProp\).
The condition (i) is trivial.
The condition (ii) follows by:
\[
\begin{array}{l}
\sem{\HFLte\p \All{\lab}{\form'}:\etyp}(\sqcup_{i\in\omega}\HFLenv_i) 
=\set{\st\mid \forall \st'\in \St.\st\Ar{a}\st'\imply \st'\in \sem{\HFLte\pHFL\form':\typProp}(\sqcup_{i\in\omega}\HFLenv_i)}\\
=\set{\st\mid \forall \st'\in \St.\st\Ar{a}\st'\imply \st'\in 
\sqcup_{i\in\omega}(\sem{\HFLte\pHFL\form':\typProp}(\HFLenv_i))}\\
\qquad\hfill \mbox{ (by induction hypothesis)}\\
= \sqcup_{i\in\omega}
\set{\st\mid \forall \st'\in \St.\st\Ar{a}\st'\imply \st'\in 
\sem{\HFLte\pHFL\form':\typProp}(\HFLenv_i)}
\hfill \mbox{(*)}\\
= \sqcup_{i\in\omega}\sem{\HFLte\pHFL\All{\lab}\form':\typProp}(\HFLenv_i).
\end{array}
\]
To see the direction \(\subseteq\) in step (*), suppose 
\(\forall \st'\in \St.\st\Ar{a}\st'\imply \st'\in 
\sqcup_{i\in\omega}(\sem{\HFLte\pHFL\form':\typProp}(\HFLenv_i))\) holds. 
Since \(\St\) is finite, the set \(\set{\st'\mid \st\Ar{a}\st'}\) is a finite set
\(\set{\st_1,\ldots,\st_k}\).  
For each \(j\in\set{1,\ldots,k}\), there exists \(i_j\) such that
\(\st_j \in \sem{\HFLte\pHFL\form':\typProp}(\HFLenv_{i_j})\). Let \(i' = \max(i_1,\ldots,i_k)\).
Then we have \(\st'\in \sem{\HFLte\pHFL\form':\typProp}(\HFLenv_{i'})\)
for every 
\(\st'\in\set{\st_1,\ldots,\st_k}\). We have thus 
\(\forall s'\in \St.\st\Ar{a}\st'\imply s'\in 
\sem{\HFLte\pHFL\form':\typProp}(\HFLenv_{i'})\), which implies \(\st\) belongs to the
set in the righthand side of (*).

To see the converse (i.e., \(\supseteq\)), suppose 
\(\forall \st'\in \St.\st\Ar{a}\st'\imply \st'\in 
\sem{\HFLte\pHFL\form':\typProp}(\HFLenv_i)\) for some \(i\in\omega\).
Then, 
\(\forall \st'\in \St.\st\Ar{a}\st'\imply \st'\in 
\sqcup_{i\in\omega}(\sem{\HFLte\pHFL\form':\typProp}(\HFLenv_i))\) follows immediately.
\item Case \rn{HT-Abs}:
In this case, \(\form = \lambda X\COL\etyp_1.\form'\), with \(\HFLte,X\COL\etyp_1\p \form':\typ\)
and \(\etyp = \etyp_1\to \typ\). 
To prove the condition (i), suppose \(\HFLenv_1\eqcont{\lts,\HFLte}\HFLenv_2\). 
Let \(f_j=\sem{\HFLte\p\lambda X\COL\etyp_1.\form':\etyp_1\to\typ}(\HFLenv_j)\) for \(j\in\set{1,2}\).
We first check \(f_j\in \Cont_{\lts,\etyp_1\to\typ}\).
Suppose \(x_1\eqcont{\lts,\etyp_1}x_2\). 
Then, by the induction hypothesis,
we have
\[
\qquad f_j\,x_1 = \sem{\HFLte,X\COL\etyp_1\p \form':\typ}(\HFLenv_j\set{X\mapsto x_1})
\eqcont{\lts,\typ}\sem{\HFLte,X\COL\etyp_1\p \form':\typ}(\HFLenv_j\set{X\mapsto x_2})
= f_j\,x_2.\]
To check the second condition for \(f_j\in \Cont_{\lts,\etyp}\), 
let \(\set{y_i}_{i\in \omega}\in \Incseq{\Cont_{\lts,\etyp_1}}\). 
Then we have:
\[
\begin{array}{l}
f_j(\sqcup_{i\in\omega} y_i) = \sem{\HFLte,X\COL\etyp_1\p\form':\typ}(\HFLenv_j\set{X\mapsto \sqcup_{i\in\omega} y_i})\\
= \sem{\HFLte,X\COL\etyp_1\p\form':\typ}(\sqcup_{i\in\omega}(\HFLenv_j\set{X\mapsto y_i}))\\
\Eqcont{} \sqcup_{i\in\omega}\sem{\HFLte,X\COL\etyp_1\p\form':\typ}(\HFLenv_j\set{X\mapsto y_i})
\qquad \mbox{(by the induction hypothesis)}\\
= \sqcup_{i\in\omega}f_j(y_i)
\end{array}
\]
as required.

To show
\(f_1\eqcont{\lts,\etyp_1\to\typ}f_2\), assume again that
\(x_1\eqcont{\lts,\etyp_1}x_2\).
Then, \(\HFLenv_1\set{X\mapsto x_1}\eqcont{\HFLte,X\COL\etyp_1}\HFLenv_2\set{X\mapsto x_2}\).
Therefore, by the induction hypothesis, we have 
\[\qquad f_1(x) = \sem{\HFLte,X\COL\etyp_1\p \form':\typ}(\HFLenv_1\set{X\mapsto x_1})
\eqcont{\lts,\typ}
\sem{\HFLte,X\COL\etyp_1\p \form':\typ}(\HFLenv_2\set{X\mapsto x_2})
=f_2(x).\]
This completes the proof of the condition (i).

To prove the condition (ii), suppose 
\(\HFLenv_0 \sqsubseteq \HFLenv_1\sqsubseteq \HFLenv_2 \sqsubseteq 
\cdots \) with \(\HFLenv_i\in \csem{\HFLte}\) for each \(i\in\omega\).
We need to show
\( \sem{\HFLte\p \form:\etyp}(\sqcup_{i\in\omega} \HFLenv_i) \eqcont{\lts,\etyp} \sqcup_{i\in\omega} \sem{\HFLte\p\form:\etyp}(\HFLenv_i) \).
We first check that both sides of the equality belong to \(\Cont_{\lts,\etyp}\).
By Lemma~\ref{lem:continuity-of-limit},
\(\sqcup_{i\in\omega} \HFLenv_i \in \csem{\HFLte}\). Thus, by
the condition (i) (where we set both \(\HFLenv_1\) and \(\HFLenv_2\) to
\(\sqcup_{i\in\omega} \HFLenv_i\)), we have
\(\sem{\HFLte\p \form:\etyp}(\sqcup_{i\in\omega} \HFLenv_i) \eqcont{\lts,\etyp}
\sem{\HFLte\p \form:\etyp}(\sqcup_{i\in\omega} \HFLenv_i) \),
which implies
\(\sem{\HFLte\p \form:\etyp}(\sqcup_{i\in\omega} \HFLenv_i) \in \Cont_{\lts,\etyp}\).
For the righthand side, by the condition (i), we have
\(\sem{\HFLte\p\form:\etyp}(\HFLenv_i)  \in \Cont_{\lts,\etyp}\) for each \(i\).
By Lemma~\ref{lem:continuity-of-limit}, we have
\(\sqcup_{i\in\omega} \sem{\HFLte\p\form:\etyp}(\HFLenv_i) \) as required.

It remains to check that \(x\eqcont{\lts,\etyp_1} y\) implies
\( \sem{\HFLte\p \form:\etyp}(\sqcup_{i\in\omega} \HFLenv_i)(x) \eqcont{\lts,\typ}
\sqcup_{i\in\omega} \sem{\HFLte\p\form:\etyp}(\HFLenv_i)(y) \).
Suppose \(x\eqcont{\lts,\etyp_1} y\).
  Then we have
\[
\begin{array}{l}
\sem{\HFLte\p \form:\etyp}(\sqcup_{i\in\omega}\HFLenv_i)(x)
= 
\sem{\HFLte,X\COL\etyp_1\p \form':\etyp}((\sqcup_{i\in\omega}\HFLenv_i)\set{X\mapsto x})\\
= 
\sem{\HFLte,X\COL\etyp_1\p \form':\etyp}(\sqcup_{i\in\omega}(\HFLenv_i\set{X\mapsto x}))\\
\Eqcont{} 
\sem{\HFLte,X\COL\etyp_1\p \form':\etyp}(\sqcup_{i\in\omega}(\HFLenv_i\set{X\mapsto y}))\\
\hfill \mbox{(by the induction hypothesis, (i))}\\
\Eqcont{} 
\sqcup_{i\in\omega}\sem{\HFLte,X\COL\etyp_1\p \form':\etyp}(\HFLenv_i\set{X\mapsto y})\\
\hfill \mbox{(by the induction hypothesis, (ii))}\\
= 
\sqcup_{i\in\omega}\sem{\HFLte\p \form:\etyp}(\HFLenv_i)(y)
\end{array}
\]
as required.


\item Case \rn{HT-App}: In this case, 
\(\form = \form_1\form_2\) and \(\etyp=\typ\), with \(\HFLte\p \form_1:\etyp_2\to \typ\),
\(\HFLte\p \form_2:\etyp_2\). The condition (i) follows immediately from the induction hypothesis.
To prove the condition (ii), suppose 
\(\HFLenv_0 \sqsubseteq \HFLenv_1\sqsubseteq \HFLenv_2 \sqsubseteq 
\cdots \) with \(\HFLenv_i\in \csem{\HFLte}\) for each \(i\in\omega\). Then we have
\[
\begin{array}{l}
\sem{\HFLte\p \form_1\form_2:\typ}(\sqcup_{i\in\omega}\HFLenv_i) 
= 
\sem{\HFLte\p \form_1:\etyp_2\to\typ}(\sqcup_{i\in\omega}\HFLenv_i) 
\sem{\HFLte\p \form_2:\etyp_2}(\sqcup_{i\in\omega}\HFLenv_i) \\
\Eqcont{} 
(\sqcup_{i\in\omega}\sem{\HFLte\p \form_1:\etyp_2\to\typ}(\HFLenv_i))
(\sqcup_{j\in\omega}\sem{\HFLte\p \form_2:\etyp_2}(\HFLenv_j)) \\
\hfill \mbox{ (by the induction hypothesis)}\\
= 
\sqcup_{i\in\omega}
(\sem{\HFLte\p \form_1:\etyp_2\to\typ}(\HFLenv_i)
(\sqcup_{j\in\omega}\sem{\HFLte\p \form_2:\etyp_2}(\HFLenv_j))\\
\Eqcont{}
\sqcup_{i\in\omega}\sqcup_{j\in\omega}
(\sem{\HFLte\p \form_1:\etyp_2\to\typ}(\HFLenv_i)
(\sem{\HFLte\p \form_2:\etyp_2}(\HFLenv_j))\\
\hfill \mbox{ (by the continuity of \(\sem{\HFLte\p \form_1:\etyp_2\to\typ}(\HFLenv_i)\)
   and Lemma~\ref{lem:continuity-of-limit})}\\
= 
\sqcup_{i\in\omega}
(\sem{\HFLte\p \form_1:\etyp_2\to\typ}(\HFLenv_i)
(\sem{\HFLte\p \form_2:\etyp_2}(\HFLenv_i))\\
= 
\sqcup_{i\in\omega}
(\sem{\HFLte\p \form_1\form_2:\typ}(\HFLenv_i)
\end{array}
\]
as required. 
\end{itemize}
\qed
\end{proof}

The following is an immediate corollary of the lemma above (see, e.g., \cite{SangiorgiBookIndCoInd}).
\begin{lemma}[fixpoint of continuous functions]
\label{lem:fixpoint}
Let \(\lts\) be an LTS.
If \(f\in \Cont_{\lts,\typ\to\typ}\), then 
\(\LFP{\typ}(f) \eqcont{\lts,\typ} \bigsqcup_{i\in\omega} f^i(\bot_{\lts,\typ})\).
\end{lemma}

  \begin{proof} 
    By the continuity of \(f\), we have
    \(    f(\bigsqcup_{i\in\omega} f^{i}(\bot_{\lts,\typ})) \eqcont{\lts,\typ}
    \bigsqcup_{i\in\omega} f^i(\bot_{\lts,\typ})\).
    Thus, by (transfinite) induction, we have \(f^\beta(\bot) \eqcont{\lts,\typ} \bigsqcup_{i\in\omega} f^i(\bot_{\lts,\typ})\).
    for every ordinal \(\beta\geq \omega\).
    Since \(\LFP{\typ}(f)=f^\beta(\bot)\) for some ordinal \(\beta\)~\cite{SangiorgiBookIndCoInd}, we have
    the required result.

    {The detail of the transfinite induction is given as follows.
  For a given ordinal number \( \beta \), let us define \( f^{\beta}(\bot_{\lts,\typ}) \) by transfinite induction.
  If \( \beta = \beta' + 1 \), then \( f^{\beta}(\bot_{\lts,\typ}) = f(f^{\beta'}(\bot_{\lts,\typ})) \).
  If \( \beta \) is a limit ordinal, then \( f^{\beta}(\bot_{\lts,\typ}) = \bigsqcup_{\beta' < \beta} f^{\beta'}(\bot_{\lts,\typ}) \).
  We prove that \( f^{\omega}(\bot_{\lts,\typ}) \Eqcont{} f^{\beta}(\bot_{\lts,\typ}) \) for every \( \beta > \omega \) by transfinite induction.
  We have proved the claim for \( \beta = \omega + 1 \), i.e.~\( f^{\omega}(\bot_{\lts,\typ}) \Eqcont f^{\omega+1}(\bot_{\lts,\typ}) \).
  \begin{itemize}
  \item If \( \beta = \beta' + 1 \), by the induction hypothesis, \( f^{\omega}(\bot_{\lts,\typ}) \Eqcont{} f^{\beta'} \).
    Since \( f \) is continuous,
    \[
    f^{\beta}(\bot_{\lts,\typ}) = f(f^{\beta'}(\bot_{\lts,\typ})) \Eqcont{} f(f^{\omega}(\bot_{\lts,\typ})) \Eqcont f^{\omega}(\bot_{\lts,\typ}).
    \]
  \item Assume that \( \beta \) is a limit ordinal and \( \typ = \etyp_1 \to \dots \to \etyp_k \to \typProp \).
   Assume also \(x_i\eqcont{\lts,\etyp_i} y_i\)  for each \( 1 \le i \le k \).
    By the induction hypothesis, for every \( \beta' < \beta \), we have
    \[
    f^{\beta'}(\bot_{\lts,\typ}) \Eqcont{} f^{\omega}(\bot_{\lts,\typ})
    \]
    and thus
    \[
    f^{\beta'}(\bot_{\lts,\typ})\,x_1\, \dots\, x_k \Eqcont{} f^{\omega}(\bot_{\lts,\typ})\,y_1\, \dots\, y_k.
    \]
    Since \( \Eqcont{} \) on \( \Cont_{\lts,\typProp} \) is the standard equivalence \( = \), for every \( \beta' < \beta \), we have
    \[
    f^{\beta'}(\bot_{\lts,\typ})\,x_1\,\dots\, x_k = f^{\omega}(\bot_{\lts,\typ})\,y_1\, \dots\ y_k.
    \]
    Since the limit of a function is defined pointwise,
    \begin{align*}
    \left(\bigsqcup_{\beta' < \beta} f^{\beta'}(\bot_{\lts,\typ})\right)\,x_1\, \dots\, x_k
    &= \bigsqcup_{\beta' < \beta} (f^{\beta'}(\bot_{\lts,\typ})\,x_1\, \dots\, x_k) \\
    &= \bigsqcup_{\beta' < \beta} (f^{\omega}(\bot_{\lts,\typ})\,y_1\, \dots\, y_k) \\
    &= f^{\omega}(\bot_{\lts,\typ})\,y_1\, \dots\, y_k.
    \end{align*}
    Thus, we have \( f^{\beta}(\bot_{\lts,\typ}) \Eqcont{} f^{\omega} \).
  \end{itemize}
  Since \( \LFP{\typ}(f) = f^{\beta}(\bot_{\lts,\typ}) \) for some ordinal \( \beta \), we have
  \[
  \bigsqcup_{i \in \omega} f^i(\bot_{\lts,\typ}) = f^{\omega}(\bot_{\lts,\typ}) \Eqcont{} f^{\beta}(\bot_{\lts,\typ}) = \LFP{\typ}(f)
  \]
  as required.
} \qed
\end{proof}

We are now ready to prove Theorem~\ref{th:mayreach}. Below we extend
continuous functions to those on tuples, and make use of
Beki\'c property between a simultaneous recursive definition of multiple functions
and a sequence of recursive definitions; see, e.g., \cite{WinskelTEXTBOOK}, Chapter~10.
\begin{proof}[of Theorem~\ref{th:mayreach}]
Given a program \(\prog = (\progd, \term)\) with \(\progd = \set{f_1\,\seq{x_1}=\term_1,\ldots,f_n\,\seq{x_n}=\term_n}\),
we write \(\prog^{(i)}\) for the recursion-free program \((\progd^{(i)}, [f_1^{(i)}/f_1,\ldots,f_n^{(i)}/f_n]\term)\)
where:
\iffull
\[
\begin{array}{l}
\progd^{(i)} =
  \set{f_k^{(j+1)}\,\seq{x}_k=[f_1^{(j)}/f_1,\ldots,f_n^{(j)}/f_n]\term_k \mid j\in\set{1,\ldots,i-1}, k\in\set{1,\ldots,n}}\\
\qquad \cup \set{f_k^{(0)}\, \seq{x}_k = \Vunit \mid k\in\set{1,\ldots,n}}.
\end{array}
\]
\else
\(\progd^{(i)} =
  \set{f_k^{(j+1)}\,\seq{x}_k=[f_1^{(j)}/f_1,\ldots,f_n^{(j)}/f_n]\term_k \mid j\in\set{1,\ldots,i-1}, k\in\set{1,\ldots,n}}
\cup \set{f_k^{(0)}\, \seq{x}_k = \Vunit \mid k\in\set{1,\ldots,n}}\).
\fi
Then, we obtain the required result as follows.
\[
\begin{array}{l}
\lts_0 \models \HESf_{\prog,\may}\\
\IFF  
\stunique \in \sem{\trMay{\term}}(\set{(f_1,\ldots,f_n)\mapsto \LFP{}(\lambda (f_1,\ldots,f_n).(\lambda \seq{x_1}.\sem{\trMay{\term_1}},\ldots,
\lambda \seq{x_n}.\sem{\trMay{\term_n}}))})\\
\hfill \mbox{(Beki\'{c} property)}\\
\IFF  
\stunique \in \sem{\trMay{\term}}(\set{(f_1,\ldots,f_n)\mapsto\\
\qquad\qquad\qquad \bigsqcup_{i\in\omega}
(\lambda (f_1,\ldots,f_n).(\lambda \seq{x_1}.\sem{\trMay{\term_1}},\ldots,\lambda \seq{x_n}.\sem{\trMay{\term_n}}))^i
 (\lambda \seq{x_1}.\emptyset,\ldots,\lambda \seq{x_n}.\emptyset)})\\
 \hfill \mbox{(Lemmas~\ref{lem:continuity} and \ref{lem:fixpoint})}\\
\IFF  
\stunique \in \sem{\trMay{\term}}(\set{(f_1,\ldots,f_n)\mapsto
\\\qquad\qquad\qquad
(\lambda (f_1,\ldots,f_n).(\lambda \seq{x_1}.\sem{\trMay{\term_1}},\ldots,\lambda \seq{x_n}.\sem{\trMay{\term_n}}))^i
 (\lambda \seq{x_1}.\emptyset,\ldots,\lambda \seq{x_n}.\emptyset)})\\
 \hfill \mbox{ for some $i$ }
\mbox{(Lemmas~\ref{lem:continuity})}\\
\IFF  
\lts_0\models  \HESf_{\prog^{(i)},\may} 
\mbox{ for some $i$}
 \hfill \mbox{(by the definition of \(\prog^{(i)}\))}\\
\IFF  \lab\in\Traces(\prog^{(i)}) 
\mbox{ for some $i$}\hfill \mbox{(Lemma~\ref{lem:mayreach-for-recfree-program})}\\
\IFF  \lab\in\Traces(\prog)
\hfill \qed
\end{array}
\]
\end{proof}

\subsection{Proofs for Section~\ref{sec:mustreach}}
\label{sec:proofs-mustreach}

We first prepare lemmas corresponding to Lemmas~\ref{lem:mayreachsem-preserved-by-reduction}--\ref{lem:mayreach-for-recfree-program}.
\begin{lemma}
\label{lem:mustreachsem-preserved-by-reduction}
Let \((\progd,\term)\) be a program. If \(\term\redv{\progd}\term'\), then
\(\sem{\trMust{(\progd,\term)}} = \sem{\trMust{(\progd,\term')}}\).
\end{lemma}
\begin{proof}
Almost the same as the proof of Lemma~\ref{lem:mayreachsem-preserved-by-reduction}. \qed
\end{proof}

\begin{lemma}
\label{lem:mustreach-base}
Let \((\progd,\term)\) be a program and \(\term\nonredv{\progd}\). Then,
\(\lts_0 \models \trMust{(\progd,\term)}\) if and only if \(\term\) is of the form
\(C[\evexp{\lab}{\term_1},\ldots,\evexp{\lab}{\term_k}]\), where \(C\) is a (multi-hole) context
generated by the syntax: \(C ::= \Hole_i \mid C_1\nondet C_2\).
\end{lemma}
\begin{proof}
The proof proceeds by induction on the structure of \(\term\).
By the condition \(\term\nonredv{\progd}\), \(\term\) is generated by the following grammar:
\[
\term ::= \Vunit \mid \evexp{\lab}\term' \mid \term_1\nondet\term_2.
\]
\begin{itemize}
\item Case \(\term=\Vunit\): The result follows immediately, 
as \(\term\) is not of the form \(C[\evexp{\lab}{\term_1},\ldots,\evexp{\lab}{\term_k}]\) and
\(\trMust{\term}=\FALSE\).
\item Case \(\term=\evexp{\lab}\term'\): 
The result follows immediately, as
\(\term\) is of the form 
\(C[\evexp{\lab}{\term_1},\ldots,\evexp{\lab}{\term_k}]\) (where \(C=\Hole_1\) and \(k=1\)) and
\(\trMust{\term}=\TRUE\).
\item Case \(\term_1\nondet\term_2\):
Because \(\trMust{(\term_1\nondet\term_2)} = \trMust{\term_1}\land \trMust{\term_2}\),
\(\lts_0 \models \trMust{(\progd,\term)}\) if and only if
\(\lts_0 \models \trMust{(\progd,\term_i)}\) for each \(i\in\set{1,2}\).
By the induction hypothesis, the latter is equivalent to 
the property that \(\term_i\) is of the form 
\(C_i[\evexp{\lab}{\term_{i,1}},\ldots,\evexp{\lab}{\term_{i,k_i}}]\) for each \(i\in\set{1,2}\),
which is equivalent to the property that 
\(\term\) is of the form 
\(C[\evexp{\lab}{\term_1},\ldots,\evexp{\lab}{\term_k}]\)
(where \(C=C_1\nondet C_2\)).
\end{itemize}
\qed
\end{proof}

\begin{lemma}
\label{lem:mustreach-for-recfree-program}
Let \(\prog\) be a recursion-free program. Then, \(\Must_\lab(\prog)\) if and only if
\(\lts_0 \models \HESf_{\prog,\must}\) for \(\lts_0 = (\set{\st_\init},\emptyset,\emptyset,\st_\init)\).
\end{lemma}
\begin{proof}
Since \(\prog=(\progd,\term_0)\) is recursion-free, there exists a finite, normalizing reduction sequence
\(\term_0 \redvs{\progd} \term \nonredv{\progd}\).
We show the required property by induction on 
the length \(n\) of this reduction sequence.
\begin{itemize}
\item Case \(n=0\): Since \(\term_0\nonredv{\progd}\),
  the result follows immediately from Lemma~\ref{lem:mustreach-base}.
\item Case \(n>0\): In this case, 
\(\term_0 \redv{\progd} \term_1\redvs{\prog} \term \).
Thus, by the induction hypothesis,
the definition of the reduction semantics and Lemma~\ref{lem:mustreachsem-preserved-by-reduction},
\(\Must_\lab(\progd,\term_0)\) if and only if 
\(\Must_\lab(\progd,\term_1)\), if and only if 
\(\lts_0 \models \HESf_{(\progd,\term_1),\must}\), if and only if 
\(\lts_0 \models \HESf_{(\progd,\term_0),\must}\).
\end{itemize}
\end{proof}
\begin{proof}[Theorem~\ref{th:must-reachability}]
Let \(\prog^{(i)}\) be the the recursion-free program  defined in the proof of Theorem~\ref{th:mayreach}.
Then, we obtain the required result as follows.
\[
\begin{array}{l}
\lts_0 \models \HESf_{\prog,\must}\\
\IFF  
\stunique \in \sem{\trMust{\term}}(\set{(f_1,\ldots,f_n)\mapsto \LFP{}(\lambda (f_1,\ldots,f_n).(\lambda \seq{x_1}.\sem{\trMust{\term_1}},\ldots,
\lambda \seq{x_n}.\sem{\trMust{\term_n}}))})\\
\hfill \mbox{(Beki\'{c} property)}\\
\IFF  
\stunique \in \sem{\trMust{\term}}(\set{(f_1,\ldots,f_n)\mapsto\\
\qquad\qquad\qquad \bigsqcup_{i\in\omega}
(\lambda (f_1,\ldots,f_n).(\lambda \seq{x_1}.\sem{\trMust{\term_1}},\ldots,\lambda \seq{x_n}.\sem{\trMust{\term_n}}))^i
 (\lambda \seq{x_1}.\emptyset,\ldots,\lambda \seq{x_n}.\emptyset)})\\
 \hfill \mbox{((Lemmas~\ref{lem:continuity} and \ref{lem:fixpoint})}\\
\IFF  
\stunique \in \sem{\trMust{\term}}(\set{(f_1,\ldots,f_n)\mapsto
\\\qquad\qquad\qquad
(\lambda (f_1,\ldots,f_n).(\lambda \seq{x_1}.\sem{\trMust{\term_1}},\ldots,\lambda \seq{x_n}.\sem{\trMust{\term_n}}))^i
 (\lambda \seq{x_1}.\emptyset,\ldots,\lambda \seq{x_n}.\emptyset)})
\mbox{ for some $i$}\\
\IFF  
\lts_0\models  \HESf_{\prog^{(i)},\must} 
\mbox{ for some $i$}\\
 \hfill \mbox{(by the definition of \(\prog^{(i)}\))}\\
\IFF  \Must_\lab(\prog^{(i)}) 
\mbox{ for some $i$}\hfill \mbox{(Lemma~\ref{lem:mustreach-for-recfree-program})}\\
\IFF  \Must_\lab(\prog).
\end{array}
\]
\qed
\end{proof}

\section{Proofs for Section~\ref{sec:path}}
\label{sec:proofs-trace}

We first modify the reduction semantics in Figure~\ref{fig:red-semantics}
by adding the following rule for distributing events with respect to \(\nondet\):
\infrule[R-Dist]{}{E[\evexp{\lab}{(t_1\nondet t_2)}] \redv{\progd} E[(\evexp{\lab}t_1)\nondet (\evexp{\lab}t_2)]}
We write \(\redvd{\progd}\) for this modified version of the reduction relation.
We call an evaluation context \emph{event-free} if it is generated only by the following syntax:
\[ E ::= \Hole \mid  \EC\nondet \term \mid \term\nondet \EC. \]
\begin{lemma}
\label{lem:path-base}
Let \(\prog=(\progd,\term)\) be a program such that \(\term\negredvd{\progd}\). Then,
\(\FinTraces(\prog)\subseteq L\) if and only if 
\(\lts_L\models \HESf_{\prog,\Path}\).
\end{lemma}
\begin{proof}
Since \(\term\negredvd{\progd}\), \(\term\) must be of the form:
\(\nondet_{i\in\set{1,\ldots,m}} (\evexp{\lab_{i,1}}\cdots \evexp{\lab_{i,n_i}}\Vunit)\)
(where \(\nondet_{i\in\set{1,\ldots,m}} \term_i\) denotes a combination of \(t_1,\ldots,t_m\) with \(\nondet\)).
Thus,
\[
\begin{array}{l}
\FinTraces(\prog)\subseteq L\\
\IFF \lab_{i,1}\cdots \lab_{i,n_i}\in L \mbox{ for every }i\in\set{1,\ldots,m} \\
\IFF \lts_L \models \Some{\lab_{i,1}}\cdots \Some{\lab_{i,n_i}}\TRUE\mbox{ for every }i\in\set{1,\ldots,m}
\qquad\hfill \mbox{(by the construction of \(\lts_L\))}\\
\IFF \lts_L \models \HESf_{\prog,\Path}
\hfill \mbox{(by the definition of \(\HESf_{\prog,\Path}\))}.
\end{array}
\]
\qed
\end{proof}

\begin{lemma}
\label{lem:path-preserved-by-reduction}
Let \((\progd,\term)\) be a program, and \(L\) be a prefix-closed regular language. 
If \(\term\redvd{\progd}\term'\), then
\(\Sem{\lts_L}{\trPath{(\progd,\term)}} = \Sem{\lts_L}{\trPath{(\progd,\term')}}\).
\end{lemma}
\begin{proof}
Let \(\progd = \set{f_1\,\seq{x}_1=\term_1,\ldots,f_n\,\seq{x}_n=\term_n}\),
and \((F_1,\ldots,F_n)\) be the greatest fixpoint of
\[\lambda (X_1,\ldots,X_n).(\Sem{\lts_L}{\lambda \seq{x}_1.\trPath{\term_1}}(\set{f_1\mapsto X_1,\ldots,f_n\mapsto X_n}),\ldots,
\Sem{\lts_L}{\lambda \seq{x}_n.\trPath{\term_n}}(\set{f_1\mapsto X_1,\ldots,f_n\mapsto X_n})).\]
By the Beki\'{c} property, \(\Sem{\lts_L}{\trPath{(\progd,\term)}}
= \Sem{\lts_L}{\trPath{\term}}(\set{f_1\mapsto F_1,\ldots,f_n\mapsto F_n})\).
Thus, it suffices to show that \(\term\redv{\progd}\term'\) implies
\(\Sem{\lts_L}{\trPath{\term}}(\HFLenv)=
\Sem{\lts_L}{\trPath{\term'}}(\HFLenv)\)
for \(\HFLenv=\set{f_1\mapsto F_1,\ldots,f_n\mapsto F_n}\).
We show it by case analysis on the rule used for deriving \(\term\redvd{\progd}\term'\).
\begin{itemize}
\item Case \rn{R-Fun}: In this case, \(\term = E[f_i\,\seq{t}]\) and \(\term'=E[[\seq{t}/\seq{x}_i]s]\)
with \(\progd(f_i)=\lambda \seq{x}_i.s\). Since \((F_1,\ldots,F_n)\) is a fixpoint, we have:
\[
\begin{array}{l}
\Sem{\lts_L}{f_i\,\seq{t}}(\HFLenv)
 = F_i(\Sem{\lts_L}{\seq{t}}(\HFLenv)\\
 = \Sem{\lts_L}{\lambda \seq{x}_i.s}(\HFLenv) (\Sem{\lts_L}{\seq{t}}(\HFLenv))\\
 = \Sem{\lts_L}{[\seq{t}/\seq{x}_i]s}(\HFLenv)
\end{array}
\]
Thus, we have 
\(\Sem{\lts_L}{\trPath{\term}}(\HFLenv)=
\Sem{\lts_L}{\trPath{\term'}}(\HFLenv)\) as required.
\item Case \rn{R-IfT}:
In this case, \(\term = E[\ifexp{p(\term_1',\ldots,\term_k')}{\term_1}{\term_2}]\) and \(\term' = E[\term_1]\)
with \((\Sem{\lts_L}{\term_1'},\ldots,\Sem{\lts_L}{\term_k'})\in \sem{p}\). 
Thus,\(\trPath{\term}=(p(\term_1',\ldots,\term_k')\imply \trPath{\term_1})\land (\neg p(\term_1',\ldots,\term_k')\imply \trPath{\term_2})\).
Since \((\Sem{\lts_L}{\term_1'},\ldots,\Sem{\lts_L}{\term_k'})\in \sem{p}\), 
\((\Sem{\lts_L}{\term_1'},\ldots,\Sem{\lts_L}{\term_k'})\not \in \sem{\neg p}\). Thus,
\(\Sem{\lts_L}{\trPath{(\ifexp{p(\term_1',\ldots,\term_k')}{\term_1}{\term_2})}}(\HFLenv) = 
\Sem{\lts_L}{(\TRUE\imply\trPath{\term_1})\land (\FALSE\imply \trPath{\term_2})}(\HFLenv)
= 
\Sem{\lts_L}{\trPath{\term_1}}(\HFLenv)\).
We have, therefore,
\(\Sem{\lts_L}{\trPath{\term}}(\HFLenv)=
\Sem{\lts_L}{\trPath{\term'}}(\HFLenv)\) as required.
\item Case \rn{R-IfF}:
Similar to the above case.
\item Case \rn{R-Dist}:
In this case, \(\term=\EC[\evexp{\lab}{(t_1\nondet t_2)}]\) and \(\term'=\EC[(\evexp{\lab}t_1)\nondet (\evexp{\lab}t_2)]\).
Thus, it suffices to show that \(\Sem{\lts_L}{\trPath{(\evexp{\lab}{(t_1\nondet t_2)})}}(\HFLenv)
= \Sem{\lts_L}{\trPath{((\evexp{\lab}t_1)\nondet (\evexp{\lab}t_2))}}(\HFLenv)\).
We have:
\[
\begin{array}{l}
\trPath{(\evexp{\lab}{(t_1\nondet t_2)})} = \Some{\lab}(\trPath{t_1}\land \trPath{t_2})\\
\trPath{((\evexp{\lab}t_1)\nondet (\evexp{\lab}t_2))} = 
(\Some{\lab}\trPath{t_1})\land (\Some{\lab}\trPath{t_2})
\end{array}
\]
Since \(\lts_L\) has at most one \(\lab\)-transition from each state,
both formulas are equivalent, i.e., 
\(\Sem{\lts_L}{\trPath{(\evexp{\lab}{(t_1\nondet t_2)})}}(\HFLenv)
= \Sem{\lts_L}{\trPath{((\evexp{\lab}t_1)\nondet (\evexp{\lab}t_2))}}(\HFLenv)\).
\end{itemize}
\qed
\end{proof}
\begin{lemma}
\label{lem:path-for-recfree-program}
Let \(\prog\) be a recursion-free program and \(L\) be a regular, prefix-closed language. 
Then, \(\FinTraces(\prog)\subseteq L\) if and only if \(\lts_L\models \HESf_{\prog,\Path}\).
\end{lemma}
\begin{proof}
Since \(\prog=(\progd,\term_0)\) is recursion-free, there exists a finite, normalizing reduction sequence
\(\term_0 \redvds{\progd} \term \negredvd{\progd}\). We show the required property by induction on 
the length \(n\) of this reduction sequence. The base case follows immediately from Lemma~\ref{lem:path-base}. For the induction step (where \(n>0\)),
we have:
\(\term_0 \redv{\progd} \term_1\redvs{\progd} \term \).
By the induction hypothesis, 
\(\FinTraces(\progd,\term_1)\subseteq L\) if and only if 
\(\lts_L \models \HESf_{(\progd,\term_1),\Path}\).
Thus, by the definition of the reduction semantics and 
Lemma~\ref{lem:path-preserved-by-reduction},
\(\FinTraces(\progd,\term_0)\subseteq L\), if and only if 
\(\FinTraces(\progd,\term_1)\subseteq L\), if and only if 
\(\lts_L \models \HESf_{(\progd,\term_1),\Path}\), if and only if
\(\lts_L \models \HESf_{(\progd,\term_0),\Path}\).
\qed
\end{proof}
To prove Theorem~\ref{th:path}, we introduce the (slightly non-standard) notion of co-continuity,
which is dual of the continuity in Definition~\ref{def:continuity}.
\begin{definition}
\label{def:cocontinuity}
For an LTS 
\(\lts = (\St{}, \Act{}, \mathbin{\TR}, \st_\init)\) and a type \(\etyp\),
the set of \emph{co-continuous} elements \(\Cocont_{\lts,\etyp}\subseteq \D_{\lts,\etyp}\)
and the equivalence relation
\(\eqcocont{\lts,\etyp}\subseteq
\Cocont_{\lts,\etyp}\times \Cocont_{\lts,\etyp}\) 
are defined by induction on \(\etyp\) as follows.
\[
\begin{array}{l}
\Cocont_{\lts,\typProp} = \D_{\lts,\typProp}\\
\Cocont_{\lts,\INT} = \D_{\lts,\INT}\\
\Cocont_{\lts,\etyp\to\typ} = \set{f\in \D_{\lts,\etyp\to\typ}\mid 
  \forall x_1,x_2\in \Cocont_{\lts,\etyp}.(x_1\eqcocont{\lts,\etyp}x_2\imply
  f(x_1)\eqcocont{\lts,\typ} f(x_2))\\\qquad\qquad\qquad\qquad
\land 
\forall \set{y_i}_{i\in \omega}\in \Decseq{\Cocont_{\lts,\etyp}}. 
f(\Glb_{i\in\omega} y_i)\eqcocont{\lts,\typ}\Glb_{i\in\omega}f(y_i)
}.\\
\eqcocont{\lts,\typProp} \,=\, \set{(x,x)\mid x\in \Cocont_{\lts,\typProp}}\\
  \eqcocont{\lts,\INT} \,=\, \set{(x,x)\mid x\in \Cocont_{\lts,\INT}}\\
    \eqcocont{\lts,\etyp\to\typ} \,=\,\\\qquad
    \set{(f_1,f_2)\mid f_1,f_2\in \Cocont_{\lts,\etyp\to\typ}\land
      \forall x_1,x_2\in \Cocont_{\lts,\etyp}.x_1\eqcocont{\lts,\etyp}x_2
      \imply f_1(x_1)\eqcocont{\lts,\typ}f_2(x_2)}.
\end{array}
\]
Here, \(\Decseq{\Cocont_{\lts,\etyp}}\) denotes the set of decreasing infinite sequences
\(a_0\sqsupseteq a_1\sqsupseteq a_2 \sqsupseteq \cdots\) consisting of elements of 
\(\Cocont_{\lts,\etyp}\).
We just write \(\Eqcocont{}\) for \(\eqcocont{\lts,\etyp}\)  when \(\lts\) and \(\etyp\) are
clear from the context.
\end{definition}

The following lemma is analogous to Lemma~\ref{lem:continuity}.
\begin{lemma}[cocontinuity of fixpoint-free functions]
\label{lem:cocontinuity}
Let \(\lts\) be an LTS.
If \(\form\) is a closed, fixpoint-free HFL formula of type \(\typ\), then 
\(\Sem{\lts}{\form} \in \Cocont_{\lts,\typ}\).
\end{lemma}
\begin{proof}
The proof is almost the same as that of Lemma~\ref{lem:continuity}.
We write \(\cocsem{\HFLte}\) for the set of 
valuations: \(\set{\HFLenv\in \Sem{\lts}{\HFLte} \mid \HFLenv(f)\in \Cocont_{\lts,\etyp}\mbox{ for each
\(f\COL \etyp\in \HFLte\)}}\).
We show the following property by induction on
the derivation of \(\HFLte \p \form:\etyp\):
\begin{quote}
If \(\form\) is fixpoint-free and
\(\HFLte \p \form:\etyp\), then 
\begin{itemize}
\item[(i)] \(\Sem{\lts}{\HFLte\p \form:\etyp}(\HFLenv)\in \Cocont_{\lts,\etyp}\) for
every \(\HFLenv\in \cocsem{\HFLte}\); and
\item[(ii)] For any decreasing sequence of interpretations \(\HFLenv_0 \sqsupseteq \HFLenv_1\sqsupseteq \HFLenv_2 \sqsupseteq 
\cdots \) such that \(\HFLenv_i\in \csem{\HFLte}\) for each \(i\in\omega\), 
\(\Sem{\lts}{\HFLte\p \form:\etyp}(\Glb_{i\in\omega}\HFLenv_i) \Eqcocont{} \Glb_{i\in\omega}
\Sem{\lts}{\HFLte\p \form:\etyp}(\HFLenv_i)\).
\end{itemize}
\end{quote}
Then, the lemma will follow as a special case, where \(\HFLte = \emptyset\).
We perform case analysis on the last rule used for deriving \(\HFLte \p \form:\etyp\).
We discuss only two cases below, as the proof is almost the same as the corresponding proof for 
Lemma~\ref{lem:continuity}.

\begin{itemize}
\item Case \rn{HT-Some}:
  In this case, \(\form=\Some{\lab}\form'\) with \(\HFLte\p \form:\typProp\) and
  \(\etyp=\typProp\).
The condition (i) is trivial since \(\Cocont_{\lts,\typProp} = \D_{\lts,\typProp}\).
We also have the condition (ii) by:
\[
\begin{array}{l}
\Sem{\lts}{\HFLte\p \Some{\lab}{\form'}:\etyp}(\Glb_{i\in\omega}\HFLenv_i) 
=\set{\st\mid \exists \st'\in \Sem{\lts}{\HFLte\pHFL\form':\typProp}(\Glb_{i\in\omega}\HFLenv_i).\ \st\Ar{a}\st'}\\
= \set{\st\mid \exists \st'\in \Glb_{i\in\omega}(\Sem{\lts}{\HFLte\pHFL\form':\typProp}(\HFLenv_i)).\ \st\Ar{a}\st'}
\qquad\hfill \mbox{ (by induction hypothesis)}\\
= \Glb_{i\in\omega}
\set{\st\mid \exists \st'\in \Sem{\lts}{\HFLte\pHFL\form':\typProp}(\HFLenv_i).\ \st\Ar{a}\st'}
\hfill \mbox{ (*)}\\
= \Glb_{i\in\omega}\Sem{\lts}{\HFLte\pHFL\Some{\lab}\form':\typProp}(\HFLenv_i).
\end{array}
\]
The step (*) is obtained as follows. Suppose
\(\exists \st'\in \Glb_{i\in\omega}.\Sem{\lts}{\HFLte\pHFL\form':\typProp}(\HFLenv_i).\ \st\Ar{a}\st'\).
Since \(\Glb_{i\in\omega}.\Sem{\lts}{\HFLte\pHFL\form':\typProp}(\HFLenv_i)\subseteq
\Sem{\lts}{\HFLte\pHFL\form':\typProp}(\HFLenv_i)\) for every \(i\), 
we have \(\exists \st'\in \Sem{\lts}{\HFLte\pHFL\form':\typProp}(\HFLenv_i).\ \st\Ar{a}\st'\) for every
\(i\); hence we have 
\[\st\in 
\textstyle \GLB_{i\in\omega}
\set{\st\mid \exists \st'\in \Sem{\lts}{\HFLte\pHFL\form':\typProp}(\HFLenv_i).\ \st\Ar{a}\st'}.\]
Conversely, suppose 
\[\st\in 
\textstyle \GLB_{i\in\omega}
\set{\st\mid \exists \st'\in \Sem{\lts}{\HFLte\pHFL\form':\typProp}(\HFLenv_i).\ \st\Ar{a}\st'},\]
i.e.,  for every \(i\), \(S_i = \set{\st'\in \Sem{\lts}{\HFLte\pHFL\form':\typProp}(\HFLenv_i)\mid \st\Ar{a}\st'}\) 
must be non-empty.
Since the set \(\set{\st'\mid \st\Ar{a}\st'}\) is finite, 
and \(S_i\) decreases monotonically, \(\bigcap_{i\in\omega} S_i\) must be non-empty.
Thus, we have 
\[\exists \st'\in \textstyle\GLB_{i\in\omega}\Sem{\lts}{\HFLte\pHFL\form':\typProp}(\HFLenv_i).\ \st\Ar{a}\st',\]
as required.
\item Case \rn{HT-All}:
In this case, \(\form=\All{\lab}\form'\) with \(\HFLte\p \form:\typProp\) and \(\etyp=\typProp\).
The condition (i) is trivial.
The condition (ii) follows by:
\[
\begin{array}{l}
\Sem{\lts}{\HFLte\p \All{\lab}{\form'}:\etyp}(\textstyle\GLB_{i\in\omega}\HFLenv_i) 
=\set{\st\mid \forall \st'\in \St.\st\Ar{a}\st'\imply \st'\in \Sem{\lts}{\HFLte\pHFL\form':\typProp}(\Glb_{i\in\omega}\HFLenv_i)}\\
=\set{\st\mid \forall \st'\in \St.\st\Ar{a}\st'\imply \st'\in 
\Glb_{i\in\omega}(\Sem{\lts}{\HFLte\pHFL\form':\typProp}(\HFLenv_i))}\\
\qquad\hfill \mbox{ (by induction hypothesis)}\\
= \Glb_{i\in\omega}
\set{\st\mid \forall \st'\in \St.\st\Ar{a}\st'\imply \st'\in 
\Sem{\lts}{\HFLte\pHFL\form':\typProp}(\HFLenv_i)}\\
= \Glb_{i\in\omega}\Sem{\lts}{\HFLte\pHFL\All{\lab}\form':\typProp}(\HFLenv_i).
\end{array}
\]
\end{itemize}
\qed
\end{proof}

The following lemma states a standard property of
cocontinuous functions~\cite{SangiorgiBookIndCoInd}, which can be proved in the same manner as Lemma~\ref{lem:fixpoint}.
\begin{lemma}[fixpoint of cocontinuous functions]
\label{lem:gfp}
Let \(\lts\) be an LTS.
If \(f\in \Cocont_{\lts,\typ\to\typ}\), then 
\(\GFP{\typ}(f) \eqcocont{\lts,\typ} \bigsqcap_{i\in\omega} f^i(\top_{\lts,\typ})\).
\end{lemma}

\begin{proof}[Proof of Theorem~\ref{th:path}]
The result follows by:
\[
\begin{array}{l}
\lts_L \not\models \HESf_{\prog,\Path}\\
\IFF  
q_0 \not\in \Sem{\lts_L}{\trPath{\term}}(\set{(f_1,\ldots,f_n)\mapsto\GFP{}(\lambda (f_1,\ldots,f_n).(\lambda \seq{x_1}.\Sem{\lts_L}{\trPath{\term_1}},\ldots,
\lambda \seq{x_n}.\Sem{\lts_L}{\trPath{\term_n}}))})\\
\hfill \mbox{(Beki\'{c} property)}\\
\IFF  
q_0 \not\in \Sem{\lts_L}{\trPath{\term}}(\set{(f_1,\ldots,f_n)\mapsto\\
\qquad\bigsqcap_{i\in\omega}
(\lambda (f_1,\ldots,f_n).(\lambda \seq{x_1}.\Sem{\lts_L}{\trPath{\term_1}},\ldots,\lambda \seq{x_n}.\Sem{\lts_L}{\trPath{\term_n}}))^i
 (\lambda \seq{x_1}.Q,\ldots,\lambda \seq{x_n}.Q)})\\
 \hfill \mbox{(Lemmas~\ref{lem:cocontinuity} and \ref{lem:gfp})}\\
\IFF  
q_0 \not\in \Sem{\lts_L}{\trPath{\term}}(\set{(f_1,\ldots,f_n)\mapsto\\\qquad
(\lambda (f_1,\ldots,f_n).(\lambda \seq{x_1}.\Sem{\lts_L}{\trPath{\term_1}},\ldots,\lambda \seq{x_n}.\Sem{\lts_L}{\trPath{\term_n}}))^i
 (\lambda \seq{x_1}.Q,\ldots,\lambda \seq{x_n}.Q)})\\\qquad\qquad
\mbox{ for some $i$}\hfill \mbox{(Lemma~\ref{lem:cocontinuity})}\\
\IFF  
\lts_L\not\models  \HESf_{\prog^{(i)},\Path} 
\mbox{ for some $i$}\\
\IFF  \Traces(\prog^{(i)}) \not\subseteq L
\mbox{ for some $i$}\hfill \mbox{(Lemma~\ref{lem:path-for-recfree-program})}\\
\IFF  \Traces(\prog) \not\subseteq L
\end{array}
\]
\qed
\end{proof}


\section{Game-based characterization of HES}
\label{sec:game-characterization}
Fix an LTS \( \lts \) and let \( \HES = (X_1^{\typ_1} =_{\munu_1} \form_1;  \cdots; X_n^{\typ_n} =_{\munu_n} \form_n) \).
The goal of this section is to construct a parity game \( \SGame_{\lts, \HES} \) characterizing the semantic interpretation of the HES \( \HES \) (or equivalently \( (\HES, X_1) \)) over the LTS \( \lts \).
The game-based characterization will be used to prove some results in Section~\ref{sec:liveness}.

\subsection{Preliminary: Parity Game}
A \emph{parity game} \( \mathcal{G} \) is a tuple \( (\Nodes_P, \Nodes_O, \Edges, \Omega) \) where
\begin{itemize}
\item \( \Nodes_P \) and \( \Nodes_O \) are disjoint sets of \emph{Proponent and Opponent nodes}, respectively,
\item \( \Edges \subseteq (\Nodes_P \cup \Nodes_O) \times (\Nodes_P \cup \Nodes_O) \) is the set of edges, and
\item \( \Omega : (\Nodes_P \cup \Nodes_O) \to \Nat \) is a \emph{priority function} whose image is bounded.
\end{itemize}
We write \( \Nodes \) for \( \Nodes_P \cup \Nodes_O \).

A \emph{play} of a parity game is a (finite or infinite) sequence \( v_1 v_2 \dots \) of nodes in \( V \) such that \( (v_i, v_{i+1}) \in E \) for every \( i \).
We write \( \cdot \) for the concatenation operation.
An infinite play \( v_1 v_2 \dots \) is said to satisfy the \emph{parity condition} if \( \max \INF (\Omega(v_1) \Omega(v_2) \dots) \) is even,
where \(\INF(m_1m_2\cdots)\) is the set of numbers that occur infinitely often in \(m_1m_2\cdots\).
A play is \emph{maximal} if either
\begin{itemize}
\item it is a finite sequence \(v_1\cdots v_n\) and the last node \(v_n\) has no successor (i.e., \(\set{v\mid (v_n,v)\in E}=\emptyset\)), or
\item it is infinite.
\end{itemize}
A maximal play is \emph{P-winning} (or simply \emph{winning}) if either
\begin{itemize}
\item it is finite and the last node is of Opponent, or
\item it is infinite and satisfies the parity condition.
\end{itemize}

Let \( \Strategy : \Nodes^* \Nodes_P \rightharpoonup V \) be a partial function that respects \( \Edges \) in the sense that \( (v_k, \Strategy(v_1 \dots v_k)) \in \Edges \) (if \(\Strategy(v_1 \dots v_k)\) is defined).
A play \( v_1 v_2 \dots v_k \) is said to \emph{conform with \( \Strategy \)} if, for every \( 1 \le i < k \) with \( v_i \in \Nodes_P \), \( \Strategy(v_1 \dots v_i) \) is defined and \( \Strategy(v_1 \dots v_i) = v_{i+1} \).
An infinite play conforms with \( \Strategy \) if so does every finite prefix of the play.
The partial function \( \Strategy \) is a \emph{P-strategy on \( \Nodes_0 \subseteq \Nodes \)} (or simply a \emph{strategy on \( \Nodes_0 \)}) if it is defined on every play that conforms with \( \Strategy \), starts from a node in \( \Nodes_0 \) and ends with a node in \( \Nodes_P \).
A strategy is \emph{P-winning on \( \Nodes_0 \subseteq \Nodes \)} (or simply \emph{winning on \( \Nodes_0 \)}) if every maximal play that conforms with the strategy and starts from \( \Nodes_0 \) is P-winning.
We say that \emph{Proponent wins the game \( \mathcal{G} \) on \( \Nodes_0 \subseteq \Nodes \)} if there exists a P-winning strategy of \( \mathcal{G} \) on \( \Nodes_0 \).
An \emph{O-strategy} and an \emph{O-winning strategy} is defined similarly.

A strategy \( \Strategy \) is \emph{memoryless} if \( \Strategy(v_1 \dots v_k) = f(v_k) \) for some \( f : \Nodes_P \to \Nodes \).

We shall consider only games of limited shape, which we call \emph{bipartite games}.
A game is \emph{bipartite} if \( \Edges \subseteq (\Nodes_P \times \Nodes_O) \cup (\Nodes_O \times \Nodes_P) \) and \( \Omega(v_O) = 0 \) for every \( v_O \in \Nodes_O \).

\emph{Parity progress measure}~\cite{Jurdzinski00} is a useful tool to show that a strategy is winning.
We give a modified version of the definition, applicable only to bipartite games.
\begin{definition}[Parity progress measure]
  Let \( \gamma \) be an ordinal number.
  Given \( (\beta_1, \dots, \beta_n), (\beta'_1, \dots, \beta'_n) \in \gamma^n \), we write \( (\beta_1, \dots, \beta_n) \ge_j (\beta'_1, \dots, \beta'_n) \) if \( (\beta_1, \dots, \beta_{n-j+1}) \ge (\beta_1', \dots, \beta'_{n-j+1}) \) by the lexicographic ordering.
  The strict inequality \( >_j \) (\( j = 1, \dots, n \)) is defined by an analogous way.

  Let \( \mathcal{G} \) be a bipartite parity game and \( \Strategy \) be a strategy of \( \mathcal{G} \) on \( \Nodes_0 \).
  Let \( n \) be the maximum priority in \( \mathcal{G} \).
  A (partial) mapping \( \varpi : \Nodes_O \rightharpoonup \gamma^n \) is a \emph{parity progress measure of \( \Strategy \) on \( \Nodes_0 \)} if it satisfies the following condition:
  \begin{quote}
    For every finite play \( \tilde{v}\cdot v_O \cdot v_P \cdot v_O' \) (\( v_O, v_O' \in \Nodes_O \) and \( v_P \in \Nodes_P \)) that starts from \( V_0 \) and conforms with the strategy \( \Strategy \), both \( \varpi(v_O) \) and \( \varpi(v_O') \) are defined and \( \varpi(v_O) \ge_{\Omega(v_P)} \varpi(v_O') \).
    Furthermore if \( \Omega(v_P) \) is odd, then \( \varpi(v_O) >_{\Omega(v_P)} \varpi(v_O') \).
  \end{quote}
\end{definition}

\begin{lemma}[\citeN{Jurdzinski00}]
  Let \( \mathcal{G} \) be a bipartite parity game and \( \Strategy \) be a strategy of \( \mathcal{G} \) on \( \Nodes_0 \subseteq \Nodes \).
  If there exists a parity progress measure of \( \Strategy \) on \( \Nodes_0 \), then \( \Strategy \) is a winning strategy on \( \Nodes_0 \).
\end{lemma}
\begin{myproof}
  Let \( \varpi \) be a parity progress measure of \( \Strategy \) on \( \Nodes_0 \).
  We prove by contradiction.

  Assume an infinite play \( v_1 v_2 \dots \) that conforms \( \Strategy \), starts from \( \Nodes_0 \) and violates the parity condition.
  Assume that \( v_1 \in \Nodes_O \); the other case can be proved similarly.
  Since \( \mathcal{G} \) is bipartite, \( v_{i} \in \Nodes_O \) if and only if \( i \) is odd.
  Furthermore \( \Omega(v_i) = 0 \) for every odd \( i \).
  
  Let \( \ell = \max \INF (\Omega(v_1) \Omega(v_2) \dots) \).
  Then there exists an even number \( k \) such that, for every even number \( i \ge k \), we have \( \Omega(v_i) \le \ell \).
  By definition, \( \varpi(v_{i-1}) \ge_{\Omega(v_{i})} \varpi(v_{i+1}) \) for every even \( i \).
  Since \( j_1 \le j_2 \) implies \( ({\ge_{j_1}}) \subseteq ({\ge_{j_2}}) \), we have \( \varpi(v_{i-1}) \ge_{\ell} \varpi(v_{i+1}) \) for every even \( i \ge k \).
  Hence we have an infinite decreasing chain \( \varpi(v_{k-1}) \ge_{\ell} \varpi(v_{k+1}) \ge_{\ell} \varpi(v_{k+3}) \ge_{\ell} \cdots \).
  Furthermore this inequality is strict if \( \Omega(v_i) = \ell \) because \( \ell \) is odd.
  Since \( \{ i \mid i \ge k, \Omega(v_i) = \ell \} \) is infinite, we have an infinite strictly decreasing chain, a contradiction.
\end{myproof}

\subsection{Preliminary: Complete-Prime Algebraic Lattice}

\begin{definition}
  Let \( (A, \le) \) be a complete lattice.
  For \( U \subseteq A \), we write \( \bigvee U \) for the least upper bound of \( U \) in \( A \).
  An element \( p \in A \) is a \emph{complete prime} if (1) \( p \neq \bot \), and (2) \( p \le \bigvee U \) implies \( p \le x \) for some \( x \in U \) (for every \( U \subseteq A \)).
  A complete lattice is \emph{complete-prime algebraic}~\cite{DBLP:journals/tcs/Winskel09} if \( x = \bigvee \{\, p \le x \mid \textrm{\( p \) : complete prime} \,\} \) for every \( x \).
\end{definition}

The following is a basic property about complete-prime algebraic lattices, which will be used later.
\begin{lemma}
\label{lem:complete-prime-leq}
Let \((A,\leq)\) be a complete prime algebraic lattice, and \(x,y\in A\). Then the followings are equivalent.
\begin{enumerate}[(i)]
\item \(x\leq y\).
\item \(p\leq x\) implies \(p\leq y\) for every complete prime \(p\).
\item \(p\nleq y\) implies \(p\nleq x\) for every complete prime \(p\).
\end{enumerate}
\end{lemma}
\begin{myproof}
(i) implies (ii): Trivial from the fact \(p\leq x\leq y\) implies \(p\leq y\).
(ii) implies (i): Assume (ii). Then, \(
\set{p\mid p\leq x, p\mbox{: complete prime}} \subseteq 
\set{p\mid p\leq y, p\mbox{: complete prime}}\). Since
\(A\) is complete-prime algebraic, \(x = \pLub 
\set{p\mid p\leq x, p\mbox{: complete prime}} \leq \pLub\set{p\mid p\leq y, p\mbox{: complete prime}}=y\).
(iii) is just a contraposition of (ii).
\end{myproof}
Let \( (A, \le) \) be a complete lattice and \( f : A \to A \) be a monotone function. We write \(\bot_A\) for the least element of \(A\).
Then \( A \) has the least fixed point of \( f \), which can be computed by iterative application of \( f \) to \( \bot_A \) as follows.
Let \( \gamma \) be an ordinal number greater than the cardinality of \( A \).
We define a family \( \{ a_i \}_{i < \gamma} \) of elements in \( A \) by:
\begin{align*}
  a_0 &:= \bot_A &
  a_{\beta + 1} &:= f(a_{\beta}) &
  a_{\beta} &:= \bigvee_{\beta' < \beta} a_{\beta'} \quad\mbox{(if \( \beta \) is a limit ordinal).}
\end{align*}
Then \( a_{\gamma} \) is the least fixed point of \( f \).
We shall write \( a_{\beta} \) as \( f^{\beta}(\bot_A) \).
\begin{lemma}\label{lem:appx:game-semantics:complete-prime-successor}
  Let \( (A, \le) \) be a complete lattice and \( f : A \to A \) be a monotone function.
  For every complete prime \( p \in A \), the minimum ordinal \( \beta \) such that \( p \le f^{\beta}(\bot) \) is a successor ordinal.
\end{lemma}
\begin{myproof}
  Assume that the minimum ordinal \( \beta \) is a limit ordinal.
  By definition, \( p \le \bigvee_{\beta' < \beta} f^{\beta'}(\bot) \).
  Since \( p \) is a complete prime, there exists \( \beta' < \beta \) such that \( p \le f^{\beta'}(\bot) \).
  Hence \( \beta \) is not the minimum, a contradiction.
\end{myproof}

\tk{Cite something.  Perhaps a paper by Winskel studying complete-prime algebraic lattices or Terui RTA 2012 or Salvati ICALP 2012.}
\begin{lemma}
\label{lem:complete-prime}
  Let \( \lts \) be an LTS.
  Then \( (\D_{\lts, \typ}, \sqleq_{\lts,\typ}) \) is complete-prime algebraic for every \( \typ \).
\end{lemma}
\begin{myproof}
  If \( \typ = \typProp \), then \( \D_{\lts, \typ} = 2^{\St} \) is ordered by the set inclusion and the complete primes are singleton sets \( \{\, q \,\} \) (\( q \in \St \)).
  For function types, \( \D_{\lts, \etyp \to \typ} \) is the set of monotone functions ordered by the pointwise ordering.
  Given an element \( d \in \D_{\lts, \etyp} \) and a complete prime \( p \in \D_{\lts, \typ} \), consider the function \( f_{d, p} \) defined by
  \[
  f_{d, p}(x) = \begin{cases} p & \mbox{(if \( d \le x \))} \\ \bot & \mbox{(otherwise).} \end{cases}
  \]
  These functions are complete primes of \( \D_{\lts, \etyp \to \typ} \).
  It is not difficult to see that every element is the least upper bound of a subset of complete primes.
\end{myproof}

\subsection{Parity Game for HES}
Let \( \HES = (X_1^{\typ_1} =_{\munu_1} \form_1;  \cdots; X_n^{\typ_n} =_{\munu_n} \form_n) \) be an HES (with ``the main formula'' \( X_1 \)) and \( \lts \) be an LTS.
\begin{definition}[Parity game for HES]
  The parity game \( \SGame_{\lts, \HES} \) is defined by the following data:
  \begin{align*}
    \Nodes_{P} :=\;& \{\, (p, X_i) \mid p \in \D_{\lts, \typ_i}, \;\textrm{\( p \): complete prime} \,\}
    \\
    \Nodes_{O} :=\;& \{\, (x_1, \dots, x_n) \mid x_1 \in \D_{\lts, \typ_1}, \dots, x_n \in \D_{\lts, \typ_n} \,\}
    \\
    \Edges
    :=\;& \{\, ((p, X_i), (x_1, \dots, x_n)) \mid p \sqleq \sem{\form_i}([X_1 \mapsto x_1, \dots, X_n \mapsto x_n]) \,\}
    \\
    \cup\;& \{\, ((x_1, \dots, x_n), (p, X_i)) \mid p \sqleq x_i, \;\textrm{\( p \): complete prime} \,\}.
  \end{align*}
  The priority of the opponent node is \( 0 \); the priority of node \( (p, X_i) \) is \( 2(n-i) \) if \( \munu_i = \nu \) and \( 2(n-i)+1 \) if \( \munu_i = \mu \).
\end{definition}

The parity game defined above is analogous to the typability game for HES defined by \citeN{Kobayashi17POPL}.
A position \((p,X_i)\in V_P\) represents the state where Proponent tries to show that
\(p \sqleq \sem{(\HES,X_i)}\). To show \(p \sqleq \sem{(\HES,X_i)}\), Proponent picks a valuation
\(\rho = [X_1\mapsto x_1,\ldots,X_n\mapsto x_n]\) such that \(p \sqleq \sem{\form_i}(\rho)\) indeed holds.
A position \((x_1,\ldots,x_n)\in V_O\) represents such a valuation, and Opponent challenges Proponent's assumption
that \(\rho = [X_1\mapsto x_1,\ldots,X_n\mapsto x_n]\) is a valid valuation, i.e.,
\(x_i\sqleq \sem{(\HES,X_i)}\) holds for each \(i\). To this end, Opponent chooses \(i\), picks
a complete prime \(p\) such that \(p\sqleq x_i\), and asks why \(p\sqleq \sem{(\HES,X_i)}\) holds,
as represented by the edge \(((x_1, \dots, x_n), (p, X_i)) \). (Note that Lemma~\ref{lem:complete-prime}
implies that \(x_i\sqleq \sem{(\HES,X_i)}\) if and only if 
\(p\sqleq \sem{(\HES,X_i)}\) for every complete prime such that \(p\sqleq x_i\); therefore, it is sufficient
for Opponent to consider only complete primes.)
As in the typability game for HES defined by \citeN{Kobayashi17POPL}, a play may continue indefinitely,
in which case, the winner is determined by the largest priority of \((p,X_i)\) visited infinitely often.

The goal of this section is to show that \( p \sqleq \sem{\HES}_{\lts} \) if and only if Proponent wins \( \SGame_{\lts,\HES} \) on \( (p, X_1) \) (Theorem~\ref{thm:appx:game-semantics:game-and-interpretation} given at the end of
this subsection).

We start from an alternative description of the interpretation of an HES.
Let us define the family \( \{ \vartheta_i \}_{i = 1, \dots, n} \) of substitutions by induction on \( n - i \) as follows:
\begin{align*}
  \vartheta_n &:= \mbox{(the identity substitution)} \\
  \vartheta_{i} &:= [X_{i+1} \mapsto \munu_{i+1} X_{i+1}. (\vartheta_{i+1}\form_{i+1} )] \circ \vartheta_{i+1}
  &\mbox{(if \( i < n \)).}
\end{align*}
The family \( \{ \psi_i \}_{i = 1, \dots, n} \) of formulas is defined by
\[
\psi_i := \munu_i X_i. ( \vartheta_{i}\form_i).
\]
Then \( \vartheta_i = [X_{i+1} \mapsto \psi_{i+1}] \circ \vartheta_{i+1} \).
The substitution \( \vartheta_i \) maps \( X_{i+1}, \dots, X_n \) to closed formulas and thus \( \tkchanged{\vartheta_i \form_j} \) (\( j \le i \)) is a formula with free variables \( X_1, \dots, X_{i} \).
The formula \( \psi_i \) has free variables \( X_1, \dots, X_{i-1} \).
In particular, \( \psi_1 = \toHFL(\HES, X_1) \).

Given \( 1 \le i \le j \le n \) and \( d_k \in \D_{\lts,\typ_k} \) (\(1 \le k \le n\)), let us define
\begin{align*}
f_{i,j}(d_1, \dots, d_{j}) &:= \sem{\vartheta_j\form_i }([X_1 \mapsto d_1, \dots, X_{j} \mapsto d_{j}]) \\
g_i(d_1, \dots, d_{i-1}) &:= \sem{\psi_i}([X_1 \mapsto d_1, \dots, X_{i-1} \mapsto d_{i-1}]).
\end{align*}
\begin{lemma}\label{lem:appx:game-semantics:alg-interpret}
  We have
  \begin{align*}
    f_{i,n}(d_1, \dots, d_n) &= \sem{\form_i}([X_1 \mapsto d_1, \dots, X_n \mapsto d_n]) \\
    f_{i,j}(d_1, \dots, d_j) &= f_{i,j+1}(d_1, \dots, d_{j}, g_{j+1}(d_1, \dots, d_{j})) & \mbox{(if \( j < n \))\hphantom{.}} \\
    g_{i}(d_1, \dots, d_{i-1}) &= \munu_i y. f_{i,i}(d_1, \dots, d_{i-1}, y).\\
g_1(\,) &= \sem{\HES}.
  \end{align*}
\end{lemma}
\begin{myproof}
  By induction on \( n - i \).
  For each \( i \), the claim is proved by induction on the structure of formulas \( \form_i \).
The last claim follows from the fact that \(\psi_1 = \toHFL(\HES,X_1)\) and: 
\[g_1(\,) = \munu_1 y.f_{1,1}(y) = \munu_1 y.\sem{\vartheta_1\form_1}([X_1\mapsto y]) = 
\sem{\munu_1 X_1.\vartheta_1\form_1} = \sem{\psi_1}.\]
\end{myproof}

We first prove the following lemma, which implies that 
the game-based characterization is complete, i.e.,
that \( p \sqleq \sem{\HES}_{\lts} \) only if Proponent wins \( \SGame_{\lts,\HES} \) on \( (p, X_1) \).
\begin{lemma}\label{lem:appx:game-semantics:completeness}
  Let \( (d_1, \dots, d_n) \) be an opponent node.
  If \( d_j \sqleq g_{j}(d_1, \dots, d_{j-1}) \) for every \( j \in\set{1, \dots, n} \), then Proponent wins \( \SGame_{\lts,\HES} \) on \( (d_1, \dots, d_n) \).
\end{lemma}
\begin{myproof}
  Let \( (\star) \) be the following condition on O-nodes \( (d_1, \dots, d_n) \):
  \begin{quote}
    \( d_j \sqleq g_{j}(d_1, \dots, d_{j-1}) \) for every \( j \in\set{1, \dots, n} \).
  \end{quote}
  Let \( \Nodes_0 = \{ (d_1, \dots, d_n) \in \Nodes_O \mid \mbox{\((d_1, \dots, d_n)\) satisfies \((\star)\)} \} \).  

  Let us first define a strategy \( \Strategy \) of \( \SGame_{\lts, \HES} \) on \( \Nodes_0 \).
  It is defined for plays of the form \( \seq{v} \cdot (d_1, \dots, d_n)  \cdot (p, X_i) \) where \( (d_1, \dots, d_n) \in \Nodes_0 \).
  The next node is determined by the last two nodes as follows.
  \begin{itemize}
  \item Case \( \munu_i = \nu \):  Then the next node is
    \[
    (d_1, \dots, d_{i-1}, c_i, c_{i+1}, \dots, c_n)
    \]
    where \( c_{j} = g_{j}(d_1, \dots, d_{i-1}, c_i, c_{i+1}, \dots, c_{j-1}) \) for \( j \ge i \).
    By the assumption \( d_i \sqleq g_{i}(d_1, \dots, d_{i-1}) \) and Lemma~\ref{lem:appx:game-semantics:alg-interpret},
    \begin{align*}
      p&\sqleq d_i \\
      &\sqleq g_i(d_1, \dots, d_{i-1}) \\
      &= \nu y. f_{i,i}(d_1, \dots, d_{i-1}, y) \\
      &= f_{i,i}(d_1, \dots, d_{i-1}, \nu y. f_{i,i}(d_1, \dots, d_{i-1}, y)) \\
      &= f_{i,i}(d_1, \dots, d_{i-1}, g_i(d_1, \dots, d_{i-1})) \\
      &= f_{i,i}(d_1, \dots, d_{i-1}, c_i) \\
      &= f_{i,i+1}(d_1, \dots, d_{i-1}, c_i, g_{i+1}(d_1, \dots, d_{i-1}, c_i)) \\
      &= f_{i,i+1}(d_1, \dots, d_{i-1}, c_i, c_{i+1}) \\
      &\;\;\vdots \\
      &= f_{i,n}(d_1, \dots, d_{i-1}, c_i, c_{i+1}, \dots, c_n) \\
      &= \sem{\form_i}([X_1 \mapsto d_1, \dots, X_{i-1} \mapsto d_{i-1}, X_{i} \mapsto c_i, \dots, X_{n} \mapsto c_n]).
    \end{align*}
  \item Case \( \munu_i = \mu \):  Then the next node is
    \[
    (d_1, \dots, d_{i-1}, e, c_{i+1}, \dots, c_n)
    \]
    where \( c_{j} = g_{j}(d_1, \dots, d_{i-1}, e, c_{i+1}, \dots, c_{j-1}) \) for \( j > i \) and \( e \) is defined as follows.
    By the condition \( (\star) \) and Lemma~\ref{lem:appx:game-semantics:alg-interpret},
    \[
    d_i \sqleq g_{i}(d_1, \dots, d_{i-1}) = \mu y. f_{i,i}(d_1, \dots, d_{i-1}, y)
    \]
    and thus
    \[
    d_i \sqleq  (\lambda y. f_{i,i}(d_1, \dots, d_{i-1}, y))^{\beta}(\bot)
    \]
    for some ordinal \( \beta \).
    By the definition of edges of \( \SGame_{\lts,\HES} \), we have \( p \sqleq d_i \) and thus
    \[
    p \sqleq (\lambda y. f_{i,i}(d_1, \dots, d_{i-1}, y))^{\beta}(\bot)
    \]
    for some \( \beta \).
    Consider the minimum ordinal among those satisfying this condition; let \( \beta_0 \) be this ordinal.
    Since \( p \) is complete prime, \( \beta_0 \) is a successor ordinal by Lemma~\ref{lem:appx:game-semantics:complete-prime-successor}, i.e.~\( \beta_0 = \beta_0' + 1 \) for some \( \beta_0' \).
    Then we define \( e =  (\lambda y. f_{i,i}(x_1, \dots, x_{i-1}, y))^{\beta'}(\bot) \).
    Now we have
    \begin{align*}
      p
      &\sqleq (\lambda y. f_{i,i}(d_1, \dots, d_{i-1}, y))(e) \\
      &= f_{i,i}(d_1, \dots, d_{i-1}, e) \\
      &= f_{i,i+1}(d_1, \dots, d_{i-1}, e, g_{i+1}(d_1, \dots, d_{i-1}, e)) \\
      &= f_{i,i+1}(d_1, \dots, d_{i-1}, e, c_{i+1}) \\
      &\;\;\vdots \\
      &= f_{i,n}(d_1, \dots, d_{i-1}, e, c_{i+1}, \dots, c_n) \\
      &= \sem{\form_i}([X_1 \mapsto d_1, \dots, X_{i-1} \mapsto d_{i-1}, X_i \mapsto e, X_{i+1} \mapsto c_{i+1}, \dots, X_{n} \mapsto c_n]).
    \end{align*}
  \end{itemize}
  It is not difficult to see that \( \Strategy \) is indeed a strategy on \( \Nodes_0 \).

  We prove that \( \Strategy \) is winning by giving a parity progress measure of \( \Strategy \) on \( \Nodes_0 \).
  Let \( \gamma \) be an ordinal greater than the cardinality of \( \D_{\lts, \typ_\ell} \) for every 
\( \ell \in\set{1, \dots, n} \).
  The (partial) mapping \( \Nodes_O \ni (d_1, \dots, d_n) \mapsto (\beta_1, \dots, \beta_n) \in \gamma^n \) is defined by:
  \begin{equation*}
    \beta_i :=
    \begin{cases}
      \min \{ \beta < \gamma \mid d_i \le (\lambda y. f_{i,i}(d_1, \dots, d_{i-1}, y))^{\beta}(\bot) \} & \mbox{(if \( \munu_i = \mu \))} \\
      0 & \mbox{(if \( \munu_i = \nu \)).}
    \end{cases}
  \end{equation*}
  This is well-defined on \( \Nodes_0 \).
  It is not difficult to see that this mapping is a parity progress measure of the strategy \( \Strategy \).
\end{myproof}

Next, we prepare lemmas used for proving 
that the game-based characterization is sound, i.e.,
that 
 existence of a winning strategy of \( \SGame_{\lts,\HES} \) on \( (p, X_1) \) implies \( p \sqleq \sem{\HES}_{\lts} \).
We first prove this result for a specific class of strategies, which we call \emph{stable strategies}.
Then we show that every winning strategy can be transformed to a stable one.
\begin{definition}[Stable strategy]
  A strategy \( \Strategy \) of \( \SGame_{\lts, \HES} \) is \emph{stable} if for every play
  \begin{equation*}
    \seq{v} \cdot (d_1, \dots, d_n) \cdot (p, X_i) \cdot (c_1, \dots, c_n)
  \end{equation*}
  that conforms with \( \Strategy \), one has \( d_1 = c_1, \dots, d_{i-1} = c_{i-1} \) and \( d_i \sqsupseteq c_i \).
\end{definition}
An important property of a stable winning strategy is as follows.
\begin{lemma}
Given \( d_1 \in \D_{\lts,\typ_1}, \dots, d_{\ell} \in \D_{\lts,\typ_\ell} \), we write \( \SGame_{\lts,\HES}^{d_1,\dots,d_{\ell}} \) for the subgame of \( \SGame_{\lts,\HES} \) consisting of nodes
\[
\{\, (p, X_i)\in V_P \mid \ell < i \,\} \cup \{\, (d_1, \dots, d_{\ell}, c_{\ell + 1}, \dots, c_n) \mid c_{\ell + 1} \in \D_{\lts,\typ_{\ell+1}}, \dots, c_{n} = \D_{\lts,\typ_n} \,\}.
\]
If a stable strategy \( \Strategy \) wins on an opponent node \( (d_1, \dots, d_n) \), then its restriction is a winning strategy of the subgame \( \SGame_{\lts,\HES}^{d_1,\dots,d_{\ell}} \) on \( (d_1, \dots, d_n) \).
\end{lemma}
\begin{myproof}
  It suffices to show that \( \Strategy \) does not get stuck in the subgame \( \SGame_{\lts,\HES}^{d_1,\dots,d_{\ell}} \).
  Since proponent nodes are restricted to \( (p, X_i) \) with \( \ell < i \), by the definition of stability, \( \Strategy \) does not change the first \( \ell \) components.
\end{myproof}

\begin{lemma}\label{lem:appx:game-semantics:stably-winning-is-sound}
  Let \( \Strategy \) be a stable winning strategy of \( \SGame_{\lts,\HES} \) on \( \Nodes_0 \subseteq \Nodes \).
  For every \( 1 \le j \le n \), we have the following propositions:
  \begin{itemize}
  \item \( (P_{j}) \): For every \( 1 \le i \le j \), if \( \tilde{v} \cdot (p, X_i) \cdot (d_1, \dots, d_n) \) conforms with \( \Strategy \) and starts from \( \Nodes_0 \), then \( p \sqleq f_{i,j}(d_1, \dots, d_{j}) \).
  \item \( (Q_{j}) \): If \( \seq{v} \cdot (p, X_j) \cdot (d_1, \dots, d_n) \) conforms with \( \Strategy \) and starts from \( \Nodes_0 \), then \( p \sqleq g_j(d_1, \dots, d_{j-1}) \).
  \end{itemize}
\end{lemma}
\begin{myproof}
  By induction on the order \( P_{n} < Q_n < P_{n-1} < Q_{n-1} < \dots < P_{1} < Q_{1} \).
  Let \( \Strategy \) be a stable winning strategy of \( \SGame_{\lts,\HES} \) on \( \Nodes_0 \subseteq \Nodes \).

  (\( P_n \))
  By the definition of the game \( \SGame_{\lts,\HES} \) and Lemma~\ref{lem:appx:game-semantics:alg-interpret}, we have
  \[
  p \sqleq \sem{\form_i}([X_1 \mapsto d_1, \dots, X_n \mapsto d_n]) = f_{i,n}(d_1, \dots, d_n).
  \]

  (\( P_j \) with \( j < n \))
  We first show that \( d_{j+1} \sqleq g_{j+1}(d_1, \dots, d_j) \).
  Since \( \D_{\lts,\typ_{j+1}} \) is complete-prime algebraic, it suffices to show \( q \sqleq g_{j+1}(d_1, \dots, d_j) \) for every complete prime \( q \sqleq d_{j+1} \).
  Given a complete prime \( q \sqleq d_{j+1} \),
  \[
  \tilde{v} \cdot (p, X_i) \cdot (d_1, \dots, d_n) \cdot (q, X_{j+1})
  \]
  is a valid play that conforms with \( \Strategy \).
  Let
  \[
  (c_1, \dots, c_n) = \Strategy(\,  \tilde{v} \cdot (p, X_i) \cdot (d_1, \dots, d_n) \cdot (q, X_{j+1})\,).
  \]
  Since \( \Strategy \) is stable, \( d_1 = c_1, \dots, d_j = c_j \).
  By the induction hypothesis, \( q \sqleq g_{i+1}(c_1, \dots, c_j) = g_{i+1}(d_1, \dots, d_j) \).
  Since \( q \sqleq d_{j+1} \) is an arbitrary complete prime, we conclude that \( d_{j+1} \sqleq g_{j+1}(d_1, \dots, d_j) \).

  Then by Lemma~\ref{lem:appx:game-semantics:alg-interpret}, we have
  \[
  p \sqleq f_{i,j+1}(d_1, \dots, d_j, d_{j+1}) \sqleq f_{i,j+1}(d_1, \dots, d_j, g_{j+1}(d_1, \dots, d_j)) = f_{i,j}(d_1, \dots, d_j),
  \]
  where the first inequality follows from the induction hypothesis and the second from \( d_{j+1} \sqleq g_{j+1}(d_1, \dots, d_j) \) and monotonicity of \( f_{i,j+1} \).

  (\( Q_j \) with \( \munu_j = \nu \))
  For every complete prime \( q \sqleq d_j \), we have a play
  \[
  \tilde{v} \cdot (p, X_i) \cdot (d_1, \dots, d_n) \cdot (q, X_j)
  \]
  that conforms with \( \Strategy \).
  Let
  \[
  (c_{q,1}, \dots, c_{q,n}) = \Strategy(\, \tilde{v} \cdot (p, X_i) \cdot (d_1, \dots, d_n) \cdot (q, X_j) \,).
  \]
  Since \( \Strategy \) is stable, \( d_1 = c_{q,1}, \dots, d_{j-1} = c_{q,j-1} \) and \( d_{j} \sqsupseteq c_{q,j} \).
  By the induction hypothesis and monotonicity of \( f_{j,j} \), we have
  \[
  q \sqleq f_{j,j}(c_{q,1}, \dots, c_{q,j-1}, c_{q,j}) \sqleq f_{j,j}(d_1, \dots, d_{j-1}, d_j).
  \]
  Since the complete prime \( q \sqleq d_j \) is arbitrary and \( \D_{\lts,\HES} \) is complete-prime algebraic, we have
  \[
  d_j = \bigsqcup_{q \sqleq d_j} q \sqleq f_{j,j}(d_1, \dots, d_{j-1}, d_j).
  \]
  Hence \( d_j \) is a prefixed point of \( \lambda y. f_{j,j}(d_1, \dots, d_{j-1}, y) \).
  So \( d_j \sqleq \nu y. f_{j,j}(d_1, \dots, d_{j-1}, y) \).
  By the induction hypothesis and Lemma~\ref{lem:appx:game-semantics:alg-interpret}, we have
  \begin{align*}
    p
    &\sqleq f_{j,j}(d_1, \dots, d_{j-1}, d_j) \\
    &\sqleq f_{j,j}(d_1, \dots, d_{j-1}, \nu y. f_{j,j}(d_1, \dots, d_{j-1}, y)) \\
    &= \nu y. f_{j,j}(d_1, \dots, d_{j-1}, y) \\
    &= g_{j}(d_1, \dots, d_j).
  \end{align*}

  (\( Q_j \) with \( \munu_j = \mu \))
  Let \( \tilde{v} \cdot (p, X_j) \cdot (d_1, \dots, d_n) \) be a play that conforms with \( \Strategy \) and starts from \( \Nodes_0 \).
  Given opponent nodes \( O_1, O_2 \in \Nodes_{O} \), we write \( O_1 \succ O_2 \) if there exists a play of the form
  \[
  \seq{v} \cdot (p, X_j) \cdot \seq{v}_1 \cdot O_1 \cdot \seq{v}_2 \cdot O_2
  \]
  that conforms with \( \Strategy \) such that every proponent node in \( \seq{v}_1 \) or \( \seq{v}_2 \) is of the form \( (q, X_j) \).
  Every play of the form
  \[
  \tilde{v} \cdot (p, X_j) \cdot (d_1, \dots, d_n) \cdot (p_1, X_j) \cdot (c_{1,1}, \dots, c_{1,n}) \cdot (p_2, X_j) \cdot (c_{2,1}, \dots, c_{2,n}) \dots
  \]
  that conforms with \( \Strategy \) eventually terminates because \( \Strategy \) is a winning strategy and the priority of \( (p_i, X_j) \) is odd (recall that opponent node \( O_i \) has priority \( 0 \)).
  Hence the relation \( \succ \) defined above is well-founded.
  We prove the following claim by induction on the well-founded relation \( \succ \):
  \begin{quote}
    Let \( \tilde{v} \cdot (p, X_j) \cdot (d_1, \dots, d_n) \cdot \tilde{v}_1 \cdot (c_1, \dots, c_n) \) be a play that conforms with \( \Strategy \) and starts from \( \Nodes_0 \).
    Suppose that every proponent node in \( \tilde{v}_1 \) is of the form \( (q, X_j) \).
    Then \( c_j \sqleq \mu y. f_{j,j}(d_1, \dots, d_{j-1}, y) \).
  \end{quote}
  It suffices to show that \( q \sqleq \mu y. f_{j,j}(d_1, \dots, d_{j-1}, y) \) for every complete prime \( q \sqleq c_j \).
  Consider the play
  \[
  \tilde{v} \cdot (p, X_j) \cdot (d_1, \dots, d_n) \cdot \tilde{v}_1 \cdot (c_1, \dots, c_n) \cdot (q, X_j) \cdot (c'_1, \dots, c'_n)
  \]
  that conforms with \( \Strategy \) and starts from \( \Nodes_0 \).
  Because \( \Strategy \) is stable and \( \tilde{v}_1 \) does not contain a node of the form \( (r, X_k) \) with \( k < j \), we have \( d_1 = c_1 = c'_1, \dots, d_{j-1} = c_{j-1} = c'_{j-1} \).
  By the induction hypothesis of the above claim, \( c'_j \sqleq \mu y. f_{j,j}(d_1, \dots, d_{j-1}, y) \).
  By the induction hypothesis of the lemma, \( q \sqleq f_{j,j}(c'_1, \dots, c'_{j-1}, c'_{j}) \).
  Hence
  \begin{align*}
    q
    &\sqleq f_{j,j}(c'_1, \dots, c'_{j-1}, c'_{j}) \\
    &\sqleq f_{j,j}(c'_1, \dots, c'_{j-1}, \mu y. f_{j,j}(d_1, \dots, d_{j-1}, y)) \\
    &\sqleq f_{j,j}(d_1, \dots, d_{j-1}, \mu y. f_{j,j}(d_1, \dots, d_{j-1}, y)) \\
    &= \mu y. f_{j,j}(d_1, \dots, d_{j-1}, y).
  \end{align*}
  So \( q \sqleq \mu y. f_{j,j}(d_1, \dots, d_{j-1}, y) \) for every complete prime \( q \sqleq c_{j} \) and thus \( c_j \sqleq \mu y. f_{j,j}(d_1, \dots, d_{j-1}, y)) \).
  By the same argument, we have \( d_j \sqleq \mu y. f_{j,j}(d_1, \dots, d_{j-1}, y) \).
  This completes the proof of the above claim.

  By the induction hypothesis, the above claim and Lemma~\ref{lem:appx:game-semantics:alg-interpret}, we have
  \begin{align*}
    p
    &\sqleq f_{j,j}(d_1, \dots, d_{j-1}, d_j) & \mbox{(induction hypothesis)} \\
    &\sqleq f_{j,j}(d_1, \dots, d_{j-1}, \mu y. f_{j,j}(d_1, \dots, d_{j-1}, y)) &\mbox{(the above claim)}\\
    &\sqleq \mu y. f_{j,j}(d_1, \dots, d_{j-1}, y) \\
    &= g_{j}(d_1, \dots, d_{j-1}) &\mbox{(Lemma~\ref{lem:appx:game-semantics:alg-interpret}).}
  \end{align*}
\end{myproof}

\begin{lemma}\label{lem:appx:game-semantics-stable-wlog}
  If there exists a winning strategy of \( \SGame_{\lts,\HES} \) on an Opponent node \( (d_1, \dots, d_n) \), there exists a stable winning strategy on a node \( (d'_1, \dots, d'_n) \) with \( d_1 \sqsubseteq d'_1, \dots, d_n \sqsubseteq d'_n \).
\end{lemma}
\begin{myproof}
  Let \( \Strategy \) be a winning strategy on the node \( (d_1, \dots, d_n) \).
  Since \( \SGame_{\lts,\HES} \) is a parity game, we can assume without loss of generality that \( \Strategy \) is memoryless.
  Given nodes \( v \) and \( v' \) of the game and \( j \in \{\, 1, \dots, n \,\} \), we say \( v' \) is \emph{reachable from \( v \) following \( \Strategy \)} if there exists a play \( v v_1 \dots v_\ell v' \) that conforms with \( \Strategy \).
  We say that \( v' \) is \emph{\( j \)-reachable from \( v \) following \( \Strategy \)} if furthermore, for every proponent node \( (p, X_i) \) in \( v_1 \dots v_k \), we have \( j < i \).

  For every \( k = 1, \dots, n \), let
  \[
  d'_k = d_k \sqcup \bigsqcup \{\, e_k \mid \mbox{\( (e_1, \dots, e_n) \) is reachable from \( (d_1, \dots, d_n) \) following \( \Strategy \)} \,\}.
  \]
  Obviously \( d_i \sqsubseteq d'_i \).
  
  Let \( \Strategy' \) be a partial function \( \Nodes^* \Nodes_P \rightharpoonup \Nodes \) defined by
  \[
  \seq{v} \cdot (c_1, \dots, c_n) \cdot (q, X_j) \mapsto (c_1, \dots, c_{j-1}, c'_j, \dots, c'_n)
  \]
  where
  \[
  c'_k = \bigsqcup \{\, e_k \mid \mbox{\( (e_1, \dots, e_n) \) is \( j \)-reachable from \( (q, X_j) \) following \( \Strategy \)} \,\}.
  \]

  We show that \( \Strategy' \) is indeed a strategy on \( (d_1', \dots, d_n') \).
  Let \( v_0 = (d'_1, \dots, d'_n) \) and assume that \( v_0 v_1 \dots v_{\ell} \) (\( \ell > 0 \)) is a play that conforms with \( \Strategy' \) and ends with an opponent node \( v_{\ell} \in \Nodes_{O} \) (hence \( \ell > 0 \) is even).
  Let \( v_{\ell} = (c_1, \dots, c_n) \).
  We first prove that for every \( k = 1, \dots n \),
  \[
  c_k \sqsupseteq \bigsqcup \{\, e_k \mid \mbox{\( (e_1, \dots, e_n) \) is \( k \)-reachable from \( v_{\ell-1} \) following \( \Strategy \)} \,\}
  \]
  by induction on \( \ell > 0 \).
  Let \( v_{\ell - 1} = (p, X_j) \).
  \begin{itemize}
  \item
    If \( k \ge j \), then by the definition of \( \Strategy' \),
    \begin{align*}
      c_k &= \bigsqcup \{\, e_k \mid \mbox{\( (e_1, \dots, e_n) \) is \( j \)-reachable from \( (p, X_j) \) following \( \Strategy \)} \,\}
      \\
      &\sqsupseteq \bigsqcup \{\, e_k \mid \mbox{\( (e_1, \dots, e_n) \) is \( k \)-reachable from \( (p, X_j) \) following \( \Strategy \)} \,\}
    \end{align*}
    because the \( k \)-reachable set is a subset of the \( j \)-reachable set.
  \item
    If \( k < j \), then \( c_k \) is the \( k \)-th component of \( v_{\ell - 2} \) by the definition of \( \Strategy' \).
    \begin{itemize}
    \item
      If \( \ell = 2 \), then
      \begin{align*}
        c_k = d'_k &= d_k \sqcup \bigsqcup \{\, e_k \mid \mbox{\( (e_1, \dots, e_n) \) is reachable from \( (d_1, \dots, d_n) \) following \( \Strategy \)} \,\}
        \\
        &\sqsupseteq \bigsqcup \{\, e_k \mid \mbox{\( (e_1, \dots, e_n) \) is \( k \)-reachable from \( v_1 \) following \( \Strategy \)} \,\}
      \end{align*}
      since \( v_1 \) is reachable from \( (d_1, \dots, d_n) \) following \( \Strategy \).
    \item
      Assume that \( \ell > 2 \).
      By the induction hypothesis for the subsequence \( v_0 v_1 \dots v_{\ell-2} \), we have
      \begin{align*}
        c_k
        &\sqsupseteq \bigsqcup \{\, e_k \mid \mbox{\( (e_1, \dots, e_n) \) is \( k \)-reachable from \( v_{\ell - 3} \) following \( \Strategy \)} \,\}
        \\
        &\sqsupseteq \bigsqcup \{\, e_k \mid \mbox{\( (e_1, \dots, e_n) \) is \( k \)-reachable from \( v_{\ell - 1} \) following \( \Strategy \)} \,\}
      \end{align*}
      since \( v_{\ell - 1} \) is \( k \)-reachable from \( v_{\ell - 3} \) following \( \Strategy \).
    \end{itemize}
  \end{itemize}
  Let \( v_0 \cdot v_1 \cdot \dots \cdot v_\ell \cdot (p, X_i) \) be a play that conforms with \( \Strategy' \) and starts from \( v_0 = (d_1', \dots, d_n') \).
  Let \( (c_1, \dots, c_n) = \Strategy'(v_0 \cdot v_1 \cdot \dots \cdot v_\ell \cdot (p, X_i)) \) and \( (e_1, \dots, e_n) = \Strategy(\,(p,X_i)\,) \).
  Since \( (e_1, \dots, e_n) \) is reachable from \( (p,X_i) \) following \( \Strategy \), the above claim shows that \( e_m \sqleq c_m \) for every \( 1 \le m \le n \).
  Since \( p \sqleq \sem{\form_i}([X_1 \mapsto e_1, \dots, X_n \mapsto e_n]) \), we have \( p \sqleq \sem{\form_i}([X_1 \mapsto c_1, \dots, X_n \mapsto c_n]) \).
  Hence \( \Strategy' \) is indeed a strategy.

  By definition and the above claim, it is easy to see that \( \Strategy' \) is a stable strategy.
  
  We show that the strategy \( \Strategy' \) is winning on \( (d_1', \dots, d_n') \).
  Assume that \( v_0 v_1 \dots \) be an infinite play that conforms with \( \Strategy' \) and starts from \( (d_1', \dots, d_n') \).
  Let \( j \) be the minimum index of variables appearing infinitely often in the play.
  Then the play can be split into
  \[
  \seq{w}_0 \cdot (p_1, X_j) \cdot \seq{w}_1 \cdot (c_{1,1}, \dots, c_{1,n}) \cdot (p_2, X_j) \cdot \seq{w}_2 \cdot (c_{2,1}, \dots, c_{2,n}) \cdot (p_3, X_j) \cdot \dots
  \]
  where \( \tilde{w}_k \) is a (possible empty) sequence of nodes, which furthermore consists of \( \Nodes_O \cup \{\, (q, X_{\ell}) \mid j < \ell \,\} \) if \( k \ge 1 \).
  For every \( m = 1, 2, \dots \), we have
  \[
  c_{m, j} = \bigsqcup \{\, e_j \mid \mbox{\( (e_1, \dots, e_n) \) is \( j \)-reachable from \( (p_m, X_j) \) following \( \Strategy \)} \,\}
  \]
  since \( j \)-th component of opponent nodes is unchanged during \( \seq{w}_{m} \).
  By the definition of edges of the game, \( p_{m+1} \sqsubseteq c_{m,j} \).
  Since \( p_{m+1} \) is a complete prime, there exists \( (e_1, \dots, e_n) \) that is \( j \)-reachable from \( (p_m, X_j) \) following \( \Strategy \) and such that \( p_{m+1} \sqsubseteq e_j \).
  By definition of \( j \)-reachability, there exists a sequence of nodes \( \seq{w}_m' \) such that \( (p_m, X_j) \cdot \seq{w}_m' \cdot (p_{m+1}, X_j) \) conforms with \( \Strategy \).
  Furthermore \( \seq{w}_m' \) consists of \( \Nodes_O \cup \{\, (q, X_{\ell}) \mid j < \ell \,\} \) if \( k \ge 1 \).
  It is easy to see that \( (p_1, X_j) \) is reachable from \( v'_0 = (d_1, \dots, d_n) \) following \( \Strategy \).
  Hence we have a sequence
  \[
  v'_0 \cdot \seq{w}_0' \cdot (p_1, X_j) \cdot \seq{w}_1' \cdot (p_2, X_j) \cdot \seq{w}_2' \cdot \dots
  \]
  that conforms with \( \Strategy \) and starts from \( v_0' = (d_1, \dots, d_n) \).
  Since \( \Strategy \) is winning on \( v_0' \), this sequence satisfies the parity condition, i.e.~the priority of \( X_j \) is even.
  Hence the original play, which conforms with \( \Strategy' \), also satisfies the parity condition.
\end{myproof}

We are now ready to prove soundness and completeness of the type-based characterization.
\begin{theorem}\label{thm:appx:game-semantics:game-and-interpretation}
  \( p \sqleq \sem{\HES}_{\lts} \) if and only if Proponent wins \( \SGame_{\lts,\HES} \) on \( (p, X_1) \).
\end{theorem}
\begin{myproof}
  (\(\Rightarrow\))
  By Lemma~\ref{lem:appx:game-semantics:completeness}, Proponent wins \( \SGame_{\lts,\HES} \) on \( (p, \bot, \dots, \bot) \).
  Since \( (p, \bot, \dots, \bot) \cdot (p, X_1) \) is a valid play, Proponent also wins on \( (p, X_1) \).
  
  (\(\Leftarrow\))
  If Proponent wins on \( (p, X_1) \), then Proponent also wins on \( (q, X_1) \) for every \( q \sqleq p \).
  Hence Proponent wins on \( (p, \bot, \dots, \bot) \).
  By Lemma~\ref{lem:appx:game-semantics-stable-wlog}, we have a stable winning strategy \( \Strategy \) of \( \SGame_{\lts,\HES} \) on \( (d_1, \dots, d_n) \) with \( p \sqleq d_1 \).
  Then \( (d_1, \dots, d_n) \cdot (p, X_1) \) is a play that conforms with \( \Strategy \) and starts from \( (d_1, \dots, d_n) \).
  Hence, by Lemma~\ref{lem:appx:game-semantics:stably-winning-is-sound}, \( p \sqleq g_1() \).
By Lemma~\ref{lem:appx:game-semantics:alg-interpret}, we have \( g_1() = \sem{\HES}_{\lts} \).
Therefore, we have \( p \sqleq \sem{\HES}_{\lts} \) as required.
\end{myproof}

\subsection{The Opposite Parity Game}
This subsection defines a game in which Proponent tries to disprove \( p \sqleq \sem{\HES}_{\lts} \) by giving an upper-bound \( \sem{\HES}_{\lts} \sqleq q \).
A similar construction can be found in Salvati and Walukiewicz~\cite{SALVATI2014340}; our construction can be seen as an infinite variant of them.

We first define the notion of complete coprimes, which is the dual of complete primes.
\begin{definition}[Complete coprime]
  Let \( (A, \le) \) be a complete lattice.
  An element \( p \in A \) is a \emph{complete coprime} if (1) \( p \neq \top \) and (2) for every \( U \subseteq A \), if \( (\bigwedge_{x \in U} x) \le p \), then \( x \le p \) for some \( x \in U \).
\end{definition}

\begin{lemma}\label{lem:appx:game-semantics:complete-prime-alg-implies-complete-coprime-alg}
  If \( (A, \le) \) is complete-prime algebraic, then for every \( x \in A \),
  \[
  x = \pGlb \{\, p \mid x \le p, \;\mbox{\( p \): complete coprime} \,\}.
  \]
\end{lemma}
\begin{myproof}
  Obviously \( x \le \pGlb \{\, p \mid x \le p, \;\mbox{\( p \): complete coprime} \,\} \).
  We show that \( x \ge \pGlb \{\, p \mid x \le p, \;\mbox{\( p \): complete coprime} \,\} \).
  Since \( (A, \le) \) is complete-prime algebraic, 
by Lemma~\ref{lem:complete-prime-leq},
it suffices to show that, for every complete prime \( q \) such that \( q \nleq x \), 
we have \(q\nleq \pGlb \{\, p \mid x \le p, \;\mbox{\( p \): complete coprime} \,\}\),
i.e., there exists a complete coprime \( p \ge x \) such that \( q \nleq p \).
  Let \( q \) be a complete prime such that \( q \nleq x \). We show that
  \[
  p = \pLub \{\, q' \mid q \nleq q', \;\mbox{\( q' \): complete prime} \,\}
  \]
is indeed a complete coprime \( p \ge x \) such that \( q \nleq p \).
  We first check \( q \nleq p \).
If it were the case that \(q\leq p\), then \( q \le p = \pLub \{\, q' \mid q \nleq q', \;\mbox{\( q' \): complete prime} \,\} \), which would imply \( q \le q' \) for some \( q' \) with \( q \nleq q' \) (as \( q \) is a complete prime), a contradiction.
  It remains to check that \( p \) is a complete coprime.
  Since \( q \nleq p \), we have \( p \neq \top \).
  Suppose that \( \pGlb_{y \in U} y \le p \).
  If \( q \le y \) for every \( y \in U \), then \( q \le p \), a contradiction.
  Hence \( q \nleq y \) for some \( y \in U \).
  Now
  \begin{align*}
    y
    &= \pLub \{\, q'' \mid q'' \le y, \;\mbox{\( q'' \): complete prime} \,\} \\
    &\le \pLub \{\, q'' \mid q \nleq q'' , \;\mbox{\( q'' \): complete prime} \,\} \\
    &= p.
  \end{align*}
  Here the second inequality comes from the 
fact that \(q''\le y\) implies \(q\nleq q''\); in fact, 
if \( q'' \le y \) and \( q \le q'' \), then \( q \le q'' \le y \), hence a contradiction.
\end{myproof}

\begin{definition}[Opposite parity game]
  The parity game \( \CounterGame_{\lts, \HES} \) is defined by the following data:
  \begin{align*}
    \Nodes_{P} :=\;& \{\, (p, X_i) \mid p \in \D_{\lts, \typ_i}, \;\textrm{\( p \): complete coprime} \,\}
    \\
    \Nodes_{O} :=\;& \{\, (d_1, \dots, d_n) \mid d_1 \in \D_{\lts, \typ_1}, \dots, d_n \in \D_{\lts, \typ_n} \,\}
    \\
    \Edges
    :=\;& \{\, ((p, X_i), (d_1, \dots, d_n)) \mid p \sqsupseteq \sem{\form_i}([X_1 \mapsto d_1, \dots, X_n \mapsto d_n]) \,\}
    \\
    \cup\;& \{\, ((d_1, \dots, d_n), (p, X_i)) \mid p \sqsupseteq d_i, \;\textrm{\( p \): complete coprime} \,\}.
  \end{align*}
  The priority of the opponent node is \( 0 \); the priority of node \( (p, X_i) \) is \( 2(n-i)+1 \) if \( \munu_i = \nu \) and \( 2(n-i)+2 \) if \( \munu_i = \mu \).
\end{definition}
The game \( \CounterGame_{\lts,\HES} \) is obtained by replacing complete primes with complete coprimes and \( \sqleq \) with \( \sqsupseteq \).
The position \((p,X_i)\) is the state where Proponent ties to show that
\(\sem{(\HES,X_i)} \sqleq p\); to this end, Proponent picks a valuation
\([X_1\mapsto d_1,\ldots,X_n\mapsto d_n]\) such that 
\((p \sqsupseteq \sem{\form_i}([X_1 \mapsto d_1, \dots, X_n \mapsto d_n])\). Opponent challenges 
the validity of the valuation
\([X_1 \mapsto d_1, \dots, X_n \mapsto d_n]\), by picking \(i\) and a complete coprime \(p'\) such that 
\(p' \sqsupseteq d_i\), and asking why \(\sem{(\HES,X_i)}\sqsubseteq p'\) holds.

\begin{lemma}\label{lem:appx:game-semantics:opposite-game}
  If Proponent wins \( \CounterGame_{\lts,\HES} \) on \( (p, X_i) \) (where \( p \) is a complete coprime), then for every complete prime \( p' \in \D_{\lts,\typ_i} \) such that \( p' \nsqsubseteq p \), Opponent wins \( \SGame_{\lts,\HES} \) on \( (p', X_i) \).
\end{lemma}
\begin{myproof}
  Let \( \Strategy \) be a P-winning strategy of \( \CounterGame_{\lts,\HES} \) on \( (p, X_i) \) (where \( p \) is a complete coprime).
  We can assume without loss of generality that \( \Strategy \) is memoryless.
  Let \( p' \in \D_{\lts,\typ_i} \) be a complete prime such that \( p' \nsqsubseteq p \).

  Let \( (q, X_j) \) be a Proponent node in \( \CounterGame_{\lts,\HES} \) and \( (q', X_{j'}) \) be a Proponent node in \( \SGame_{\lts,\HES} \).
  We say that \( (q, X_j) \) is \emph{inconsistent with \( (q', X_{j'}) \)} if \( j = j' \) and \( q' \nsqsubseteq q \).
  Similarly, given Opponent nodes \( (d_1, \dots, d_n) \) in \( \CounterGame_{\lts,\HES} \) and \( (d'_1, \dots, d'_n) \) in \( \SGame_{\lts,\HES} \), those nodes are \emph{inconsistent} if \( d'_j \nsqsubseteq d_j \) for some \( 1 \le j \le n \).
  
  Let \( (p', X_i) \cdot v'_1 \cdot v'_2 \cdot \dots \cdot v'_k \) be a play of \( \SGame_{\lts,\HES} \).
  This play is \emph{admissible} (with respect to \( \Strategy \)) if there exists a play \( (p, X_i) \cdot v_1 \cdot v_2 \cdot \dots \cdot v_k \) of \( \CounterGame_{\lts,\HES} \) that conforms with \( \Strategy \) such that \( v_j \) is inconsistent with \( v'_j \) for every \( j \in\set{1, \dots, k} \).
  For each admissible play \( (p', X_i) \cdot v'_1 \cdot v'_2 \cdot \dots \cdot v'_k \) of \( \SGame_{\lts,\HES} \), we choose such a play \( (p, X_i) \cdot v_1 \cdot v_2 \cdot \dots \cdot v_k \) of \( \CounterGame_{\lts,\HES} \).
  Suppose that our choice satisfies the following property:
  \begin{quote}
    If \( (p, X_i) \cdot v_1 \cdot v_2 \cdot \dots \cdot v_k \) is the chosen play corresponding to \( (p', X_i) \cdot v'_1 \cdot v'_2 \cdot \dots \cdot v'_k \), then \( (p, X_i) \cdot v_1 \cdot v_2 \cdot \dots \cdot v_{k-1} \) is the chosen play for \( (p', X_i) \cdot v'_1 \cdot v'_2 \cdot \dots \cdot v'_{k-1} \).
  \end{quote}
  (A way to achieve this is to introduce a well-ordering on nodes of \( \CounterGame_{\lts,\HES} \) and to choose the minimum play w.r.t.~the lexicographic ordering from those that satisfy the requirement.)

  We define an Opponent strategy \( \CounterStrategy \) for \( \SGame_{\lts,\HES} \) defined on admissible plays starting from \( (p', X_i) \).
  Let
  \[
  (p', X_i) \cdot v'_1 \cdot v'_2 \cdot \dots \cdot v'_k \cdot (d'_1, \dots, d'_n)
  \]
  be an admissible play of \( \SGame_{\lts,\HES} \).
  Let
  \[
  (p, X_i) \cdot v_1 \cdot v_2 \cdot \dots \cdot v_k \cdot (d_1, \dots, d_n)
  \]
  be the chosen play of \( \CounterGame_{\lts,\HES} \) corresponding to the above play.
  So \( v_j \) and \( v'_j \) are inconsistent for every \( j \in\set{1, \dots, k} \) and \( (d_1, \dots, d_n) \) is inconsistent with \( (d'_1, \dots, d'_n) \).
  By definition, \( d'_{\ell} \nsqsubseteq d_{\ell} \) for some \( 1 \le \ell \le n \).
  Since \( \D_{\lts,\typ_{\ell}} \) is complete-prime algebraic, by Lemma~\ref{lem:appx:game-semantics:complete-prime-alg-implies-complete-coprime-alg},
  \[
  d'_{\ell} = \Lub \{\, q' \in \D_{\lts,\typ_{\ell}} \mid q' \sqleq d'_{\ell},\; \mbox{\( q' \): complete prime} \,\}
  \]
  and
  \[
  d_{\ell} = \Glb \{\, q \in \D_{\lts,\typ_{\ell}} \mid d_{\ell} \sqleq q,\; \mbox{\( q \): complete coprime} \,\}.
  \]
  Hence \( d'_{\ell} \nsqsubseteq d_{\ell} \) implies that there exist a complete prime \( q' \sqleq d'_{\ell} \) and a complete coprime \( d_{\ell} \sqleq q \) such that \( q' \nsqsubseteq q \).
  The Opponent strategy \( \CounterStrategy \) chooses \( (q', X_{\ell}) \) as the next node in \( \SGame_{\lts,\HES} \).

  We show that \( \CounterStrategy \) is indeed an O-strategy of \( \SGame_{\lts,\HES} \) on \( (p', X_i) \).
  It suffices to show that every play \( (p', X_i) \cdot v'_1 \cdot \dots \cdot v'_k \) that conforms with the Opponent strategy \( \CounterStrategy \) is admissible.
  If \( k = 0 \), the play is obviously admissible.
  Let \( (p', X_i) \cdot v'_1 \cdot \dots \cdot v'_k \) be a play that conforms with \( \CounterStrategy \) and assume that it is admissible.
  If \( v'_k \) is an Opponent node, the next node \( v'_{k+1} \) is determined by \( \CounterStrategy \) and then \( (p', X_i) \cdot v'_1 \cdot \dots \cdot v'_k \cdot v'_{k+1} \) is admissible by the definition of \( \CounterStrategy \).
  Consider the case that \( v'_k = (p'_k, X_{i_k}) \) is an Proponent node.
  Since this play is admissible, we have the chosen play \( (p, X_i) \cdot v_1 \cdot v_2 \cdot \dots \cdot v_k \) of \( \CounterGame_{\lts,\HES} \) corresponding to the above play.
  Then \( v_j \) and \( v'_j \) are inconsistent for every \( j = 1, \dots, k \).
  Since \( v_k \) is inconsistent with \( v'_k \), \( v_{k} = (p_k, X_{i_k}) \) with \( p'_k \nsqsubseteq p_k \).
  Let \( v'_{k+1} = (d'_1, \dots, d'_n) \) be an arbitrary Opponent node in \( \SGame_{\lts,\HES} \) such that \( (v_k, v_{k+1}) \) is valid in \( \SGame_{\lts,\HES} \).
  Let \( v_{k+1} = (d_1, \dots, d_n) \) be the Opponent node in \( \CounterGame_{\lts,\HES} \) determined by \( \Strategy \).
  By definition of edges, \( p'_k \sqleq \sem{\form_i}([X_1 \mapsto d'_1, \dots, X_n \mapsto d'_n]) \) and \( \sem{\form_i}([X_1 \mapsto d_1, \dots, X_n \mapsto d_n]) \sqleq p_k \).
  Since \( \sem{\form_i} \) is monotone, \( d'_j \sqleq d_j \) for every \( j = 1, \dots, n \) would imply \( p'_k \sqleq p_k \), a contradiction.
  Hence \( d'_j \nsqsubseteq d_j \) for some \( 1 \le j \le n \), that means, \( (d'_1, \dots, d'_n) \) is inconsistent with \( (d_1, \dots, d_n) \).
  Therefore the play \( (p, X_i) \cdot v_1 \cdot v_2 \cdot \dots \cdot v_k \cdot (d'_1, \dots, d'_n) \) is admissible for every choice of \( (d'_1, \dots, d'_n) \).
  So \( \CounterStrategy \) is defined for every play that conforms with \( \CounterStrategy \) and starts from \( (p', X_i) \).

  We show that \( \CounterStrategy \) is O-winning.
  As \( \CounterStrategy \) does not get stuck, it suffices to show that every infinite play following \( \CounterStrategy \) does not satisfy the parity condition.
  Assume an infinite play
  \[
  (p', X_i) \cdot v'_1 \cdot (p'_1, X_{i'_1}) \cdot v'_2 \cdot (p'_2, X_{i'_2}) \cdot \dots \cdot v'_\ell \cdot (p'_{\ell}, X_{i'_{\ell}}) \cdot \dots
  \]
  following \( \CounterStrategy \).
  By the definition of \( \CounterStrategy \), every finite prefix of the play is admissible.
  Let \( (p', X_i) \cdot v_{k,1} \cdot \dots \cdot v_{k,k} \) be the chosen play for the prefix of the above play of length \( k \).
  Then, by the assumption on the choice, the choice of each node does not depend on the length of the prefix, i.e.~\( v_{k,i} = v_{k',i} \) for every \( i \ge 1 \) and \( k, k' \ge i \).
  Hence the above infinite play is ``admissible'' in the sense that there exists an infinite play in \( \CounterGame_{\lts,\HES} \)
  \[
  (p, X_i) \cdot v_1 \cdot (p_1, X_{i_1}) \cdot v_2 \cdot (p_2, X_{i_2}) \cdot \dots \cdot v_{\ell} \cdot (p_{\ell}, X_{i_{\ell}}) \cdot \dots
  \]
  that conforms with \( \Strategy \) and such that \( (p_{\ell}, X_{i_{\ell}}) \) and \( (p'_{\ell}, X_{i'_{\ell}}) \) are inconsistent for every \( \ell \).
  Hence \( i_{\ell} = i'_{\ell} \) for every \( \ell \).
  Because
  \[
  \mbox{(priority of \( (p_{\ell}, X_{i_{\ell}}) \) in \( \CounterGame_{\lts,\HES} \))} = \mbox{(priority of \( (p'_{\ell}, X_{i_{\ell}}) \) in \( \SGame_{\lts,\HES} \))} + 1,
  \]
  the latter play satisfies the parity condition if and only if the former play does not satisfy the parity condition.
  Since \( \Strategy \) is a winning strategy, we conclude that the former play does not satisfy the parity condition.
\end{myproof}

\section{Proof for Section~\ref{sec:liveness}}
\label{sec:proof-liveness}

\subsection{Some definitions}
Here we introduce some notions and notations, which are useful in proofs of results of Section~\ref{sec:liveness}.

\begin{definition}[Choice sequences]
  Let \( \RTyPath \) be an infinite sequence consisting of \( \PLeft \) and \( \PRight \), called a \emph{choice sequence}.
  We define a reduction relation for pairs of terms and choice sequences by
  \[
  (\term_1 \nondet \term_2;\, \PLeft \RTyPath) \stackrel{\epsilon}{\red}_D (\term_1; \RTyPath)
  \qquad
  (\term_1 \nondet \term_2;\, \PRight \RTyPath) \stackrel{\epsilon}{\red}_D (\term_2; \RTyPath)
  \]
  and \( (u; \RTyPath) \stackrel{\ell}{\red}_D (u'; \RTyPath) \) if \( u \neq u_1 \nondet u_2 \) and \( u \stackrel{\ell}{\red}_D u' \). We sometimes omit labels and just write
\( (u; \RTyPath) {\red}_D (u'; \RTyPath) \)
  for   \( (u; \RTyPath) \stackrel{\ell}{\red}_D (u'; \RTyPath) \).
  A choice sequence resolves nondeterminism of a program \( P \); a reduction sequence is completely determined by its initial term, its length and a choice sequence.
  We write \( \red_D^{k} \) for the \( k \)-step reduction relation.
  If \( (t_0, \RTyPath_0) \red_D (t_1, \RTyPath_1) \red_D \cdots \red_D (t_k, \RTyPath_k) \), then we write \( \RTyPath_0 \Vdash t_0 \red_D t_1 \red_D \cdots \red_D t_k \) and say that the reduction sequence \( t_0 \red_D t_1 \red_D \cdots \red_D t_k \) follows \( \RTyPath_0 \).
\end{definition}

We introduce a slightly more elaborate notion of the recursive call relation below.
Note that the relation \(f\,\seq{t}\recall{\progd}g\,\seq{u}\) in Definition~\ref{def:recall}
coincides with
\(f\,\seq{t}\recall{\progd}g\,\seq{u}\) redefined below.
\nk{Confirm this. I made the notation the same.}
\begin{definition}[Recursive call relation, call sequences]
  Let \( \prog = (\progd, \mainfun) \) be a program, with \( \progd = \{ f_i\,\seq{x}_i = u_i \}_{1 \le i \le n} \).
  We define \( D^{\Mark} := D \cup \{ f^{\Mark}_i\,\tilde{x} = u_i \}_{1 \le i \le n} \) where \( f^{\Mark}_1, \dots, f^{\Mark}_n \) are fresh symbols.
  So \( D^{\Mark} \) has two copies of each function symbol, one of which is marked.
  For the terms \(\seq{t}_i\) and \(\seq{t}_j\) that do not contain marked symbols,
  we write \( (f_i\,\seq{t}_i; \RTyPath) \CallSeqN{k+1}_{\progd} (f_j\,\seq{t}_j; \RTyPath') \) if
  \[
  ([\seq{t}_i/\seq{x}_i][f_1^\Mark/f_1,\ldots,f_n^\Mark/f_n]u_i;\, \RTyPath) \red^{k}_{\progd^\Mark} (f_j^\Mark\,\seq{t}'_j;\, \RTyPath')
  \]
  \tkchanged{(then \(\seq{t}_j\) is obtained by erasing all the marks in \(\seq{t}'_j\)).}
  If there exists a (finite or infinite) sequence \( (f\,\tilde{s}; \RTyPath) \CallSeqN{k_1}_D (g_1\,\tilde{u}_1; \RTyPath_1) \CallSeqN{k_2}_D (g_2\,\tilde{u}_2; \RTyPath_2) \CallSeqN{k_3}_D \cdots \), then we write \( \RTyPath \Vdash f\,\tilde{s} \CallSeqN{k_1}_D g_1\,\tilde{u}_1 \CallSeqN{k_2}_D g_2\,\tilde{u}_2 \CallSeqN{k_3}_D \cdots \) and call it a \emph{call sequence of \( f\,\tilde{s} \) following \( \RTyPath \)}.
  We often omit the number \( k_j \) of steps and the choice sequence \( \RTyPath \).
  Given a program \( P = (D, \mainfun) \), the \emph{set of call sequences of \( P \)} is a subset of (finite or infinite) sequences of function symbols defined by
  \( \Callseq(P) := \{\, \mainfun\, g_1\, g_2 \dots \mid \mainfun \CallSeq_D g_1\,\seq{u}_1 \CallSeq_D g_2\,\seq{u}_2 \CallSeq_D \cdots \,\} \).
\end{definition}
Note that by the definition above,
\( (f_i\,\seq{t}_i; \RTyPath) \CallSeqN{k}_{\progd} (f_j\,\seq{t}_j; \RTyPath') \)
implies
\( (f_i\,\seq{t}_i; \RTyPath) \red^{k}_{\progd} (f_j\,\seq{t}_j; \RTyPath') \).
We write \( \term^{\Mark} \) for the term obtained by marking all function symbols in \( \term \), i.e.~\( [f_1^{\Mark}/f_1, \ldots, f_n^{\Mark}/f_n]\term \).

\subsection{Existence of a unique infinite call-sequence}
Let \( P = (D, \mainfun) \) be a program and
\[
\RTyPath \Vdash \mainfun \stackrel{\ell_1}{\red}_D t_1 \stackrel{\ell_2}{\red}_D \cdots
\]
be the infinite reduction sequence following \( \RTyPath \).
This subsection proves that there exists a unique infinite call-sequence
\[
\RTyPath \Vdash f\,\tilde{u} \CallSeq_D g_1\,\tilde{u}_1 \CallSeq_D g_2\,\tilde{u}_2 \CallSeq_D \cdots
\]
of this reduction sequence.

\begin{lemma}
  Let \( P = (D, \mainfun) \) be a program and \( \RTyPath \in \{ \PLeft, \PRight \}^{\omega} \) be a choice sequence.
  Suppose we have an infinite reduction sequence
  \[
  \RTyPath \Vdash \mainfun \red_D t_1 \red_D t_2 \red_D \dots.
  \]
  Then there exists an infinite call-sequence
  \[
  \RTyPath \Vdash \mainfun \CallSeq_D g_1\,\tilde{u}_1 \CallSeq_D g_2\,\tilde{u}_2 \CallSeq_D \cdots.
  \]
\end{lemma}
\begin{myproof}
  Kobayashi and Ong~\cite[Appendix~B]{KO09LICSfull} have proved a similar result for simply-typed programs without integers.
  We remove integers from \( D \) as follows:
  \begin{itemize}
  \item First, replace each conditional branching \( \ifexp{p(\term'_1,\ldots,\term'_k)}{\term_1}{\term_2} \) with the nondeterministic branching \( \term_1 \nondet \term_2 \).
  \item Then replace each integer expression with the unit value.
  \end{itemize}
  Let us write \( D' \) for the resulting function definitions.
  Then there exists a choice \( \RTyPath' \) such that \( \RTyPath' \Vdash \mainfun \red_{D'} t_1' \red_{D'} t_2' \red_{D'} \cdots \) and, for every \( i \ge 1 \), \( t_i' \) is obtained from \( t_i \) by the above translation.
  By the result of Kobayashi and Ong, this sequence has an infinite call-sequence.
  This call-sequence can be transformed into a call-sequence of the original reduction sequence.
\end{myproof}

We prove uniqueness.

Suppose that \( (f\,\seq{u}; \RTyPath) \red^k_{\progd} (g\,\seq{v}; \RTyPath') \).
Given an occurrence of \( g \) in \( f\,\seq{u} \), we write \( f'\,\seq{u}' \) for the term obtained by replacing the occurrence of \( g \) with \( g^{\Mark} \).
We say that the head-occurrence \( g \) in \( g\,\seq{v} \) is a \emph{copy of the occurrence of \( g \) in \( f\,\seq{u} \)} if \( (f'\,\seq{u}'; \RTyPath) \red^k_{\progd^{\Mark}} (g^{\Mark}\,\seq{v}'; \RTyPath) \).
An occurrence of \( h \) in \( f\,\seq{u} \) is an \emph{ancestor of the head-occurrence of \( g \) in \( g\,\seq{v} \)} if
\[
(f\,\seq{u}; \RTyPath) \red^{k_0}_{\progd} (h_1\,\seq{w}_1; \RTyPath_1) \CallSeqN{k_1}_{\progd} (h_2\,\seq{w}_2; \RTyPath_2) \CallSeqN{k_2}_{\progd} \dots \CallSeqN{k_\ell}_{\progd} (h_\ell\,\seq{w}_{\ell}; \RTyPath_{\ell}) = (g\,\seq{v}; \RTyPath')
\]
where \( \ell \ge 1 \), \( k = \sum_{j = 0}^{\ell} k_{j} \) and \( h_1 \) is a copy of the occurrence of \( h \) in \( f\,\seq{u} \) (hence \( h_1 = h \) as function symbols).
It is not difficult to see that, given a reduction sequence \( (f\,\seq{u}; \RTyPath) \red^k_{\progd} (g\,\seq{v}; \RTyPath') \), there exists a unique ancestor of \( g \) in \( f\,\seq{u} \).

\begin{lemma}
  For every term \( t \) and choice sequence \( \RTyPath \in \{ \PLeft, \PRight \}^{\omega} \), there exists at most one infinite call-sequence of \( t \) following \( \RTyPath \).
\end{lemma}
\begin{myproof}
  Let \( \RTyPath \Vdash t_0 \red_D t_1 \red_D t_2 \red_D \cdots \) be the infinite reduction sequence following \( \RTyPath \).
  In this proof, we shall consider only sequences following \( \RTyPath \) and hence we shall omit \( \RTyPath \Vdash \).

  Assume that there exist difference infinite call-sequences, say,
  \[
  t_0 \CallSeq_D g_{1}\,\tilde{u}_{1} \CallSeq_D \cdots \CallSeq_D g_{k}\,\tilde{u}_{k} \CallSeq_D h_{1}\,\tilde{v}_{1} \CallSeq_D h_{2}\,\tilde{v}_{2} \CallSeq_D \cdots
  \]
  and
  \[
  t_0 \CallSeq_D g_{1}\,\tilde{u}_{1} \CallSeq_D \cdots \CallSeq_D g_{k}\,\tilde{u}_{k} \CallSeq_D k_{1}\,\tilde{w}_{1} \CallSeq_D k_{2}\,\tilde{w}_{2} \CallSeq_D  \cdots
  \]
  where \( h_{1} \) and \( k_{1} \) are different.
  Since both call sequences are infinite, the reduction sequence has infinite switching of the sequence, i.e.~the reduction sequence must be of the following form:
  \[
  t_0 \red^*_D h_{i_1}\,\tilde{v}_{i_1} \red^*_D k_{j_1}\,\tilde{w}_{j_1} \red^*_D h_{i_2}\,\tilde{v}_{i_2} \red^*_D k_{j_2}\,\tilde{w}_{j_2} \red^*_D \cdots.
  \]
  Since \( h_{i_1} \) is not an ancestor of \( k_{j_1} \), an ancestor of \( k_{j_1} \) must appear in \( \tilde{v}_{i_1} \), say, in the \( \ell_1 \)-th argument \( v_{i_1, \ell_1} \).
  Since \( k_{j_1} \) appears in the reduction sequence, \( v_{i_1, \ell_1} \) must have a head occurrence and the reduction sequence is of the form
  \[
  t_0 \red_D^* h_{i_1}\,\tilde{v}_{i_1} \red_D^* v_{i_1, \ell_1}\,\tilde{s}_{1} \red_D^* k_{j_1}\,\tilde{w}_{j_1} \red_D^* h_{i_2}\,\tilde{v}_{i_2} \red_D^* \cdots.
  \]
  Since \( h_{i_1} \) is an ancestor of \( h_{i_2} \), there is no ancestor of \( h_{i_2} \) in \( \tilde{v}_{i_1} \), in particular, in \( v_{i_1, \ell_1} \).
  So an ancestor of \( h_{i_2} \) must appear in \( \tilde{s}_1 \), say, \( s_{1, \ell_2} \), and the reduction sequence is of the form
  \[
  t_0 \red_D^* h_{i_1}\,\tilde{v}_{i_1} \red_D^* v_{i_1, \ell_1}\,\tilde{s}_{1} \red_D^* s_{1, \ell_2}\,\tilde{s}_2 \red_D^* h_{i_2}\,\tilde{v}_{i_2} \red_D^* k_{j_2}\,\tilde{w}_{j_2} \red_D^* \cdots.
  \]
  Since an ancestor of \( k_{j_2} \) appears in \( v_{i_1, \ell_1} \), there is no ancestor in \( \tilde{s}_1 \), in particular, in \( s_{1, \ell_2} \).
  Hence there is an ancestor of \( k_{j_2} \) in \( \tilde{s}_{2} \), say, \( s_{2, \ell_3} \), and the reduction sequence is
  \[
  t_0 \red_D^* h_{i_1}\,\tilde{v}_{i_1} \red_D^* v_{i_1, \ell_1}\,\tilde{s}_{1} \red_D^* s_{1, \ell_2}\,\tilde{s}_2 \red_D^* s_{2,\ell_3}\,\tilde{s}_3 \red_D^* k_{j_2}\,\tilde{w}_{j_2} \red_D^* h_{i_3}\,\tilde{v}_{i_3} \red_D^* \cdots.
  \]
  By repeatedly applying the above argument, we have the following decomposition of the reduction sequence:
  \[
  t_0 \red_D^* h_{i_1}\,\tilde{v}_{i_1} \red_D^* v_{i_1, \ell_1}\,\tilde{s}_1 \red_D^* s_{1, \ell_2}\,\tilde{s}_2 \red_D^* s_{2,\ell_3}\,\tilde{s}_3 \red_D^* s_{3,\ell_4}\,\tilde{s}_4 \red_D^* s_{4,\ell_5}\,\tilde{s}_5 \red_D^* \cdots.
  \]
  Then the order of \( s_{m, \ell_m} \) is greater than \( s_{m+1, \ell_{m+1}} \) for every \( m \) since \( s_{m+1, \ell_{m+1}} \) is an argument of \( s_{m, \ell_m} \).
  This contradicts the assumption that the program \( D \) is simply typed.
\end{myproof}


\begin{corollary}\label{cor:liveness:unique-existence-of-call-seq}
  Let \( P = (D, \mainfun) \) be a program and
  \[
  \RTyPath \Vdash \mainfun \red_D t_1 \red_D \cdots
  \]
  be the infinite reduction sequence following \( \RTyPath \).
  Then there exists a unique infinite call-sequence of this reduction sequence, i.e.~
  \[
  \RTyPath \Vdash \mainfun \CallSeq_D g_1\,\tilde{u}_1 \CallSeq_D g_2\,\tilde{u}_2 \CallSeq_D \cdots.
  \]
\end{corollary}

\subsection{Proof of Theorem~\ref{theorem:callseq}}
\label{sec:proof-csa}

Let \( \lts_0 = (\{\stunique\}, \emptyset, \emptyset,\stunique) \) be the transition system with one state and no transition.
This section shows that given a program \( \prog = (\progd, \mainfun) \) with priority \( \Omega \), 
\(\models_\CSA (P,\Pfun)\) holds if and only if \( \lts_0 \models \HESf_{(P,\Pfun),\CSA} \).
Without loss of generality, we assume below that \(\prog\) is of the form \((\progd, \mainfun)\);
note that, if \(\prog=(\progd,t)\), then we can replace it with \((\progd\cup\set{\mainfun = t},\mainfun)\).

Fix a program \( \prog = (\progd, \mainfun) \) with priority \( \Pfun \).
Suppose that \( D = \{ f_1\,\tilde{x}_1 = \term_1, \dots, f_n\,\tilde{x}_n = \term_n \} \) and let \( \Pst_i \) be the simple type of \( f_i \).
We can assume without loss of generality that \( \mainfun = f_1 \) and that \( \Pfun(f_i) \ge \Pfun(f_j) \) if \( i \le j \).
Let
\begin{align*}
  \Omega'(f_i) =
  \begin{cases}
    2(n-i) & \mbox{(if \( \Omega(f_i) \) is even)} \\
    2(n-i)+1 & \mbox{(if \( \Omega(f_i) \) is odd).}
  \end{cases}
\end{align*}
The priority assignments \( \Pfun \) and \( \Pfun' \) are equivalent in the sense that an infinite sequence \( g_1 g_2 \dots \) of function symbols satisfies the parity condition with respect to \( \Pfun \) if and only if so does to \( \Pfun' \).
Hence we can assume without loss of generality that \( 2(n-i) \le \Pfun(f_i) \le 2(n-i) + 1 \).

Given a type \( \Pst \) (resp.~\( \Pest \)) of the target language, we define a type \( \livetrans{\Pst} \) (resp.~\( \livetrans{\Pest} \)) for HFLs as follows:
\begin{equation*}
  \livetrans{\Tunit} := \typProp
  \qquad\qquad
  \livetrans{(\Pest \to \Pst)} := \livetrans{\Pest} \to \livetrans{\Pst}
  \qquad\qquad
  \livetrans{\Tint} := \typInt.
\end{equation*}
Simply \( \livetrans{(-)} \) replaces \( \Tunit \) with \( \typProp \).
The operation \( \livetrans{(-)} \)  is pointwise extended for
type environments by:
\( \livetrans{(x_1\COL\Pest_1,\ldots,x_\ell\COL\Pest_\ell)}=
   x_1\COL \livetrans{\Pest_1},\ldots,x_\ell\COL \livetrans{\Pest_\ell}\)
Given an LTS \( \lts \), let us define \( \D_{\lts, \Pst} := \D_{\lts, \livetrans{\Pst}} \).

\begin{definition}[Semantic interpretation of a program]
Let \( \STE \vdash \term \COL \Pst \).
A \emph{valuation of \( \STE \)} is a mapping \(\HFLenv\)
such that \(\dom(\HFLenv)=\dom(\STE)\)
and \(\HFLenv(x)\in 
\D_{\lts,\STE(x)} \) for each \(x\in\dom(\STE)\).
The set of valuations of \( \STE \) is ordered by the point-wise ordering.
The \emph{interpretation} of a type judgment \(\STE\pST \term :\Pest\)
is a (monotonic) function from the valuations of \( \STE \) to \( \D_{\lts,\Pst} \) inductively defined
by:
\begin{align*}
  \sem{\STE \pST \unitexp : \Tunit}(\HFLenv) &:= \{ \stunique \} \\
  \sem{\STE \pST x : \Pst}(\HFLenv) &:= \HFLenv(x) \\
  \sem{\STE \pST n : \Tint}(\HFLenv) &:= n \\
  \sem{\STE \pST \term_1 \OP \term_2 : \typInt}(\HFLenv) &:= (\sem{\STE \pST \term_1 : \typInt}(\HFLenv)) \sem{\OP} (\sem{\STE \pST \term_2 : \typInt}(\HFLenv)) \\
  \sem{\STE \pST \ifexp{p(\term'_1,\ldots,\term'_k)}{\term_1}{\term_2} : \Tunit}(\HFLenv) &:=
  \begin{cases}
    \sem{\STE \pST \term_1 : \Tunit}(\HFLenv) & \mbox{(if \( (\sem{\term'_1}(\HFLenv), \dots, \sem{\term'_k}(\HFLenv)) \in \sem{p} \))} \\
    \sem{\STE \pST \term_2 : \Tunit}(\HFLenv) & \mbox{(if \( (\sem{\term'_1}(\HFLenv), \dots, \sem{\term'_k}(\HFLenv)) \notin \sem{p} \))} \\
  \end{cases}
  \\
  \sem{\STE \pST \evexp{\lab}{\term} : \Tunit}(\HFLenv) &:= \sem{\STE \pST \term : \Tunit}(\HFLenv) \\
  \sem{\STE \pST \term_1\,\term_2 : \Pst}(\HFLenv) &:= (\sem{\STE \pST \term_1 : \Pest \to \Pst}(\HFLenv))(\sem{\STE \pST \term_2 : \Pest}(\HFLenv)) \\
  \sem{\STE \pST \term_1 \nondet \term_2 : \Tunit}(\HFLenv) &:= \sem{\STE \pST \term_1 : \Tunit}(\HFLenv) \sqcap \sem{\STE \pST \term_2 : \Tunit}(\HFLenv) \\
  \sem{\STE \pST \lambda x. \term : \Pest \to \Pst}(\HFLenv) &:=
  \{ (d, \sem{\STE, x \COL \Pest \pST \term : \Pst}(\HFLenv [x \mapsto d])) \mid d\in
     \D_{\lts,\Pest}\}.
\end{align*}
\end{definition}
We often write just \(\sem{t}\) for \(\sem{\STE \pST t: \Pest}\) (on the assumption that
the simple type environment for \(t\) is implicitly determined).

\begin{lemma}\label{lem:appx:liveness:csa:term-interpretation}
  \( \sem{\STE\pST \term:\Pest}_{\lts_0}([x_1 \mapsto d_1, \dots, x_n \mapsto d_n]) =
  \sem{\livetrans{\STE}\pST \livetrans{\term}:\livetrans{\Pest}}_{\lts_0}([x_1 \mapsto d_1, \dots, x_n \mapsto d_n]) \).
\end{lemma}
\begin{myproof}
  By induction on the structure of \( \term \).
\end{myproof}

\subsubsection{Completeness}
We show that \(\models_\CSA (P,\Pfun)\) implies \( \lts_0 \models \HESf_{(P,\Pfun),\CSA}\).
We use game-based characterization of \( \lts_0 \models \HESf_{(P,\Pfun),\CSA}\).

\begin{definition}[Parity game for a program]
The parity game \( \SGame_{\lts, (P,\Pfun)} \) is defined as follows:
\begin{align*}
  \Nodes_{P} :=\;& \{\, (p, f_i) \mid p \in \D_{\lts, \Pst_i}, \;\textrm{\( p \): complete prime} \,\}
  \\
  \Nodes_{O} :=\;& \{\, (d_1, \dots, d_n) \mid d_1 \in \D_{\lts, \Pst_1}, \dots, d_n \in \D_{\lts, \Pst_n} \,\}
  \\
  \Edges
  :=\;& \{\, ((p, f_i), (d_1, \dots, d_n)) \mid p \sqleq \sem{\lambda \tilde{x}_i. \term_i}([f_1 \mapsto d_1, \dots, f_n \mapsto d_n]) \,\}
  \\
  \cup\;& \{\, ((d_1, \dots, d_n), (p, f_i)) \mid p \sqleq d_i, \;\textrm{\( p \): complete prime} \,\}.
\end{align*}
The priority of \( (d_1, \dots, d_n) \) is \( 0 \) and that of \( (p, f_i) \) is the priority \( \Omega(f_i) \) of \( f_i \).
\end{definition}

\begin{corollary}\label{cor:appx:liveness-from-call-seq:game-equivalence}
  \( \SGame_{\lts_0, (\prog,\Pfun)} \) is isomorphic to \( \SGame_{\lts_0, \HESf_{(\prog,\Pfun),\CSA}} \), i.e.~there exists a bijection of nodes that preserves the owner and the priority of each node.
\end{corollary}
\begin{myproof}
  The bijection on nodes is given by \( (d_1, \dots, d_n) \leftrightarrow (d_1, \dots, d_n) \) and \( (p, f_i) \leftrightarrow (p, f_i) \).
  This bijection preserves edges because of Lemma~\ref{lem:appx:liveness:csa:term-interpretation}.
\end{myproof}

We define a strategy of \( \SGame_{\lts_0,(P,\Pfun)} \) on \( (\{\stunique\}, f_1) \) by inspecting the reduction sequences of \( P \).
The technique used here is the same as \cite{KO09LICS}.
In this construction, we need to track occurrences of a function symbol or a term in a reduction sequence.
This can be achieved by marking symbols and terms if needed as in Definition~\ref{def:recall}.
In the sequel, we shall not introduce explicit marking and use the convention that, if the same metavariable appears in a reduction sequence (e.g.~\( f\,u \red^* u\,\seq{v} \)), then the first one is the origin of all others.

The occurrence of \( t \) in \( t\,u_1\,\dots\,u_k \) (\( k \ge 0 \)) is called a \emph{head occurrence}.
Given a head occurrence of \( t \) of a term reachable from \( \mainfun \) (i.e.~\( \mainfun \red_\progd^* t\,\tilde{u} \) for some \( \tilde{u} \)), we assign an element \( d_{t,\tilde{u}} \in \D_{\lts_0,\Pst} \) where \( \Pst \) is the simple type of \( t \) by induction on the order of \( \Pst \).
\begin{itemize}
\item
  Case \( \seq{u} \) is empty:
  Then \( t : \Tunit \) and we define \( d_{t,\epsilon} = \{\stunique\} \).
\item
  Case \( t\,\tilde{u} = t\,u_1\,\dots\,u_k \):
  Note that the order of \( u_j \) (\( 1 \le j \le k \)) is less than that of \( t \).
  For every \( 1 \le j \le k \), we define \( e_j \) as follows.
  \begin{itemize}
  \item Case \( u_j : \Tint \):
    Then \( e_j = \sem{u_j} \).
    (Note that \( u_j \) has no free variable and no function symbol.)
  \item Case \( u_j : \Pst \):
    Let \( e_j \) be the element defined by
    \[
    e_j = \LUB \{\, d_{u_j, \tilde{v}} \mid t\,\tilde{u} \red^*_\progd u_j\,\tilde{v} \,\}.
    \]
  \end{itemize}
  Then \( d_{t,\tilde{u}} \) is the function defined by
  \[
  d_{t,\tilde{u}}(x_1, \dots, x_k) :=
  \begin{cases}
    \{ \stunique \} & \mbox{(if \( e_j \sqleq x_j \) for all \( j \in \{ 1, \dots, k \} \))} \\
    \{\,\} & \mbox{(otherwise).}
  \end{cases}
  \]
  Note that \( d_{t,\tilde{u}} \) is the minimum function such that \( d_{t,\tilde{u}}(e_1, \dots, e_k) = \{\stunique\} \).
\end{itemize}
  
Let \( \theta = [s_1/x_1, \dots, s_k/x_k] \) and assume \( \mainfun \red_\progd^* (\theta t_0)\,\tilde{u} \).
We define a mapping from function symbols and variables
to elements of their semantic domains as follows.
\begin{itemize}
\item
  For a function symbol \( f_i \):
  \[
  \HFLenv_{t_0, \tilde{u}, \theta}(f_i) = \LUB \{\, d_{f_i, \tilde{v}} \mid (\theta t_0) \seq{u} \red^*_\progd f_i \tilde{v}, \,\mbox{\( f_i \) originates from \( t_0 \)} \,\}
  \]
\item
  For a variable \( x_i : \Tint \):
  \[
  \HFLenv_{t_0,\tilde{u},\theta}(x_i) = \sem{s_i}
  \]
\item
  For a variable \( x_i : \Pst \):
  \[
  \HFLenv_{t_0, \tilde{u}, \theta}(x_i) = \LUB \{\, d_{s_i, \tilde{v}} \mid (\theta t_0) \seq{u}\red^*_\progd s_i \tilde{v}, \,\mbox{\( s_i \) originates from \( \theta(x_i) \)} \,\}.
  \]
\end{itemize}

\begin{lemma}\label{lem:appx:call-seq-to-hes:completeness:arithmetic-expressions}
  Let \( \theta = [s_1/x_1, \dots, s_k/x_k] \) and assume \( \mainfun \red_\progd^* (\theta t_0)\,\tilde{u} \).
  Let \( v \) be an arithmetic expression (i.e.~a term of type \( \Tint \)) appearing in \( t_0 \).
  Then \( \sem{\theta v} = \sem{v}(\HFLenv_{t_0, \tilde{u}, \theta}) \).
\end{lemma}
\begin{myproof}
  By induction on \( v \).
  \begin{itemize}
  \item Case \( v = x_i \):
    Then \( \sem{\theta x_i} = \sem{s_i} = \HFLenv_{t_0, \tilde{u}, \theta}(x_i) = \sem{x_i}(\HFLenv_{t_0, \tilde{u}, \theta}) \).
  \item Case \( v = v_1 \OP v_2 \):
    By the induction hypothesis, \( \sem{\theta v_j} = \sem{v_j}(\HFLenv_{t_0, \tilde{u}, \theta}) \) for \( j = 1, 2 \).
    Then \( \sem{\theta (v_1 \OP v_2)} = \sem{\theta v_1} \sem{\OP} \sem{\theta v_2} = \sem{v_1}(\HFLenv_{t_0, \tilde{u}, \theta}) \sem{\OP} \sem{v_2}(\HFLenv_{t_0, \tilde{u}, \theta}) = \sem{v_1 \OP v_2}(\HFLenv_{t_0, \tilde{u}, \theta}) \).
  \end{itemize}
\end{myproof}

\begin{lemma}\label{lem:appx:call-seq-to-hes:infinite-subject-expansion}
  Let \( \theta = [s_1/x_1, \dots, s_k/x_k] \) and assume \( \mainfun \red_\progd^* (\theta t_0)\,\tilde{u} \).
  Then \( d_{\theta t_0,\tilde{u}} \sqleq \sem{t_0}(\HFLenv_{t_0, \tilde{u}, \theta}) \).
\end{lemma}
\begin{myproof}
  By induction on the structure of \( t_0 \).
  \begin{itemize}
  \item Case \( t_0 = \unitexp \):
    Then \( \sem{t_0} = \{\stunique\} = d_{\unitexp,\epsilon} \).
  \item Case \( t_0 = x_i \) for some \( 1 \le i \le k \):
    Then \( (\theta t_0) \tilde{u} = s_i\,\tilde{u} \).
    By definition, \( d_{s_i,\tilde{u}} = 
    \HFLenv_{t_0,\tilde{u},\theta}(x_i) \).
  \item Case \( t_0 = f_i \):
    Then \( (\theta t_0) \tilde{u} = f_i\,\tilde{u} \).
    Since this occurrence of \( f_i \) originates from \( t_0 \), we have
    \( d_{f_i,\tilde{u}} = 
    \HFLenv_{t_0,\tilde{u},\theta}(f_i) \).
  \item Case \( t_0 = n \) or \(t_0= \term_1 \OP \term_2 \):
    Never occurs because the type of \(t_0\) must be of the form
    \(\Pest_1\to\cdots\to\Pest_\ell\to \Tunit\).
  \item Case \( t_0 = \ifexp{p(\term'_1,\ldots,\term'_k)}{\term_1}{\term_2} \):
    Then \( \tilde{u} \) is an empty sequence.
    Assume that \( (\sem{\theta \term'_1}, \dots, \sem{\theta \term'_k}) \in \sem{p} \).
    The other case can be proved in a similar manner.

    Then we have \( \theta t_0 \red_\progd \theta \term_1 \).
    By definition, \( d_{\theta t_0, \epsilon} = d_{\theta \term_1, \epsilon} = \{\stunique\} \).
    
    Here, \(\red_\progd^+\) is the transitive closure of \(\red_\progd\).
    Since every reduction sequence \( \theta t_0 \red_\progd^+ s\,\tilde{v} \) can be factored into \( \theta t_0 \red_\progd \theta \term_1 \red_\progd^* s\,\tilde{v} \), we have \( \HFLenv_{t_0, \epsilon, \theta} = \HFLenv_{\term_1, \epsilon, \theta} \).
    By the induction hypothesis, \( \{\stunique\} \sqleq \sem{\term_1}(\HFLenv_{\term_1,\epsilon,\theta}) \).
    Hence \( \{\stunique\} \sqleq \sem{\term_1}(\HFLenv_{t_0,\epsilon,\theta}) \).
    
    By Lemma~\ref{lem:appx:call-seq-to-hes:completeness:arithmetic-expressions}, \( \sem{\theta \term'_j} = \sem{\term'_j}(\HFLenv_{t_0,\tilde{u},\theta}) \) for every \( j = 1, \dots, k \).
    Hence \( (\sem{\theta \term'_1}, \dots, \sem{\theta \term'_k}) \in \sem{p} \) implies \( (\sem{\term'_1}(\HFLenv_{t_0,\tilde{u},\theta}), \dots, \sem{\term'_k}(\HFLenv_{t_0,\tilde{u},\theta})) \in \sem{p} \).
    By definition, \( \sem{t_0}(\HFLenv_{t_0,\tilde{u},\theta}) = \sem{\term_1}(\HFLenv_{t_0,\tilde{u},\theta}) \sqsupseteq \{\stunique\} \).
  \item Case \( t_0 = \evexp{\lab}{\term} \):
    Then \( t_0 \) has type \( \Tunit \) and \( \tilde{u} \) is the empty sequence.
    We have \( \theta t_0 = \evexp{\lab}{(\theta \term)} \).
    Since every reduction sequence \( \theta t_0 \red_\progd^+ s\,\tilde{v} \) can be factored into \( \theta t_0 \red_\progd \theta \term \red_\progd^* s\,\tilde{v} \), we have \( \HFLenv_{t_0, \epsilon, \theta} = \HFLenv_{\term, \epsilon, \theta} \).
    By definition, \( d_{\theta t_0, \epsilon} = d_{\theta \term, \epsilon} = \{ \stunique \} \) and \( \sem{t_0}(\HFLenv_{t_0, \epsilon, \theta}) = \sem{\term}(\HFLenv_{t_0, \epsilon, \theta}) = \sem{\term}(\HFLenv_{\term, \epsilon,\theta}) \).
    By the induction hypothesis, \( d_{\theta \term, \epsilon} \sqleq \sem{\term}(\HFLenv_{\term, \epsilon, \theta}) \).
    Hence \( d_{\theta t_0, \epsilon} \sqleq \sem{t_0}(\HFLenv_{t_0, \epsilon, \theta}) \).
  \item Case \( t_0 = t_1\,t_2 \):
    Suppose that \( (\theta t_1) (\theta t_2) \tilde{u} \red_\progd^* (\theta t_2)\,\tilde{v} \).
    We first show that \( \HFLenv_{t_2, \tilde{v}, \theta}(x) \sqleq \HFLenv_{t_1 t_2, \tilde{u}, \theta}(x) \) for every \( x \in \{ f_1, \dots, f_n \} \cup \{ x_1, \dots,x_k \} \).
    \begin{itemize}
    \item Case \( x = x_i : \Tint \):
      Then \( \HFLenv_{t_2, \tilde{v}, \theta}(x_i) = \HFLenv_{t_1t_2, \tilde{u}, \theta}(x_i) = \sem{s_i} \).
    \item Case \( x = x_i : \Pst \):
      Then
      \begin{align*}
        \HFLenv_{t_2, \tilde{v}, \theta}(x_i)
        &= \LUB \{\, d_{s_i, \tilde{w}} \mid (\theta t_2)\,\tilde{v} \red^*_\progd s_i \tilde{w}, \,\mbox{\( s_i \) originates from \( \theta(x_i) \)} \,\} \\
        &\sqleq \LUB \{\, d_{s_i, \tilde{w}} \mid (\theta (t_1\,t_2))\,\tilde{u} \red^*_\progd s_i \tilde{w}, \,\mbox{\( s_i \) originates from \( \theta(x_i) \)} \,\} \\
        &= \HFLenv_{t_1 t_2, \tilde{u}, \theta}(x_i)
      \end{align*}
      because \( (\theta t_2)\,\tilde{v} \red^*_\progd s_i \tilde{w} \) (where \( s_i \) originates from \( \theta x_i \)) implies \( (\theta (t_1\,t_2))\,\tilde{u} \red^*_\progd (\theta t_2) \tilde{v} \red^*_\progd s_i \tilde{w} \) (where \( s_i \) originates from \( \theta x_i \)).
    \item Case \( x = f_i \):
      \begin{align*}
        \HFLenv_{t_2, \tilde{v}, \theta}(f_i)
        &= \LUB \{\, d_{f_i, \tilde{w}} \mid (\theta t_2)\,\tilde{v} \red^*_\progd f_i \tilde{w}, \,\mbox{\( f_i \) originates from \( t_2 \)} \,\} \\
        &\sqleq \LUB \{\, d_{f_i, \tilde{w}} \mid (\theta (t_1\,t_2))\,\tilde{u} \red^*_\progd f_i \tilde{w}, \,\mbox{\( f_i \) originates from \( t_1 t_2 \)} \,\} \\
        &= \HFLenv_{t_1 t_2, \tilde{u}, \theta}(f_i)
      \end{align*}
      because \( (\theta t_2)\,\tilde{v} \red^*_\progd f_i \tilde{w} \) (where \( f_i \) originates from \( t_2 \)) implies \( (\theta (t_1\,t_2))\,\tilde{u} \red^*_\progd (\theta t_2) \tilde{v} \red^*_\progd f_i \tilde{w} \) (where \( f_i \) originates from \( t_1 t_2 \)).      
    \end{itemize}

    It is easy to see that \( \HFLenv_{t_1, (\theta t_2)\,\tilde{u}, \theta} \sqleq \HFLenv_{t_1t_2, \tilde{u}, \theta} \).
    By the induction hypothesis,
    \[
    d_{\theta t_1, (\theta t_2)\,\tilde{u}} \sqleq \sem{t_1}(\HFLenv_{t_1, (\theta t_2)\,\tilde{u}, \theta}) \sqleq \sem{t_1}(\HFLenv_{t_1t_2, \tilde{u}, \theta}).
    \]
    By the definition of \( d_{\theta t_1, (\theta t_2)\,\tilde{u}} \),
    \[
    d_{\theta t_1, (\theta t_2)\,\tilde{u}} \big(\LUB \{ d_{\theta t_2, \tilde{v}} \mid (\theta t_1)\,(\theta t_2)\,\tilde{u} \red_\progd^* (\theta t_2) \tilde{v} \}\big) = d_{(\theta t_1)\,(\theta t_2),\,\tilde{u}}.
    \]
    By the induction hypothesis, for every reduction sequence \( (\theta t_1)\,(\theta t_2)\,\tilde{u} \red^*_\progd (\theta t_2) \tilde{v} \), we have \( d_{\theta t_2, \tilde{v}} \sqleq \sem{t_2}(\HFLenv_{t_2, \tilde{v}, \theta}) \).
    Since \( \HFLenv_{t_2, \tilde{v}, \theta}(x) \sqleq \HFLenv_{t_1t_2, \tilde{u}, \theta}(x) \) for every \( x \), we have
    \[
    d_{\theta t_2, \tilde{v}} \sqleq \sem{t_2}(\HFLenv_{t_2, \tilde{v}, \theta}) \sqleq \sem{t_2}(\HFLenv_{t_1t_2, \tilde{u}, \theta}).
    \]
    Because the reduction sequence \( (\theta t_1)\,(\theta t_2)\,\tilde{u} \red^*_\progd (\theta t_2) \tilde{v} \) is arbitrary,
    \[
    \big(\LUB \{ d_{\theta t_2, \tilde{v}} \mid (\theta t_1)\,(\theta t_2)\,\tilde{u} \red_\progd^* (\theta t_2) \tilde{v} \}\big) \sqleq \sem{t_2}(\HFLenv_{t_1t_2, \tilde{u}, \theta}).
    \]
    Therefore, by monotonicity,
    \begin{align*}
      d_{(\theta t_1)\,(\theta t_2),\,\tilde{u}}
      &= d_{\theta t_1, (\theta t_2)\,\tilde{u}} \big(\LUB \{ d_{\theta t_2, \tilde{v}} \mid (\theta t_1)\,(\theta t_2)\,\tilde{u} \red_\progd^* (\theta t_2) \tilde{v} \}\big) \\
      &\sqleq \sem{t_1}(\HFLenv_{t_1t_2, \tilde{u}, \theta}) \big( \sem{t_2}(\HFLenv_{t_1t_2, \tilde{u}, \theta}) \big) \\
      &= \sem{t_1\,t_2}(\HFLenv_{t_1t_2,\tilde{u},\theta}).
    \end{align*}
  \item Case \( t_0 = \term_1 \nondet \term_2 \):
    Then \( \tilde{u} = \epsilon \).
    We have \( \theta t_0 \red_\progd \theta \term_1 \) and \( \theta t_0 \red_\progd \theta \term_2 \).
    By the definition of \( d_{\theta \term_1, \epsilon} \) and the induction hypothesis,
    \[
    \{\stunique\} = d_{\theta \term_1,\epsilon} \sqleq \sem{\term_1}(\HFLenv_{\term_1, \epsilon, \theta}).
    \]
    Similarly
    \[
    \{\stunique\} = d_{\theta \term_2, \epsilon} \sqleq \sem{\term_2}(\HFLenv_{\term_2, \epsilon, \theta}).
    \]
    Because \( \theta t_0 \red_\progd \theta \term_1 \), we have \( \HFLenv_{\term_1, \epsilon, \theta} \sqleq \HFLenv_{t_0, \epsilon, \theta} \).
    Similarly \( \HFLenv_{\term_2, \epsilon, \theta} \sqleq \HFLenv_{t_0, \epsilon, \theta} \).
    By monotonicity of interpretation,
    \[
    \{\stunique\} = d_{\theta \term_1,\epsilon} \sqleq \sem{\term_1}(\HFLenv_{\term_1, \epsilon, \theta}) \sqleq \sem{\term_1}(\HFLenv_{t_0, \epsilon, \theta})
    \]
    and
    \[
    \{\stunique\} = d_{\theta \term_2,\epsilon} \sqleq \sem{\term_2}(\HFLenv_{\term_2, \epsilon, \theta}) \sqleq \sem{\term_2}(\HFLenv_{t_0, \epsilon, \theta}).
    \]
    Hence
    \[
    \{\stunique\} \sqleq \sem{\term_1}(\HFLenv_{t_0, \epsilon, \theta}) \sqcap \sem{\term_2}(\HFLenv_{t_0, \epsilon, \theta}) = \sem{t_0}(\HFLenv_{t_0, \epsilon, \theta}).
    \]
\end{itemize}
\end{myproof}

The strategy \( \Strategy_{\prog} \) of \( \SGame_{\lts_0,\prog} \) on \( (\{\stunique\}, \mainfun) \) is defined as follows.
Each play in the domain of \( \Strategy_{\prog} \)
\[
(\{\stunique\}, \mainfun) \cdot O_1 \cdot (p_1, g_1) \cdot O_2 \cdot \dots \cdot O_k \cdot (p_k, g_k)
\]
is associated with a call-sequence
\begin{align*}
  &
  \mainfun \CallSeqN{m_1}_\progd g_1\,\tilde{u}_1 \CallSeqN{m_2}_\progd g_2\,\tilde{u}_2 \CallSeqN{m_3}_\progd \dots \CallSeqN{m_{k-1}}_\progd g_{k-1}\,\tilde{u}_{k-1} \CallSeqN{m_k}_\progd g_k\,\tilde{u}_k
\end{align*}
such that \( p_j \sqleq d_{g_j, \tilde{u}_j} \) for every \( j = 1, 2, \dots, k \).
Let \( \xi_i \) be the finite sequence over \( \{\PLeft,\PRight\} \) describing the choice made during \( g_{i-1}\,\tilde{u}_{i-1} \CallSeqN{m_i}_{\progd} g_{i}\,\tilde{u}_{i} \) (where \( g_0\,\tilde{u}_0 = \mainfun \)).
The \emph{canonical associated call-sequence} of the play is the minimum one ordered by the lexicographic ordering on \( (m_1, \xi_1, m_2, \xi_2, \dots, m_k, \xi_k) \).

Assume that the above call-sequence is canonical.
The next step of this reduction sequence is \( [\tilde{u}_k/\tilde{x}_k]t \) if \( (g_k\,\seq{x}_k = t) \in \progd \).
Let \( \vartheta = [\tilde{u}_k/\tilde{x}_k] \).
In this situation, the strategy \( \Strategy_{\prog} \)
chooses \( (\HFLenv_{t, \epsilon, \vartheta}(f_1), \dots, \HFLenv_{t, \epsilon, \vartheta}(f_n)) \) as the next node.
This is a valid choice, i.e.:
\begin{lemma}
  \( p_k \sqleq \sem{D(g_k)}([f_1 \mapsto \HFLenv_{t, \epsilon, \vartheta}(f_1), \dots, f_n \mapsto \HFLenv_{t, \epsilon, \vartheta}(f_n)]) \).
\end{lemma}
\begin{myproof}
  By Lemma~\ref{lem:appx:call-seq-to-hes:infinite-subject-expansion}, we have:
  \[ d_{\vartheta t, \epsilon} = \{\stunique\} \sqleq \sem{t}(\HFLenv_{t, \epsilon,\vartheta})
  = (\sem{\lambda\seq{x}_k.t}([f_1 \mapsto \HFLenv_{t, \epsilon, \vartheta}(f_1), \dots, f_n \mapsto \HFLenv_{t, \epsilon, \vartheta}(f_n)]))
(\HFLenv_{t, \epsilon, \vartheta}(\seq{x}_k)).
  \]
  By definition, \(d_{g_k, \tilde{u}_k}\) is the least element such that
  \(  d_{g_k, \tilde{u}_k}(\HFLenv_{t, \epsilon, \vartheta}(\seq{x}_k)) = \set{\stunique}\).
  Therefore
  \[
  d_{g_k, \tilde{u}_k} \sqleq \sem{\lambda \seq{x}_k. t}([f_1 \mapsto \HFLenv_{t, \epsilon, \vartheta}(f_1), \dots, f_n \mapsto \HFLenv_{t, \epsilon, \vartheta}(f_n)]).
  \]
  Since \( p_k \sqleq d_{g_k, \tilde{u}_k} \), we have
  \[
  p_k \sqleq d_{g_k, \tilde{u}_k} \sqleq \sem{\lambda \tilde{x}_k. t}([f_1 \mapsto \HFLenv_{t, \epsilon, \vartheta}(f_1), \dots, f_n \mapsto \HFLenv_{t, \epsilon, \vartheta}(f_n)]).
  \]
\end{myproof}

Now we have a play
\[
(\{\stunique\}, \mainfun) \cdot O_1 \cdot (p_1, g_{1}) \cdot O_2 \cdot \dots \cdot O_k \cdot (p_k, g_{k}) \cdot (\HFLenv_{t, \epsilon, \theta}(f_1), \dots, \HFLenv_{t, \epsilon, \theta}(f_n))
\]
associated with the call-sequence
\begin{align*}
  &
  \mainfun \CallSeqN{m_1}_\progd g_1\,\tilde{u}_1 \CallSeqN{m_2}_\progd g_2\,\tilde{u}_2 \CallSeqN{m_3}_\progd \dots \CallSeqN{m_{k-1}}_\progd g_{k-1}\,\tilde{u}_{k-1} \CallSeqN{m_k}_\progd g_k\,\tilde{u}_k
\end{align*}
Let \( (p_{k+1}, g_{k+1}) \) be the next opponent move.
By definition of the game,
\[
p_{k+1} \sqleq \HFLenv_{t, \epsilon, \theta}(g_{k+1}) = \LUB \{\, d_{g_{k+1}, \tilde{v}} \mid \theta t \red^*_\progd g_{k+1}\,\tilde{v}, \,\mbox{\( g_{k+1} \) originates from \( t \)} \,\}.
\]
This can be more simply written as
\[
p_{k+1} \sqleq \HFLenv_{t, \epsilon, \theta}(g_{k+1}) = \LUB \{\, d_{g_{k+1}, \tilde{v}} \mid g_{k}\,\seq{u}_k \CallSeq g_{k+1}\,\seq{v} \,\}.
\]
Since \( p_{k+1} \) is a complete prime, \( p_{k+1} \sqleq d_{g_{k+1}, \tilde{v}} \) for some \( \tilde{v} \) with \( g_{k}\,\seq{u}_k \CallSeq g_{k+1}\,\seq{v} \).
Then we have an associated call-sequence
\begin{align*}
  &
  \mainfun \CallSeqN{m_1}_\progd g_1\,\tilde{u}_1 \CallSeqN{m_2}_\progd g_2\,\tilde{u}_2 \CallSeqN{m_3}_\progd \dots \CallSeqN{m_{k-1}}_\progd g_{k-1}\,\tilde{u}_{k-1} \CallSeqN{m_k}_\progd g_k\,\tilde{u}_k \CallSeqN{m_{k+1}}_{\progd} g_{k+1}\,\tilde{u}_{k+1}
\end{align*}
of
\[
(\{\stunique\}, \mainfun) \cdot O_1 \cdot (p_1, g_{1}) \cdot O_2 \cdot \dots \cdot O_k \cdot (p_k, g_{k}) \cdot (\HFLenv_{t, \epsilon, \theta}(f_1), \dots, \HFLenv_{t, \epsilon, \theta}(f_n)) \cdot (p_{k+1}, g_{k+1}).
\]

\begin{lemma}\label{lem:appx:liveness-from-call-seq:parity-call-seq-to-winning}
  \( \Strategy_{\prog} \) is a winning strategy of \( \SGame_{\lts_0,(\prog,\Pfun)} \) on \( (\{\stunique\}, \mainfun) \) if every infinite call-sequence of \( \prog \) satisfies the parity condition.
\end{lemma}
\begin{myproof}
  The above argument shows that \( \Strategy_{\prog} \) is a strategy of \( \SGame_{\lts_0,(\prog,\Pfun)} \) on \( (\{\stunique\}, \mainfun) \).
  We prove that it is winning.
  Assume an infinite play
  \[
  (\{\stunique\},\mainfun) \cdot v_1 \cdot (p_1,g_1) \cdot v_2 \cdot (p_2,g_2) \cdot \dots
  \]
  that conforms with \( \Strategy_{\prog} \) and starts from \( (\{\stunique\},\mainfun) \).
  Then each odd-length prefix is associated with the canonical call-sequence
  \[
  \RTyPath_{k} \Vdash
  \mainfun \CallSeqN{m_{k,1}}_{\progd} g_1\,\seq{u}_{k,1} \CallSeqN{m_{k,2}}_{\progd} g_2\,\seq{u}_{k,2} \CallSeqN{m_{k,3}}_{\progd} \dots \CallSeqN{m_{k,k}}_{\progd} g_k\,\seq{u}_{k,k}.
  \]
  Let \( \xi_{k,i} \) be the sequence of \( \{\PLeft,\PRight\} \) describing the choice made during \( g_{i-1}\,\seq{u}_{k,i-1} \CallSeqN{m_{k,i}}_{\progd} g_i\,\seq{u}_i \).
  By the definitions of \( \Strategy_{\prog} \) and canonical call-sequence (which is the minimum with respect to the lexicographic ordering on \( (m_{k,1}, \xi_{k,1}, \dots, m_{k,k}, \xi_{k,k}) \)), a prefix of the canonical call-sequence is the canonical call-sequence of the prefix.
  In other words, \( k_{k,i} = k_{k',i} \) and \( \xi_{k,i} = \xi_{k',i} \) for every \( i \ge 1 \) and \( k, k' \le i \).
  Let us write \( k_i \) for \( k_{i,i} = k_{i+1,i} = \cdots \) and \( \seq{u}_i \) for \( \seq{u}_{i,i} = \seq{u}_{i+1,i} = \cdots \).
  Let \( \RTyPath = \xi_{1,1} \xi_{2,2} \xi_{3,3} \dots \).
  Then we have
  \[
  \RTyPath \Vdash
  \mainfun \CallSeqN{m_1}_{\progd} g_1\,\seq{u}_{1} \CallSeqN{m_2}_{\progd} g_2\,\seq{u}_{2} \CallSeqN{m_3}_{\progd} \dots.
  \]
  Since every infinite call-sequence satisfies the parity condition, the infinite play is P-winning.
\end{myproof}

\subsubsection{Soundness}
The proof of soundness (i.e.~\( \lts_0 \models \HESf_{(P,\Pfun),\CSA}\) implies \(\models_\CSA (P,\Pfun)\)) can be given in a manner similar to the proof of completeness.
To show the contraposition, assume
that there exists an infinite call-sequence that violates the parity condition.
Let \( \RTyPath \) be the choice of this call-sequence.
We construct a strategy \( \Strategy_{\prog,\RTyPath} \) of the opposite game \( \CounterGame_{\lts_0,\HESf_{(P,\Pfun),\CSA}} \) by inspecting the reduction sequence.
Then, by Lemma~\ref{lem:appx:game-semantics:opposite-game}, there is no winning strategy for \( \SGame_{\lts_0,\HESf_{(P,\Pfun),\CSA}} \), which implies 
\( \lts_0 \not\models \HESf_{(P,\Pfun),\CSA}\) by Theorem~\ref{thm:appx:game-semantics:game-and-interpretation}.
\begin{definition}[Opposite parity game for a program]
The parity game \( \CounterGame_{\lts, (\prog,\Pfun)} \) is defined as follows:
\begin{align*}
  \Nodes_{P} :=\;& \{\, (p, f_i) \mid p \in \D_{\lts, \Pst_i}, \;\textrm{\( p \): complete coprime} \,\}
  \\
  \Nodes_{O} :=\;& \{\, (d_1, \dots, d_n) \mid d_1 \in \D_{\lts, \Pst_1}, \dots, d_n \in \D_{\lts, \Pst_n} \,\}
  \\
  \Edges
  :=\;& \{\, ((p, f_i), (d_1, \dots, d_n)) \mid p \sqsupseteq \sem{\lambda \tilde{x}_i. \term_i}([f_1 \mapsto d_1, \dots, f_n \mapsto d_n]) \,\}
  \\
  \cup\;& \{\, ((d_1, \dots, d_n), (p, f_i)) \mid p \sqsupseteq d_i, \;\textrm{\( p \): complete coprime} \,\}.
\end{align*}
The priority of \( (d_1, \dots, d_n) \) is \( 0 \) and that of \( (p, f_i) \) is \( \Omega(f_i) + 1\).
\end{definition}

\begin{corollary}\label{cor:appx:liveness-from-call-seq:soundness:game-equivalence}
  \( \CounterGame_{\lts_0, (\prog,\Pfun)} \) is isomorphic to \( \CounterGame_{\lts_0, \HESf_{(\prog,\Pfun),\CSA}} \), i.e.~there exists a bijection of nodes that preserves the owner and the priority of each node.
\end{corollary}
\begin{myproof}
  The bijection on nodes is given by \( (d_1, \dots, d_n) \leftrightarrow (d_1, \dots, d_n) \) and \( (p, f_i) \leftrightarrow (p, f_i) \).
  This bijection preserves edges because of Lemma~\ref{lem:appx:liveness:csa:term-interpretation}.
\end{myproof}

Let \( \RTyPath_0 \in \{\PLeft,\PRight\}^{\omega} \).
Given a head occurrence of \( t \) of a term reachable from \( \mainfun \) following \( \RTyPath_0 \) (i.e.~\( (\mainfun; \RTyPath_0) \red_\progd^* (t\,\tilde{u}; \RTyPath) \) for some \( \tilde{u} \)), we assign an element \( \bar{d}_{t,\tilde{u},\RTyPath} \in \D_{\lts_0,\Pst} \) where \( \Pst \) is the simple type of \( t \) by induction on the order of \( \Pst \).
\begin{itemize}
\item
  Case \( \seq{u} \) is empty:
  Then \( t : \Tunit \) and we define \( \bar{d}_{t,\epsilon,\RTyPath} = \{\,\} \).
\item
  Case \( \seq{u} = u_1 \dots u_k \):
  Then \( t\,\tilde{u} = t\,u_1\,\dots\,u_k \).
  Note that the order of \( u_j \) (\( 1 \le j \le k \)) is less than that of \( t \).
  For every \( 1 \le j \le k \), we define \( \bar{e}_j \) as follows.
  \begin{itemize}
  \item Case \( u_j : \Tint \):
    Then \( \bar{e}_j = \sem{u_j} \).
    (Note that \( u_j \) has no free variable and no function symbol.)
  \item Case \( u_j : \Pst \):
    Let \( \bar{e}_j \) be the element defined by
    \[
    \bar{e}_j = \GLB \{\, \bar{d}_{u_j, \tilde{v}, \RTyPath'} \mid (t\,\tilde{u}; \RTyPath) \red^*_\progd (u_j\,\tilde{v}; \RTyPath') \,\}.
    \]
    (By the convention, \( u_j \) in the right-hand side originates from the \( j \)-th argument of \( t \).)
  \end{itemize}
  Then \( \bar{d}_{t,\tilde{u},\RTyPath} \) is the function defined by
  \[
  \bar{d}_{t,\tilde{u},\RTyPath}(x_1, \dots, x_k) :=
  \begin{cases}
    \{ \, \} & \mbox{(if \( \bar{x}_j \sqleq \bar{e}_j \) for all \( j \in \{ 1, \dots, k \} \))} \\
    \{ \stunique \} & \mbox{(otherwise).}
  \end{cases}
  \]
  Note that \( \bar{d}_{t,\tilde{u},\RTyPath} \) is the maximum function such that \( \bar{d}_{t,\tilde{u},\RTyPath}(\bar{e}_1, \dots, \bar{e}_k) = \{\,\} \).
\end{itemize}
  
Let \( \theta = [s_1/x_1, \dots, s_k/x_k] \) and assume \( (\mainfun; \RTyPath_0) \red_\progd^* ((\theta t_0)\,\tilde{u}; \RTyPath) \).
We define a mapping from function symbols and variables
to elements of their semantic domains as follows.
\begin{itemize}
\item
  For a function symbol \( f_i \):
  \[
  \bar{\HFLenv}_{t_0, \tilde{u}, \theta, \RTyPath}(f_i) = \GLB \{\, \bar{d}_{f_i, \tilde{v},\RTyPath'} \mid
  ((\theta t_0) \seq{u}; \RTyPath) \red^*_{\progd} (f_i \tilde{v}; \RTyPath'), \,\mbox{\( f_i \) originates from \( t_0 \)} \,\}
  \]
\item
  For a variable \( x_i : \Tint \):
  \[
  \bar{\HFLenv}_{t_0,\tilde{u},\theta,\RTyPath}(x_i) = \sem{s_i}
  \]
\item
  For a variable \( x_i : \Pst \):
  \[
  \bar{\HFLenv}_{t_0, \tilde{u}, \theta,\RTyPath}(x_i) = \GLB \{\, \bar{d}_{s_i, \tilde{v},\RTyPath'} \mid
  ((\theta t_0) \seq{u}; \RTyPath) \red^*_\progd (s_i \tilde{v}; \RTyPath'), \,\mbox{\( s_i \) originates from \( \theta(x_i) \)} \,\}.
  \]
\end{itemize}

\begin{lemma}\label{lem:appx:call-seq-to-hes:soundness:arithmetic-expressions}
  Let \( \theta = [s_1/x_1, \dots, s_k/x_k] \) and assume \( (\mainfun, \RTyPath_0) \red_{\progd}^* ((\theta t_0)\,\tilde{u}; \RTyPath) \).
  Let \( v \) be an arithmetic expression (i.e.~a term of type \( \Tint \)) appearing in \( t_0 \).
  Then \( \sem{\theta v} = \sem{v}(\bar{\HFLenv}_{t_0, \tilde{u}, \theta, \RTyPath}) \).
\end{lemma}
\begin{myproof}
  Similar to the proof of Lemma~\ref{lem:appx:call-seq-to-hes:completeness:arithmetic-expressions}.
\end{myproof}

\begin{lemma}\label{lem:appx:call-seq-to-hes:soundness:infinite-subject-expansion}
  Let \( \theta = [s_1/x_1, \dots, s_k/x_k] \) and \( \RTyPath_0 \in \{\PLeft,\PRight\}^{\omega} \).
  Assume that the unique reduction sequence starting from \( (\mainfun; \RTyPath_0) \) does not terminate and that \( (\mainfun; \RTyPath_0) \red_\progd^* ((\theta t_0)\,\tilde{u}, \RTyPath) \).
  Then \( \bar{d}_{\theta t_0,\tilde{u},\RTyPath} \sqsupseteq \sem{t_0}(\bar{\HFLenv}_{t_0, \tilde{u}, \theta,\RTyPath}) \).
\end{lemma}
\begin{myproof}
  By induction on the structure of \( t_0 \).
  \begin{itemize}
  \item Case \( t_0 = \unitexp \):
    This case never occurs since the reduction sequence does not terminate.
  \item Case \( t_0 = x_i \) for some \( 1 \le i \le k \):
    Then \( (\theta t_0) \tilde{u} = s_i\,\tilde{u} \).
    By definition, \( \bar{d}_{s_i,\tilde{u},\RTyPath} = 
    \bar{\HFLenv}_{t_0,\tilde{u},\theta,\RTyPath}(x_i) \).
  \item Case \( t_0 = f_i \):
    Then \( (\theta t_0) \tilde{u} = f_i\,\tilde{u} \).
    Since this occurrence of \( f_i \) originates from \( t_0 \), we have \( \bar{d}_{f_i,\tilde{u},\RTyPath} \sqsupseteq \bar{\HFLenv}_{t_0,\tilde{u},\theta,\RTyPath}(f_i) \).
  \item Case \( t_0 = n \) or \( \term_1 \OP \term_2 \):
    Never occurs because the type of \(t_0\) must be of the form
    \(\Pest_1\to\cdots\to\Pest_\ell\to \Tunit\).
  \item Case \( t_0 = \ifexp{p(\term'_1,\ldots,\term'_k)}{\term_1}{\term_2} \):
    Then \( \tilde{u} \) is the empty sequence.
    Assume that \( (\sem{\theta \term'_1}, \dots, \sem{\theta \term'_k}) \in \sem{p} \).
    The other case can be proved by a similar way.

    Then we have \( (\theta t_0; \RTyPath) \red_\progd (\theta \term_1; \RTyPath) \).
    By definition, \( \bar{d}_{\theta t_0, \epsilon, \RTyPath} = \bar{d}_{\theta \term_1, \epsilon, \RTyPath} = \{\,\} \).
    
    Since every reduction sequence \( (\theta t_0; \RTyPath) \red_\progd^+ (s\,\tilde{v}; \RTyPath') \) can be factored into \( (\theta t_0; \RTyPath) \red_\progd (\theta \term_1; \RTyPath) \red_\progd^* (s\,\tilde{v}; \RTyPath') \), we have \( \bar{\HFLenv}_{t_0, \epsilon, \theta, \RTyPath} = \bar{\HFLenv}_{\term_1, \epsilon, \theta, \RTyPath} \).
    By the induction hypothesis, \( \{\,\} \sqsupseteq \sem{\term_1}(\bar{\HFLenv}_{\term_1,\epsilon,\theta,\RTyPath}) \).
    Hence \( \{\,\} \sqsupseteq \sem{\term_1}(\bar{\HFLenv}_{t_0,\epsilon,\theta,\RTyPath}) \).
    
    By Lemma~\ref{lem:appx:call-seq-to-hes:soundness:arithmetic-expressions}, \( \sem{\theta \term'_j} = \sem{\term'_j}(\bar{\HFLenv}_{t_0,\tilde{u},\theta,\RTyPath}) \) for every \( j = 1, \dots, k \).
    Hence \( (\sem{\theta \term'_1}, \dots, \sem{\theta \term'_k}) \in \sem{p} \) implies \( (\sem{\term'_1}(\bar{\HFLenv}_{t_0,\tilde{u},\theta,\RTyPath}), \dots, \sem{\term'_k}(\bar{\HFLenv}_{t_0,\tilde{u},\theta,\RTyPath})) \in \sem{p} \).
    By definition, \( \sem{t_0}(\bar{\HFLenv}_{t_0,\tilde{u},\theta,\RTyPath}) = \sem{\term_1}(\bar{\HFLenv}_{t_0,\tilde{u},\theta,\RTyPath}) \sqsubseteq \{\,\} \).
  \item Case \( t_0 = \evexp{\lab}{\term} \):
    Then \( t_0 \) has type \( \Tunit \) and \( \tilde{u} \) is the empty sequence.
    We have \( \theta t_0 = \evexp{\lab}{(\theta \term)} \).
    Since every reduction sequence \( (\theta t_0; \RTyPath) \red_\progd^+ (s\,\tilde{v}; \RTyPath') \) can be factored into \( (\theta t_0; \RTyPath) \red_\progd (\theta \term;\RTyPath) \red_\progd^* (s\,\tilde{v}; \RTyPath') \), we have \( \bar{\HFLenv}_{t_0, \epsilon, \theta,\RTyPath} = \bar{\HFLenv}_{\term, \epsilon, \theta, \RTyPath} \).
    By definition, \( \bar{d}_{\theta t_0, \epsilon,\RTyPath} = \bar{d}_{\theta \term, \epsilon,\RTyPath} = \{ \, \} \) and \( \sem{t_0}(\bar{\HFLenv}_{t_0, \epsilon, \theta,\RTyPath}) = \sem{\term}(\bar{\HFLenv}_{t_0, \epsilon, \theta, \RTyPath}) = \sem{\term}(\bar{\HFLenv}_{\term, \epsilon,\theta}) \).
    By the induction hypothesis, \( \bar{d}_{\theta \term, \epsilon, \RTyPath} \sqsupseteq \sem{\term}(\bar{\HFLenv}_{\term, \epsilon, \theta, \RTyPath}) \).
    Hence \( \bar{d}_{\theta t_0, \epsilon, \RTyPath} \sqsupseteq \sem{t_0}(\bar{\HFLenv}_{t_0, \epsilon, \theta, \RTyPath}) \).
  \item Case \( t_0 = t_1\,t_2 \):
    Suppose that \( ((\theta t_1) (\theta t_2) \tilde{u}; \RTyPath) \red_\progd^* ((\theta t_2)\,\tilde{v}; \RTyPath') \).
    We first show that \( \bar{\HFLenv}_{t_2, \tilde{v}, \theta,\RTyPath'}(x) \sqsupseteq \bar{\HFLenv}_{t_1 t_2, \tilde{u}, \theta, \RTyPath}(x) \) for every \( x \in \{ f_1, \dots, f_n \} \cup \{ x_1, \dots,x_k \} \).
    \begin{itemize}
    \item Case \( x = x_i : \Tint \):
      Then \( \bar{\HFLenv}_{t_2, \tilde{v}, \theta, \RTyPath'}(x_i) = \bar{\HFLenv}_{t_1t_2, \tilde{u}, \theta, \RTyPath}(x_i) = \sem{s_i} \).
    \item Case \( x = x_i : \Pst \):
      Then
      \begin{align*}
        \bar{\HFLenv}_{t_2, \tilde{v}, \theta, \RTyPath'}(x_i)
        &= \GLB \{\, \bar{d}_{s_i, \tilde{w}, \RTyPath''} \mid ((\theta t_2)\,\tilde{v}; \RTyPath') \red^*_\progd (s_i \tilde{w}; \RTyPath''), \,\mbox{\( s_i \) originates from \( \theta(x_i) \)} \,\} \\
        &\sqsupseteq \GLB \{\, \bar{d}_{s_i, \tilde{w}, \RTyPath''} \mid ((\theta (t_1\,t_2))\,\tilde{u}; \RTyPath) \red^*_\progd (s_i \tilde{w}; \RTyPath''), \,\mbox{\( s_i \) originates from \( \theta(x_i) \)} \,\} \\
        &= \bar{\HFLenv}_{t_1 t_2, \tilde{u}, \theta, \RTyPath}(x_i)
      \end{align*}
      because \( ((\theta t_2)\,\tilde{v}; \RTyPath') \red^*_\progd (s_i \tilde{w}; \RTyPath'') \) (where \( s_i \) originates from \( \theta(x_i) \)) implies \( ((\theta (t_1\,t_2))\,\tilde{u}; \RTyPath) \red^*_\progd ((\theta t_2) \tilde{v}; \RTyPath') \red^*_\progd (s_i \tilde{w}; \RTyPath'') \) (where \( s_i \) originates from \( \theta(x_i) \)).
    \item Case \( x = f_i \):
      \begin{align*}
        \bar{\HFLenv}_{t_2, \tilde{v}, \theta, \RTyPath'}(f_i)
        &= \GLB \{\, \bar{d}_{f_i, \tilde{w}, \RTyPath''} \mid ((\theta t_2)\,\tilde{v}; \RTyPath') \red^*_\progd (f_i \tilde{w}; \RTyPath''), \,\mbox{\( f_i \) originates from \( t_2 \)} \,\} \\
        &\sqsupseteq \GLB \{\, \bar{d}_{f_i, \tilde{w}, \RTyPath''} \mid ((\theta (t_1\,t_2))\,\tilde{u}; \RTyPath) \red^*_\progd (f_i \tilde{w}; \RTyPath''), \,\mbox{\( f_i \) originates from \( t_1 t_2 \)} \,\} \\
        &= \bar{\HFLenv}_{t_1 t_2, \tilde{u}, \theta, \RTyPath}(f_i)
      \end{align*}
      because \( ((\theta t_2)\,\tilde{v}; \RTyPath') \red^*_\progd (f_i \tilde{w}; \RTyPath'') \) (where \( f_i \) originates from \( t_2 \)) implies \( ((\theta (t_1\,t_2))\,\tilde{u}; \RTyPath) \red^*_\progd ((\theta t_2) \tilde{v}; \RTyPath') \red^*_\progd (f_i \tilde{w}; \RTyPath'') \) (where \( f_i \) originates from \( t_1 t_2 \)).      
    \end{itemize}

    It is easy to see that \( \bar{\HFLenv}_{t_1, (\theta t_2)\,\tilde{u}, \theta, \RTyPath} \sqsupseteq \bar{\HFLenv}_{t_1t_2, \tilde{u}, \theta, \RTyPath} \).
    By the induction hypothesis,
    \[
    \bar{d}_{\theta t_1, (\theta t_2)\,\tilde{u}, \RTyPath} \sqsupseteq \sem{t_1}(\bar{\HFLenv}_{t_1, (\theta t_2)\,\tilde{u}, \theta, \RTyPath}) \sqsupseteq \sem{t_1}(\bar{\HFLenv}_{t_1t_2, \tilde{u}, \theta, \RTyPath}).
    \]
    By the definition of \( \bar{d}_{\theta t_1, (\theta t_2)\,\tilde{u}, \RTyPath} \),
    \[
    \bar{d}_{\theta t_1, (\theta t_2)\,\tilde{u}, \RTyPath} \big(\GLB \{ \bar{d}_{\theta t_2, \tilde{v}, \RTyPath'} \mid ((\theta t_1)\,(\theta t_2)\,\tilde{u}; \RTyPath) \red_\progd^* ((\theta t_2) \tilde{v}; \RTyPath') \}\big) = \bar{d}_{(\theta t_1)\,(\theta t_2),\,\tilde{u},\RTyPath}.
    \]
    By the induction hypothesis, for every reduction sequence \( ((\theta t_1)\,(\theta t_2)\,\tilde{u}; \RTyPath) \red^*_\progd ((\theta t_2) \tilde{v}; \RTyPath') \), we have \( \bar{d}_{\theta t_2, \tilde{v}, \RTyPath'} \sqsupseteq \sem{t_2}(\bar{\HFLenv}_{t_2, \tilde{v}, \theta, \RTyPath'}) \).
    Since \( \bar{\HFLenv}_{t_2, \tilde{v}, \theta, \RTyPath'}(x) \sqsupseteq \bar{\HFLenv}_{t_1t_2, \tilde{u}, \theta, \RTyPath'}(x) \) for every \( x \), we have
    \[
    \bar{d}_{\theta t_2, \tilde{v}, \RTyPath'} \sqsupseteq \sem{t_2}(\bar{\HFLenv}_{t_2, \tilde{v}, \theta, \RTyPath'}) \sqsupseteq \sem{t_2}(\bar{\HFLenv}_{t_1t_2, \tilde{u}, \theta, \RTyPath}).
    \]
    Because the reduction sequence \( ((\theta t_1)\,(\theta t_2)\,\tilde{u}; \RTyPath) \red^*_\progd ((\theta t_2) \tilde{v}; \RTyPath) \) is arbitrary,
    \[
    \big(\GLB \{ \bar{d}_{\theta t_2, \tilde{v}, \RTyPath'} \mid ((\theta t_1)\,(\theta t_2)\,\tilde{u}; \RTyPath) \red_\progd^* ((\theta t_2) \tilde{v}; \RTyPath') \}\big) \sqsupseteq \sem{t_2}(\bar{\HFLenv}_{t_1t_2, \tilde{u}, \theta, \RTyPath}).
    \]
    Therefore, by monotonicity of interpretation,
    \begin{align*}
      \bar{d}_{(\theta t_1) (\theta t_2), \,\tilde{u}, \RTyPath}
      &= \bar{d}_{\theta t_1, (\theta t_2)\,\tilde{u}, \RTyPath} \big(\GLB \{ \bar{d}_{\theta t_2, \tilde{v}, \RTyPath'} \mid ((\theta t_1)\,(\theta t_2)\,\tilde{u}; \RTyPath) \red_\progd^* ((\theta t_2) \tilde{v}; \RTyPath') \}\big) \\
      &\sqsupseteq \sem{t_1}(\bar{\HFLenv}_{t_1t_2, \tilde{u}, \theta, \RTyPath}) \big( \sem{t_2}(\bar{\HFLenv}_{t_1t_2, \tilde{u}, \theta, \RTyPath}) \big) \\
      &= \sem{t_1\,t_2}(\bar{\HFLenv}_{t_1t_2,\tilde{u},\theta,\RTyPath}).
    \end{align*}
  \item Case \( t_0 = \term_1 \nondet \term_2 \):
    Then \( \tilde{u} = \epsilon \).
    Assume that \( \RTyPath = \PLeft \RTyPath' \); the other case can be proved similarly.
    
    Then we have \( (\theta t_0; \RTyPath) \red_\progd (\theta \term_1; \RTyPath') \).
    By the definition of \( \bar{d}_{\theta \term_1, \epsilon, \RTyPath'} \) and the induction hypothesis,
    \[
    \{\,\} = \bar{d}_{\theta \term_1,\epsilon, \RTyPath'} \sqsupseteq \sem{\term_1}(\bar{\HFLenv}_{\term_1, \epsilon, \theta, \RTyPath'}).
    \]
    Because \( (\theta t_0; \RTyPath) \red_\progd (\theta \term_1; \RTyPath') \), we have \( \bar{\HFLenv}_{\term_1, \epsilon, \theta, \RTyPath'} \sqsupseteq \bar{\HFLenv}_{t_0, \epsilon, \theta, \RTyPath} \).
    By monotonicity of interpretation,
    \[
    \{\,\} = \bar{d}_{\theta \term_1,\epsilon, \RTyPath'} \sqsupseteq \sem{\term_1}(\bar{\HFLenv}_{\term_1, \epsilon, \theta, \RTyPath'}) \sqsupseteq \sem{\term_1}(\bar{\HFLenv}_{t_0, \epsilon, \theta, \RTyPath}).
    \]
    Hence
    \[
    \{\,\} \sqsupseteq \sem{\term_1}(\bar{\HFLenv}_{t_0, \epsilon, \theta, \RTyPath}) \sqcap \sem{\term_2}(\bar{\HFLenv}_{t_0, \epsilon, \theta, \RTyPath}) = \sem{t_0}(\bar{\HFLenv}_{t_0, \epsilon, \theta, \RTyPath}).
    \]
\end{itemize}
\end{myproof}

The strategy \( \Strategy_{\prog, \RTyPath_0} \) of \( \SGame_{\lts_0,\prog} \) on \( (\{\stunique\}, \mainfun) \) is defined as follows.
Each play in the domain of \( \Strategy \)
\[
(\{\stunique\}, \mainfun) \cdot O_1 \cdot (p_1, g_1) \cdot O_2 \cdot \dots \cdot O_k \cdot (p_k, g_k)
\]
is associated with a call-sequence
\begin{align*}
  &
  (\mainfun; \RTyPath_0) \CallSeqN{m_1}_\progd (g_1\,\tilde{u}_1; \RTyPath_1) \CallSeqN{m_2}_\progd (g_2\,\tilde{u}_2; \RTyPath_2) \CallSeqN{m_3}_\progd \dots \CallSeqN{m_{k-1}}_\progd (g_{k-1}\,\tilde{u}_{k-1}; \RTyPath_{k-1}) \CallSeqN{m_k}_\progd (g_k\,\tilde{u}_k; \RTyPath_{k})
\end{align*}
such that \( p_j \sqsupseteq \bar{d}_{g_j, \tilde{u}_j} \) for every \( j = 1, 2, \dots, k \).
The \emph{canonical associated call-sequence} of the play is the minimum one ordered by the lexicographic ordering of reduction steps \( (m_1, m_2, \dots, m_k) \).

Assume that the above call-sequence is canonical.
The next step of this reduction sequence is \( ([\tilde{u}_k/\tilde{x}_k]t; \RTyPath_k) \) if \( (g_k\,\seq{x}_k = t) \in \progd \).
Let \( \vartheta = [\tilde{u}_k/\tilde{x}_k] \).
In this situation, the strategy \( \Strategy_{\prog,\RTyPath} \) chooses \( (\bar{\HFLenv}_{t, \epsilon, \vartheta, \RTyPath_k}(f_1), \dots, \bar{\HFLenv}_{t, \epsilon, \vartheta, \RTyPath_k}(f_n)) \) as the next node.
This is a valid choice, i.e.:
\begin{lemma}
  \( p_k \sqsupseteq \sem{\progd(g_k)}([f_1 \mapsto \bar{\HFLenv}_{t, \epsilon, \vartheta, \RTyPath_k}(f_1), \dots, f_n \mapsto \bar{\HFLenv}_{t, \epsilon, \vartheta, \RTyPath_k}(f_n)]) \).
\end{lemma}
\begin{myproof}
  By Lemma~\ref{lem:appx:call-seq-to-hes:soundness:infinite-subject-expansion}, we have:
  \[ \bar{d}_{\vartheta t, \epsilon, \RTyPath_k} = \{\,\} \sqsupseteq \sem{t}(\bar{\HFLenv}_{t, \epsilon,\vartheta, \RTyPath_k})
  = \sem{\lambda \seq{x}_k. t}([f_1 \mapsto \bar{\HFLenv}_{t, \epsilon, \vartheta,\RTyPath_k}(f_1), \dots, f_n \mapsto \bar{\HFLenv}_{t, \epsilon, \vartheta,\RTyPath_k}(f_n)])(\bar{\HFLenv}_{t, \epsilon, \vartheta,\RTyPath}(\seq{x}_k) ).\]
  By definition,
  \(\bar{d}_{g_k, \tilde{u}_k, \RTyPath_k}\) is the greatest element such that
  \( \bar{d}_{g_k, \tilde{u}_k, \RTyPath_k}(\bar{\HFLenv}_{t, \epsilon, \vartheta,\RTyPath}(\seq{x}_k) )
   = \{\,\}\).
  Therefore
  \[
  \bar{d}_{g_k, \tilde{u}_k} \sqsupseteq \sem{\lambda \seq{x}_k. t}([f_1 \mapsto \bar{\HFLenv}_{t, \epsilon, \vartheta,\RTyPath_k}(f_1), \dots, f_n \mapsto \bar{\HFLenv}_{t, \epsilon, \vartheta,\RTyPath_k}(f_n)]).
  \]
  Since \( p_k \sqsupseteq \bar{d}_{g_k, \tilde{u}_k, \RTyPath_k} \), we have
  \[
  p_k \sqsupseteq \bar{d}_{g_k, \tilde{u}_k, \RTyPath_k} \sqsupseteq \sem{\lambda \tilde{x}_k. t}([f_1 \mapsto \bar{\HFLenv}_{t, \epsilon, \vartheta,\RTyPath_k}(f_1), \dots, f_n \mapsto \bar{\HFLenv}_{t, \epsilon, \vartheta,\RTyPath_k}(f_n)]).
  \]
\end{myproof}
Now we have a play
\[
(\{\stunique\}, \mainfun) \cdot O_1 \cdot (p_1, g_1) \cdot O_2 \cdot \dots \cdot O_k \cdot (p_k, g_k) \cdot (\bar{\HFLenv}_{t, \epsilon, \theta, \RTyPath_k}(f_1), \dots, \bar{\HFLenv}_{t,\epsilon,\theta,\RTyPath_k}(f_n))
\]
associated with the call sequence
\begin{align*}
  &
  (\mainfun; \RTyPath_0) \CallSeqN{m_1}_\progd (g_1\,\tilde{u}_1; \RTyPath_1) \CallSeqN{m_2}_\progd (g_2\,\tilde{u}_2; \RTyPath_2) \CallSeqN{m_3}_\progd \dots \CallSeqN{m_{k-1}}_\progd (g_{k-1}\,\tilde{u}_{k-1}; \RTyPath_{k-1}) \CallSeqN{m_k}_\progd (g_k\,\tilde{u}_k; \RTyPath_{k})
\end{align*}
Let \( (p_{k+1}, g_{k+1}) \) be the next opponent move.
By definition of the game,
\begin{align*}
  p_{k+1} \sqsupseteq \bar{\HFLenv}_{t, \epsilon, \theta, \RTyPath}(g_{k+1})
  &= \GLB \{\, \bar{d}_{g_{k+1}, \tilde{v}, \RTyPath} \mid (\theta t; \RTyPath_k) \red^*_\prog (g_{k+1} \tilde{v}; \RTyPath), \,\mbox{\( g_{k+1} \) originates from \( t \)} \,\} \\
  &= \GLB \{\, \bar{d}_{g_{k+1}, \tilde{v}, \RTyPath} \mid (g_k\,\tilde{u}_k; \RTyPath_k) \CallSeq_{\progd} (g_{k+1}\,\tilde{v}; \RTyPath) \,\}.
\end{align*}
Since \( p_{k+1} \) is a complete coprime, \( p_{k+1} \sqsupseteq \bar{d}_{g_{k+1}, \tilde{v}, \RTyPath} \) for some \( \tilde{v} \) and \( \RTyPath \) with \( (g_k\,\tilde{u}_k; \RTyPath_k) \CallSeq_{\progd} (g_{k+1}\,\tilde{v}; \RTyPath) \).
Hence we have an associated call-sequence
\begin{align*}
  &
  (\mainfun; \RTyPath_0) \CallSeqN{m_1}_\progd (g_1\,\tilde{u}_1; \RTyPath_1) \CallSeqN{m_2}_\progd (g_2\,\tilde{u}_2; \RTyPath_2) \CallSeqN{m_3}_\progd \dots \CallSeqN{m_k}_\progd (g_k\,\tilde{u}_k; \RTyPath_{k}) \CallSeqN{m_{k+1}} (g_{k+1}\,\tilde{u}; \RTyPath)
\end{align*}
of
\[
(\{\stunique\}, \mainfun) \cdot O_1 \cdot (p_1, g_1) \cdot O_2 \cdot \dots \cdot O_k \cdot (p_k, g_k) \cdot (\bar{\HFLenv}_{t, \epsilon, \theta, \RTyPath_k}(f_1), \dots, \bar{\HFLenv}_{t,\epsilon,\theta,\RTyPath_k}(f_n)) \cdot (p_{k+1}, g_{k+1})
\]
as desired.

\begin{lemma}\label{lem:appx:liveness-from-call-seq:soundness}
  If the (unique) call sequence following \( \RTyPath_0 \) does not satisfy the parity condition, then \( \Strategy_{\prog,\RTyPath_0} \) is an winning strategy of \( \CounterGame_{\lts_0,(\prog,\Pfun)} \) on \( (\{\}, \mainfun) \).
\end{lemma}
\begin{myproof}
  The above argument shows that \( \Strategy_{\prog,\RTyPath_0} \) is a strategy of \( \CounterGame_{\lts_0,(\prog,\Pfun)} \) on \( (\{\}, \mainfun) \).
  We prove that this is winning.

  Assume an infinite play
  \[
  (\{\},\mainfun) \cdot v_1 \cdot (p_1,g_1) \cdot v_2 \cdot (p_2,g_2) \cdot \dots
  \]
  that conforms with \( \Strategy_{\prog,\RTyPath} \) and starts from \( (\{\stunique\},\mainfun) \).
  Then each odd-length prefix is associated with the canonical call-sequence
  \[
  (\mainfun; \RTyPath_0) \CallSeqN{k_{m,1}}_{\progd} (g_1\,\seq{u}_{m,1}; \RTyPath_{m,1}) \CallSeqN{k_{m,2}}_{\progd} (g_2\,\seq{u}_{m,2}; \RTyPath_{m,2}) \CallSeqN{k_{m,3}}_{\progd} \dots \CallSeqN{k_{m,m}}_{\progd} (g_m\,\seq{u}_{m,m}; \RTyPath_{m,m}).
  \]
  By the definitions of \( \Strategy_{\prog,\RTyPath_0} \) and canonical call-sequence (which is the minimum with respect to the lexicographic ordering on \( (k_{m,1}, \dots, k_{m,m}) \)), a prefix of the canonical call-sequence is the canonical call-sequence of the prefix.
  In other words, \( k_{m,i} = k_{m',i} \) for every \( i \ge 1 \) and \( m, m' \ge i \).
  Since the reduction sequence is completely determined by a choice \( \RTyPath_0 \), the initial term \( \mainfun \) and the number of steps, we have \( \RTyPath_{m,i} = \RTyPath_{m',i} \) for every \( i \ge 1 \) and \( m, m' \ge i \).
  Let us write \( k_i \) for \( k_{i,i} = k_{i+1,i} = \cdots \), \( \seq{u}_i \) for \( \seq{u}_{i,i} = \seq{u}_{i+1,i} = \cdots \) and \( \RTyPath_i \) for \( \RTyPath_{i,i} = \RTyPath_{i+1,i} = \cdots \).
  Now we have an infinite call-sequence
  \[
  \RTyPath_0 \Vdash
  \mainfun \CallSeqN{k_1}_{\progd} g_1\,\seq{u}_{1} \CallSeqN{k_2}_{\progd} g_2\,\seq{u}_{2} \CallSeqN{k_3}_{\progd} \dots
  \]
  following \( \RTyPath_0 \).
  By Corollary~\ref{cor:liveness:unique-existence-of-call-seq}, this is the unique call-sequence following \( \RTyPath_0 \).
  Hence, by the assumption, this infinite call-sequence does not satisfy the parity condition.
  This means that this play is P-winning (recall the definition of the priorities of \( \CounterGame_{\lts_0,(\prog,\Pfun)} \)).
\end{myproof}

\begin{myproof}[Proof of Theorem~\ref{theorem:callseq}]
  We prove the following result:
  \begin{quote}
    \(\models_\CSA (\prog,\Pfun)\) if and only if \( \lts_0 \models \HESf_{(P,\Pfun),\CSA} \).
  \end{quote}
  
  \( (\Rightarrow) \)
  Suppose that \( \models_\CSA (\prog,\Pfun) \).
  Then by Lemma~\ref{lem:appx:liveness-from-call-seq:parity-call-seq-to-winning},
  there exists a winning strategy of \( \SGame_{\lts_0,(\prog,\Pfun)} \) on \( (\{\stunique\}, \mainfun) \).
  Then by Corollary~\ref{cor:appx:liveness-from-call-seq:game-equivalence}, there also exists a winning strategy of \( \SGame_{\lts_0,\HESf_{(P,\Pfun),\CSA}} \) on \( (\{\stunique\}, \mainfun) \).
  By Theorem~\ref{thm:appx:game-semantics:game-and-interpretation}, this implies \( \stunique \in \sem{\HESf_{(\prog,\Pfun),\CSA}} \).
  By definition, \( \lts_0 \models \HESf_{(P,\Pfun),\CSA} \).
  
  \( (\Leftarrow) \)
  We prove the contraposition.
  Suppose that \( \models_\CSA (\prog,\Pfun) \) does not hold.
  Then there exist a choice \( \RTyPath_0 \) and an infinite call-sequence following \( \RTyPath_0 \) that does not satisfy the parity condition.
  Then by Lemma~\ref{lem:appx:liveness-from-call-seq:soundness}, \( \Strategy_{\prog,\RTyPath_0} \) is a winning strategy of \( \CounterGame_{\lts_0,(\prog,\Pfun)} \) on \( (\{\},\mainfun) \).
  By Corollary~\ref{cor:appx:liveness-from-call-seq:soundness:game-equivalence}, there also exists a winning strategy of \( \CounterGame_{\lts_0,\HESf_{(\prog,\Pfun),\CSA}} \) on \( (\{\},\mainfun) \).
  Then by Lemma~\ref{lem:appx:game-semantics:opposite-game}, Opponent wins \( \SGame_{\lts_0,\HESf_{(\prog,\Pfun),\CSA}} \) on \( (\{\stunique\}, \mainfun) \).
  Hence by Theorem~\ref{thm:appx:game-semantics:game-and-interpretation}, \( \stunique \notin \sem{\HESf_{(\prog,\Pfun),\CSA}} \).
\end{myproof}

\subsection{Proof of Theorem~\ref{thm:linevess:eff-sel-sound-and-complete}}
Assume a total order on the set of states of the automaton \( \PWA \), fixed in the sequel.
Recall \rn{(IT-Event)} rule in Fig.~\ref{fig:inter}:
\infrule[IT-Event]{
  \delta(q, \lab) = \{q_1, \dots, q_n\} \\
  \ITE \RaiseP \Omega(q_i) \pInter \term \COL q_i \ESRel \term'_i \quad\text{(for each \(i \in \{1, \dots, n\}\))}
}{
  \ITE
  \pInter
  (\evexp{\lab}{\term}) \COL q  \ESRel (\evexp{\lab}{\term'_1 \nondet \cdots \nondet \term'_n})
}
Since the order of \( \term'_1, \dots, \term'_n \) is not important, we can assume without loss of generality that \( q_1 < q_2 < \dots < q_n \) holds for every instance of the above rule used in a derivation.

Here we prove Theorem~\ref{thm:linevess:eff-sel-sound-and-complete} for the translation using the following rule \rn{IT-Event'} instead of \rn{IT-Event}.
\infrule[IT-Event']{
  \delta(q, \lab) = \{q_1, \dots, q_n\} \qquad q_1 < q_2 < \dots < q_n \\
  \ITE \RaiseP \Omega(q_i) \pInter \term \COL q_i \ESRel \term'_i \quad\text{(for each \(i \in \{1, \dots, n\}\))}
}{
  \ITE
  \pInter
  (\evexp{\lab}{\term}) \COL q  \ESRel (\evexp{\lab}{\term'_1 \nondet \cdots \nondet \term'_n})
}

\subsubsection{Notations}
Given an intersection type \( \Itype = \bigwedge_{1 \le i \le k} (\Atype_i, m_i) \) and a variable \( x \), we abbreviate the sequence \( x_{\Atype_1, m_1} x_{\Atype_2, m_2} \dots x_{\Atype_k, m_k} \) as \( \dup{x,\Itype} \).
For \( \Itype = \INT \), we write \( \dup{x,\Itype} \) for \( x_{\INT} \).

For an intersection type \( \Itype = \bigwedge_{1 \le i \le k} (\Atype_i, m_i) \), we write \( \ZEnv{x : \Itype} \) for the intersection type environment \( \{ x : (\Atype_i, m_i, 0) \mid 1 \le i \le k \} \).
If \(\Itype = \Tint\), \(\ZEnv{x \COL \Itype}\) means \(\{x \COL \Tint\}\).
Similarly, for a top-level environment \( \TopEnv \), we write \( \TopEnvZ \) to mean \( \{ x : (\Atype, m, 0) \mid x : (\Atype, m) \in \TopEnv \} \).

Given an intersection type environment \( \ITE \), we write \( \ITE^{\Mark} \) for \( \{ x^{\Mark} : (\Atype, m, m') \mid x : (\Atype, m, m') \in \ITE \} \cup \{ x^{\Mark} : \INT \mid x : \INT \in \ITE \} \).
Similarly, for a top-level environment \( \TopEnv \), we write \( \TopEnv^{\Mark} \) for \( \{ x^{\Mark} : (\Atype, m) \mid x : (\Atype, m) \in \TopEnv \} \).

For an intersection type environment \( \ITE \) and a top-level environment \( \TopEnv \), we write \(\ITE \EnvRel \TopEnv\) if
\(f \COL (\Atype, m) \in \TopEnv\) for each \((f \COL (\Atype, m, m')) \in \ITE\).


\subsubsection{Modified type-based translation}
Given programs \( \prog \) and \( \prog' \) with \( \TopEnv \pInter \prog \ESRel (\prog', \Pfun) \), the first step to prove Theorem~\ref{thm:linevess:eff-sel-sound-and-complete} is to compare reduction sequences of \( \prog \) and \( \prog' \).
There is a little gap between reduction sequences of \( \prog \) and \( \prog' \) since a reduction sequence of \( \prog' \) expresses a reduction sequence of \( \prog \) together with a run of \( \PWA \) over the events generated by the reduction sequence of \( \prog \).
In particular a nondeterministic branch in \( \term' \) comes from either
\begin{itemize}
\item a non-deterministic branch \( \term_1 \nondet \term_2 \) in \( P \), or
\item non-determinism of the transition rule of the automaton \( \PWA \).
\end{itemize}
To fill the gap, we shall distinguish between the two kinds of non-deterministic branches, by using \( \enondet \) for the latter.

Formally let us introduce a new binary construct \( \enondet \) to the syntax of terms.
Here we use the convention that \( \nondet \) and \( \enondet \) are right associative, i.e.~\( t_1 \nondet t_2 \nondet t_3 \) (resp.~\( t_1 \enondet t_2 \enondet t_3 \)) means \( t_1 \nondet (t_2 \nondet t_3) \) (resp.~\( t_1 \enondet (t_2 \enondet t_3) \)).
The operational behavior of \( \enondet \) is the same as that of \( \nondet \), i.e.,
\[
(\term_1 \enondet \term_2;\, \PLeft \RTyPath) \stackrel{\epsilon}{\red}_D (\term_1; \RTyPath)
\qquad
(\term_1 \enondet \term_2;\, \PRight \RTyPath) \stackrel{\epsilon}{\red}_D (\term_2; \RTyPath)
\]
where \( \RTyPath \in \{\PLeft,\PRight\}^{\omega} \) is a choice sequence.
Hence we have
\[
((\evexp{\lab}{(\term_1 \enondet \dots \enondet \term_n)});\,
\underbrace{\PRight \dots \PRight}_{i-1} \PLeft  \RTyPath)
\Preds{\lab}{\progd}
(\term_{i};\, \RTyPath)
\]
for \( 1 \le i < n \) and
\[
((\evexp{\lab}{(\term_1 \enondet \dots \enondet \term_n)});\,
\underbrace{\PRight \dots \PRight}_{n-1} \RTyPath)
\Preds{\lab}{\progd}
(\term_{n};\, \RTyPath).
\]
\begin{definition}[Modified type-based transformation]
  The \emph{modified type-based transformation judgment} \( \ITE \pInter \term : \Atype \eESRel \term' \) is a quadruple where \( \term \) is a term without \( \enondet \) and \( \term' \) is a term possibly having \( \enondet \).
  The \emph{modified type-based transformation relation} is defined by the rules in Fig.~\ref{fig:appx:liveness:inter-alternative}.
  The translation \(\pInter \prog \eESRel (P', \Pfun)\) of programs is defined in the same way as \(\pInter \prog \ESRel (P', \Pfun)\).
\end{definition}
\begin{figure}[t]
  \infrule[IT-Event-Alt]{
    \delta(q, \lab) = \{q_1, \dots, q_n\} \quad\text{(\(q_1 < q_2 < \dots < q_n\))} \\
    \ITE \RaiseP \Omega(q_i) \pInter \term \COL q_i \eESRel \term'_i \quad\text{(for each \(i \in \{1, \dots, n\}\))}
  }{
    \ITE
    \pInter
    (\evexp{\lab}{\term}) \COL q  \eESRel (\evexp{\lab}{\term'_1 \enondet \cdots \enondet \term'_n})
  }
  \caption{Modified type-based transformation rules.  Other rules are obtained by replacing \( \ESRel \) with \( \eESRel \) in the rules in Figure~\ref{fig:inter} except for \textsc{(IT-Event)}.}
  \label{fig:appx:liveness:inter-alternative}
\end{figure}
The next lemma establishes the connection between the original and modified transformations.
Given a program \( \prog' \) with \( \enondet \), we write \( [\nondet/\enondet]\prog' \) for the program obtained by replacing \( \enondet \) with \( \nondet \).
\begin{lemma}~\label{lem:appx:liveness:relating-modified-and-original-translations2}
  If \(\TopEnv \pInter \prog \eESRel (\prog', \Pfun)\), then \( \TopEnv \pInter \prog \ESRel ([\nondet/\enondet]\prog', \Pfun) \).
  Conversely, if \( \TopEnv \pInter \prog \ESRel (\prog', \Pfun) \), then there exists \( \prog'' \) such that \( \TopEnv \pInter \prog \eESRel (\prog'', \Pfun) \) and \( \prog' = [\nondet/\enondet]\prog'' \).
\end{lemma}
\begin{myproof}
  Straightforward induction.
\end{myproof}

Suppose that \( \TopEnv \pInter P \eESRel (P', \Pfun) \).
Since \( P' \) simulates both \( P \) and the automaton \( \PWA \), an infinite reduction sequence of \( P' \) induces a pair of an infinite event sequence \( \seq{\ell} \) and a run \( R \) of \( \PWA \) over \( \seq{\ell} \).
\begin{definition}[Induced event sequence and run]
  Let \( D' \) be a function definition possibly containing \( \enondet \).
  The reduction relation \( (\termalt, \RTyPath, q) \Preds[N]{\seq{\ell},\runseq}{D'} (\termalt', \RTyPath', q') \), meaning that the \( N \)-step reduction following \( \RTyPath \) from \( \termalt \) to \( \termalt' \) generates events \( \seq{\ell} \) associated with a run \( q \runseq \) ending with \( q' \), is defined by the following rules.
  \infrule{
    (\termalt, \RTyPath) \Preds[N]{\epsilon}{D'} (\termalt', \RTyPath')
  }{
    (\termalt, \RTyPath, q) \Preds[N]{\epsilon,\epsilon}{D'} (\termalt', \RTyPath', q)
  }
  \infrule{
    \delta(q, \lab) = \{q_{1}, \dots, q_{n}\} \quad\text{ (\( q_1 < q_2 < \dots < q_n \))}
    \andalso
    (\termalt_i, \RTyPath, q_i) \Preds[N]{\seq{\ell},\runseq}{D'} (\termalt', \RTyPath', q')
  }{
    (\evexp{\lab}{(\termalt_1 \enondet \dots \enondet \termalt_n)}, \underbrace{\PRight \dots \PRight}_{i-1}\PLeft \RTyPath, q)
    \Preds[N+i+1]{\lab \seq{\ell}, q_i \runseq}{D'} (\termalt', \RTyPath', q')
  }
  \infrule{
    \delta(q, \lab) = \{q_{1}, \dots, q_{n}\} \quad\text{ (\( q_1 < q_2 < \dots < q_n \))}
    \andalso
    (\termalt_n, \RTyPath, q_n) \Preds[N]{\seq{\ell},\runseq}{D'} (\termalt', \RTyPath', q')
  }{
    (\evexp{\lab}{(\termalt_1 \enondet \dots \enondet \termalt_n)}, \underbrace{\PRight \dots \PRight}_{n-1} \RTyPath, q)
    \Preds[N+n]{\lab \seq{\ell}, q_n \runseq}{D'} (\termalt', \RTyPath', q')
  }
  If \( (\termalt_0, \RTyPath_0, q_0) \Preds[N_1]{\seq{\ell}_1,\runseq_1}{D'} (\termalt_1, \RTyPath_1, q_1) \Preds[N_2]{\seq{\ell}_2,\runseq_2}{D'} \dots \), we write
  \(
  \RTyPath_0 \Vdash (\termalt_0, q_0)  \Preds[N_1]{\seq{\ell}_1,\runseq_1}{D'} (\termalt_1, q_1) \Preds[N_2]{\seq{\ell}_2,\runseq_2}{D'} \dots
  \).
  If the number of steps (resp.~\( \RTyPath \)) is not important, we write as \( \RTyPath \Vdash (\termalt, q) \Preds{\seq{\ell},\runseq}{D'} (\termalt', q') \) (resp.~\( (\termalt,q) \Preds[N]{\seq{\ell},\runseq}{D'} (\termalt', q') \)).
  Other notations such as \( (\termalt,q) \Preds{\seq{\ell},\runseq}{D'} (\termalt',q') \) are defined similarly.
\end{definition}

It is easy to show that, if \( (\termalt, q) \Preds{\seq{\ell},\runseq}{D'} (\termalt', q') \), then \( q\runseq \) is a run of \( \PWA \) over \( \seq{\ell} \).
Obviously \( \RTyPath \Vdash (\termalt, q) \Preds{\seq{\ell},\runseq}{D'} (\termalt', q') \) implies \( \RTyPath \Vdash \termalt \Preds{\seq{\ell}}{D'} \termalt' \).
As proved later in Lemma~\ref{lem:appx:liveness:summary},
the converse also holds: Given state \( q \) and \( \termalt \Preds{\seq{\ell}}{D'} \termalt' \) (with a mild condition), there exist \( \runseq \) and \( q' \) such that \( (\termalt, q) \Preds{\seq{\ell},\runseq}{D'} (\termalt', q') \).

\subsubsection{Basic properties of the type system and transformation}
We prove Weakening Lemma and Substitution Lemma.

\begin{lemma}[Weakening]\label{lem:appx:liveness:weakening}
 If \(\ITE \pInter \term \COL \Atype \eESRel \term'\) and \(\ITE \subseteq \ITE'\),
 then \(\ITE' \pInter \term \COL \Atype \eESRel \term'\).
\end{lemma}
\begin{myproof}
  Straightforward induction on structure of the derivations.
  We discuss only one case below.
  \begin{itemize}
   \item Case for \rn{IT-Event-Alt}:
    Then we have
    \begin{align*}
      \Atype = q \qquad
      \delta(q, \lab) = \{q_1, \dots, q_n\} \qquad
      \term &= \evexp{\lab}{\termalt} \qquad
      \term' = \evexp{\lab}{(\termalt_1' \enondet \dots \enondet \termalt_n')} \\
     \forall i \in \{1, \dots, n\}. &\ITE \RaiseP \Omega(q_i) \pInter \termalt \COL \Atype \eESRel \termalt'_i.
    \end{align*}
    Since \(\ITE \subseteq \ITE'\), for each \(i \in \{1, \dots, n\}\), we have \(\ITE \RaiseP \Omega(q_i) \subseteq \ITE' \RaiseP \Omega(q_i)\).
    By the induction hypothesis,
    \[
     \forall i \in \{1, \dots, n\}. \ITE' \RaiseP \Omega(q_i) \pInter \termalt \COL \Atype \eESRel \termalt'_i.
    \]
    By using \rn{IT-Event-Alt}, we obtain
    \(\ITE' \pInter \term \COL \Atype \eESRel \term'\).
  \end{itemize}
\end{myproof}

In order to simplify the statement of the substitution lemma, we introduce the following abbreviation.
For \( \Itype = \bigwedge_{1 \le i \le l} (\Atype_i, m_i) \),
we write 
\[
\ITE \pInter \term \COL \Itype \eESRel \seq{\termalt}
\quad\stackrel{\mathrm{def}}{\Leftrightarrow}\quad
\forall i \in \{1, \dots, l\}. \ITE \RaiseP m_i \pInter \term \COL \Atype_i \eESRel \termalt_i.
\]
where \( \seq{\termalt} = s_1s_2 \dots s_l \).
If \( \Itype = \Tint \), the judgment \( \ITE \pInter \term \COL \Itype \eESRel \seq{\termalt} \) has the obvious meaning (in this case, \( \seq{\termalt} \) is of length \( 1 \)).
By using this notation, the application rule can be simply written as
\infrule{
  \ITE \pInter \term_1 \COL \Itype \to \Atype \eESRel \term'_1
  \andalso
  \ITE \pInter \term_2 \COL \Itype \eESRel \seq{\term}'_2
}{
  \ITE \pInter \term_1\,\term_2 \COL \Atype \eESRel \term'_1\,\seq{\term}'_2
}

\begin{lemma}[Substitution]\label{lem:appx:liveness:substitution}
  Assume that
  \begin{align*}
    &
    \ITE \pInter u \COL \Itype \eESRel \seq{u}'
    \\
    &
    ((\ITE \cup \ZEnv{x \COL \Itype}) \RaiseP n) \cup \ITE' \pInter \term \COL \Atype \eESRel \term' \qquad
    \text{(\( x \notin \dom(\ITE) \)).}
 \end{align*}
 Then
 \[
 (\ITE \RaiseP n) \cup \ITE' \pInter [u/x] \term \COL \Atype \eESRel [\seq{u}'/\dup{x,\Itype}]t'.
 \]
\end{lemma}
\begin{myproof}
  The proof proceeds by induction on the derivation of
  \(((\ITE \cup \ZEnv{x \COL \Itype}) \RaiseP n) \cup \ITE' \pInter \term \COL \Atype \eESRel \term'\),
  with case analysis on the last rule used.
  
  \begin{itemize}
  \item Case for \rn{IT-Var}:
        The case where \(\term \neq x\) is trivial.
        Assume that \(\term = x\).
        The type \( \Itype \) is either \( \Tint \) or \( \bigwedge_{1 \le i \le l} (\Atype_i, m_i) \).
        The former case is easy; we prove the latter case.

        Assume that \( \Itype = \bigwedge_{1 \le i \le l} (\Atype_i, m_i) \).
        Since  \(((\ITE \cup \ZEnv{x \COL \Itype}) \RaiseP n) \cup \ITE' \pInter x \COL \Atype \eESRel \term'\),
        there exists \(k \in \{1, \dots, l\}\) such that \( \Atype = \Atype_k \), \( \term' = x_{\theta_k,m_k} \) and \( n = m_k \).
        By the assumption \( \ITE \pInter u \COL \Itype \eESRel \seq{u}' \), we have \(\ITE \RaiseP m_k \pInter u \COL \Atype_k \eESRel u'_k\).
        Since \([u/x]x = u\) and \([\seq{u}'/\dup{x,\Itype}]x_{\Atype_k, m_k} = u'_k\), we have
        \begin{align*}
          \ITE \RaiseP m_k &\pInter [u/x]x \COL \Atype_k \eESRel [\seq{u'}/\dup{x,\Itype}] x_{\Atype_k, m_k}.
        \end{align*}
        By Weakening (Lemma~\ref{lem:appx:liveness:weakening}), we have
        \begin{align*}
          (\ITE \RaiseP m_k) \cup \ITE' &\pInter [u/x]x \COL \Atype_k \eESRel [\seq{u'}/\dup{x,\Itype}] x_{\Atype_k, m_k}.
        \end{align*}
        as required.
  \item Case for \rn{IT-Event-Alt}:
        Then \(\term = \evexp{\lab}{\term_1}\) and we have
        \begin{align*}
         \Atype = q \qquad
         \delta(q, \lab) = \{q_1, \dots, q_k\} \qquad q_1 < q_2 < \dots < q_k \\
         \term'=\evexp{\lab}{(\term'_1 \enondet \dots \enondet \term'_k)} \\
         \forall j \in \{1, \dots, k\}. (((\ITE \cup \ZEnv{x \COL \Itype}) \RaiseP n) \cup \ITE') \RaiseP \Omega(q_j) &\pInter \term_1 \COL q_j \eESRel \term'_j.
        \end{align*}
        For each \(j \in \{1, \dots, k\}\), we have
        \begin{align*}
          (((\ITE \cup \ZEnv{x \COL \Itype}) \RaiseP n) \cup \ITE') \RaiseP \Omega(q_j) 
          & = ((\ITE \cup \ZEnv{x \COL \Itype}) \RaiseP n \RaiseP \Omega(q_j)) \cup (\ITE' \RaiseP \Omega(q_j)) \\
          ((\ITE \RaiseP n) \cup \ITE') \RaiseP \Omega(q_j) 
          & = (\ITE \RaiseP n \RaiseP \Omega(q_j)) \cup (\ITE' \RaiseP \Omega(q_j)).
        \end{align*}
        Thus, one can apply the induction hypotheses, obtaining
        \begin{align*}
         \forall j \in \{1, \dots, k\}. ((\ITE  \RaiseP n) \cup \ITE') \RaiseP \Omega(q_j) &\pInter [u/x] \term_1 \COL q_j \eESRel  [\seq{u'}/\dup{x,\Itype}] \term'_j.
        \end{align*}
        By using \rn{IT-Event-Alt}, we have the following as required:
        \begin{align*}
         (\ITE  \RaiseP n) \cup \ITE' &\pInter [u/x] (\evexp{\lab}{\term_1}) \COL q \eESRel  [\seq{u'}/\dup{x,\Itype}] (\evexp{\lab}{(\term'_1 \enondet \dots \enondet \term'_k)}).
        \end{align*}
  \item Case for \rn{IT-App}:
        Then \(\term=\term_1\;\term_2\) and we have
        \begin{align*}
         \term'=\term'_1\;\term'_{2,1}\dots \term'_{2,k} \qquad
         ((\ITE \cup \ZEnv{x \COL \Itype}) \RaiseP n) \cup \ITE' &\pInter \term_1 \COL \bigwedge_{1 \le j \le k}(\Atype'_{j}, m'_{j}) \to \Atype
         \eESRel \term'_1 \\
         \forall j \in \{1,\dots,k\}. (((\ITE \cup \ZEnv{x \COL \Itype}) \RaiseP n) \cup \ITE') \RaiseP m'_{j} &\pInter \term_2 \COL \Atype'_{j} \eESRel \term'_{2,j}.
        \end{align*}
        By the induction hypothesis, we have
        \begin{align*}
         (\ITE \RaiseP n) \cup \ITE' \pInter [u/x]\term_1 \COL \bigwedge_{1 \le j \le k}(\Atype'_{j}, m'_{j}) \to \Atype
         \eESRel [\seq{u'}/\dup{x,\Itype}]\term'_1 \\
         \forall j \in \{1,\dots,k\}. ((\ITE \RaiseP n) \cup \ITE') \RaiseP m'_{j} \pInter [u/x]\term_2 \COL \Atype'_{j} \eESRel [\seq{u'}/\dup{x,\Itype}] \term'_{2,j}.
        \end{align*}
        By using \rn{IT-App}, we obtain:
        \begin{align*}
         (\ITE \RaiseP n) \cup \ITE' &\pInter [u/x](\term_1\;\term_2) \COL \Atype
         \eESRel [\seq{u'}/\dup{x,\Itype}](\term'_1\;\term'_{2,1}\dots\term'_{2,k})
        \end{align*}
        as required.
  \item Case for \rn{IT-Abs}:
        Then \(\term=\lambda y.\termalt\) and we have
        \begin{align*}
         \Atype = \bigwedge_{1 \le j \le k} (\Atype'_j, m'_j) \to \Atype'_0 \qquad
         &\term'=\lambda y_{\Atype'_1,m'_1}\dots y_{\Atype'_k,m'_k}. s' \\
         y \not\in dom((\ITE \cup \ZEnv{x \COL \Itype}) \RaiseP n) \qquad
         &((\ITE \cup \ZEnv{x \COL \Itype}) \RaiseP n) \cup \ITE' \cup \ZEnv{y \COL \bigwedge_{1 \le j \le k} (\Atype'_j, m'_j)} \pInter \termalt \COL \Atype'_0
         \eESRel \termalt'
        \end{align*}
        By the induction hypothesis, we have
        \begin{align*}
         (\ITE \RaiseP n) \cup \ITE' \cup \ZEnv{y \COL \bigwedge_{1 \le j \le k} (\Atype'_j, m'_j)} &\pInter [u/x]\termalt \COL \Atype'_0
         \eESRel [\seq{u'}/\dup{x,\Itype}] \termalt'.
        \end{align*}
        By using \rn{IT-Abs}, we obtain:
        \begin{align*}
         (\ITE \RaiseP n) \cup \ITE'
         &\pInter [u/x]\lambda y.\termalt \COL \bigwedge_{1 \le j \le k} (\Atype'_j, m'_j) \to \Atype'_0
         \eESRel [\seq{u'}/\dup{x,\Itype}]
         \lambda y_{\Atype'_1,m'_1}\dots y_{\Atype'_k,m'_k}. \termalt'
        \end{align*}
        as required.
 \end{itemize}
\end{myproof}

\subsubsection{Simulations in both directions}
Given \( \TopEnv \pInter P \eESRel (P', \Pfun) \), we show that the following two data are equivalent:
\begin{enumerate}
\item A reduction sequence of \( P \) together with a run of \( \PWA \) over the generated event sequence.
\item A reduction sequence of \( P' \).
\end{enumerate}
We first prove the direction \( (1)\Rightarrow(2) \).
\begin{lemma}\label{lem:appx:livesess:sim1}
  Assume
  \begin{itemize}
   \item \( \TopEnv \pInter D \eESRel D' \),
   \item \( \ITE \pInter \term \COL q \eESRel \termalt \),
   \item \( \ITE \EnvRel \TopEnv \), and
   \item \(\term \Preds{\seq{\ell}}{D} \term' \).
  \end{itemize}
  Let \( q\runseq \) be an arbitrary run of \( \PWA \) over \( \seq{\ell} \) and \( q' \) be the last state of \( q\runseq \).
  Then there exist 
  a term \( \termalt' \) and a type environment \( \ITE' \EnvRel \TopEnv \) such that
  \[
  (\termalt, q) \Preds{\seq{\ell}, \runseq}{D'} (\termalt', q')
  \]
  and
  \[
  \ITE' \pInter \term' \COL q' \eESRel \termalt'.
  \]
\end{lemma}
\begin{myproof}
  By induction on the length of the reduction sequence \(\term \Preds{\seq{\ell}}{D} \term' \).
  The claim trivially holds if the length is \( 0 \); we assume that the length is not \( 0 \).
  The proof proceeds by case analysis on the shape of \(\term\).
 \begin{itemize}
  \item Case \( \term = \evexp{\lab}{\term_1} \):
        Then the reduction sequence \(\term \Preds{\seq{\ell}}{D} \term' \) is of the form
        \[
        \evexp{\lab}{\term_1} \Pred{\lab}{D} \term_1 \Preds{\seq{\ell}'}{D} \term'
        \]
        with \( \seq{\ell} = \lab \seq{\ell'} \).
        Since the last rule used to derive \( \ITE \pInter \term \COL q \eESRel \termalt \) is \rn{IT-Event-Alt}, we have:
        \begin{align*}
         \delta(q, \lab) &= \{q_1, \dots, q_k\} \qquad\text{ (\( q_1 < q_2 < \dots < q_k \))} \\
         \forall i \in \{1, \dots, k\}. \ITE \RaiseP \Omega(q_i) &\pInter \term_1 \COL q_i \eESRel \termalt_i \\
         \termalt &= \evexp{\lab}{(\termalt_1 \enondet \dots \enondet \termalt_k)}.
        \end{align*}
        Since \( q\runseq \) is a run over \( \seq{\ell} = \lab \seq{\ell}' \), it must be of the form \( q\runseq = q q_i \runseq' \) for some \( 1 \le i \le k \), where \( q_i \runseq' \) is a run over \( \seq{\ell}' \).
        By applying the induction hypothesis to \( \term_1 \Preds{\seq{\ell}'}{D} \term' \), there exists \( \termalt' \) such that
        \[
        \ITE' \pInter \term' \COL q' \eESRel \termalt'
        \qquad
        (\termalt_i, q_i) \Preds{\seq{\ell}', \runseq'}{D'} (\termalt', q')
        \]
        for some \( \ITE' \EnvRel \TopEnv \).
        Then we have
        \[
        (\evexp{\lab}{(\termalt_1 \enondet \dots \enondet \termalt_k)}, q) \Preds{\lab,q_i}{D'} (\termalt_i,q_i) \Preds{\seq{\ell}',\runseq'}{D'} \termalt'.
        \]
  \item Case \(\term = \term_1 \nondet \term_2\):
        Suppose that the reduction sequence is \( \term \Pred{\epsilon}{D} \term_1 \Preds{\seq{\ell}}{D} \term' \);
        the case that \( \term \Pred{\epsilon}{D} \term_2 \Preds{\seq{\ell}}{D} \term' \) can be proved similarly.
        Since the last rule used to derive \( \ITE \pInter \term \COL q \eESRel \termalt \) is \rn{IT-NonDet}, we have:
        \begin{align*}
         \termalt = \termalt_{1} \nondet \termalt_{2} \qquad
         \ITE \pInter \term_{1} \COL q \eESRel \termalt_{1} \qquad
         \ITE \pInter \term_{2} \COL q \eESRel \termalt_{2}
        \end{align*}
        By applying the induction hypothesis to \( \term_1 \Preds{\seq{\ell}}{D} \term' \), we have \( \termalt' \) such that
        \[
        \ITE' \pInter \term' \COL q' \eESRel \termalt'
        \qquad
        (\termalt_1, q) \Preds{\seq{\ell},\runseq}{D'} (\termalt',q')
        \]
        for some \( \ITE' \EnvRel \TopEnv \).
        Then
        \[
        (\termalt_1 \nondet \termalt_2,q) \Preds{\epsilon,\epsilon}{D'} (\termalt_1, q) \Preds{\seq{\ell},\runseq}{D'} (\termalt', q').
        \]
  \item Case \(\term = \ifexp{p(\term_{1}, \dots, \term_{n})}{\term_{n+1}}{\term_{n+2}}\):
        Suppose \((\sem{\term_{1}}, \dots, \sem{\term_{n}}) \in \sem{p}\).
        Then the reduction sequence is \( \term \Pred{\epsilon}{D} \term_{n+1} \Preds{\seq{\ell}}{D} \term' \).
        Since the last rule used to derive \( \ITE \pInter \term \COL q \eESRel \termalt \) is \rn{IT-IF}, we have:
        \begin{align*}
         \termalt &= \ifexp{p(\termalt_{1}, \dots, \termalt_{n})}{\termalt_{n+1}}{\termalt_{n+2}} \qquad
         \forall i \in \{1,\dots, n\}.
         \ITE \pInter \term_{i} \COL \Tint \eESRel \termalt_{i} \\
         \ITE &\pInter \term_{n+1} \COL q \eESRel \termalt_{n+1} \qquad
         \ITE \pInter \term_{n+2} \COL q \eESRel \termalt_{n+2}
        \end{align*}
        By the well-typedness of the program,
        for every \(i \in \{1,\dots,n\}\), \(\term_{i}\) consists of only integers and integer operations.
        Thus, we have \(\sem{\term_i} = \sem{\termalt_i}\) for each \(i \in \{1,\dots,n\}\) and
        \(\termalt \Pred{\epsilon}{D'} \termalt_{n+1}\).
        
        By applying the induction hypothesis to \( \term_{n+1} \Preds{\seq{\ell}}{D} \term' \), we have
        \[
        \ITE' \pInter \term' \COL q' \eESRel \termalt'
        \qquad
        (\termalt_{n+1}, q) \Preds{\seq{\ell},\runseq}{D'} (\termalt',q')
        \]
        for some \( \ITE' \EnvRel \TopEnv \) and \( \termalt' \).
        Then
        \[
        (\termalt, q) \Preds{\epsilon,\epsilon}{D'} (\termalt_{n+1}, q) \Preds{\seq{\ell},\runseq}{D'} (\termalt',q').
        \]

        The case where
        \((\sem{\term_{1}}, \dots, \sem{\term_{n}}) \notin \sem{p}\) is similar.
  \item Case \(\term = f\;\term_1\dots \term_n\):
        Suppose \(D(f) = \lambda x_1\dots x_n. \term_0\).
        In this case, the reduction sequence is \( \term \Pred{\epsilon}{D} [\term_1/x_1]\dots[\term_n/x_n]\term_0 \Preds{\seq{\ell}}{D} \term' \).
        Since \( \ITE \pInter \term : q \eESRel \termalt \), we have \( f : (\Atype, m, m) \in \ITE \) such that
        \begin{align*}
         \Atype = \Itype_1 \to \dots \Itype_n \to q \\
         \forall i \in \{ 1, \dots, n\}.\; \ITE \pInter \term_i \COL \Itype_i \eESRel \seq{\termalt}_i \\
         \termalt = f_{\Atype,m}\,\seq{\termalt}_1\,\dots\,\seq{\termalt}_n.
        \end{align*}
        Since \(\TopEnv \pInter D \eESRel D' \), 
        we have
        \begin{align*}
          D'(f_{\Atype,m}) = \lambda \dup{x_1, \Itype_1}\dots\dup{x_n, \Itype_n}. \termalt_0
          \\
         \TopEnvZ \cup \ZEnv{x_1 \COL \Itype_1} \cup \dots \cup \ZEnv{x_n \COL \Itype_n}
         &\pInter \term_0 \COL q
         \eESRel \termalt_0
        \end{align*}
        By using Weakening (Lemma~\ref{lem:appx:liveness:weakening}),
        we have
        \begin{align*}
         \forall i \in \{ 1, \dots, n \}.\; \TopEnvZ \cup \ITE &\pInter \term_{i} \COL \Itype_{i} \eESRel \seq{\termalt}_{i} \\
         \TopEnvZ \cup \ITE \cup \ZEnv{x_1 \COL \Itype_1} \cup \dots \cup \ZEnv{x_n \COL \Itype_n}
         &\pInter \term_0 \COL q \eESRel \termalt_0
        \end{align*}
        Since
        \[
        (\TopEnvZ \cup \ITE \cup \ZEnv{x_1 \COL \Itype_1} \cup \dots \cup \ZEnv{x_n \COL \Itype_n}) \RaiseP 0 =
        \TopEnvZ \cup \ITE \cup \ZEnv{x_1 \COL \Itype_1} \cup \dots \cup \ZEnv{x_n \COL \Itype_n},
        \]
        by using Substitution Lemma~\ref{lem:appx:liveness:substitution} repeatedly,
        we have
        \begin{align*}
         \TopEnvZ \cup \ITE &\pInter [\term_{1}/x_1] \dots [\term_{n}/x_n] \term_0 \COL q
         \eESRel [\seq{\termalt_1}/\dup{x_1, \Itype_1}]\dots [\seq{\termalt_n}/\dup{x_n, \Itype_n}] \termalt_0.
        \end{align*}
        Since \( (\TopEnvZ \cup \ITE) \EnvRel \TopEnv \), we can apply the induction hypothesis to \( [\term_1/x_1]\dots[\term_n/x_n]\term_0 \Preds{\seq{\ell}}{D} \term' \); we have
        \[
        \ITE' \pInter \term' \COL q' \eESRel \termalt'
        \qquad
        (\termalt'', q) \Preds[N']{\seq{\ell},\runseq}{D'} (\termalt', q')
        \]
        for some \( \ITE' \EnvRel \TopEnv \) and \( \termalt' \), where
        \[
        \termalt'' = [\seq{\termalt_1}/\dup{x_1, \Itype_1}]\dots [\seq{\termalt_n}/\dup{x_n, \Itype_n}] \termalt_0.
        \]
        We have \( (\termalt, q) \Preds{\epsilon,\epsilon}{D'} (\termalt'',q) \Preds{\seq{\ell},\runseq}{D'} (\termalt',q') \) as desired.
 \end{itemize}
\end{myproof}

An infinite analogue of this lemma can be obtained as a corollary.
\begin{corollary}\label{cor:appx:liveness:infinite-forward}
  Assume
  \[
  \TopEnv \pInter D \eESRel D'
  \qquad
  \ITE \pInter \term_0 \COL q_0 \eESRel \termalt_0
  \qquad
  \ITE \EnvRel \TopEnv
  \]
  and an infinite reduction sequence
  \[
  \term_0 \Pred{\ell_1}{D} \term_1 \Pred{\ell_1}{D} \term_2 \Pred{\ell_2}{D} \cdots
  \]
  generating an infinite event sequence \( \seq{\ell} = \ell_1 \ell_2 \cdots \).
  Let \( q_0\runseq \) be an infinite run of \( \PWA \) over the infinite sequence \( \seq{\ell} \).
  Then there exist \( \{ (q_i, \runseq_i, \termalt_i) \}_{i \in \omega} \) such that
  \[
  (\termalt_0, q_0) \Preds{\ell_1,\runseq_1}{D'} (\termalt_1, q_1) \Preds{\ell_2,\runseq_2}{D'} (\termalt_2, q_2) \Preds{\ell_3, \runseq_3}{D'} \cdots
  \]
  and \( \runseq = \runseq_1 \runseq_2 \dots \).
\end{corollary}
\begin{myproof}
  By using Lemma~\ref{lem:appx:livesess:sim1}, one can construct by induction on \( i > 0 \) a family \( \{ (q_i, \runseq_i, \termalt_i, \ITE_i) \}_{i \in \omega} \) that satisfies
  \[
  (\termalt_{i-1},q_{i-1}) \Preds{\ell_i,\runseq_i}{D'} (\termalt_i, q_i)
  \qquad
  \ITE_i \EnvRel \TopEnv
  \qquad
  \ITE_i \pInter \term_i \COL q_i \eESRel \termalt_i
  \qquad
  \text{\( \runseq_1 \runseq_2 \dots \runseq_i \) is a prefix of \( \runseq \).}
  \]
  Since the length of \( \runseq_1 \runseq_2 \dots \runseq_i \) is equivalent to that of \( \ell_1 \ell_2 \dots \ell_i \), the sequence \( \runseq_1 \runseq_2 \dots \) is indeed an infinite sequence and thus equivalent to \( \runseq \).
\end{myproof}

We show the converse in a bit stronger form.
\begin{lemma}\label{lem:appx:liveness:summary}
  Assume
  \begin{align*}
    \TopEnv \pInter D \eESRel D' \qquad
    \ITE_1 \EnvRel \TopEnv \qquad
    \ITE_2 \EnvRel \TopEnv^\Mark \qquad
    \ITE_1 \cup \ITE_2 &\pInter \term \COL q \eESRel \termalt \qquad
    (\termalt, \RTyPath) \Preds[N]{\seq{\ell}}{D'^\Mark} (\termalt', \RTyPath')
  \end{align*}
  where \( \term \) does not contain \( \lambda \)-abstraction and \( \termalt' \) is not of the form \( \termalt'_1 \enondet \termalt'_2 \).
  Then there exist a run \( q\runseq \) over \( \seq{\ell} \) and a state \(q'\) such that
  \[
  (\termalt,\RTyPath,q) \Preds[N]{\seq{\ell},\runseq}{D'^\Mark} (\termalt', \RTyPath', q')
  \]
  Furthermore there exist a term \( \term' \) and a type environment \( \ITE_1' \) such that
  \[
  \ITE_1' \EnvRel \TopEnv \qquad
  \ITE_1' \cup (\ITE_2 \RaiseP{\MAX(\Omega(\runseq))}) \pInter \term' \COL q' \eESRel \termalt' \qquad
  \term \Preds{\seq{\ell}}{D^\Mark} \term'.
  \]
\end{lemma}
\begin{myproof}
  \tk{Now the proof does not rely on the simulation lemmas.}
  By induction on the length of the reduction sequence.
  The claim trivially holds if the length is \( 0 \).
  Assume that the length is greater than \( 0 \).
  
  The proof proceeds by case analysis on the shape of \(\termalt\).
  By the assumption
  \(\termalt \Preds{\seq{\ell}}{D'^\Mark} \termalt'\), it suffices to consider
  only the cases where the shape of \(\termalt\) matches the lefthand side of a transition rule.
 \begin{itemize}
  \item Case \(\termalt = \evexp{\lab}{\termalt_0}\):
        By the shape of \(\termalt\),
        the last rule used on the derivation of \(\ITE \cup \ITE' \pInter \term \COL q \eESRel \termalt\) is
        either \rn{IT-Event-Alt} or \rn{IT-App} (\rn{IT-Abs} is not applicable since \( \term \) is assumed to have no abstraction).
        By induction on the derivation, one can prove that \( \term = (\evexp{\lab}{\term_0})\,\term_1\,\dots\,\term_k \) for some \( k \).
        By the simple-type system, one cannot apply \( \evexp{\lab}{\term_0} \) to a term; hence \( k = 0 \) (i.e.~\( \term = \evexp{\lab}{\term_0} \)) and the last rule used on the derivation is \rn{IT-Event-Alt}.
        Then we have:
        \begin{align*}
         \delta(q, \lab) &= \{q_1, \dots, q_n\} \text{ (\( q_1 < q_2 < \dots < q_n \))} \\
         \forall i \in \{1, \dots, n\}. (\ITE_1 \cup \ITE_2) \RaiseP \Omega(q_i) &\pInter \term_0 \COL q_i \eESRel \termalt_i \\
         \termalt_0 = (\termalt_1 \enondet \dots \enondet \termalt_l)
        \end{align*}
        Recall that \( \termalt' \) is not of the form \( \termalt'_1 \enondet \termalt'_2 \); hence the reduction sequence
        \[
        (\termalt, \RTyPath) \Preds[N]{\seq{\ell}}{D'^\Mark} (\termalt', \RTyPath')
        \]
        must be of the form
        \[
        (\termalt, \RTyPath) \Preds[k]{\lab}{D'^{\Mark}} (\termalt_i, \RTyPath'') \Preds[N-k]{\seq{\ell}'}{D'^{\Mark}} (\termalt', \RTyPath')
        \]
        where \( k = i + 1 \) if \( i < n \) and \( k = i \) if \( i = n \).
        By the induction hypothesis, there exist a run \( q_i\runseq' \) over \( \ell' \), a term \( \term' \) and a type environment \( \ITE' \) with \( \ITE_1' \EnvRel \TopEnv \) such that
        \[
        (\termalt_i, \RTyPath'', q_i) \Preds[N-k]{\seq{\ell}',\runseq'}{D'^\Mark} (\termalt',\RTyPath', q')
        \]
        and
        \[
        \ITE_1' \cup (\ITE_2 \RaiseP \Omega_{\PWA}(q_i) \RaiseP \MAX(\Omega_{\PWA}(\runseq'))) \pInter \term' : q' \eESRel \termalt'.
        \]
        We have \( (\ITE_2 \RaiseP \Omega_{\PWA}(q_i) \RaiseP \MAX(\Omega_{\PWA}(\runseq'))) = (\ITE_2 \RaiseP \MAX(\Omega_{\PWA}(q_i \runseq'))) \) and \( (\termalt,\RTyPath,q) \Preds[k]{\lab,q_i}{D'^\Mark} (\termalt_i, \RTyPath'',q_i) \Preds[N-k]{\seq{\ell}',\runseq'}{D'^\Mark} (\termalt',\RTyPath',q') \).
        By construction, \( qq_i\runseq' \) is a run over \( \lab \seq{\ell}' \).
  \item Case \(\termalt = \termalt_1 \nondet \termalt_2\):
        As with the previous case,
        the last rule used on the derivation is \rn{IT-NonDet}.
        Then we have
        \begin{align*}
          (\ITE_1 \cup \ITE_2) \pInter \term_{1} \COL q \eESRel \termalt_{1} \qquad
          (\ITE_1 \cup \ITE_2) \pInter \term_{2} \COL q \eESRel \termalt_{2} \qquad
          \term = \term_{1} \nondet \term_{2}
        \end{align*}
        Suppose \( \RTyPath = \PLeft \RTyPath'' \); the other case can be proved by a similar way.
        Then the reduction sequence
        \(
        (\termalt, \RTyPath) \Preds[N]{\seq{\ell}}{D'^\Mark} (\termalt',\RTyPath')
        \)
        must be of the form
        \[
        (\termalt, \RTyPath) \Pred{\epsilon}{D'^{\Mark}} (\termalt_1,\RTyPath'') \Preds[N-1]{\seq{\ell}}{D'^{\Mark}} (\termalt', \RTyPath').
        \]
        Hence, by applying the induction hypothesis to \( \ITE_1 \cup \ITE_2 \pInter \term_{1} \COL q \eESRel \termalt_{1} \), we obtain the desired result.
  \item Case \(\termalt = \ifexp{p(\termalt_{1}, \dots, \termalt_{n})}{\termalt_{n+1}}{\termalt_{n+2}}\):
        As with the previous cases,
        the last rule used on the derivation is \rn{IT-If}.
        %
        Hence we have
        \begin{align*}
          &
          \term = \ifexp{p(\term_{1}, \dots, \term_{n})}{\term_{n+1}}{\term_{n+2}} \\&
          \forall i \in \{1,\dots, n\}.
          \ITE_1 \cup \ITE_2 \pInter \term_{i} \COL \Tint \eESRel \termalt_{i} \\&
          \ITE_1 \cup \ITE_2 \pInter \term_{n+1} \COL q \eESRel \termalt_{n+1} \qquad
          \ITE_1 \cup \ITE_2 \pInter \term_{n+2} \COL q \eESRel \termalt_{n+2}
        \end{align*}
        Recall that the result type of a function cannot be the integer type in our language.
        This implies that \( \term_i \) (\( 1 \le i \le n \)) consists only of constants and numerical operations, and thus \( \term_i = \termalt_i \) for every \( 1 \le i \le n \).

        Suppose that \((\sem{\termalt_{1}}, \dots, \sem{\termalt_{n}}) = (\sem{\term_1}, \dots, \sem{\term_n}) \in \sem{p}\).
        Then the reduction sequence
        \(
        (\termalt, \RTyPath) \Preds[N]{\seq{\ell}}{D'^\Mark} (\termalt', \RTyPath')
        \)
        must be of the form
        \[
        (\termalt,\RTyPath) \Pred{\epsilon}{D'^{\Mark}} (\termalt_{n+1},\RTyPath) \Preds[N-1]{\seq{\ell}}{D'^{\Mark}} (\termalt', \RTyPath').
        \]
        We obtain the desired result by applying the induction hypothesis to \( \ITE_1 \cup \ITE_2 \pInter \term_{n+1} \COL q \eESRel \termalt_{n+1} \); the corresponding reduction sequence is \( \term \Pred{\epsilon}{D^{\Mark}} \term_{n+1} \Preds{\seq{\ell}}{D^{\Mark}} \term' \).
        The case that \((\sem{\termalt_{1}}, \dots, \sem{\termalt_{n}}) \notin \sem{p}\) can be proved similarly.
  \item Case \(\termalt = g\;\termalt_1\dots \termalt_n\) with \(g\in\dom(D')\):
        \newcommand\bdy{\mathit{body}}
        By the shape of \(\termalt\),
        the last rule used on the derivation is \rn{IT-App} or \rn{IT-AppInt}.
        By induction on the derivation,
        we have \(\term = f \term_1 \dots \term_k\) for some \(k\).
        Then \( \termalt = f_{\Atype,m}\,\seq{\termalt}_1'\,\dots\,\seq{\termalt}_k' \) and
        \begin{align*}
          f : (\Atype, m, m) \in \ITE_1 \\
          \Atype = \Itype_1 \to \dots \Itype_{k} \to q \\
          \forall i \in \{ 1, \dots, k \}.\; \ITE_1 \cup \ITE_2 \pInter \term_i \COL \Itype_i \eESRel \seq{\termalt}_i'.
        \end{align*}
        Assume that \( D(f) = \lambda x_1\dots x_k. \term_\bdy \) and \( D'(f_{\Atype, m}) = \lambda y_1\dots y_n. \termalt_\bdy \).
        Since \( \TopEnv \pInter D \eESRel D' \) and \( f : (\Atype,m,m) \in \ITE_1 \EnvRel \TopEnv \), we have
        \begin{align*}
         \TopEnvZ \pInter \lambda x_1\dots x_k. \term_\bdy \COL \Atype
         \eESRel
         \lambda y_1\dots y_n. \termalt_\bdy.
        \end{align*}
        Hence \( y_1\dots y_n = \dup{x_1,\Itype_1} \dots \dup{x_k,\Itype_k} \) and
        \[
        \TopEnvZ \cup [x_1 : \Itype_1] \cup \dots \cup [x_k : \Itype_k] \pInter \term_\bdy \COL q
        \eESRel
        \termalt_\bdy.
        \]
        By using Weakening (Lemma~\ref{lem:appx:liveness:weakening}) and Substitution Lemma (Lemma~\ref{lem:appx:liveness:substitution}) repeatedly,
        we have
        \begin{align*}
         \TopEnvZ \cup \ITE_1 \cup \ITE_2 &\pInter [\term_{1}/x_1] \dots [\term_{k}/x_k] \term_\bdy \COL q
         \eESRel [\seq{u_1}/\dup{x_1, \Itype_1}]\dots [\seq{u_k}/\dup{x_k, \Itype_k}] \termalt_{\bdy}.
        \end{align*}
        The reduction sequence
        \(
        (\termalt,\RTyPath) \Preds[N]{\seq{\ell}}{D'^\Mark} (\termalt',\RTyPath')
        \)
        must be of the form
        \[
        (\termalt,\RTyPath) \Pred{\epsilon}{D'^{\Mark}} ([\seq{u_1}/\dup{x_1, \Itype_1}]\dots [\seq{u_k}/\dup{x_k, \Itype_k}] \termalt_\bdy, \RTyPath)
        \Preds[N-1]{\seq{\ell}}{D'^{\Mark}} (\termalt', \RTyPath').
        \]
        By applying the induction hypothesis to the above judgment, we complete the proof; the corresponding reduction sequence is \( \term \Pred{\epsilon}{D^\Mark}  [\term_{1}/x_1] \dots [\term_{k}/x_k] \term_\bdy \Preds{\seq{\ell}}{D^{\Mark}} \term' \) where \( \term' \) is the term obtained by the induction hypothesis.
      \item Case \(\termalt = g^\Mark\;\termalt_1\dots \termalt_n\) with \(g\in\dom(D')\):
            Similar to the above case.
 \end{itemize}
\end{myproof}

\begin{corollary}\label{cor:appx:liveness:infinite-backward}
  Assume
  \begin{align*}
    \TopEnv \pInter D \eESRel D' \qquad
    \ITE \EnvRel \TopEnv \qquad
    \ITE \pInter \term_0 \COL q_0 \eESRel \termalt_0
  \end{align*}
  and an infinite reduction sequence
  \[
  (\termalt_0, \RTyPath_0) \Preds[N_1]{\ell_1}{D'} (\termalt_1, \RTyPath_1) \Preds[N_2]{\ell_2}{D'} (\termalt_2, \RTyPath_2) \Preds[N_3]{\ell_3}{D'} \cdots.
  \]
  Suppose that, for every \( i \), \( \termalt_i \) is not of the form \( \termalt_{i1} \enondet \termalt_{i2} \).
  Then there exist \( \{ (q_i, \runseq_i, \term_i) \}_{i \in \omega} \) such that
  \[
  (\termalt_0, \RTyPath_0, q_0) \Preds[N_1]{\ell_1,\runseq_1}{D'} (\termalt_1, \RTyPath_1, q_1) \Preds[N_2]{\ell_2,\runseq_2}{D'} (\termalt_2, \RTyPath_2, q_2) \Preds[N_3]{\ell_3, \runseq_3}{D'} \cdots
  \]
  and
  \[
  \term_0 \Preds{\ell_1}{D} \term_1 \Preds{\ell_2}{D} \term_2 \Preds{\ell_3}{D} \cdots.
  \]
  Furthermore \( q_0\runseq_1\runseq_2 \dots \) is an infinite run of \( \PWA \) over \( \ell_1 \ell_2 \dots \).
\end{corollary}
\begin{myproof}
  By using Lemma~\ref{lem:appx:liveness:summary}, one can define a family \( \{ (q_i, \runseq_i, \term_i, \ITE_i) \}_{i \in \omega} \) that satisfies
  \begin{align*}
    (\termalt_{i-1},\RTyPath_{i-1},q_{i-1}) \Preds[N_i]{\ell_i,\runseq_i}{D'^\Mark} (\termalt_i, \RTyPath_i,q_i) \\
    \ITE_i \EnvRel \TopEnv \\
    \ITE_i \pInter \term_i \COL q_i \eESRel \termalt_i \\
    \term_{i-1} \Preds{\seq{\ell}}{D^\Mark} \term_i
  \end{align*}
  by induction on \( i > 0 \).
  Since \( \term_0 \) and \( \termalt_0 \) do not contain marked symbols as well as the bodies of function definitions in \( D \) and \( D' \), \( \term_i \) and \( \termalt_i \) dose not have marked symbols for every \( i \).
  Hence \( (\termalt_{i-1},\RTyPath_{i-1},q_{i-1}) \Preds[N_i]{\ell_i,\runseq_i}{D'} (\termalt_i, \RTyPath_i,q_i) \) and \( \term_{i-1} \Preds{\seq{\ell}}{D} \term_i \) for every \( i > 0 \).
  Furthermore \( q_0\runseq_1 \dots \runseq_i \) is a run over \( \ell_1 \dots \ell_i \) for every \( i \).
  Hence the infinite sequence \( q_0 \runseq_1 \runseq_2 \dots \) is an infinite run over \( \ell_1 \ell_2 \dots \).
  So \( \{(q_i,\runseq_i,\term_i)\}_{i \in \omega} \) satisfies the requirements.
\end{myproof}


\begin{lemma}\label{lem:appx:liveness:eff-sel-ev-seq-vs-choice-seq}
  Assume that \( \TopEnv \pInter P \eESRel (P', \Pfun) \).
  Let \( P = (\mainfun, D) \) and \( P' = (\mainfun', D') \).
  The following conditions are equivalent.
  \begin{enumerate}
  \item
    \( \exists \seq{\ell} \in \InfTraces(\prog). \exists \runseq \COL \text{run of \(\PWA\) over \(\seq{\ell}\)}.\;
    \MAX\mathbf{Inf}(\Omega_{\PWA}(\runseq)) \text{ is \emph{even}.} \)
  \item
    There exist \( \RTyPath \) and an infinite reduction sequence
    \[
    \RTyPath \Vdash (\mainfun', \qinit) = (\termalt_0, q_0) \Preds{\ell_1,\runseq_1}{D'} (\termalt_1, q_1) \Preds{\ell_2,\runseq_2}{D'} \cdots
    \]
    such that \( \runseq_1 \runseq_2 \cdots \) is an infinite sequence and 
    \( \MAX\mathbf{Inf}(\Omega_{\PWA}(\runseq_1 \runseq_2 \cdots)) \) is even.
  \end{enumerate}
\end{lemma}
\begin{myproof}
  (\( (1) \Rightarrow (2) \))
  Assume an infinite reduction sequence
  \[
  \RTyPath \Vdash \mainfun = \term_0 \Pred{\ell_1}{D} \term_1 \Pred{\ell_2}{D} \term_2 \Pred{\ell_3}{D} \cdots
  \]
  and an infinite run \( \qinit \runseq \) of \( \PWA \) over \( \ell_1 \ell_2 \dots \) (here we implicitly assume that \( \ell_1 \ell_2 \dots \) is an infinite sequence).
  Since \( \TopEnv \pInter P \eESRel (P', \Pfun) \), we have \( [\TopEnv] \pInter \mainfun : \qinit \eESRel \mainfun' \).
  Then by Corollary~\ref{cor:appx:liveness:infinite-forward},
  there exist \( \{ (q_i, \runseq_i, \termalt_i) \}_{i \in \omega} \) such that
  \[
  (\mainfun', \qinit) = (\termalt_0, q_0) \Preds{\ell_1,\runseq_1}{D'} (\termalt_1, q_1) \Preds{\ell_2,\runseq_2}{D'} (\termalt_2, q_2) \Preds{\ell_3, \runseq_3}{D'} \cdots
  \]
  and \( \runseq = \runseq_1 \runseq_2 \dots \).
  Then \( \MAX\mathbf{Inf}(\Omega_{\PWA}(\runseq_1 \runseq_2 \cdots)) = \MAX\mathbf{Inf}(\Omega_{\PWA}(\qinit \runseq_1 \runseq_2 \cdots))  = \MAX\mathbf{Inf}(\Omega_{\PWA}(\runseq)) \) is even by the assumption.

  (\( (2) \Rightarrow (1) \))
  Assume an infinite reduction sequence
  \[
  \RTyPath \Vdash (\mainfun', \qinit) = (\termalt_0, q_0) \Preds{\ell_1,\runseq_1}{D'} (\termalt_1, q_1) \Preds{\ell_2,\runseq_2}{D'} \cdots
  \]
  such that \( \runseq_1 \runseq_2 \cdots \) is an infinite sequence and \( \MAX\mathbf{Inf}(\Omega_{\PWA}(\runseq_1 \runseq_2 \cdots)) \) is even.
  By Corollary~\ref{cor:appx:liveness:infinite-backward}, there exists \( \{\term_i\}_i \) such that
  \[
  \mainfun \Preds{\ell_1}{D} \term_1 \Preds{\ell_2}{D} \cdots.
  \]
  By construction, \( \runseq = \qinit \runseq_1 \runseq_2 \cdots \) is an infinite run over \( \ell_1 \ell_2 \cdots \).
  Since \( \runseq \) is infinite by the assumption, \( \ell_1 \ell_2 \cdots \) is also infinite.
  Then \( \MAX\mathbf{Inf}(\Omega_{\PWA}(\runseq)) = \MAX\mathbf{Inf}(\Omega_{\PWA}(\qinit \runseq_1 \runseq_2 \cdots)) = \MAX\mathbf{Inf}(\Omega_{\PWA}(\runseq_1 \runseq_2 \cdots)) \) is even by the assumption.
\end{myproof}

\begin{lemma}\label{lem:appx:liveness:callseq-summary-one}
 Assume that
 \begin{align*}
  \TopEnv \pInter D \eESRel D' \qquad
  &\ITE_0 \EnvRel \TopEnv \qquad
  \ITE_0 \pInter \term : q \eESRel g\,\seq{s} \qquad
  & (g\,\seq{s}; \RTyPath) \CallSeqN{N}_{D'} (h_{\Atype,m}\,\seq{u}; \RTyPath')
 \end{align*}
 and \( (g\,\seq{s}, \RTyPath, q) \Preds[N]{\seq{\ell},\runseq}{D'} (h_{\Atype,m}\,\seq{u}, \RTyPath', q') \).
 Then \( m=\MAX(\Omega(\runseq)) \).
\end{lemma}
\begin{myproof}
  Since \(\ITE_0 \pInter t \COL q \eESRel g\,\seq{s} \),
  there exist \( \Atype_0 \) and \( m_0 \) such that
  \begin{align*}
    & \term = f\,t_1\,\dots\,t_n \\
    & \seq{\termalt} = \seq{s}_1\,\dots\,\seq{s}_n \\
    & f \COL (\Atype_0, m_0, m_0) \in \ITE_0 \\
    & \Atype_0 = \Itype_1 \to \dots \to \Itype_n \to q \\
    & \forall i \in \{ 1,\dots,n \}.\; \ITE_0 \pInter t_i \COL \Itype_i \eESRel \seq{s}_i.
  \end{align*}
  Suppose \( D(f) = \lambda x_1 \dots x_n. t_0 \).
  Since \( \TopEnv \pInter D \eESRel D' \), 
  we have
  \begin{align*}
    & D'(f_{\Atype_0, m_0}) = \lambda \seq{x}_1 \dots \seq{x}_n. s_0
    &\TopEnvZ \cup \ZEnv{x_1 \COL \Itype_1} \cup \dots \cup \ZEnv{x_n \COL \Itype_n}
    &\pInter t_0 \COL q \eESRel s_0
  \end{align*}
  where \(\seq{x}_i = \dup{x_i, \Itype_i}\) for each \(i \in \{1, \dots, n\}\).
  By easy induction on the structure of the derivation, we have
  \begin{align*}
    \TopEnvZ \cup \TopEnvZM \cup \ZEnv{x_1 \COL \Itype_1} \cup \dots \cup \ZEnv{x_n \COL \Itype_n}
    &\pInter \term_0^\Mark \COL q \eESRel \termalt_0^\Mark.
  \end{align*}

  By using Weakening (Lemma~\ref{lem:appx:liveness:weakening}) and Substitution Lemma (Lemma~\ref{lem:appx:liveness:substitution}),
  we have
  \begin{align*}
    \TopEnvZ \cup \TopEnvZM \cup \ITE_0
    &\pInter [\term_1/x_1]\dots[\term_n/x_n] \term_0^\Mark \COL q \eESRel [\seq{s}_1/\seq{x}_1]\dots[\seq{s}_n/\seq{x}_n]s_0^\Mark.
  \end{align*}
  
  Now, by the definition of call sequence, we have
  \begin{align*}
    ([\seq{\termalt}_1/\seq{x}_1]\dots[\seq{\termalt}_n/\seq{x}_n] s_0^\Mark, \RTyPath) \Preds[N-1]{\seq{\ell}}{D'^\Mark} (h_{\Atype,m}^\Mark\;\seq{u'}, \RTyPath'). \label{eq:appx:lem:summary4:2}
  \end{align*}
  By Lemma~\ref{lem:appx:liveness:summary}, there exist \( \runseq' \), \( \term' \) and \( \ITE' \EnvRel \TopEnv \) such that
  \[
  ([\seq{\termalt}_1/\seq{x}_1]\dots[\seq{\termalt}_n/\seq{x}_n] s_0^\Mark, \RTyPath,q) \Preds[N-1]{\seq{\ell}, \runseq'}{D'^\Mark} (h_{\Atype,m}^\Mark\;\seq{u'}, \RTyPath', q') \label{eq:appx:lem:summary4:2:1}
  \]
  and
  \[
  \ITE' \cup (\TopEnvZM \RaiseP \MAX(\Omega_{\PWA}(\runseq'))) \pInter \term' : q' \eESRel h_{\Atype,m}^\Mark\;\seq{u'}
  \]
  where \( q' \) is the last state in \( \runseq' \).
  
  By \( (g\,\seq{s}, \RTyPath, q) \Preds[N]{\seq{\ell},\runseq}{D'} (h_{\Atype,m}\,\seq{u}, \RTyPath', q') \), we have
  \[
  (g\,\seq{s}, \RTyPath, q) \Pred{\epsilon,\epsilon}{D'} ([\seq{\termalt}_1/\seq{x}_1]\dots[\seq{\termalt}_n/\seq{x}_n] s_0, \RTyPath, q) \Preds[N-1]{\seq{\ell},\runseq}{D'} (h_{\Atype,m}\,\seq{u}, \RTyPath', q').
  \]
  Since the mark does not affect the induced run of the automaton, we have
  \[
  ([\seq{\termalt}_1/\seq{x}_1]\dots[\seq{\termalt}_n/\seq{x}_n] s_0^\Mark, \RTyPath, q) \Preds[N-1]{\seq{\ell},\runseq}{D'} (\termalt'', \RTyPath', q')
  \]
  for some \( \termalt'' \) (which is equivalent to \( h_{\Atype,m}\,\seq{u} \) except for marks).
  Comparing this reduction sequence with
  \[
  ([\seq{\termalt}_1/\seq{x}_1]\dots[\seq{\termalt}_n/\seq{x}_n] s_0^\Mark, \RTyPath,q) \Preds[N-1]{\seq{\ell}, \runseq'}{D'^\Mark} (h_{\Atype,m}^\Mark\;\seq{u'}, \RTyPath', q') \label{eq:appx:lem:summary4:2:2}
  \]
  given above, we conclude that \( \termalt'' = h_{\Atype,m}^\Mark\;\seq{u'} \) and \( \runseq = \runseq' \).
  Thus
  \[
  \ITE' \cup (\TopEnvZM \RaiseP \MAX(\Omega(\runseq))) \pInter \term' : q' \eESRel h_{\Atype,m}^\Mark\;\seq{u'}.
  \]
  Hence \( h^\Mark : (\Atype, m, m) \in (\TopEnvZM \RaiseP \MAX(\Omega(\runseq))) \) and thus \( m = \MAX(\Omega(\runseq)) \) by definition of \( \TopEnvZM \).
\end{myproof}

\begin{lemma}\label{lem:liveness:summarising-priority}
  Assume that \( \TopEnv \pInter P \eESRel (P', \Pfun) \).
  Let \( P = (\mainfun, D) \) and \( P' = (\mainfun', D') \).
  For each choice sequence \( \RTyPath \in \{\PLeft,\PRight\}^{\omega} \), the following conditions are equivalent.
  \begin{enumerate}
  \item There exists an infinite reduction sequence
    \[
    \RTyPath \Vdash (\mainfun', \qinit) = (\termalt_0, q_0) \Preds{\ell_1,\runseq_1}{D'} (\termalt_1, q_1) \Preds{\ell_2,\runseq_2}{D'} \cdots
    \]
    such that \( \runseq_1 \runseq_2 \cdots \) is an infinite sequence and 
    \( \MAX\mathbf{Inf}(\Omega_{\PWA}(\runseq_1 \runseq_2 \cdots)) \) is even.
  \item There exists an infinite call-sequence
    \[
    \RTyPath \Vdash \mainfun' = g^0_{\Atype_0,m_0}\,\tilde{u}_0 \CallSeqN{k_0}_{D'} g^1_{\Atype_1,m_1}\,\tilde{u}_1 \CallSeqN{k_1}_{D'} g^2_{\Atype_2, m_2}\,\tilde{u}_2 \CallSeqN{k_2}_{D'} \cdots.
    \]
    such that \( \MAX \mathbf{Inf}(\Omega(\seq{g})) \) is odd.
  \end{enumerate}
\end{lemma}
\begin{myproof}
  (\( (1) \Rightarrow (2) \))
  Assume an infinite reduction sequence
  \[
  \RTyPath \Vdash (\mainfun', \qinit) = (\termalt_0, q_0) \Preds{\ell_1,\runseq_1}{D'} (\termalt_1, q_1) \Preds{\ell_2,\runseq_2}{D'} \cdots.
  \]
  By Corollary~\ref{cor:liveness:unique-existence-of-call-seq}, we have a (unique) infinite call-sequence
  \[
  \RTyPath \Vdash \mainfun' = g^0_{\Atype_0,m_0}\,\tilde{u}_0 \CallSeqN{k_0}_{D'} g^1_{\Atype_1,m_1}\,\tilde{u}_1 \CallSeqN{k_1}_{D'} g^2_{\Atype_2, m_2}\,\tilde{u}_2 \CallSeqN{k_2}_{D'} \cdots.
  \]
  So the given reduction sequence can be rewritten as \tk{Here is a (quite) small gap, indeed.}
  \[
  \RTyPath \Vdash (\mainfun, \qinit) = (g^0_{\Atype_0,m_0}\,\tilde{u}_0, q'_0) \Preds[k_0]{\seq{\ell}'_1,\runseq_1'}{D'} (g^1_{\Atype_1,m_1}\,\tilde{u}_1, q'_1) \Preds[k_1]{\seq{\ell}'_2,\runseq'_2}{D'} (g^2_{\Atype_2, m_2}\,\tilde{u}_2, q'_2) \Preds[k_2]{\seq{\ell}'_3,\runseq'_3}{D'} \cdots.
  \]
  Note that \( \runseq_1 \runseq_2 \dots = \runseq_1' \runseq_2' \dots \).
  By Lemma~\ref{lem:appx:liveness:callseq-summary-one}, \( m_i = \MAX(\Omega_\PWA(R'_i)) \).
  Since \( \MAX\mathbf{Inf}(\Omega_{\PWA}(\runseq_1 \runseq_2 \cdots)) \) is even,
  \begin{align*}
    \MAX\mathbf{Inf}(\Omega(\seq{g}))
    &=
    \MAX\mathbf{Inf}(m_1+1, m_2+1, \dots)
    \\
    &=
    \MAX\mathbf{Inf}(\MAX(\Omega_\PWA(R'_1))+1,\MAX(\Omega(R'_2))+1, \dots)
    \\
    &=
    \MAX\mathbf{Inf}(\MAX(\Omega_{\PWA}(R'_1)),\MAX(\Omega_\PWA(R'_2)), \dots) +1
    \\
    &=
    \MAX\mathbf{Inf}(\Omega_{\PWA}(R'_1R'_2\dots)) +1
    \\
    &=
    \MAX\mathbf{Inf}(\Omega_{\PWA}(R_1R_2\dots)) +1
  \end{align*}
  is odd.

  (\( (2) \Rightarrow (1) \))
  Assume an infinite call-sequence
  \[
  \RTyPath \Vdash \mainfun' = g^0_{\Atype_0,m_0}\,\tilde{u}_0 \CallSeqN{k_0}_{D'} g^1_{\Atype_1,m_1}\,\tilde{u}_1 \CallSeqN{k_1}_{D'} g^2_{\Atype_2, m_2}\,\tilde{u}_2 \CallSeqN{k_2}_{D'} \cdots.
  \]
  This is an infinite reduction sequence
  \[
  \RTyPath \Vdash \mainfun' = g^0_{\Atype_0,m_0}\,\tilde{u}_0 \Preds[k_0]{\seq{\ell}_1}{D'} g^1_{\Atype_1,m_1}\,\tilde{u}_1 \Preds[k_1]{\seq{\ell}_2}{D'} g^2_{\Atype_2, m_2}\,\tilde{u}_2 \Preds[k_2]{\seq{\ell}_3}{D'} \cdots.
  \]
  By the assumption on the program, \( \seq{\ell}_1 \seq{\ell}_2 \dots \) is an infinite sequence.
  Then by Corollary~\ref{cor:appx:liveness:infinite-backward},
  \[
  \RTyPath \Vdash (\mainfun', q_0) = (g^0_{\Atype_0,m_0}\,\tilde{u}_0, q_0) \Preds[k_0]{\seq{\ell}_1, \runseq_1}{D'} (g^1_{\Atype_1,m_1}\,\tilde{u}_1, q_1) \Preds[k_1]{\seq{\ell}_2,\runseq_2}{D'} (g^2_{\Atype_2, m_2}\,\tilde{u}_2, q_2) \Preds[k_2]{\seq{\ell}_3,\runseq_3}{D'} \cdots
  \]
  for some \( \{(q_i,\runseq_i)\}_{i} \).
  By Lemma~\ref{lem:appx:liveness:callseq-summary-one}, \( m_i = \MAX(\Omega_{\PWA}(R_i)) \).
  Hence
  \begin{align*}
    \MAX\mathbf{Inf}(\Omega_{\PWA}(R_1R_2\dots))
    &=
    \MAX\mathbf{Inf}(\MAX(\Omega_{\PWA}(R_1)),\MAX(\Omega_\PWA(R_2)), \dots)
    \\
    &=
    \MAX\mathbf{Inf}(m_1  m_2 \dots)
    \\
    &=
    \MAX\mathbf{Inf}(\Omega(\seq{g})) - 1
  \end{align*}
  Since \( \MAX \mathbf{Inf}(\Omega(\seq{g})) \) is odd, \( \MAX\mathbf{Inf}(\Omega_{\PWA}(R_1R_2\dots)) \) is even.
  The sequence \( \runseq_1 \runseq_2 \dots \) is infinite since \( \seq{\ell}_1 \seq{\ell}_2 \dots \) is.
\end{myproof}

\subsubsection{Proof of Theorem~\ref{thm:linevess:eff-sel-sound-and-complete}}
Assume that \( \TopEnv \pInter \prog \ESRel (\prog_0, \Pfun) \).
Then by Lemma~\ref{lem:appx:liveness:relating-modified-and-original-translations2}, there exists \( \prog_1 \) such that \( \TopEnv \pInter \prog \eESRel (\prog_1, \Pfun) \) and \( \prog_0 = [\nondet/\enondet]\prog_1 \).
Obviously \( \models_\CSA (\prog_1, \Pfun) \) if and only if \( \models_\CSA (\prog_0, \Pfun) \).
Let \( \prog_1 = (\progd_1, \mainfun_1) \).

We prove the lemma by establishing the equivalence of the following propositions:
\begin{enumerate}
\item
  \( \InfTraces(\prog) \cap \Lang(\PWA) \neq \emptyset \).
\item
  \( \exists \seq{\ell} \in \InfTraces(\prog). \exists \runseq \COL \text{run of \(\PWA\) over \(\seq{\ell}\)}.\;
  \MAX\mathbf{Inf}(\Omega_{\PWA}(\runseq)) \text{ is \emph{even}.} \)
\item
  There exist \( \RTyPath \) and an infinite reduction sequence
  \[
  \RTyPath \Vdash (\mainfun_1, \qinit) = (\termalt_0, q_0) \Preds{\ell_1,\runseq_1}{D_1} (\termalt_1, q_1) \Preds{\ell_2,\runseq_2}{D_1} \cdots
  \]
  such that \( \runseq_1 \runseq_2 \cdots \) is an infinite sequence and 
  \( \MAX\mathbf{Inf}(\Omega_{\PWA}(\runseq_1 \runseq_2 \cdots)) \) is even.
\item There exist \( \RTyPath \) and an infinite call-sequence
  \[
  \RTyPath \Vdash \mainfun_1 = g^0_{\Atype_0,m_0}\,\tilde{u}_0 \CallSeqN{k_0}_{D_1} g^1_{\Atype_1,m_1}\,\tilde{u}_1 \CallSeqN{k_1}_{D_1} g^2_{\Atype_2, m_2}\,\tilde{u}_2 \CallSeqN{k_2}_{D_1} \cdots.
  \]
  such that \( \MAX \mathbf{Inf}(\Omega(\seq{g})) \) is odd.
\item \( \neg (\models_\CSA (\prog_1, \Pfun)) \).
\item \( \neg (\models_\CSA (\prog_0, \Pfun)) \).
\end{enumerate}

\((1) \Leftrightarrow (2)\): By definition.

\((2) \Leftrightarrow (3)\): By Lemma~\ref{lem:appx:liveness:eff-sel-ev-seq-vs-choice-seq}.

\((3) \Leftrightarrow (4)\): By Lemma~\ref{lem:liveness:summarising-priority}.

\((4) \Leftrightarrow (5)\): By definition

\((5) \Leftrightarrow (6)\): Obvious.


\subsection{Proof of Theorem~\ref{thm:liveness:eff-sel-effective}}
\label{sec:proof-effectiveness}

Here we provide a proof of Theorem~\ref{thm:liveness:eff-sel-effective}.
The proof also implies that the reduction from the temporal verification to
\HFLZ{} model checking can be performed \emph{in polynomial time}.

For each simple type \(\Pest\), we define the set \(\Theta_{\Pest}\) of \emph{canonical intersection types} by:
\[
\begin{array}{l}
\Theta_{\Tunit} = \set{q \mid q\in Q_\PWA}\qquad
\Theta_{\INT} = \set{\INT}\\
\Theta_{\INT\to\Pst} = \set{\INT\to \Atype \mid \Atype\in \Theta_{\Pst}}\\
\Theta_{\Pst_1\to\Pst_2} = \set{\IT\set{(\Atype_1,m)\mid \Atype_1\in \Theta_{\Pst_1}, m\in\set{0,\ldots,M}}\to \Atype_2 \mid 
\Atype_2\in \Theta_{\Pst_2}}
\end{array}
\]
Here, \(M\) is the largest priority used in \(\PWA\).
Note that, for any simple type \(\Pst\;(\neq \INT)\),  the size of the set \(\Theta_{\Pest}\) is \(|Q|\).
For a simple type environment \(\STE\), we define the canonical type environment \(\Gamma_{\STE,m}\) as:
\[
\set{x\COL \INT \mid x\COL\INT\in \STE}
\cup 
\set{x\COL (\Atype,m',m)\mid x\COL\Pst\in \STE, \Atype\in \Theta_{\Pst}, m'\in\set{0,\ldots,M}}.
\]

\begin{lemma}
\label{lem:canonical-translation}
Suppose \(\STE \p t:\Pest\). Then, for any \(m\in\set{0,\ldots,M}\) and \(\Atype\in \Theta_{\Pest}\),
there exists \(t'\) such that \(\Gamma_{\STE,m}\pInter t:\Atype\Ra t'\). Furthermore, \(t'\) can be effectively
constructed.
\end{lemma}
\begin{proof}
The proof proceeds by induction on the derivation of \(\STE \p t:\Pest\), with case analysis on the last
rule used.
\begin{itemize}
\item Case \rn{LT-Unit}: In this case, \(t=\unitexp\) and \(\Pest=\Tunit\). By rule \rn{IT-Unit}, we have
\(\Gamma_{\STE,m}\pInter t:\Atype\Ra t'\) for \(t'=\unitexp\).
\item Case \rn{LT-Var}: In this case, \(t=x\) and \(\STE=\STE',x\COL\Atype\). If \(\Pest=\INT\), then 
\(\Atype=\INT\) and \(\Gamma_{\STE,m}=\Gamma_{\STE',m}, x\COL \INT\). By rule \rn{IT-VarInt}, we have \(\Gamma_{\STE,m}\pInter t:\Atype\Ra t'\)
for \(t'=x_\INT\), as required.
If \(\Pest\neq \INT\), then \(\Gamma_{\STE,m} = \Gamma_{\STE',m}\cup \set{x\COL(\Atype',m',m)\mid \Atype'\in\Theta_{\Pest},
m'\in\set{0,\ldots,M}}\ni x\COL (\Atype,m,m)\). Therefore, by rule \rn{IT-Var}, we have
\(\Gamma_{\STE,m}\pInter t\COL\Atype\ESRel t'\) for \(t'=x_{\Atype,m}\), as required.
\item Case \rn{LT-Int}: In this case, \(t=n\) and \(\Pest=\INT\).
By rule \rn{IT-Int}, we have
\(\Gamma_{\STE,m}\pInter t:\Atype\ESRel t'\) for \(t'=n\), as required.
\item Case \rn{LT-Op}: In this case, \(\term=\term_1\OP \term_2\) and \(\Atype=\INT\), with \(\STE\p \term_i:\INT\)
for each \(i\in\set{1,2}\). By the induction hypothesis,
there exists \(\term_i'\) such that \(\Gamma_{\STE,m}\pInter \term_i:\Atype\ESRel \term_i'\) for each \(i\).
By rule \rn{IT-Op}, we have 
\(\Gamma_{\STE,m}\pInter t:\Atype\ESRel t'\) for \(t'=\term'_1\OP\term'_2\), as required.
\item Case \rn{LT-Ev}: In this case, \(\term=\evexp{\lab}{\term_1}\) and \(\Pest=\Tunit\),
with \(\STE\p \term_1:\Tunit\) and \(\Atype=q\in Q_{\PWA}\). Let \(\delta_{\PWA}(q,\lab)=\set{q_1,\ldots,q_n}\)
and \(m_i = \max(m,\Pfun_\A(q_i))\). By the induction hypothesis,
we have \(\Gamma_{\STE,m_i}\pInter \term_1:q_i\ESRel \term_i'\) for each \(i\in\set{1,\ldots,n}\).
Since \(\Gamma_{\STE,m_i} = \Gamma_{\STE,m}\RaiseP \Pfun_{\PWA}(q_i)\), by applying rule \rn{IT-Event},
we obtain \(\Gamma_{\STE,m}\pInter t:\Atype\ESRel t'\) for
\(t'=\evexp{\lab}{(\term'_1 \nondet \cdots \nondet \term'_n)}\), as required.
\item Case \rn{LT-If}:
In this case, \(\term = \ifexp{p(\term_1, \dots, \term_k)}{\term_{k+1}}{\term_{k+2}}\) and
\(\Pest=\Tunit\), with
(i) \(\STE\p \term_i:\INT\) for each \(i\in\set{1,\ldots,k}\),
(ii) \(\STE\p \term_i:\Tunit\) for each \(i\in\set{k+1,k+2}\), and (iii) \(\Atype=q\in Q_{\PWA}\).
By the induction hypothesis, we have:
\(\Gamma_{\STE,m}\pInter \term_i:\INT\ESRel \term'_i\) for each \(i\in\set{1,\ldots,k}\), and 
\(\Gamma_{\STE,m}\pInter \term_i:\Atype\ESRel \term'_i\) for each \(i\in\set{k+1,k+2}\).
Thus, by rule \rn{IT-If}, we have 
\(\Gamma_{\STE,m}\pInter t:\Atype\ESRel t'\) for
\(t'=\ifexp{p(\term'_1, \dots, \term'_k)}{\term'_{k+1}}{\term'_{k+2}}\) as required.
\item Case \rn{LT-App}:
In this case, \(\term=\term_1\term_2\) with
\(\STE\p \term_1:\Pest_2\to\Pest\) and \(\STE\p \term_2:\Pest_2\).
Let \(\set{(\Atype_1,m_1),\ldots,(\Atype_k,m_k)} = \Theta_{\Pest_2}\times \set{0,\ldots,M}\)
with \((\Atype_1,m_1)<\cdots < (\Atype_k,m_k)\), and 
let \(\Atype'\) be \(\IT_{1\leq i\leq k} (\Atype_i,m_i)\to \Atype\).
By the induction hypothesis, we have
\(\Gamma_{\STE,m}\pInter \term_1:\Atype'\ESRel \term_1'\)
and \(\Gamma_{\STE,m}\RaiseP m_i\pInter\term_2:\Atype_i\ESRel\term'_{2,i}\) for each \(i\in\set{1,\ldots,k}\).
Thus, by rule \rn{IT-App}, we have
\(\Gamma_{\STE,m}\pInter t:\Atype\ESRel t'\) for
\(t'=\term_1' \term'_{2,1}\cdots \term'_{2,k}\) as required.
\item Case \rn{LT-NonDet}:
In this case, \(\term=\term_1\nondet\term_2\) and \(\Pest=\Tunit\), with
\(\STE\p \term_i:\Tunit\) for each \(i\in\set{1,2}\). By the assumption \(\Atype\in  \Theta_{\Pest}\), \(\Atype\in Q_{\PWA}\).
By the induction hypothesis, we have \(\Gamma_{\STE,m}\pInter \term_i:\Atype\ESRel \term'_i\) for each \(i\).
By rule \rn{IT-NonDet}, we have 
\(\Gamma_{\STE,m}\pInter t:\Atype\ESRel t'\) for
\(t'=\term_1' \nondet\term'_2\) as required.
\end{itemize}
\qed
\end{proof}

Theorem~\ref{thm:liveness:eff-sel-effective} is an immediate corollary of the lemma above.

\begin{proof}[Theorem~\ref{thm:liveness:eff-sel-effective}]
Let \(P=(\set{f_1 = \term_1,\ldots,f_n=\term_n},\term)\), and
 \(\STE\) be a simple type environment for \(P\), i.e.,
\(\STE \p P\). 
Let \(\TopEnv\)  be:
\[ \set{f_i\COL (\Atype, m)\mid i\in\set{1,\ldots,n}, \Atype\in \Theta_{\STE(f_i)}}.\]
For each \(f_i\COL(\Atype,m)\in \TopEnv\), let \(\term_{i,\Atype}'\) and \(\term'\) be
the terms given by Lemma~\ref{lem:canonical-translation}, such that
\(\Gamma_{\STE,0} \pInter \term_i:\Atype\ESRel \term_{i,\Atype}\)
and 
\(\Gamma_{\STE,0} \pInter \term:q_I\ESRel \term'\).
Let \(P'\) be:
\[(\set{f_{i,\Atype,m}=\term'_{i,\Atype} \mid f_i\COL(\Atype,m)\in \TopEnv}, \term'),\]
and \(\Pfun\) be 
\[\set{f_{i,\Atype,m}\mapsto m+1 \mid f_i\COL(\Atype,m)\in \TopEnv}.\]
Then \( \TopEnv \pInter P \Rightarrow (P',\Pfun) \) holds, and
\(\TopEnv, P',\Pfun\) can be effectively constructed as described above. \qed
\end{proof}

Any program \(P=(\set{f_1 = \term_1,\ldots,f_n=\term_n},\term)\) can be normalized (with an at most polynomial
blow-up of the size), so that
each of \(\term_i\) and \(\term\) is one of the following forms:
\begin{itemize}
\item
\(\unitexp\)
\item \( y\,(z_1\,x_{1,1}\,\cdots\,x_{1,\ell_1})\,\cdots\,(z_k\,x_{k,1}\,\cdots\,x_{k,\ell_k}) \)
\item \( \evexp{\lab}{(y\,x_1\,\cdots\,x_k)}\)
\item \(\ifexp{p(\term_1, \dots, \term_k)}{(y\,x_1\,\cdots\,x_k)}{(y'\,x'_1\,\cdots\,x'_{k'})}\)
\end{itemize}
\noindent
where \(k,k',\ell_i\) may be \(0\). The normalization can be performed by
adding auxiliary functions: See \cite{Kobayashi13JACM}, Section~4.3 for a similar normalization.
For a normalized program \(P\), 
 the size of \(P'\) obtained by the transformation in the proof above
is polynomial in the size of \(P\) and \(|Q_{\PWA}|\).
Furthermore, \(P'\) and \(\Pfun\) can be constructed in polynomial time.
Thus, the whole reduction from temporal property verification to call-sequence analysis (hence also to
\HFLZ{} model checking) can be
performed in polynomial time.

\section{An Example of the Translation of Section~\ref{sec:liveness}}
\label{sec:ex-tr-derivation}
Here we show derivation trees for the translation in Example~\ref{ex:tr}.

The body of \(g\) is translated as follows (where we omit irrelevant type bindings),
where \(\ITE_0=k\COL (q_a,0,0), k\COL (q_b,1,0) \).
\[
\small
\infers{\ITE_0 \pInter (\evexp{\Lab{a}}k)\nondet (\evexp{\Lab{b}}k):q_a
\Ra (\evexp{\Lab{a}}k_{q_a,0})\nondet (\evexp{\Lab{b}}k_{q_b,1})}
{
\infers{\ITE_0 \pInter (\evexp{\Lab{a}}k):q_a
\Ra (\evexp{\Lab{a}}k_{q_a,0})}
{k\COL (q_a,0,0), k\COL (q_b,1,0) \pInter k: q_a
\Ra k_{q_a,0}}
 &
\infers{\ITE_0 \pInter (\evexp{\Lab{b}}k):q_a
\Ra (\evexp{\Lab{a}}k_{q_b,1})}
{k\COL (q_a,0,1), k\COL (q_b,1,1) \pInter k:q_b
\Ra k_{q_b,1}}
}
\]

The body of \(f\) is translated as follows.
\[
\small
\infers{\ITE_1,x\COL\INT\pInter \ifexp{x>0}{g\,(f(x-1))}{(\evexp{\Lab{b}}f\,5)}: q_a\Ra \term_{f,q_a}}
{\ITE_1,x\COL\INT\pInter x:\INT\Ra x_\INT & 
\ITE_1,x\COL\INT\pInter 0:\INT\Ra 0 & \pi_1 & \pi_2 }
\]
Here, \(\ITE_1\) is:
\[
\begin{array}{l}
g\COL ((q_a,0)\land (q_b,1)\to q_a, 0, 0), \quad
g\COL ((q_a,0)\land (q_b,1) \to q_b, 0, 0),\\
f\COL (\INT\to q_a, 0, 0), \quad f\COL (\INT\to q_b, 1, 0)
\end{array}
\]
and \(\pi_1\) and \(\pi_2\) are:
\[
\small
\pi_1 = \raisebox{-2ex}{
\infers
{\ITE_1,x\COL\INT\pInter {g\,(f(x-1))}: q_a\Ra 
g_{(q_a,0)\land (q_b,1)\to q_a, 0}\,(f_{\INT\to q_a, 0}(x_\INT-1))\,
(f_{\INT\to q_b, 1}(x_\INT-1))}
{\ITE_1,x\COL\INT\pInter {g}: (q_a,0)\land (q_b,1)\to q_a \Ra g_{(q_a,0)\land (q_b,1)\to q_a, 0} 
 & \pi_3 & \pi_4}}
\]
\vspace*{1ex}
\[
\small
\pi_2 = \raisebox{-4ex}{
\infers
{\ITE_1,x\COL\INT\pInter \evexp{\Lab{b}}f\,5: q_a\Ra \evexp{\Lab{b}}f_{\INT\to q_b, 1} 5}
{\infers{\ITE_1\RaiseP 1,x\COL\INT\pInter f\,5: q_a\Ra f_{\INT\to q_b, 1} 5}
 {\ITE_1\RaiseP 1,x\COL\INT\pInter f: \INT\to q_b\Ra f_{\INT\to q_b, 1}
  & \ITE_1\RaiseP 1,x\COL\INT\pInter 5: \INT\Ra 5
}}}
\]
\vspace*{1ex}

\[
\footnotesize
\pi_3 =\raisebox{-2ex}{
\infers{
\ITE_1\RaiseP 0,x\COL\INT\pInter {f(x-1)}: q_a \Ra f_{\INT\to q_a, 0}(x_\INT-1)}
{\ITE_1\RaiseP 0, x\COL\INT\pInter f:\INT\to q_a\Ra f_{\INT\to q_a, 0}
 & \infers{\ITE_1\RaiseP 0, x\COL\INT\pInter x-1:\INT\Ra x_\INT-1}{\cdots}}}
\]
\vspace*{1ex}

\[
\footnotesize
\pi_4 =\raisebox{-2ex}{
\infers{
\ITE_1\RaiseP 1,x\COL\INT\pInter {f(x-1)}: q_b \Ra f_{\INT\to q_b, 1}(x_\INT-1)}
{\ITE_1\RaiseP 1, x\COL\INT\pInter f:\INT\to q_b\Ra f_{\INT\to q_b, 1}
 & \infers{\ITE_1\RaiseP 1, x\COL\INT\pInter x-1:\INT\Ra x_\INT-1}{\cdots}}}
\]

\section{Proving \HFLZ{} formulas in Coq}
\label{sec:coq}

Given program verification problems, the reductions presented in 
Sections~\ref{sec:reachability}--\ref{sec:liveness} yield \HFLZ{} model checking problems,
which may be thought as a kind of ``verification conditions'' (like those for Hoare triples).
Though we plan to develop automated/semi-automated tools for discharging the ``verification conditions'',
we demonstrate here that it is also possible to use an interactive theorem prover to do so.

Here we use Coq proof assistant, and consider \HFLN (HFL extended with natural numbers)
instead of \HFLZ{}.
Let us consider the termination of the following program:
\begin{verbatim}
let sum n k =
  if n<=0 then k 0 
  else sum (n-1) (fun r -> k(n+r))
in sum m (fun r->())
\end{verbatim}
Here, we assume \(m\) ranges over the set natural numbers.

The translation in Section~\ref{sec:mustreach} yields the following \HFLN{} formulas.
\[
(\mu \mathtt{sum}.\lambda n.\lambda k.(n\leq 0\imply k\, 0)\land (n>0\imply \mathtt{sum}\,(n-1)\,\lambda r.k(n+r)))\,
m\, (\lambda r.\TRUE).
\]
The goal is to prove that for every \(m\),
the formula is satisfied by the trivial model \(\lts_0= (\set{\stunique},\emptyset,\emptyset,\stunique)\).

In order to avoid the clumsy issue of representing variable bindings in Coq, we represent
the \emph{semantics} of the above formula in Coq.
The following definitions correspond to those of \(\D_{\lts,\typ}\) in Section~\ref{sec:pre}.

\begin{verbatim}
(* syntax of simple types: 
   "arint t" and "ar t1 t2" represent nat->t and t1->t2 respectively *)
Inductive ty: Set :=
    o: ty
  | arint: ty -> ty    
  | ar: ty -> ty -> ty.

(* definition of semantic domains, minus the monotonicity condition *)
Fixpoint dom (t:ty): Type :=
  match t with
     o => Prop
   | arint t' => nat -> dom t'
   | ar t1 t2 => (dom t1) -> (dom t2)
  end. 
\end{verbatim}
Here, we use \texttt{Prop} as the semantic domain \(\D_{\lts_0,\typProp} = \set{\emptyset, \set{\stunique}}\)
and represent \(\set{\stunique}\) as a proposition \texttt{True}.

Above, we have deliberately omitted the monotonicity condition, which is separately defined
by induction on simple types, as follows.
\begin{verbatim}
Fixpoint ord (t:ty) {struct t}: dom t -> dom t -> Prop :=
  match t with
	    o => fun x: dom o => fun y: dom o => (x -> y)
  | arint t' =>
     fun x: dom (arint t') => fun y: dom (arint t') =>	    
          forall z:nat, ord t' (x z) (y z)
  | ar t1 t2 => 
      fun x: dom (ar t1 t2) => fun y: dom (ar t1 t2) =>	    
       forall z w:dom t1, ord t1 z z -> ord t1 w w ->
       ord t1 z w -> ord t2 (x z) (y w)
  end.
 
Definition mono (t: ty) (f:dom t) :=
    ord t f f.
\end{verbatim}
Here, \texttt{ord}\ \(\typ\) corresponds to \(\sqleq_{\lts_0,\typ}\) in Section~\ref{sec:pre},\footnote{Note, however,
that since
\texttt{dom t} may be inhabited by non-monotonic functions, 
``\texttt{ord}\ \(\typ\)'' is not reflexive.}
and the monotonicity condition on \(f\) is expressed as the reflexivity condition \(\mathtt{ord}\; \typ\; f\; f\).

We can then state 
the claim that the sum program is terminating for every \(m\) as the following theorem:

\begin{verbatim}
Definition sumt := arint (ar (arint o) o).
Definition sumgen :=
  fun sum: dom sumt => 
   fun n:nat=> fun k:nat->Prop =>
    (n<=0 -> k 0) /\ (n>0 -> sum (n-1) (fun r:nat=>k(r+n))).
	   
Theorem sum_is_terminating:
  forall sum: dom sumt,
  forall FPsum: (* sum is a fixpoint of sumgen *)
     (forall n:nat, forall k:nat->Prop, sum n k <-> sumgen sum n k),
  forall LFPsum: (* sum is the least fixpoint of sumgen *)
     (forall x:dom sumt,
       mono sumt x ->
       ord sumt (sumgen x) x -> ord sumt sum x),		      			       
  forall m:nat, sum m (fun r:nat => True).
(* can be automatically generated up to this point *)
Proof.
(* this part should be filled by a user *)
...
Qed.
\end{verbatim}
Here, \texttt{sumgen} is the semantics of the argument of the \(\mu\)-operator
(i.e., \(\lambda \texttt{sum}.\lambda n.\lambda k. (n\leq 0\imply k\, 0)\land (n>0\imply \mathtt{sum}\,(n-1)\,\lambda r.k(n+r))\)),
and the first three ``\texttt{forall } ...'' assumes that \(\texttt{sum}\) is the least fixpoint of it,
and the last line says that \(\texttt{sum}\;m\;\lambda x.\texttt{True}\) is equivalent to \(\texttt{True}\) for every \(m\).

Note that except the proof (the part ``\texttt{...}''), all the above script
can be \emph{automatically} generated based on the development in the paper (like \texttt{Why3}\cite{why3}, but
\emph{without} any invariant annotations). 

The following is a proof of the above theorem.
\begin{verbatim}
Proof.
  intros.
  (* apply induction on m *)
  induction m; 
  apply FPsum; 
  unfold sumgen; simpl;auto.
  split; auto.
  omega.
  (* induction step *)
  assert (m-0=m); try omega.
  rewrite H; auto.
Qed.
\end{verbatim}

More examples
\ifsubmission
are found at \url{http://www-kb.is.s.u-tokyo.ac.jp/~koba/papers/hfl_in_coq.zip}
(provided also as supplementary materials).
\else
are found at \url{http://www-kb.is.s.u-tokyo.ac.jp/~koba/papers/hfl_in_coq.zip}.
\fi
The Coq proofs for some of those examples are much longer.
For each example, however, there are only a few places where human insights are required
(like ``\texttt{induction m}'' above). We, therefore,
 expect that the proofs can be significantly shortened by preparing appropriate libraries.

\section{HORS- vs HFL-based Approaches to Program Verification}
\label{sec:hors-vs-hfl}

In this section, we provide a more detailed comparison between our new HFL-based approach and
HORS-based approaches~\cite{Kobayashi09POPL,KSU11PLDI,Kobayashi13JACM,Ong11POPL,Kuwahara2014Termination,MTSUK16POPL,Watanabe16ICFP} 
to program verification. Some familiarity with HORS-based approaches may be required to fully understand the comparison.

HORS model checking algorithms~\cite{Ong06LICS,KO09LICS} usually consist of two phases,
one for computing a kind of higher-order ``procedure summaries'' in the form of variable profiles~\cite{Ong06LICS}
or intersection types~\cite{KO09LICS} (that summarize which states are visited between two function calls), 
and the other for solving games, which consists in nested least/greatest fixpoint computations.
In the case of finite-data programs, the combination of the two phases provides a sound and complete verification
algorithm~\cite{Ong06LICS,KO09LICS}. 
To deal with infinite-data programs, however,
the HORS-based 
approaches~\cite{Kobayashi09POPL,KSU11PLDI,Kobayashi13JACM,Ong11POPL,Kuwahara2014Termination,MTSUK16POPL,Watanabe16ICFP} 
had to apply various transformations to map verification problems to HORS model checking problems in a sound but
incomplete manner (incompleteness is inevitable because, in the presence of values from infinite data domains, 
the former is undecidable whereas the latter is decidable),
as illustrated on the lefthand side of Figure~\ref{fig:hors-vs-hfl}.
 A problem about this approach is that
the second phase of HORS model checking -- nested least/greatest fixpoint computations -- actually does not help
much to solve the original problem, because least fixpoint computations are required for proving liveness (such as termination), 
but the liveness of infinite-data programs usually depends on properties about infinite data domains such as ``there is no
infinite decreasing sequence of natural numbers,'' which are not available after the transformations to HORS.
For this reason, the previous techniques for proving termination~\cite{Kuwahara2014Termination} and 
fair termination~\cite{MTSUK16POPL} used HORS model checking only as a backend of a safety property checker, where
only greatest fixpoint computations are performed in the second phase; reasoning about liveness was performed during
the transformation to HORS model checking.

\begin{figure}
\includegraphics[scale=0.5]{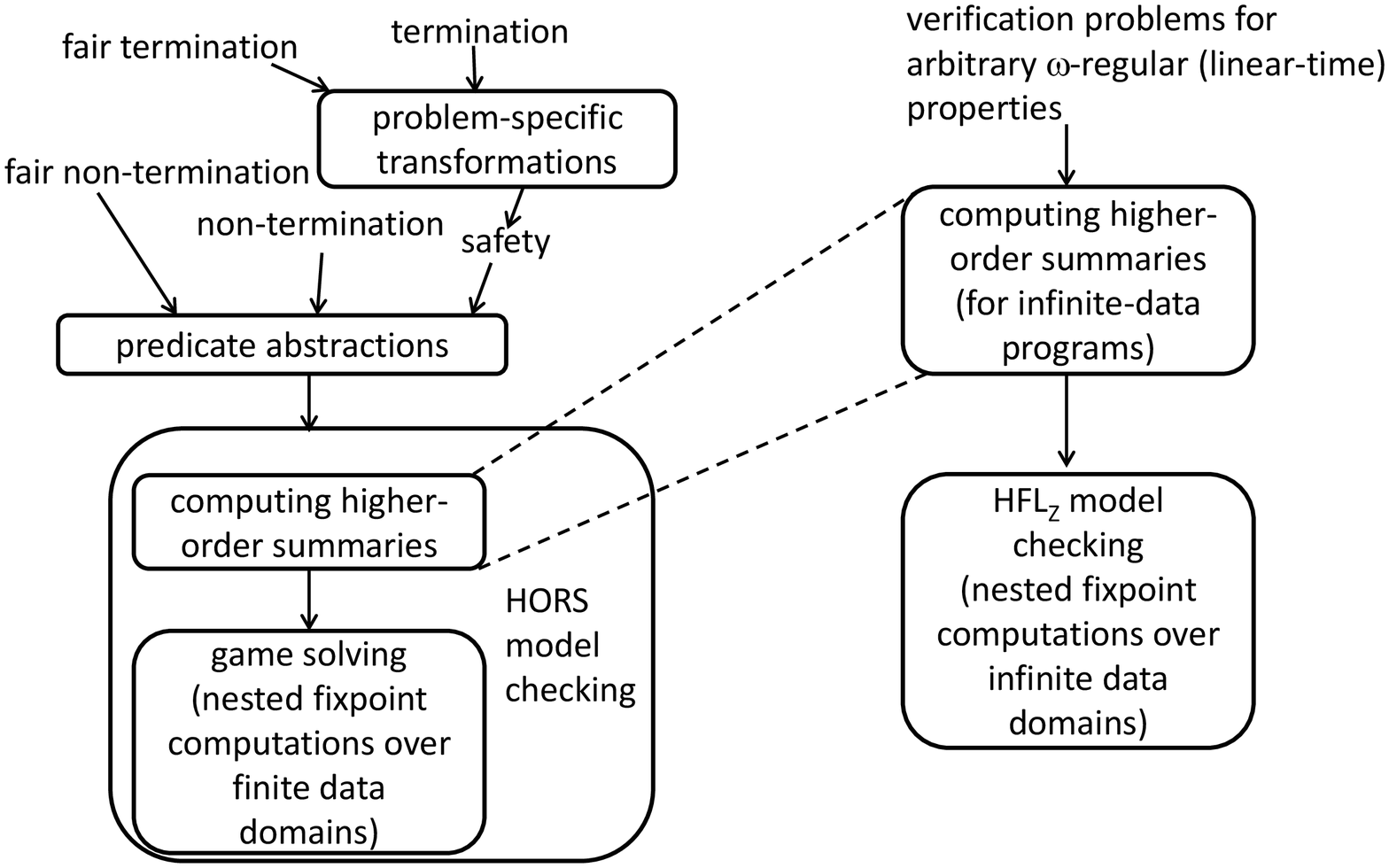}
\caption{Comparison between HORS- and HFL-based approaches. In the HORS-based approaches (shown on the lefthand side), 
there are actually some
feedback loops (due to counterexample-guided abstraction refinement, etc.), which are omitted.
The right hand side shows the approach based on the reduction in Section~\ref{sec:liveness}. In the reductions
in Section~\ref{sec:reachability} and \ref{sec:path}, the first phase is optimized for degenerate cases.}
\label{fig:hors-vs-hfl}
\end{figure}

As shown on the righthand side of
Figure~\ref{fig:hors-vs-hfl}, in our new HFL-based approach,
 we have extended the first phase of HORS model checking -- the computation of higher-order 
procedure summaries -- to deal with infinite data programs, and moved it up front; we then
formalized the remaining problems as HFL model checking 
problems. That can be viewed as the essence of the reduction in Section~\ref{sec:liveness}, and the reductions
in Sections~\ref{sec:reachability} and \ref{sec:path} are those for degenerate cases.
Advantages of the new approach include: (i) necessary information on infinite data is
available in the second phase of least/greatest computations (cf. the discussion above on the HORS-based approach),
 and (ii) various verification problems boil down to the issue of how to prove
least/greatest fixpoint formulas; thus we can reuse and share the techniques developed for different verification problems.
The price to pay is that the first phase (especially a proof of its correctness)
is technically more involved, because now we need to deal with infinite-data
programs instead of HORS's (which are essentially finite-data functional programs). 
That explains long proofs in Appendices (especially Appendix~\ref{sec:game-characterization} and \ref{sec:proof-liveness}).

\end{document}